\documentclass[10pt,prd,aps,amsfonts,eqsecnum,superscriptaddress,nofootinbib,notitlepage,longbibliography,final]{revtex4-2}

\usepackage[a4paper, left=2cm, right=2cm, top=2cm, bottom=2cm]{geometry}
\usepackage{setspace}
\usepackage[english]{babel}
\addto\extrasenglish{%
	\def\appendixautorefname{Appendix}%
}
\usepackage[utf8]{inputenc}
\usepackage[T1]{fontenc}
\usepackage{lmodern}
\usepackage{amsmath,amsthm,mathdots,amssymb,bm,bbm,mathrsfs,mathtools}
\usepackage[final]{graphicx} 
\usepackage{tikz,pgfplots}\pgfplotsset{compat=1.16}
\usetikzlibrary{positioning,fit,calc}

\usepackage{thm-restate}

\newtheorem{theorem}{Theorem}
\newtheorem{corollary}{Corollary}

\newtheorem{assumption}{Assumption}

\newtheorem{lemma}{Lemma}

\newtheorem{proposition}{Proposition}

\newtheorem{remark}{Remark}
\newtheorem{definition}{Definition}

\newcommand{\beq}{\begin{equation}}
	\newcommand{\eeq}{\end{equation}}

\newcommand{\Stat}{\mathrm{Stat}}
\newcommand{\Eval}{\mathrm{Eval}}

\newcommand{\Exp}{\mathrm{Exp}}

\theoremstyle{definition}

\usepackage{comment}

\usepackage{mleftright}\mleftright 

\usepackage{enumitem}

\usepackage[colorlinks=true, urlcolor=violet, linkcolor=blue, citecolor=red, hyperindex=true, linktocpage=true, pagebackref=true, draft=false]{hyperref}
\usepackage{microtype,color} 
\usepackage{xcolor}

\newcommand{\vertiii}[1]{{\left\vert\kern-0.25ex\left\vert\kern-0.25ex\left\vert #1 \right\vert\kern-0.25ex\right\vert\kern-0.25ex\right\vert}}

\newcommand{\Lword}[1]{\text{Lindbladian}}

\DeclarePairedDelimiterX{\braket}[1]{\langle}{\rangle}{#1}
\DeclarePairedDelimiterX\ketbra[2]{| }{|}{#1 \delimsize\rangle\!\delimsize\langle #2}	
\DeclarePairedDelimiterX\dotp[2]{\langle}{\rangle}{#1, #2}

\DeclareMathAlphabet{\dutchcal}{U}{dutchcal}{m}{n}
\SetMathAlphabet{\dutchcal}{bold}{U}{dutchcal}{b}{n}
\DeclareMathAlphabet{\dutchbcal} {U}{dutchcal}{b}{n}

\makeatletter
\DeclareRobustCommand*{\pmzerodot}{%
	\nfss@text{%
		\sbox0{$\vcenter{}$}
		\sbox2{0}%
		\sbox4{0\/}%
		\ooalign{%
			0\cr
			\hidewidth
			\kern\dimexpr\wd4-\wd2\relax 
			\raise\dimexpr(\ht2-\dp2)/2-\ht0\relax\hbox{%
				\if b\expandafter\@car\f@series\@nil\relax
				\mathversion{bold}%
				\fi
				$\cdot\m@th$%
			}%
			\hidewidth
			\cr
			\vphantom{0}
		}%
	}%
}

\makeatletter
\def\l@subsubsection#1#2{}
\makeatother

\usetikzlibrary{quantikz2}

\begin{document}
\renewcommand{\appendixautorefname}{Appendix}
\renewcommand{\chapterautorefname}{Chapter}
\renewcommand{\sectionautorefname}{Section}
\renewcommand{\subsubsectionautorefname}{Section}

\title{Physically natural metric-measure Lindbladian ensembles and their learning hardness}
\author{Caisheng Cheng }
\email{ccs112202@mail.ustc.edu.cn}
\affiliation{Hefei National Laboratory for Physical Sciences at Microscale and Department of Modern
	Physics, University of Science and Technology of China, Hefei, Anhui, China}
\affiliation{Shanghai Branch, CAS Centre for Excellence and Synergetic Innovation Centre in Quantum Information and Quantum Physics, University of Science and Technology of China, Shanghai 201315, China}
\affiliation{Shanghai Research Center for Quantum Sciences, Shanghai 201315, China}

\author{Ruicheng Bao}
\thanks{C.~Cheng and R.~Bao contributed equally to this work.}
\affiliation{Department of Physics, Graduate School of Science,
	The University of Tokyo, Hongo, Bunkyo-ku, Tokyo 113-0033, Japan}

\begin{abstract}
	In open quantum systems, a basic question at the interface of quantum information, statistical physics, and many-body dynamics is how well can one infer the structure of noise and dissipation generators from finite-time measurement statistics alone. Motivated by this question, we study the learnability and cryptographic applications of random open-system dynamics generated by Lindblad-Gorini-Kossakowski-Sudarshan (GKSL) master equations.
Working in the affine hull of the GKSL cone, we introduce physically motivated ensembles of random local Lindbladians via a linear parametrisation around a reference generator.
On top of this geometric structure, we extend statistical query (SQ) and quantum-process statistical query (QPStat) frameworks to the open-system setting and prove exponential (in the parameter dimension $M$) lower bounds on the number of queries required to learn random Lindbladian dynamics.
In particular, we establish average-case SQ-hardness for learning output distributions in total variation distance and average-case QPStat-hardness for learning Lindbladian channels in diamond norm.
To support these results physically, we derive a linear-response expression for the ensemble-averaged total variation distance and verify the required nonvanishing scaling in a random local amplitude-damping chain.
Finally, we design two Lindbladian physically unclonable function (Lindbladian-PUF) protocols based on random Lindbladian ensembles with distribution-level and tomography-based verification, thereby providing open-system examples where learning hardness can be translated into cryptographic security guarantees.
\end{abstract}

\maketitle
\vspace{-1cm}

\begingroup
\small  
\tableofcontents
\endgroup

\section{Introduction}

Understanding the universal emergence of chaotic and ergodic behavior in many-body quantum systems is one of the central goals of modern quantum many-body physics.
Early work on random matrices and single-particle quantum chaos, together with the formulation of the eigenstate thermalization hypothesis, has shown that introducing small uncertainties into the parameters of a chaotic Hamiltonian already suffices to produce universal fluctuation behavior in spectral statistics and physical observables~\cite{wigner1967random,bohigas1984characterization,muller2004semiclassical,deutsch1991quantum,rigol2008thermalization}.
Over the past decade, local random quantum circuits, unitary $t$-designs, and pseudorandom unitaries (PRUs) have developed into a mature technical framework~\cite{emerson2003pseudo,gross2007evenly,dankert2005efficient,dankert2009exact,brandao2016local,nakata2016efficient,haferkamp2022random,ji2018pseudorandom,brakerski2019pseudo,metger2024simple,chen2024efficient,ananth2025pseudorandom,ma2025construct,laracuente2024approximate,schuster2024random,lami2025anticoncentration,grevink2025will,cui2025unitary,schuster2025strong,schuster2025hardness}:
on the one hand, they quantify how quickly physically implementable circuits become indistinguishable from Haar-random unitaries; on the other hand, such pseudorandom dynamics have found widespread applications in randomized benchmarking and noise characterization, random-circuit sampling and demonstrations of quantum advantage, as well as in foundational problems in quantum gravity and black-hole information~\cite{emerson2005scalable,knill2008randomized,elben2023randomized,guta2020fast,huang2020predicting,zhao2021fermionic,aaronson2016complexity,boixo2018characterizing,bouland2019complexity,arute2019quantum,morvan2023phase,abanin2025constructive,hayden2007black,kim2020ghost,akers2022black,akers2024holographic,yang2025complexity}.
From an information-theoretic and complexity-theoretic perspective, these works collectively suggest that ``typical dynamics'' in closed systems are often both highly chaotic in physics and highly complex computationally.

In parallel, quantum learning theory has systematically investigated, from the viewpoint of access models, the advantages and limitations of quantum machine learning for both classical and quantum tasks~\cite{valiant1984theory,SQLearning,arunachalam2017guest}.
In particular, in the setting of random circuits for closed systems, learning complexity is closely tied to pseudorandomness.
In this framework, the learner may be a classical, quantum, or hybrid algorithm that constructs a hypothesis about an unknown object (a Boolean function, a probability distribution, a quantum state, or a quantum process) through limited access to it~\cite{HaahHarrowJiWuYu2017SampleOptimal,Wang_2017,Aaronson2018ShadowTomography,ArunachalamDeWolf2018QuantumSampleComplexity,Aaronson_2019,classical_shadow_tomography,cotler2020quantum,chen2021exponentialseparationslearningquantum,Huang_2021,Aharonov_Cotler_Qi_2022,nisq,Huang_2023, chen2021exponentialseparationslearningquantum,Huang_2022, Aharonov_Cotler_Qi_2022}.
Classical models of random examples and statistical queries (SQ) have been extended to the quantum regime, leading to quantum sample models and quantum statistical queries (QSQs)~\cite{QPAC,arunachalam2018optimal,arunachalam2020quantum,arunachalam2023role,hinsche2023one,nietner2023average,nietner2023unifying}.
Within these access models, a range of positive and negative results have been obtained for learning quantum states, shadow tomography, learning probability distributions, and learning quantum processes~\cite{bisio2009optimal,aaronson2018shadow,childs2022quantum,montanaro2017learning,mohseni2008quantum,chung2018sample,haah2023query}.
Together with constructions of unitary designs and PRUs, these results paint the following picture for closed systems:
on the one hand, there exist efficient algorithms that, under strong access models, can predict local properties of highly complex quantum processes; on the other hand, for ensembles of random circuits, typical output distributions or channels exhibit pronounced average-case learning hardness under restricted access models.

In sharp contrast to this closed-system picture, realistic quantum devices are inevitably coupled to their environment, and their dynamics are more naturally described by the Lindblad-Gorini-Kossakowski-Sudarshan (GKSL) master equation. On the one hand, dissipation and decoherence limit the applicability of closed-system random-circuit models; on the other hand, open systems have motivated a wide range of schemes that exploit dissipation to prepare steady states and engineer quantum dynamics~\cite{verstraete2009quantum,diehl2008quantum,kraus2008preparation,barreiro2011open}.
Recent work on ensembles of random Lindbladians and their typical relaxation properties~\cite{denisov2019universal,tarnowski2023random,bao2025initial}, as well as proposals for learning Lindbladians from stationary states, time series, or trajectory data~\cite{bairey2020learning,wang2024simulation,wallace2025learning}, has revealed a number of intriguing phenomena, including universal spectral statistics, spectral-gap structures, and insensitivity to the initial state.
However, from the viewpoint of learning complexity and cryptography, there is still a striking gap on the open-system side:
so far there is no family inside the GKSL cone that simultaneously enjoys physical naturalness and a well-behaved metric-measure structure in the sense of a ``random Lindbladian design'', and there is also a lack of quantitative characterizations of the learnability of typical open-system dynamics in the SQ/QPStat access models, let alone constructions of open-system analogues of pseudorandom processes or cryptographic protocol primitives based on such ensembles.

The goal of this work is to develop a unified framework that fills this gap, in a way that runs in parallel with unitary designs/PRUs and SQ/QPStat learning theory, while at the same time highlighting structures and physical meanings that are specific to open systems.
Our main contributions can be summarized as follows:

\begin{itemize}
	\item[1.] \textbf{Random Lindbladian ensembles and metric-measure structure.}
	Working inside the affine hull of the GKSL cone, we start from local Lindbladians and construct two physically natural ensembles of random Lindbladians via a linear parametrization
	\(
	\mathcal{L}(\theta)
	=
	\mathcal{L}_{\mathrm{ref}}
	+\sum_{j=1}^M\theta_j G_j.
	\)
	One ensemble is a ``Lindblad Haar-type'' spherical ensemble, where the parameter lies on the unit sphere $S^{M-1}$ with a $1/\sqrt{M}$ normalization; the other is a product-measure ensemble where each local coupling strength fluctuates independently within a small interval.
	We treat both in a unified way as metric probability spaces $(\Theta,d,\mu)$ and show that, for any bounded SQ(analogy with QPStat) test function $\varphi:X\to[-1,1]$, the induced function
	$F_\varphi(\theta):=P_{\mathcal{L}(\theta)}[\varphi]$
	has a dimension-independent Lipschitz constant on parameter space, and satisfies, respectively, a Lévy-type concentration inequality on the sphere and a McDiarmid-type bounded-differences concentration inequality under the product measure.
	
	\item[2.] \textbf{SQ and QPStat hardness results in open systems.}
	On the distribution level, we fix a single-shot measurement interface (evolution time $t$, input state $\rho_{\mathrm{in}}$, and a POVM $\{M_x\}$ on a finite output set $X$), and only access the channel output distribution $P_{\mathcal{L}}$ via SQ queries.
	We show that, for the Lindbladian ensembles above, both the binary hypothesis-testing problem of deciding whether $\mathcal{L}$ is drawn from a sub-ensemble or from a reference distribution, and the task of learning the output distribution $P_{\mathcal{L}}$ in the average case, require any SQ algorithm (including SQ algorithms implemented by quantum procedures) to use at least exponentially many queries in the parameter dimension $M$, even if success is only required on a subset of parameters.
	On the channel level, we introduce a QPStat access model suitable for Lindbladians, allowing the learner to make process-level statistical queries about expectation values of the channels $\Lambda_t^{(\mathcal{L}(\theta))}$.
	We prove that, for the same Lindbladian ensembles, any QPStat algorithm that attempts to learn typical random Lindbladian channels up to constant accuracy in diamond distance or in average channel distance must use exponentially many process queries, thereby establishing an average-case QPStat hardness result for channel learning in the open-system setting.
	Our results thus show that ``open-system versions'' of random processes exhibit intrinsic average-case learning hardness under weak access models.
	
	\item[3.] \textbf{Open-system Porter-Thomas-type average hypothesis and linear-response analysis.}
	To support the above average-case lower bounds, we formulate and analyze an open-system analogue of a Porter-Thomas-type ``average TV-distance hypothesis'': after suitable normalization and scaling, the ensemble-averaged total variation distance
	$\mathbb{E}_{\mathcal{L}} d_{\mathrm{TV}}(P_{\mathcal{L}},Q)$
	converges, in the high-dimensional limit, to a strictly positive constant $m_0>0$.
	Using a Kubo-type perturbative formula, we perform a linear-response expansion for a random local Lindbladian model, explicitly compute the norm of the first-order response vector and its scaling behavior in terms of the parameter dimension $M$, evolution time $t$, and the structure of local jump operators, and analytically and numerically confirm the constant scaling of $m_0$ in a concrete model of a local random amplitude-damping chain.
	This unifies our analysis of Lindbladian ensembles into a physically meaningful picture that echoes Porter-Thomas-based analyses in closed systems.
	
	\item[4.] \textbf{PUF protocols based on random Lindbladians.}
	Finally, we connect the above learning-hardness results with cryptographic primitives by constructing two types of physically unclonable function (PUF) protocols based on random Lindbladians.
	The first is a ``distribution-level verification'' scheme that is entirely based on output distributions; its security follows directly from the SQ hardness of learning the output distributions, and the verification interface is deliberately constrained so that the scheme remains secure under a given resource bound for an adversary.
	The second is a ``tomography-based verification'' scheme that allows a finite number of QPStat queries or experimental tomography; by defining a tomographic fingerprint in an appropriate operator basis, it translates unclonability of channels in diamond distance into the QPStat hardness of learning Lindbladian channels.
	Within our framework we also discuss how, in the open-system setting, one can explicitly balance verification cost against security, thereby providing an example of how ``SQ/QPStat hardness implies cryptographic security'' in the open-system regime.
\end{itemize}

The remainder of the paper is organized as follows.
We begin by summarizing and discussing our main theorems and results.
In Sec.~\ref{zhang12}, we review GKSL generators, the SQ and QPStat models, and the relevant tools from measure concentration.
Secs.~\ref{zhang1}-\ref{zhang3} establish, respectively, the geometric structure of Lindbladian ensembles, their Lipschitz-Levy-type concentration properties, and lower bounds on the SQ/QPStat learning complexity of output distributions and channels.
Secs.~\ref{zhang4}-\ref{zhang5} present the linear-response analysis and a concrete model of a random amplitude-damping chain that support the Porter-Thomas-type average hypothesis.
Sec.~\ref{sec:crqpuf_applications} then discusses in detail the Lindbladians-PUF protocols based on random Lindbladians and their security analysis.

\section{Summary of results}
\subsection{From random unitary circuits to random Lindbladians: an ensemble-construction blueprint}

In the literature on random quantum circuits and typicality, the construction of random unitaries usually follows a common blueprint~\cite{mele2024introduction,brandao2016local,haferkamp2022random,nietner2025average}:
one first fixes a circuit space $\mathcal C$ (e.g., a brickwork local circuit of fixed depth $d$ whose layers consist of two-qubit gates), and then specifies a physically natural random ensemble $\mu_{\mathcal C}$ on $\mathcal C$. In an idealized setting this is Haar measure on $U(d)$; in more hardware-like brickwork circuits, each local two-qubit gate is drawn independently from Haar, leading to a product-measure-type distribution on $\mathcal C$. For sufficiently large depth, such local random circuits approximate Haar $t$-designs for fixed $t$~\cite{brandao2016local,haferkamp2022random}. Given a fixed measurement architecture, each circuit $C\in\mathcal C$ induces a classical output distribution $P_C$, and the measure $\mu_{\mathcal C}$ then defines an ensemble of random circuit output distributions. On this basis one studies concentration properties on $(\mathcal C,d_{\mathrm{circ}},\mu_{\mathcal C})$ (e.g., concentration of Lipschitz functions $C\mapsto P_C[\varphi]$) as well as average-case SQ hardness with respect to $\mu_{\mathcal C}$~\cite{nietner2025average}.

Our first contribution is to construct a random Lindbladian ensemble and a central conceptual starting point of this work is to transplant this random-circuit blueprint systematically to the Lindbladian setting.
Instead of postulating a measure directly on the abstract GKSL cone, we first build a linear coordinate system in a finite-dimensional superoperator space, and then choose a metric and a measure in the resulting parameter space; this yields ensembles of random Lindbladians in a controlled way. As summarized schematically in Fig.~\ref{fig:metric_prob_spaces}, our random
Lindbladian ensembles are always defined on a metric probability space
$(\Theta,d,\mu)$, with the Haar-sphere and product-measure models
corresponding to two concrete choices of parameter space, metric, and sampling
measure.
Concretely, our construction proceeds as follows:

\begin{enumerate}
	\item \textbf{Complete linear parametrization.}
	For fixed Hilbert-space dimension $d$, the affine hull of the GKSL cone can be identified with an affine subspace of a finite-dimensional real vector space.
	Lemma~\ref{lem:affine_span_L} shows that there exist a basis $\{G_j\}_{j=1}^M$ and a reference generator $\mathcal L_{\mathrm{ref}}\in\mathsf{GKSL}$ such that
	\[
	\mathcal L
	= \mathcal L_{\mathrm{ref}} + \sum_{j=1}^M \theta_j(\mathcal L)\,G_j
	\]
	for every Lindbladian $\mathcal L\in\mathsf{GKSL}$, where the feasible parameters $\theta(\mathcal L)\in\mathbb R^M$ lie in a convex cone $\mathcal C_{\mathrm{GKSL}}\subset\mathbb R^M$.
	In other words, in suitable linear coordinates the form
	$
	\mathcal L=\mathcal L_{\mathrm{ref}}+\sum_j\theta_j(\mathcal L)G_j
	$
	covers all GKSL generators; in this sense the parametrization is even more flexible than those constrained by group structure in random-unitary models.
	
	\item \textbf{From $(\theta,d,\mu)$ to random Lindbladian ensembles.}
	Once a linear coordinate system is fixed, we do not choose a measure directly on $\mathsf{GKSL}$. Instead, we first choose on the parameter space $\Theta\subset\mathbb R^M$
	\[
	\text{(i) a parameter domain $\Theta$, \quad (ii) a metric $d(\cdot,\cdot)$, \quad (iii) a probability measure $\mu$,}
	\]
	thereby forming a metric probability space $(\Theta,d,\mu)$.
	We then push forward $\mu$ through
	\[
	\theta\longmapsto \mathcal L(\theta)
	:= \mathcal L_{\mathrm{ref}}+\sum_{j=1}^M\theta_j G_j
	\]
	to obtain a prior $\mu_{\mathcal L}:=\mathcal L_\#\mu$ on Lindbladians.
	This mirrors the random-circuit step ``choose a circuit space $\mathcal C$ and a circuit measure $\mu_{\mathcal C}$'', except that we now work in a linear superoperator space rather than a unitary group.
	
	\item \textbf{Two concrete ensembles designed for concentration and SQ hardness.}
	To make the subsequent learning-complexity analysis work, we require $(\Theta,d,\mu)$ to satisfy:
	\begin{itemize}
		\item[(a)] For every SQ(analogy with QPStat) test function $\varphi:X\to[-1,1]$, the map
		\(
		\theta\mapsto F_\varphi(\theta):=P_{\mathcal L(\theta)}[\varphi]
		\)
		has a uniform Lipschitz constant with respect to $d$, as provided by Corollary~\ref{cor:PL_Lipschitz_L} together with the linear parametrization;
		\item[(b)] The metric probability space $(\Theta,d,\mu)$ itself satisfies a strong concentration inequality (of Lévy or McDiarmid type), so that $F_\varphi$ exhibits high-dimensional concentration under $\mu$, yielding exponentially small upper bounds on
		\(
		\operatorname{frac}(\mu_{\mathcal L},Q,\tau)
		\).
	\end{itemize}
	Guided by these requirements, we introduce two physically natural and geometrically controlled ensembles:
	\begin{itemize}
		\item[i.] \emph{Spherical random-perturbation ensemble} (Definition~\ref{def:param_sphere_ensemble} and Remark~\ref{ball lin}):
		in $\mathbb R^M$ take the unit sphere $S^{M-1}$ with its rotation-invariant measure, and define
		\[
		\mathcal L_{\mathrm{sph}}(\theta)
		= \mathcal L_{\mathrm{ref}}
		+ \frac{\delta}{\sqrt M}\sum_{j=1}^M \theta_j G_j,\qquad \theta\in S^{M-1},
		\]
		where $\delta>0$ is a small parameter.
		By choosing $\delta$ sufficiently small (or normalizing the intersection $S^{M-1}\cap\mathsf{GKSL}$), we ensure $\mathcal L_{\mathrm{sph}}(\theta)\in\mathsf{GKSL}$ for all $\theta$.
		The Euclidean/spherical geometry lets us apply Lévy concentration (Lemma~\ref{lem:levy_sphere}) and obtain the $\exp(-\Omega(M))$-type bound in Theorem~\ref{thm:frac_bound_open_system}.
		\item[ii.] \emph{Independent local-coupling (product-measure) ensemble} (Definition~\ref{def:prod_measure_ensemble}):
		here $\Theta_{\mathrm{prod}}=[-\delta/M,\delta/M]^M$, each coordinate $\theta_j$ is drawn independently, and the metric is $\ell_1$, $d_1(\theta,\theta')=\|\theta-\theta'\|_1$.
		The parametrization
		\[
		\mathcal L_{\mathrm{prod}}(\theta)
		= \mathcal L_{\mathrm{ref}} + \sum_{j=1}^M \theta_j G_j
		\]
		models the assumption ``each local coupling constant is an independent small random number.''
		Then $F_\varphi(\theta)$ has a dimension-independent Lipschitz constant in $d_1$ (Lemma~\ref{lem:prod_theta_Lipschitz}), and McDiarmid’s inequality (Lemma~\ref{lem:mcdiarmid}) yields Theorem~\ref{thm:frac_bound_open_system_prod}.
	\end{itemize}
	In both constructions the small parameter $\delta$ plays a dual role, because it controls the deviation of $\mathcal L(\theta)$ from the reference $\mathcal L_{\mathrm{ref}}$ (staying in a local GKSL neighborhood), and it enables the linear-response analysis in Sec.~\ref{zhang4}.
\end{enumerate}

\begin{figure}[t]
	\centering
	\includegraphics[width=0.85\textwidth]{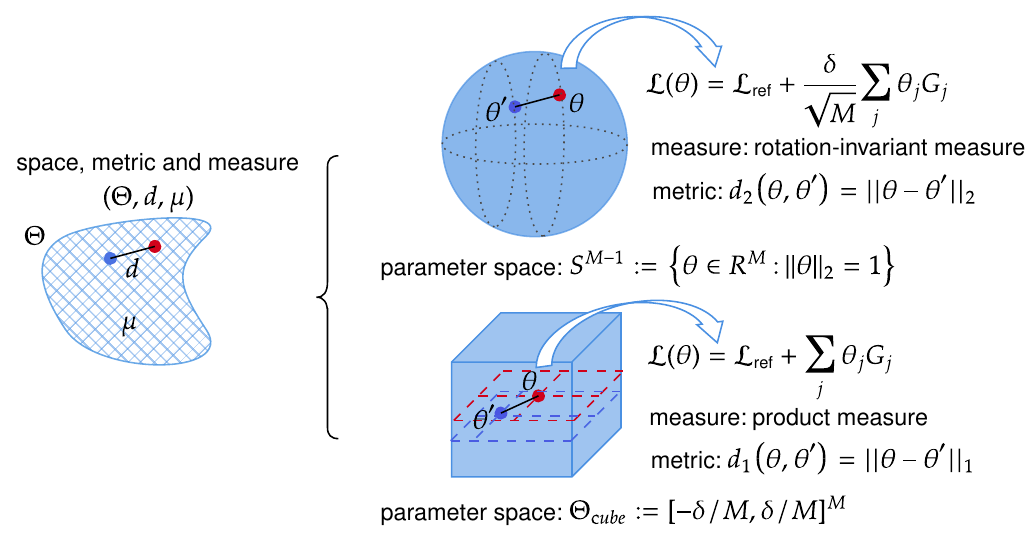}
	\caption{Schematic view of the metric probability spaces used to define
		random Lindbladian ensembles. 
		Left: an abstract parameter space $(\Theta,d,\mu)$, where $\Theta$ is the
		parameter domain, $d$ is the chosen metric on $\Theta$, and $\mu$ is the
		ensemble (sampling) measure.
		Top right: the Haar-sphere ensemble, where parameters are sampled
		uniformly from the high-dimensional sphere
		$S^{M-1}:=\{\theta\in\mathbb R^M:\|\theta\|_2=1\}$ with respect to the
		rotation-invariant (Haar-induced) surface measure.  The metric is the
		Euclidean distance $d_2(\theta,\theta')=\|\theta-\theta'\|_2$, and the
		Lindbladian is parameterized as
		$\mathcal L(\theta)=\mathcal L_{\mathrm{ref}}+\frac{\delta}{\sqrt{M}}
		\sum_j \theta_j G_j$.
		Bottom right: the product-measure ensemble, where each coupling
		$\theta_j$ is an independent random variable supported on
		$[-\delta/M,\delta/M]$, so that
		$\Theta_{\mathrm{cube}}:=[-\delta/M,\delta/M]^M$.
		Here the metric is the $\ell_1$ distance
		$d_1(\theta,\theta')=\|\theta-\theta'\|_1$, the measure is the product
		measure $\mu_{\mathrm{prod}}$, and the Lindbladian is parameterized as
		$\mathcal L(\theta)=\mathcal L_{\mathrm{ref}}+\sum_j \theta_j G_j$.}
	\label{fig:metric_prob_spaces}
\end{figure}

From this perspective, a structural contribution of our work is the systematic route
\[
\text{``linear parametrization + (parameter space, metric, measure)''}
\]
for generating random Lindbladian ensembles: in finite dimension we cover all GKSL generators (Lemma~\ref{lem:affine_span_L}); with appropriate choices of $\Theta,d,\mu$ we can directly plug Lévy or McDiarmid concentration tools into the SQ/QPStat learning-complexity framework.
Conceptually this mirrors traditional random-circuit constructions (Haar or local gate sets), but in the open-system setting it yields a linear, explicitly parametrized, and physically natural realization.

\subsection{Setup}

Once the geometric structure of $(\Theta,d,\mu)$ and the linear parametrization
$\theta\mapsto\mathcal L(\theta)$ are in place, we will study learning and decision
problems for random Lindbladian dynamics under two complementary information-access
models,The overall architecture of our model is summarized in
Fig.~\ref{fig:setup_SQ_interface}.

\begin{itemize}
	\item[(i)] a distribution-level SQ interface, where one fixes a time $t>0$,
	an input state $\rho_{\mathrm{in}}$ and a POVM $\{M_x\}_{x\in X}$, and only sees the
	induced classical output distribution $P_{\mathcal L(\theta)}\in\mathcal D_X$ through
	statistical queries $P_{\mathcal L(\theta)}[\varphi]$ with tolerance $\tau$;
	\item[(ii)] a process-level QPStat interface, where for each parameter
	$\theta$ one regards the CPTP map $E_\theta = e^{t\mathcal L(\theta)}$ as the
	object to be learned, and a QPStat oracle returns approximate expectation values
	$\operatorname{Tr}[O\,E_\theta(\rho)]$ for any state $\rho\in\mathsf S(\mathcal H)$ and operator $O\in\mathcal B(\mathcal H)$ with $\|\rho\|_1\le1$ and $\|O\|_\infty\le1$ with prescribed tolerance.
\end{itemize}

On the distribution side, using the exponentially small upper bounds from
Theorems~\ref{thm:frac_bound_open_system} and \ref{thm:frac_bound_open_system_prod},
together with the general SQ learning and decision-theoretic lemmas of
Feldman and of Nietner \emph{et al.}~\cite{feldman2017general,nietner2025average},
we obtain in Sec.~\ref{result 2} exponential SQ lower bounds
for average-case learning tasks on
random-Lindbladian output distributions. For ease of comparison with the random-circuit
literature on Haar designs, we formulate our SQ hardness results primarily for the
spherical random-perturbation ensemble (Definition~\ref{def:param_sphere_ensemble})
as a baseline, while the independent local-coupling ensemble illustrates that the
framework extends to other physically natural random GKSL models as well.

On the process side, we introduce in Sec.~\ref{zhang8} a QPStat oracle model adapted
to open-system dynamics, together with a general diamond-norm lower-bound framework.
We show that for the same parameter-sphere ensemble $(\Theta,d,\mu)$, the map
$\theta\mapsto E_\theta=e^{t\mathcal L(\theta)}$ inherits a uniform Lipschitz control
with respect to suitable operator norms, and the Lévy- or McDiarmid-type concentration
on $(\Theta,d,\mu)$ again yields exponentially small distinguishable fractions.
Specialized to our linear random-Lindbladian embeddings, this leads to an exponential
QPStat query lower bound for learning the channels $\{E_\theta\}$ in diamond distance,
as stated in Theorem~\ref{thm:lindblad_qpstat_lower_bound}. In Sec.~\ref{sec:crqpuf_applications},
these two access models will be combined with our SQ and QPStat hardness results to
give distribution-level and tomography-based Lindbladian-PUF constructions.

\begin{figure}[t]
	\centering
	\includegraphics[width=0.8\textwidth]{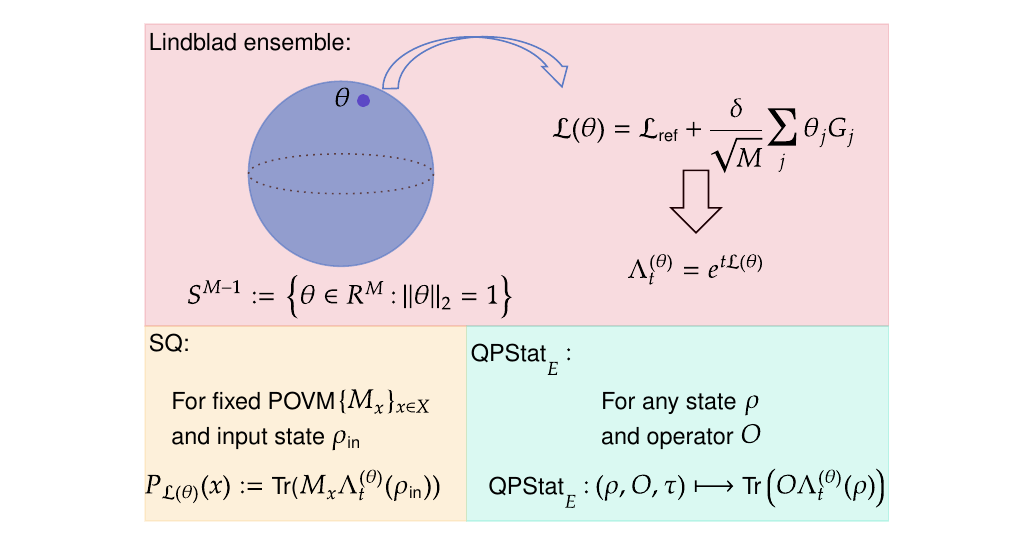}
	\caption{Schematic setup of our random Lindbladian ensemble and the two
		information-access models considered in this work.
		Top: a Lindbladian family is parametrized by points $\theta$ on the
		high-dimensional sphere $S^{M-1}$ via the linear embedding
		$\mathcal L(\theta)=\mathcal L_{\mathrm{ref}}+ \frac{\delta}{\sqrt{M}}\sum_j
		\theta_j G_j$, and generates a quantum Markov semigroup
		$\Lambda_t^{(\theta)} = e^{t\mathcal L(\theta)}$.
		Bottom left: the distribution-level SQ interface, where a fixed input state
		$\rho_{\mathrm{in}}$ and POVM $\{M_x\}_{x\in X}$ define the output distribution
		$P_{\mathcal L(\theta)}(x)=\operatorname{Tr}\bigl(M_x\Lambda_t^{(\theta)}(\rho_{\mathrm{in}})\bigr)$,
		and the learner only accesses expectations $P_{\mathcal L(\theta)}[\varphi]$.
		Bottom right: the process-level QPStat interface for the channel
		$E_\theta:=\Lambda_t^{(\theta)}$, where a QPStat oracle, given a state
		$\rho$, an observable $O$ and a tolerance $\tau$, returns an approximation to
		$\operatorname{Tr}\bigl(O\,\Lambda_t^{(\theta)}(\rho)\bigr)$.}
	\label{fig:setup_SQ_interface}
\end{figure}

We now spell out the concrete random local Lindbladian ensemble that will serve as
a common backbone for both the SQ and the QPStat analyses.

\begin{definition}[Random local Lindbladian ensemble]\label{def:random_local_L_recall}
	Consider a system of $N$ qubits with Hilbert space
	$\mathcal H=(\mathbb C^2)^{\otimes N}$ of dimension $d=2^N$.
	Let $\mathcal B(\mathcal H)$ denote the algebra of bounded operators on $\mathcal H$.
	\begin{enumerate}
		\item Fix a reference Lindbladian generator
		$\mathcal L_{\rm ref}:\mathcal B(\mathcal H)\to\mathcal B(\mathcal H)$
		that satisfies the GKSL conditions.
		Denote the corresponding quantum Markov semigroup (QMS) by
		\[
		\Lambda_t^{({\rm ref})}:=e^{t\mathcal L_{\rm ref}}.
		\]
		\item Construct a family of local Lindblad dissipators $\{G_j\}_{j=1}^M$, where each
		$G_j$ is a linear combination of dissipators
		\[
		\mathcal D_\alpha(\rho)
		=L_\alpha\rho L_\alpha^\dagger - \frac12\{L_\alpha^\dagger L_\alpha,\rho\},
		\]
		associated with local jump operators $L_\alpha$.
		Assume there exists a constant $C_G>0$ (independent of $M$) such that
		\begin{equation}\label{eq:Gj_norm_bound_recall}
			\|G_j\|_{1\to1}\le C_G,\qquad j=1,\dots,M.
		\end{equation}
		\item Take the parameter space to be the real sphere
		\[
		S^{M-1}:=\{\theta\in\mathbb R^M:\ \|\theta\|_2=1\},
		\]
		and let $\mu_{\rm sph}$ denote the uniform (Haar) measure on this sphere.
		For a given perturbation amplitude $\delta>0$, define the random Lindbladian
		\begin{equation}\label{eq:L_theta_delta}
			\mathcal L(\theta)
			:=
			\mathcal L_{\rm ref}
			+\frac{\delta}{\sqrt M}\sum_{j=1}^M\theta_j G_j,
			\qquad \theta\in S^{M-1}.
		\end{equation}
		For $\delta$ sufficiently small, Lemma~\ref{lem:convexity_GKSL_slice} guarantees
		that all $\mathcal L(\theta)$ are still valid GKSL generators. Denote the corresponding
		QMS by
		\[
		\Lambda_t^{(\theta)}:=e^{t\mathcal L(\theta)}.
		\]
		\item Fix an evolution time $t>0$, an initial state $\rho_{\rm in}\in\mathcal S_d^+$,
		and a POVM $\{M_x\}_{x\in X}$ on a finite outcome set $X$.
		For each $\theta$, define the output measurement distribution
		\begin{equation}\label{eq:P_theta_Q_def}
			P_{\mathcal L(\theta)}(x)
			:= \operatorname{Tr}\!\bigl(M_x\,\Lambda_t^{(\theta)}(\rho_{\rm in})\bigr),
			\qquad x\in X.
		\end{equation}
		In the SQ model, the learner only has access to $P_{\mathcal L(\theta)}$ through
		expectations $P_{\mathcal L(\theta)}[\varphi]$.
		\item At the same time, for each $\theta$ we regard the CPTP map
		\[
		E_\theta := \Lambda_t^{(\theta)}
		\]
		as the quantum channel to be learned in the QPStat model. The QPStat oracle
		studied in Sec.~\ref{zhang8} provides approximate values
		$\operatorname{Tr}[O\,E_\theta(\rho)]$ for any state $\rho\in\mathsf S(\mathcal H)$ and operator $O\in\mathcal B(\mathcal H)$ with $\|\rho\|_1\le1$ and $\|O\|_\infty\le1$ with prescribed tolerance. 
	\end{enumerate}
\end{definition}

\subsection{Learning hardness and cryptographic applications}
\label{subsec:summary_learning}

Our second group of results concerns the learning complexity of random Lindbladian
ensembles and their cryptographic applications.  On the distribution side, we show
that the classical output statistics induced by a physically natural random-Lindbladian
ensemble are exponentially hard to learn in the SQ model in
average-case.  On the process side, we introduce a QPStat oracle model for
expectation-value queries to quantum channels and prove an analogous exponential
hardness result for learning random Lindbladian channels in diamond norm.  Finally,
we combine these two query models to design two families of Lindbladians
physically unclonable functions (Lindbladians-PUFs) whose security is based on the above
learning-hardness results.  The technical backbone of our average-case analysis is a
linear-response formula for the mean TV distance together with an explicit
amplitude-damping model where the relevant constants can be computed exactly.

\paragraph*{SQ hardness for output distributions.}

We first fix a physically natural random-Lindbladian ensemble
$\mu_{\mathcal L}$ constructed via the linear parametrization
$\theta\mapsto\mathcal L(\theta)$ on a parameter sphere $\Theta=S^{M-1}$ of
dimension $M$, together with a one-shot measurement interface
$(t,\rho_{\mathrm{in}},\{M_x\}_{x\in X})$, which induces a classical output
distribution $P_{\mathcal L}\in\mathcal D_X$ for each Lindbladian
$\mathcal L$.  In the SQ model, the learner only accesses this distribution
through statistical queries $P_{\mathcal L}[\varphi]$ with tolerance $\tau$,
rather than individual samples.

A key geometric input is that, for any SQ test function
$\varphi:X\to[-1,1]$, the associated map
\[
F_\varphi(\theta):=P_{\mathcal L(\theta)}[\varphi]
\]
has a dimension-independent Lipschitz constant on the parameter sphere
(Lemma~\ref{lem:theta_Lipschitz}), and that $\Theta$ satisfies a
Lévy-type concentration inequality (Theorem~\ref{thm:frac_bound_open_system}).
Equivalently, if $Q$ denotes the ensemble-average distribution
$Q[\varphi]=\mathbb E_{\mathcal L\sim\mu_{\mathcal L}}P_{\mathcal L}[\varphi]$,
then the maximal distinguishable fraction
\[
\max_{\varphi:X\to[-1,1]}
\Pr_{\mathcal L\sim\mu_{\mathcal L}}\bigl[|P_{\mathcal L}[\varphi]-Q[\varphi]|>\tau\bigr]
\]
decays as $\exp(-\Omega(M))$.

To capture typical-instance complexity, we introduce the open-system
Porter-Thomas mean assumption (Assumption~\ref{ass:mean_TV}), which asserts
that for some reference distribution $Q$ the mean TV distance
$\mathbb E_{\mathcal L\sim\mu_{\mathcal L}}d_{\mathrm{TV}}(P_{\mathcal L},Q)$
converges to a constant $m_0>0$ as $M\to\infty$.  Combined with a
Lipschitz bound for $d_{\mathrm{TV}}(P_{\mathcal L},Q)$ in the parameter
$\theta$ (Lemma~\ref{lem:TV_Lipschitz_theta}), this yields a ``far from
$Q$'' volume estimate (Theorem~\ref{thm:far_from_Q_open}) and hence
average-case SQ lower bounds.

\begin{theorem}[Deterministic average-case SQ hardness for decision problems (informal)]
	\label{thm:intro_det_decision_zh}
	Assume the conclusions of Theorem~\ref{thm:frac_bound_open_system} and
	Assumption~\ref{ass:mean_TV}, so that Theorem~\ref{thm:far_from_Q_open}
	applies.  Fix accuracy parameters $\epsilon,\tau\in(0,1)$ and a coverage
	parameter $\beta\in(0,1)$ independent of $M$.  Consider the decision
	problem: given SQ oracle access to an unknown distribution
	$P\in\{Q\}\cup\{P_{\mathcal L}:\mathcal L\sim\mu_{\mathcal L}\}$, decide
	whether $P=Q$ or $P$ comes from the Lindbladian ensemble $\mu_{\mathcal L}$.
	
	Then for sufficiently large $M$, any deterministic SQ algorithm that, using
	at most $q$ queries with tolerance $\tau$, solves this task correctly on a
	subset of Lindbladians of $\mu_{\mathcal L}$-measure at least $\beta$ must
	satisfy
	\[
	q_{\mathrm{det}}^{\mathrm{avg\text{-}dec}}(\epsilon,\tau;\beta)
	\ge \beta\,\exp(c M),
	\]
	where $c>0$ depends only on the geometric constants of the ensemble and is
	independent of $M$.
\end{theorem}

\begin{theorem}[Randomized average-case SQ hardness for decision problems (informal)]
	\label{thm:intro_rand_decision_zh}
	Under the same ensemble assumptions as in
	Theorem~\ref{thm:intro_det_decision_zh}, take constants
	$\epsilon,\tau\in(0,1)$ and $\alpha>1/2$, $\beta\in(0,1)$.
	Any randomized SQ algorithm that, with success probability at least
	$\alpha$ over its internal randomness, correctly decides between
	$P=Q$ and $P\sim\mu_{\mathcal L}$ on a subset of Lindbladians of
	$\mu_{\mathcal L}$-measure at least $\beta$ must satisfy
	\[
	q_{\mathrm{rand}}^{\mathrm{avg\text{-}dec}}(\epsilon,\tau;\alpha,\beta)
	\ge (\alpha-\tfrac12)\,\beta\,\exp(c M),
	\]
	where $c>0$ is independent of $M$.
\end{theorem}

\begin{theorem}[Average-case SQ hardness of learning open-system dynamics (informal)]
	\label{thm:intro_avg_learning_zh}
	Under the ensemble assumptions of Theorem~\ref{thm:intro_det_decision_zh},
	take constants $\epsilon,\tau\in(0,1)$ and parameters
	$\alpha>1/2$, $\beta\in(0,1)$.  Consider the SQ learning task:
	given SQ oracle access to an unknown $\mathcal L\sim\mu_{\mathcal L}$
	via a fixed time $t$, input state $\rho_{\mathrm{in}}$, and finite POVM
	interface, output a hypothesis distribution $\widehat P$ satisfying
	$d_{\mathrm{TV}}(\widehat P,P_{\mathcal L})<\epsilon$.
	We only require this goal to hold on a subset of Lindbladians of
	$\mu_{\mathcal L}$-measure at least $\beta$, and for randomized algorithms
	the success probability is required to be at least $\alpha$.
	
	Then for sufficiently large $M$ one has
	\[
	q_{\mathrm{det}}^{\mathrm{avg\text{-}learn}}(\epsilon,\tau;\beta)
	\gtrsim \beta\,\exp(c M),\qquad
	q_{\mathrm{rand}}^{\mathrm{avg\text{-}learn}}(\epsilon,\tau;\alpha,\beta)
	\gtrsim (\alpha-\tfrac12)\,\beta\,\exp(c M),
	\]
	so that both deterministic and randomized average-case SQ learning
	complexities are exponential in $M$.
\end{theorem}

Our average-case learning-complexity results are proved using
``far-from-$Q$'' techniques, but we also establish ``far-from-$D$''
results which show that typical outputs of random Lindbladians cannot be
approximated by any fixed low-complexity classical model.

Specifically, Theorem~\ref{thm:far_from_all_D_open} shows that for a
parameter-sphere ensemble satisfying Lipschitz-Lévy concentration and
the open-system Porter-Thomas mean assumption, with probability
exponentially close to $1$ the output distribution $P_{\mathcal L}$ is
at least a constant distance $\varepsilon_*>0$ in total variation from
any predetermined classical distribution $D$, and the measure of the
``bad'' set decays exponentially in $M$.  From a complexity-theoretic
standpoint this is closely related to the conjecture of
Aaronson and Chen~\cite{aaronson2016complexity} on random circuit
outputs and underlies the intuition that there is no fixed low-complexity
family of classical distributions that can approximate the outputs of
most random Lindbladians.

\paragraph*{QPStat hardness for learning random Lindbladian channels.}

The SQ results above are phrased at the level of classical output
distributions.  In Sec.~\ref{zhang8}, we also consider a process
learning model where the learner has QPStat oracle access to a channel
$E_\theta:=\Lambda_t^{(\theta)}=e^{t\mathcal L(\theta)}$:
for any input state $\rho$ and observable $O$, an oracle query returns
an approximation to $\operatorname{Tr}(O\,E_\theta(\rho))$ up
to additive tolerance $\tau$.

For the same parameter-sphere ensemble $(\Theta,\mu_\Theta)$ we show
that learning the underlying channel up to small diamond distance is
also exponentially hard on average.

\begin{theorem}[Exponential QPStat hardness of learning random Lindbladian channels (informal)]
	\label{thm:intro_QPStat}
	Let $E_\theta=\Lambda_t^{(\theta)}$ be the channels generated
	by the spherical random-perturbation Lindbladian ensemble of
	Definition~\ref{def:random_local_L_recall}, and let $\mu_\Theta$ denote
	the parameter-sphere measure.  Consider QPStat oracle access with
	tolerance $\tau$ as in Definition~\ref{def:QPStat_oracle}, and fix
	accuracy $\epsilon>0$ and coverage $\beta\in(0,1)$.
	
	Any (possibly quantum) QPStat learning algorithm that, using at most
	$q$ QPStat queries, outputs with success probability at least
	$\alpha>1/2$ a hypothesis channel $\widehat{E}$ satisfying
	$\|\widehat{E}-E_\theta\|_\diamond\le\epsilon$ on a
	subset of parameters of $\mu_\Theta$-measure at least $\beta$ must
	satisfy
	\[
	q \ge
	(\alpha-\tfrac12)\,\beta\,
	\exp\!\Bigl(
	c\,\frac{M\tau^2}{L_*^2}
	\Bigr)
	= (\alpha-\tfrac12)\,\beta\,\exp\!\bigl(\Omega(M)\bigr),
	\]
	where $L_*$ is a dimension-independent Lipschitz constant for the map
	$\theta\mapsto\operatorname{Tr}\bigl(O\Lambda_t^{(\theta)}(\rho)\bigr)$
	(Lemma~\ref{lem:qpstat_concentration}), and $c>0$ is independent of $M$.
	A precise formulation is given in Theorem~\ref{thm:lindblad_qpstat_lower_bound}.
\end{theorem}

Thus, for our random-Lindbladian ensembles, not only the induced output
distributions but also the underlying channels themselves are
exponentially hard to learn in natural expectation-value oracle models.

\paragraph*{PUFs from random Lindbladian ensembles.}

The learning-hardness results obtained above have direct implications
for cryptography.  In Sec.~\ref{sec:crqpuf_applications} we turn them
into two families of PUFs based on random Lindbladian ensembles.

\begin{itemize}
	\item[(A)] Distribution-level Lindbladian-PUF with SQ interface (Scheme A): the
	challenge is a parameter $\theta$, the device implements the
	corresponding random Lindbladian channel $E_\theta=e^{t\mathcal
		L(\theta)}$, and we only consider the resulting classical output
	distribution $P_\theta$ after a fixed measurement
	$(t,\rho_{\mathrm{in}},\{M_x\})$.  In the trusted factory phase, the
	verifier (manufacturer) has sample-level access and enrolls a
	probability fingerprint of $P_\theta$; after deployment, the physical
	device exposed to any external user (including a modelling adversary
	$A$) provides only an SQ interface
	$\operatorname{Stat}_\tau(P_\theta)$.  Scheme~A uses an exponentially
	large family of Hadamard-type test functions
	$\mathcal F_{\mathrm{Had}}=\{\varphi_{\mathrm{BIT}}:\mathrm{BIT}\in\{0,1\}^L\}$
	(Definition~\ref{def:hadamard_tests}) as the challenge space: in each
	authentication round, the verifier samples fresh
	$\varphi_{\mathrm{BIT}}$ from $\mathcal F_{\mathrm{Had}}$ and checks
	whether the SQ replies agree with the enrolled fingerprint values
	(Definition~\ref{def:crqpuf_sq_hadamard}).  Because
	$\mathcal F_{\mathrm{Had}}$ forms a tomographically complete
	Hadamard-type basis for $\mathcal D_X$, passing verification implies
	that the adversary has effectively learned the underlying output
	distribution in total variation distance; at the same time, any
	table-lookup strategy that tries to pre-compute the SQ answers on all
	challenges in $\mathcal F_{\mathrm{Had}}$ would require exponentially
	many SQ queries.  The SQ security theorem
	(Theorem~\ref{thm:crqpuf_sq_hadamard_security}) formalizes this: with
	no structural restrictions on the internal strategy of $A$, any
	adversary that, using only SQ access $\operatorname{Stat}_\tau(P_\theta)$,
	impersonates the honest device with non-negligible success probability
	must issue exponentially many SQ queries in the parameter dimension
	~$M$, whereas the honest verifier’s enrollment and verification costs
	remain polynomial in $|X|$ and $1/\tau$.
	
	\item[(B)] Tomography-based Lindbladian-PUF (Scheme B): in an extended QPStat
	model with an informationally complete family of input operators and
	observables, we define a tomographic fingerprint map $\mathcal T$ that
	sends each channel $E_\theta$ to a finite list of expectation values
	(Definition~\ref{def:tomographic_fingerprint}).  Lemma~\ref{lem:fingerprint_norm_equivalence}
	shows that $\mathcal T$ is real-linear and injective and that the
	induced tomographic norm $\|\cdot\|_{\mathrm{tom}}$ is equivalent to
	the diamond norm $\|\cdot\|_\diamond$.  In Scheme~B the verifier
	stores $\mathcal T(E_\theta)$ at enrolment and later authenticates by
	asking the prover to reproduce (up to tolerance) the entire
	tomographic fingerprint matrix
	(Definition~\ref{def:tomographic_crqpuf}).  Passing verification is
	therefore equivalent, at the level of information theory, to having
	learned $E_\theta$ in diamond distance.
\end{itemize}

In both schemes, any modelling adversary that manages to produce a
successful classical emulator of the PUF necessarily solves a
corresponding SQ or QPStat learning problem and therefore faces
exponential query complexity.

\begin{theorem}[Lindbladian-PUF security from learning hardness (informal)]
	\label{thm:intro_CRQPUF}
	For the random-Lindbladian ensembles considered above:
	\begin{itemize}
		\item[(i)] Scheme~A (probability-fingerprint Lindbladian-PUF with SQ
		interface).  Under the SQ hardness assumptions of
		Theorems~\ref{thm:intro_avg_learning_zh}
		, together with the Hadamard
		test-function construction in
		Definition~\ref{def:hadamard_tests}, any adversary~$A$ which, after
		the device leaves the factory, interacts with it \emph{only} via the
		SQ oracle $\operatorname{Stat}_\tau(P_\theta)$, and which attempts
		to impersonate the honest device in Scheme~A with success
		probability at least $\alpha_{learn}>1/2$ on a
		$\mu_{\mathcal L}$-fraction at least $\beta_{learn}$ of challenges, must
		make
		\[
		q_{\mathrm{adv}}^{\mathrm{SQ}}
		\;\gtrsim\;
		(\alpha_{learn}-\tfrac12)\,\beta_{learn}\,\exp(c M)
		\]
		statistical queries
		(Theorem~\ref{thm:crqpuf_sq_hadamard_security}).  Here we do not
		impose any structural restrictions on the internal strategy of~$A$:
		it may be an arbitrary (possibly randomized) algorithm, as long as
		all its interactions with the physical device go through the SQ
		interface.  In contrast, the honest verifier’s enrollment and
		verification costs scale only polynomially in the system size and
		$1/\tau$.
		
		\item[(ii)] Scheme~B (tomography-based Lindbladian-PUF).  In the extended
		QPStat model of Definition~\ref{def:extended_qpstat}, impersonating
		the device in Scheme~B up to a small tomographic-norm
		(equivalently, diamond-norm) error would yield a QPStat learner that
		$\varepsilon$-learns the underlying channel on a non-negligible
		fraction of the ensemble.  By Theorem~\ref{thm:intro_QPStat} and its
		formal version Theorem~\ref{thm:lindblad_qpstat_lower_bound}, any
		such adversary must use exponentially many QPStat queries in the
		parameter dimension $M$ (Theorem~\ref{thm:crqpuf_tomography_security}).
		At the same time, the tomographic fingerprint $\mathcal T(E_\theta)$
		has $D^2\sim d^4$ entries, so the honest verifier’s enrollment and
		verification costs also scale on the order of $d^4$; in typical
		local Lindbladian models with parameter dimension $M$ polynomial in
		$d$, this cost is comparable to the information content of the
		channel itself.  Scheme~B should therefore be viewed as an
		information-theoretic benchmark rather than an efficiently
		verifiable cryptographic protocol.
	\end{itemize}
\end{theorem}

Conceptually, these Lindbladian-PUF constructions illustrate how the learning
hardness of random Lindbladian ensembles can be turned into positive
cryptographic security guarantees, both at the level of classical output
distributions and at the level of full quantum channels.  Scheme~A
realizes a genuinely efficiently verifiable Lindbladian-PUF whose
verification cost is polynomial while any SQ-modelling attack is
exponentially hard; Scheme~B, in contrast, achieves diamond-norm control
but pays the natural tomographic cost of order $d^4$, highlighting the
trade-off between security and verification cost in tomography-based
settings.

\paragraph*{Technical tools: linear response and an explicit model.}

To support the average-case results above, we develop two technical
tools.  First, we derive a linear-response analytic expression for the
mean TV distance
\[
m(\delta):=\mathbb E_{\theta\sim\mu_\Theta}
d_{\mathrm{TV}}\bigl(P_{\mathcal L(\theta)},Q\bigr),
\]
which explicitly tracks the dependence on the parameter dimension $M$
and on the first-order response-coefficient vectors $\{a_{j,x}\}$, and
reduces the physical requirement of a nonzero mean constant to an
extensivity condition on $\sum_x\|a^{(x)}\|_2$.

\begin{theorem}[Linear-response analytic scaling of the mean TV distance (informal)]
	\label{thm:intro_linear_response_informal}
	Under the setting of Definition~\ref{def:random_local_L_recall} and
	Lemma~\ref{lem:perturbation_linear_response}, let
	\[
	F_{\mathrm{TV}}(\theta)
	:= d_{\mathrm{TV}}\bigl(P_{\mathcal L(\theta)},Q\bigr)
	= \frac12\sum_{x\in X}\bigl|P_{\mathcal L(\theta)}(x)-Q(x)\bigr|,
	\quad
	m(\delta):=\mathbb E_{\theta\sim\mu_\Theta}F_{\mathrm{TV}}(\theta),
	\]
	and define
	\[
	a^{(x)}:=\bigl(a_{1,x},\dots,a_{M,x}\bigr)\in\mathbb R^M,
	\qquad
	\|a^{(x)}\|_2
	:=\Bigl(\sum_{j=1}^M a_{j,x}^2\Bigr)^{1/2},
	\]
	where $a_{j,x}$ are given by the first-order response formula
	Eq.~\eqref{eq:ajx_def}.
	
	Then in the small-noise limit $\delta\to 0$ (linear-response regime),
	one has
	\[
	m(\delta)
	= \delta\, m_0(M)+O(\delta^2),
	\qquad
	m_0(M)
	:=\frac{\kappa_M}{2M}\sum_{x\in X}\|a^{(x)}\|_2,
	\]
	where $\kappa_M$ is a coefficient determined by spherical geometry
	(Eq.~\eqref{eq:kappa_M_def}) and satisfies
	$\kappa_M\to\sqrt{2/\pi}$ as $M\to\infty$.
	
	In particular, in the high-dimensional limit,
	\[
	m(\delta)
	= \delta\cdot\frac{\sqrt{2/\pi}}{2M}
	\sum_{x\in X}\|a^{(x)}\|_2
	+O(\delta^2)+o(\delta).
	\]
	This provides a precise linear-response analytic version of the
	open-system Porter-Thomas mean assumption,
	$\mathbb E_\theta d_{\mathrm{TV}}(P_{\mathcal L(\theta)},Q)\approx m_0>0$.
	
	\begin{itemize}
		\item[(i)] If the first-order response coefficients satisfy an
		extensivity scaling in the limit $M\to\infty$,
		\[
		\frac1M\sum_{x\in X}\|a^{(x)}\|_2 \longrightarrow c_a>0,
		\]
		then there exists a constant $c_a'>0$ independent of $M$ such that
		$m(\delta)\approx c_a'\,\delta$, so the mean TV distance remains of
		order $O(\delta)$ with a nonzero constant prefactor.
		\item[(ii)] Conversely, if
		$\sum_x\|a^{(x)}\|_2=o(M)$, then $m_0(M)\to 0$ and the mean TV
		distance decreases with system size; in this regime the typical
		$P_{\mathcal L(\theta)}$ approaches $Q$ in TV distance and SQ-hardness
		criteria based on $d_{\mathrm{TV}}(P_{\mathcal L},Q)$ cannot yield
		nontrivial lower bounds.
	\end{itemize}
	The rigorous version is given in Theorem~\ref{thm:mean_TV_linear_response}
	and Remark~\ref{rmk:mean_TV_PT}.
\end{theorem}

Second, we verify the existence of a positive mean constant $m_0>0$ in a
simple yet fully explicit one-dimensional random local amplitude-damping
model, by computing exactly the mean norm of the first-order
response coefficients in that model.

\begin{theorem}[Explicit computation of the mean scaling constant in a random local amplitude-damping model (informal)]
	\label{thm:intro_amplitude_damping_informal}
	In the random local amplitude-damping Lindbladian model of
	Definition~\ref{def:local_amp_damp_model}, consider a one-dimensional
	chain of $N$ qubits,
	\[
	\mathcal H = (\mathbb C^2)^{\otimes N},\quad M:=N,
	\]
	take the reference evolution to be trivial,
	$\mathcal L_{\mathrm{ref}}=0$, and the initial state to be
	$\rho_{\mathrm{in}}=\ket{1\cdots 1}\!\bra{1\cdots 1}$.
	At each site $j$, introduce a local amplitude-damping jump operator
	\[
	J_j^{(-)}:=\sqrt\gamma\,\sigma_j^{-},\qquad
	G_j:=\mathcal D_j^{(-)},
	\]
	and measure the output in the computational-basis POVM
	$\{M_x=\ket{x}\!\bra{x}\}_{x\in\{0,1\}^N}$.
	
	In this specific model, one can compute the first-order response
	coefficients $\{a_{j,x}\}$ explicitly and show that
	\[
	\frac1M\sum_{x\in X}\|a^{(x)}\|_2
	= \gamma t\Bigl(1+\frac1{\sqrt M}\Bigr),
	\]
	which, when substituted into the general linear-response expression,
	yields
	\[
	m(\delta)
	= \delta\cdot\frac{\kappa_M}{2}\,\gamma t
	\Bigl(1+\frac1{\sqrt M}\Bigr)
	+O(\delta^2),
	\]
	where $\kappa_M$ is as above.
	Letting $M\to\infty$ and using $\kappa_M\to\sqrt{2/\pi}$, we obtain
	\[
	m(\delta)
	\xrightarrow[M\to\infty]{}
	m_\ell\,\delta +O(\delta^2),
	\qquad
	m_\ell:=\gamma t\,\sqrt{\frac1{2\pi}}>0.
	\]
	That is,
	\[
	\mathbb E_{\theta}\,d_{\mathrm{TV}}\bigl(P_{\mathcal L(\theta)},Q\bigr)
	= m_\ell\,\delta +O(\delta^2),\qquad m_\ell>0,
	\]
	with a linear prefactor $m_\ell$ independent of $M$.
	This shows that, after appropriate normalization, the requirement
	$m_0>0$ in the open-system Porter-Thomas mean assumption can be
	rigorously verified in a simple physically realizable toy model.
\end{theorem}

These two technical results explain why the mean-distance assumption used
in our average-case SQ hardness theorems is not ad hoc but follows from
explicit geometric and model-dependent considerations.

\section{Discussion and outlook}

In this work, we start from a seemingly simple yet fundamental question:
under finite evolution time and finite measurement resources, is it possible to reliably reconstruct the noise and dissipation structure acting on a high-dimensional open quantum system, using only statistical information from a single-time evolution?
To this end, we construct a class of physically natural random Lindbladian ensembles on the affine hull of the GKSL generator cone, and systematically apply the frameworks of classical statistical query (SQ) learning and quantum process statistical query (QPStat) learning to the open-system setting.
From this unified perspective, we obtain a collection of exponential lower bounds on average-case learning complexity, and then turn these complexity results into cryptographic security conditions for Lindbladian-PUF protocols driven by random Lindbladians, thereby establishing a clear link between open-system dynamics, learning theory, and quantum cryptography.

Conceptually, the core structural ingredients of this work are the combination of linear parametrization and metric-measure spaces.
We first show that the affine hull of the GKSL cone admits a complete linear parametrization in terms of a finite-dimensional basis of superoperators together with a reference Lindbladian, giving rise to a family of local Lindbladians of the form
\(
\mathcal L(\theta)
\)
in the superoperator space.
On this basis, instead of imposing a measure directly on the GKSL cone, we first choose a geometrically well-behaved metric-measure structure \((\Theta,d,\mu)\) on parameter space (for instance, a high-dimensional spherical ensemble with rotational symmetry or a product ensemble with independent local couplings), and then push it forward to the Lindbladian space.
The idea is inspired by constructions of Haar or design measures for random quantum circuits, but technically we must handle additional constraints coming from Liouvillian semigroups and local Lindbladians, in order to obtain a genuine metric-measure Lindbladian ensemble.

Within this geometric and measure-theoretic framework, we analyze the learning complexity of random Lindbladians in a systematic way.
On the distribution level, once a single-shot measurement interface is fixed, each Lindbladian induces a classical output distribution, and the learner is only allowed to access expectations of this distribution through SQ queries.
We prove that, for any bounded test function, the map from parameters to expectations has a dimension-independent Lipschitz constant on parameter space, and in the spherical and product-measure models it satisfies Lévy and McDiarmid concentration inequalities, respectively.
Combined with existing general SQ lower-bound techniques, this directly yields exponentially scaling query complexities in the parameter dimension for both testing and learning problems:
even if one only requires success on a subset of parameters of nonzero measure, average-case learning of typical random Lindbladian output distributions remains extremely hard in the SQ sense.
On the channel level, we formulate a QPStat access model tailored to Lindbladians, which allows the learner to query expectation values of the evolved channel with a finite tolerance.
We show that, under the same parameter-space structure, the map from parameters to channels also satisfies a uniform Lipschitz condition with respect to appropriate operator norms, and thereby obtain an exponential QPStat query lower bound for average-case learning in diamond distance.
This result demonstrates that, even in the open-system setting and even when we only care about finite-time CPTP evolution, typical Lindbladian channels exhibit strong average-case non-learnability at the process level, in structural analogy with SQ or QPStat hardness results for random unitary circuits in closed systems.

To endow the above average-case lower bounds with a clear physical content, we introduce and analyze an open-system version of a Porter-Thomas-type ``average total variation distance hypothesis'': under suitable normalization and parameter scaling, the average total variation distance between the random Lindbladian output distribution and a fixed reference distribution converges in the high-dimensional limit to a strictly positive constant.
Using a Kubo-like first-order response expansion, we reduce this hypothesis to geometric properties of the vector of first-order response coefficients.
In a concrete model of a local random amplitude-damping chain, we explicitly compute the scaling behavior of these coefficients and the dependence of the average TV distance on the parameter dimension and evolution time, and show that in physically natural parameter regimes of small noise and local jump operators the average distance stays bounded away from zero.
Together with numerical simulations, this establishes a unified physical picture linking microscopic Lipschitz-Lévy control to macroscopic scaling of the average TV distance.

On the cryptographic side, we construct two types of Lindbladian-PUF protocols based on random Lindbladians.
The first family of schemes operates entirely at the distribution level. The verification interface is deliberately restricted to statistical-query access to output distributions, and its security follows directly from the SQ hardness of distribution learning.
The second family allows a finite number of QPStat queries or experimental tomography; in an appropriate operator basis we introduce a tomographic fingerprint, and use the equivalence between the associated tomographic fingerprint norm and the diamond norm to translate channel unclonability into the QPStat hardness of learning Lindbladian channels.
Although the verification cost of the latter is not ideal from a cryptographic standpoint, taken together the two schemes demonstrate a basic message in the open-system setting:
as long as one carefully constrains the information structure available to the verifier and the adversary in the protocol design, average-case exponential lower bounds on learning complexity can be naturally converted into security constraints for Lindbladian-PUFs.

Looking ahead, we expect the framework developed here to be extendable and deepenable in several directions.
A direct extension is to consider more general Lindbladian ensembles and non-Markovian effects.
The present work focuses mainly on linear-perturbation models of local GKSL generators, which are naturally restricted.
It is therefore of interest to move on to Davies generators satisfying detailed balance, thermal Lindbladians with energy-selective structure, or even effective Lindbladian approximations with weak non-Markovian memory kernels.
Understanding how to maintain good metric-measure structure and concentration properties in these more complex noise models is a key step toward characterizing the ``typical complexity'' of more realistic open systems.

Second, our results indicate that random Lindbladians already exhibit a certain form of open-system pseudorandomness under SQ and QPStat models.
A natural long-term goal is to build on our construction to define and characterize ``pseudorandom Lindbladians'' and to relate them to quantum pseudorandom circuits, noise engineering, and open-system cryptography.
This may eventually lead to open-system analogues of unitary $t$-designs and pseudorandom unitaries.

Third, it is natural to look for quantum learning advantages under stronger access models.
The present work concentrates on relatively weak access via SQ and QPStat.
An important question is whether, in models that allow quantum examples, quantum statistical queries, or trajectory-level data, there exist open-system tasks that are classically SQ-hard but QSQ-easy, thereby demonstrating genuine quantum learning advantages rooted in dissipation and the environment.
In this regard, embedding the Lindbladian ensembles introduced here into the framework of generative quantum advantage may lead to new separation results.

Moreover, it is worthwhile to explore connections to experimental noise and NISQ hardware in more depth.
Our local random Lindbladian and amplitude-damping chain models share certain structural features with noise encountered on superconducting qubits, trapped ions, and other platforms, but they remain idealized.
It is therefore natural and important to fit random Lindbladian models satisfying our assumptions to specific hardware architectures, to use experimental data to verify Lipschitz-Lévy concentration and the scaling of the average TV distance, and to realize simplified prototypes of Lindbladian-PUFs directly on hardware.

Finally, the potential links between learning complexity, statistical mechanics, and thermodynamics are particularly intriguing.
Lindbladian dynamics inherently encodes nonequilibrium thermodynamic quantities such as entropy production rates and heat currents.
Combined with our average-case hard-learning results, this raises the question of whether hard-to-learn Lindbladian dynamics are necessarily accompanied by some extreme thermodynamic behavior, and whether one can thereby establish closer connections between learning complexity, mixing times, and nonequilibrium statistical mechanics.

Overall, this work shows that in the setting of open systems and Lindbladian dynamics, ``typicality'' and ``complexity'' remain two sides of the same coin:
on the one hand, random Lindbladians exhibit strong measure concentration and typical behavior in parameter space; on the other hand, under weak access models such as SQ and QPStat, their output distributions and channels display inherent average-case non-learnability.
We hope that this work can serve as a starting point for further investigations of randomness, complexity, and learnability in open systems, and foster more interplay with quantum advantage, noise engineering, and quantum cryptography in the future.

\emph{Acknowledgments:} We are grateful to Ming-Cheng Chen, Chao-Yang Lu, Yi-En Liang and Zhong-Xia Shang for valuable discussions. C. C. is supported by the National Natural Science Foundation of China (No.124B100020). R. B. is supported by JSPS KAKENHI Grant No. 25KJ0766.

\section{Preliminaries}\label{zhang12}

\subsection{Norm conventions and basic inequalities}\label{subsec:norms} 

In this subsection we unify several types of norm notation used in this paper, for later use in Lipschitz constants, concentration estimates, and stability analysis of Lindbladians. Our main reference are books~\cite{bhatia2013matrix} and~\cite{watrous2018theory}.

\begin{definition}[Euclidean norms on parameter space]\label{def:euclidean_norm}
	Let the parameter vector be $\theta=(\theta_1,\dots,\theta_M)\in\mathbb{R}^M$. Define
	\begin{equation}
		\|\theta\|_2 := \Bigl(\sum_{j=1}^M \theta_j^2\Bigr)^{1/2},\qquad
		\|\theta\|_1 := \sum_{j=1}^M |\theta_j|,\qquad
		\|\theta\|_\infty := \max_{1\le j\le M}|\theta_j|.
	\end{equation}
	In particular, the parameters of random Lindbladians in this paper are typically taken from the unit sphere
	$S^{M-1}:=\{\theta\in\mathbb{R}^M:\|\theta\|_2=1\}$.
\end{definition}

\begin{definition}[Schatten norms of operators]\label{def:schatten_norms}
	Let $\mathcal{H}\cong\mathbb{C}^d$ and $\mathcal{B}(\mathcal{H})$ the algebra of bounded operators on it.
	For any $A\in\mathcal{B}(\mathcal{H})$, denote its singular values by $s_1,\dots,s_d\ge0$.
	For $p\in[1,\infty]$, define the Schatten $p$-norms by
	\begin{align}
		\|A\|_1 &:= \sum_{i=1}^d s_i = \operatorname{Tr}\sqrt{A^\dagger A} \quad\text{(trace norm)},\\
		\|A\|_2 &:= \Bigl(\sum_{i=1}^d s_i^2\Bigr)^{1/2}
		= \bigl(\operatorname{Tr}(A^\dagger A)\bigr)^{1/2} \quad\text{(Hilbert-Schmidt norm)},\\
		\|A\|_\infty &:= \max_i s_i
		= \sup_{\|\psi\|=1}\|A\psi\| \quad\text{(operator norm)}.
	\end{align}
	In particular, when $A$ is an observable, $\|A\|_\infty$ is its spectral radius (maximum absolute eigenvalue).
\end{definition}

\begin{remark}[Relations between Schatten norms]\label{rem:schatten_relations}
	In $d$ dimensions the Schatten norms satisfy the standard sandwich inequalities:
	\begin{equation}
		\|A\|_\infty
		\le \|A\|_2
		\le \|A\|_1
		\le \sqrt{d}\,\|A\|_2
		\le d\,\|A\|_\infty.
	\end{equation}
	In this paper the trace distance $\|\rho-\sigma\|_1$ is used to quantify the geometric distance between quantum states, while $\|O\|_\infty$ appears in Hölder's inequality,
	\(
	|\operatorname{Tr}(OX)| \le \|O\|_\infty\|X\|_1,
	\)
	and serves as a basic tool in many Lipschitz bounds.
\end{remark}

\begin{definition}[superoperator norms induced by Schatten norms, cf.~\cite{watrous2018theory,watrous2004notes}]\label{def:induced_norms}
	Let $\mathcal T:\mathcal{B}(\mathcal{H})\to\mathcal{B}(\mathcal{H})$ be a linear superoperator.
	For given Schatten norms $\|\cdot\|_p,\|\cdot\|_q$, define the induced $p\to q$ norm by
	\begin{equation}
		\|\mathcal T\|_{p\to q}
		:= \sup_{X\neq 0} \frac{\|\mathcal T(X)\|_q}{\|X\|_p}.
	\end{equation}
	The case most frequently used in this paper is $p=q=1$:
	\begin{equation}
		\|\mathcal T\|_{1\to 1}
		:= \sup_{X\neq0} \frac{\|\mathcal T(X)\|_1}{\|X\|_1}.
	\end{equation}
	When we write
	\(
	\|\Lambda_t^{(\mathcal L')}-\Lambda_t^{(\mathcal L)}\|_{1\to1},
	\)
	we mean the $1\to1$ norm induced by the trace norm in the above sense.
\end{definition}

\begin{remark}[Induced $1\to1$ norm and contraction of trace distance]\label{rem:1to1_contraction}
	If $\Lambda_t$ is a CPTP map (for instance, an element of a QMS generated by a GKSL generator $\mathcal L$), then for any operator $X$,
	\begin{equation}
		\|\Lambda_t(X)\|_1 \le \|X\|_1,
	\end{equation}
	and hence $\|\Lambda_t\|_{1\to1}\le 1$.
	For density matrices $\rho,\sigma$, this is equivalent to contraction of trace distance:
	\(
	\|\Lambda_t(\rho)-\Lambda_t(\sigma)\|_1\le\|\rho-\sigma\|_1.
	\)
	In this paper, when discussing the difference between two Lindbladian evolutions, we often control the TV distance of the corresponding measurement distributions via
	$\|\Lambda_t^{(\mathcal L')}-\Lambda_t^{(\mathcal L)}\|_{1\to1}$.
\end{remark}

\begin{definition}[Diamond norm]\label{def:diamond_norm}
	Let $\mathcal T:\mathcal{B}(\mathcal{H})\to\mathcal{B}(\mathcal{H})$ be a linear superoperator.
	Its diamond norm is defined as
	\begin{equation}
		\|\mathcal T\|_\diamond
		:= \sup_{k\ge1} \|\mathcal T\otimes\operatorname{id}_{\mathcal{H}_k}\|_{1\to1},
	\end{equation}
	where $\mathcal{H}_k\cong\mathbb{C}^k$ and $\operatorname{id}_{\mathcal{H}_k}$ is the identity channel.
	In finite dimensions, the supremum can be restricted to $k=d=\dim\mathcal H$.
	In particular, for two quantum channels $\Phi,\Psi$ one has
	\begin{equation}
		\frac12\|\Phi-\Psi\|_\diamond
		= \sup_{\rho_{SA}} \frac12\bigl\|
		(\Phi\otimes\operatorname{id}_A)(\rho_{SA})
		- (\Psi\otimes\operatorname{id}_A)(\rho_{SA})
		\bigr\|_1,
	\end{equation}
	i.e., the diamond distance gives the worst-case output trace distance between channels when an ancilla is allowed.
\end{definition}

\begin{remark}[Relation between diamond norm and induced $1\to1$ norm]\label{rem:diamond_vs_1to1}
	In general one has
	\begin{equation}
		\|\mathcal T\|_{1\to1} \le \|\mathcal T\|_\diamond \le d\,\|\mathcal T\|_{1\to1},
	\end{equation}
	where $d=\dim\mathcal H$.
	Thus, in the absence of ancillary systems,
	$\|\mathcal T\|_{1\to1}$ is a lower bound for $\|\mathcal T\|_\diamond$;
	when the dimension is fixed, they differ only by a polynomial factor.
	In this paper we mainly use $\|\cdot\|_{1\to1}$ to establish Lipschitz bounds; when we need to compare with the standard channel distance (diamond norm), one can convert between them via the above inequalities.
\end{remark}

\begin{definition}[Supremum norm of classical functions and measurement operators]\label{def:sup_norm_classical}
	On a classical domain $X$, for a function $\varphi:X\to\mathbb{R}$ define
	\(
	\|\varphi\|_\infty := \sup_{x\in X}|\varphi(x)|.
	\)
	In the statistical query framework we usually require $\varphi:X\to[-1,1]$,
	so that $\|\varphi\|_\infty\le 1$.
	
	In the quantum setting, given a POVM or projective measurement,
	we often associate a bounded observable $O_\varphi$ to a test function $\varphi$ such that
	\[
	P[\varphi] = \mathbb{E}_{x\sim P}[\varphi(x)]
	= \operatorname{Tr}(O_\varphi \rho),
	\]
	and $\|O_\varphi\|_\infty\le 1$(because $\left\|O_{\varphi}\right\|_{\infty} \leq\|\varphi\|_{\infty}$ in this case, see Lemma~\ref{lem:Ophi_norm} for proof).
	The matrix Hölder inequality
	\(
	|\operatorname{Tr}(O_\varphi X)|\le\|O_\varphi\|_\infty\|X\|_1
	\)
	ensures a Lipschitz control of statistical query outputs in terms of the trace distance between states.
\end{definition}

\begin{remark}[Unified conventions in this paper]\label{rem:norm_conventions_summary}
	Summarizing, the default meanings of the norm symbols in this paper are:
	\begin{itemize}
		\item $\|\cdot\|_2$: for parameter vectors it is the Euclidean norm, and for operators the Hilbert-Schmidt norm (depending on context);
		\item $\|\cdot\|_1$: for operators it is the trace norm, and for differences of distributions it is twice the total variation distance;
		\item $\|\cdot\|_\infty$: for operators it is the spectral radius, and for functions the supremum norm;
		\item $\|\cdot\|_{1\to1}$: the superoperator norm induced by the trace norm;
		\item $\|\cdot\|_\diamond$: the standard diamond norm in quantum information, appearing only when we compare channel distances.
	\end{itemize}
	All Lipschitz constants and concentration inequalities in what follows are based on the above norm conventions.
\end{remark}

\subsection{GKSL generators and quantum Markov semigroups}\label{subsec:GKSL}

In this subsection, we briefly review the Gorini-Kossakowski-Sudarshan-Lindblad (GKSL) theorem and related structures, and fix the notation and basic properties of Lindbladians used in this work. Our main reference is Chapter~5 of~\cite{rivas2012open}.

\begin{definition}[Quantum Markov semigroups and generators]\label{def:QMS}
	Let $\mathcal H$ be a finite-dimensional Hilbert space and $\mathcal B(\mathcal H)$ the algebra of bounded operators on $\mathcal H$.
	A family of linear maps
	\[
	\{\Lambda_t\}_{t\ge0},\qquad \Lambda_t:\mathcal B(\mathcal H)\to\mathcal B(\mathcal H),
	\]
	is called a quantum Markov semigroup (QMS) if:
	\begin{enumerate}
		\item Semigroup property: $\Lambda_0=\mathrm{id}$ and, for all $s,t\ge0$,
		\(
		\Lambda_{t+s}=\Lambda_t\circ\Lambda_s;
		\)
		\item Complete positivity and trace preservation: For each $t\ge0$, $\Lambda_t$ is completely positive and trace-preserving (CPTP);
		\item Strong continuity: For every $\rho\in\mathcal B(\mathcal H)$, the map $t\mapsto \Lambda_t(\rho)$ is continuous in the trace norm.
	\end{enumerate}
	Under these conditions there exists a unique bounded linear map
	\(
	\mathcal L:\mathcal B(\mathcal H)\to\mathcal B(\mathcal H)
	\)
	such that
	\[
	\Lambda_t = e^{t\mathcal L},\qquad
	\frac{\mathrm d}{\mathrm dt}\rho(t)\big|_{t=0} = \mathcal L(\rho(0)),
	\]
	where $\rho(t):=\Lambda_t(\rho(0))$.
	The map $\mathcal L$ is called the (Lindblad) generator of the QMS.
\end{definition}

\begin{definition}[GKSL form and Lindblad generators]\label{def:GKSL_form}
	In the finite-dimensional setting, a generator $\mathcal L$ generates a QMS as in Definition~\ref{def:QMS} if and only if it can be written in GKSL (Lindblad) form~\cite{rivas2012open}:
	\begin{equation}\label{eq:GKSL_form}
		\mathcal L(\rho)
		= -\mathrm i[H,\rho]
		+ \sum_k \gamma_k\Bigl(
		L_k\rho L_k^\dagger
		- \frac12\{L_k^\dagger L_k,\rho\}
		\Bigr),
	\end{equation}
	where
	\begin{enumerate}
		\item $H=H^\dagger$ is the effective system Hamiltonian;
		\item $\{L_k\}$ is a family of Lindblad (jump) operators;
		\item the rates satisfy $\gamma_k\ge0$;
		\item $\{\cdot,\cdot\}$ denotes the anticommutator.
	\end{enumerate}
	A map $\mathcal L$ satisfying these conditions is called a GKSL generator (or Lindbladian).
	We denote the set of all such maps by
	\(
	\mathrm{GKSL}\subset \operatorname{End}(\mathcal B(\mathcal H)).
	\)
\end{definition}

\begin{remark}[The GKSL cone and affine structure]\label{rem:GKSL_cone}
	From the form \eqref{eq:GKSL_form} it follows that:
	\begin{itemize}
		\item[i.] If $\mathcal L_1,\mathcal L_2\in\mathrm{GKSL}$ and $a,b\ge0$, then $a\mathcal L_1+b\mathcal L_2$ is again a GKSL generator (one can merge jump operators and rates accordingly).
		Thus $\mathrm{GKSL}$ forms a closed convex cone in $\operatorname{End}(\mathcal B(\mathcal H))$.
		\item[ii.] Fix a reference Lindbladian $\mathcal L_{\mathrm{ref}}\in\mathrm{GKSL}$.
		Its affine hull
		\[
		\mathrm{Aff}(\mathrm{GKSL})
		:= \bigl\{\mathcal L_{\mathrm{ref}} + V:\ V\in \operatorname{span}(\mathrm{GKSL}-\mathcal L_{\mathrm{ref}})\bigr\}
		\]
		is a finite-dimensional real affine space.
		In this work we take
		\(
		\mathrm V := \mathrm{Aff}(\mathrm{GKSL})-\mathcal L_{\mathrm{ref}}
		\)
		as the underlying linear parameter space for Lindbladian perturbations.
	\end{itemize}
	Physically, $\mathrm{GKSL}$ describes all effective generators of open-system dynamics that are completely positive, trace-preserving, and Markovian; its convexity reflects the fact that one can model the superposition of different dissipative channels.
\end{remark}

\begin{remark}[Kossakowski conditions and comparison with classical generators]\label{rem:Kossakowski}
	For a classical finite-state Markov process, the generator matrix $Q$ must satisfy
	\[
	Q_{ii}\le0,\quad Q_{ij}\ge0~(i\neq j),\quad \sum_j Q_{ij}=0,
	\]
	which guarantees that $e^{Qt}$ is a stochastic matrix.
	The GKSL theorem can be viewed as a quantum analogue of these conditions:
	by the Lumer-Phillips theorem~\cite{lumer1961dissipative} one can show that if $e^{t\mathcal L}$ forms a contraction semigroup on the space of self-adjoint operators in the trace norm, and if $\mathcal L$ is Hermitian-preserving and trace-preserving, then it must have the structure~\eqref{eq:GKSL_form}. Conversely, any $\mathcal L$ of the form~\eqref{eq:GKSL_form} generates a CPTP contraction semigroup~\cite{rivas2012open}.
\end{remark}

\begin{proposition}[Trace-norm contractivity and complete positivity]\label{prop:trace_contractive}
	Let $\mathcal L$ be a GKSL generator and $\Lambda_t=e^{t\mathcal L}$.
	Then for every $t\ge0$ and all density operators $\rho,\sigma$,
	\begin{equation}\label{eq:trace_contractive}
		\|\Lambda_t(\rho)-\Lambda_t(\sigma)\|_1
		\le \|\rho-\sigma\|_1.
	\end{equation}
	More generally, for any ancillary system $\mathcal H_{\mathrm{anc}}$, the map
	$\Lambda_t\otimes\mathrm{id}_{\mathcal H_{\mathrm{anc}}}$ is also a contraction in trace norm.
	This is equivalent to the complete positivity of $\Lambda_t$.
\end{proposition}

\begin{remark}[Heisenberg picture and dual generator]\label{rem:Heisenberg_dual}
	Given a Schr\"odinger-picture Lindbladian $\mathcal L$,
	its dual generator in the Heisenberg picture, $\mathcal L^\dagger$, satisfies
	\[
	\frac{\mathrm d}{\mathrm dt}X(t)
	= \mathcal L^\dagger(X(t)),\qquad
	\operatorname{Tr}\bigl[\rho\,\mathcal L^\dagger(X)\bigr]
	= \operatorname{Tr}\bigl[\mathcal L(\rho)\,X\bigr],
	\]
	and
	\[
	\mathcal L^\dagger(X)
	= \mathrm i[H,X]
	+ \sum_k \gamma_k\Bigl(
	L_k^\dagger X L_k
	- \frac12\{L_k^\dagger L_k,X\}
	\Bigr).
	\]
	The Heisenberg picture is often more convenient for discussing conserved quantities and energy dissipation.
	In this work we mainly operate in the Schr\"odinger picture, but the above duality is implicitly used when we introduce energy densities and observables.
\end{remark}

\begin{definition}[Stationary states, spectral gap, and primitivity]\label{def:primitive_gap}
	Let $\mathcal L$ be a GKSL generator.
	\begin{enumerate}
		\item A stationary state (fixed point) is a density operator $\sigma$ satisfying
		\(
		\mathcal L(\sigma)=0.
		\)
		If the stationary state is unique and full-rank, the semigroup is called primitive.
		\item Viewing $\mathcal L$ as a linear operator on $\mathcal B(\mathcal H)$, its spectrum typically satisfies
		\(
		\lambda_1=0,\ \Re\lambda_k\le0.
		\)
		The spectral gap is defined as
		\[
		\lambda := \min_{k\ge2} |\Re\lambda_k| > 0.
		\]
		In the primitive case, $\lambda>0$ implies exponential mixing:
		\(
		\|\Lambda_t(\rho)-\sigma\|_1\le C e^{-\lambda t}
		\)
		for some constant $C$ and all initial states $\rho$.
	\end{enumerate}
\end{definition}

\subsection{Statistical query learning model}\label{subsec:SQ_prelim}

In this section, we briefly review the basic concepts and notation used in the statistical query (SQ) learning framework~\cite{nietner2025average,feldman2017general}.

\begin{definition}[Distributions, TV distance, and expectation notation]\label{def:TV_and_expectation}
	Let $X$ be a finite domain and $\mathcal{D}_X$ the set of all probability distributions on $X$. For $P,Q\in\mathcal{D}_X$, their total variation (TV) distance is defined as
	\begin{equation}
		d_{\mathrm{TV}}(P,Q)
		:= \frac{1}{2}\sum_{x\in X} |P(x)-Q(x)|.
	\end{equation}
	Given $\epsilon>0$, the open $\epsilon$-ball centered at $P$ is defined as
	\begin{equation}
		B_\epsilon(P)
		:=\bigl\{Q\in\mathcal{D}_X:\ d_{\mathrm{TV}}(P,Q)<\epsilon\bigr\}.
	\end{equation}
	For any distribution $P$ and function $\varphi:X\to[-1,1]$, we write
	\begin{equation}
		P[\varphi]
		:= \mathbb{E}_{x\sim P}[\varphi(x)]
	\end{equation}
	for the expectation of $\varphi$ with respect to $P$. In particular, we use $\mathcal{D}_n$ to denote the class of distributions defined on $\{0,1\}^n$.
\end{definition}

\begin{definition}[Statistical query oracle]\label{def:SQ_oracle}
	Let $P\in\mathcal{D}_X$ and $\tau>0$. A statistical query oracle $\Stat_\tau(P)$ is an ideal interface that, for any query function $\varphi:X\to[-1,1]$, returns a real number $v$ satisfying
	\begin{equation}
		|v - P[\varphi]| \le \tau.
	\end{equation}
\end{definition}

\begin{remark}[Monotonicity and sample-based meaning of the tolerance parameter]\label{rem:tolerance}
	(1) If $0<\zeta<\tau$, then any $\Stat_\zeta(P)$ is automatically a valid $\Stat_\tau(P)$, since
	\(
	|v-P[\varphi]|\le\zeta<\tau
	\)
	also satisfies the condition in Definition~\ref{def:SQ_oracle}. Hence, lower bounds proved for a smaller tolerance $\zeta$ automatically apply to any larger tolerance $\tau$.\\
	(2) When $\tau$ is at least inverse-polynomial (for example, $\tau\ge 1/\mathrm{poly}(n)$), one can approximate $P[\varphi]$ from a polynomial number of samples and thereby implement $\Stat_\tau(P)$. In this regime, the complexity of the oracle is dominated by the cost of evaluating the query function $\varphi$.
\end{remark}

\begin{definition}[Representations and $\epsilon$-representations of distributions]\label{def:representation}
	In the SQ framework, one must fix a representation of distributions. Two common types of representations are:
	\begin{enumerate}
		\item \emph{Generator}: a randomized algorithm that, given random bits, outputs samples $x\sim P$.
		\item \emph{Evaluator}: when $P\in\mathcal{D}_n$, an algorithm
		\(
		\Eval_P:\{0,1\}^n\to[0,1]
		\)
		satisfying $\Eval_P(x)=P(x)$.
	\end{enumerate}
	Given a distribution $P$, if the distribution $Q$ corresponding to a given representation (generator or evaluator) satisfies
	\(
	d_{\mathrm{TV}}(P,Q)<\epsilon,
	\)
	we call this representation an \emph{$\epsilon$-representation} (or $\epsilon$-approximate representation) of $P$. In our complexity lower bounds, we impose no assumptions on the computational efficiency of these representations: the lower bounds apply equally to representations that may be computationally inefficient.
\end{definition}

\begin{definition}[$\epsilon$-learning problem in the SQ framework]\label{def:SQ_learning_problem}
	Let $\mathcal{D}$ be a class of distributions and $\epsilon,\tau\in(0,1)$. Under a fixed choice of representation for distributions, the task of $\epsilon$-learning $\mathcal{D}$ using statistical queries with tolerance $\tau$ is defined as follows: for any unknown distribution $P\in\mathcal{D}$, given access to the oracle $\Stat_\tau(P)$, design a learning algorithm that outputs an $\epsilon$-representation of $P$ (i.e., some $Q$ such that $d_{\mathrm{TV}}(P,Q)<\epsilon$).
\end{definition}

\begin{definition}[Average-case query complexity]\label{def:avg_case_complexity}
	Let $\mathcal{D}$ be a class of distributions, $\mu$ a probability measure over $\mathcal{D}$, and $\alpha,\beta\in(0,1)$.
	\begin{enumerate}
		\item Deterministic average-case complexity: For given $(\epsilon,\tau)$, a deterministic SQ learning algorithm $\mathcal{A}$ is said to $\epsilon$-learn $P$ within at most $q$ queries on average if
		\begin{equation}
			\Pr_{P\sim\mu}\bigl[
			\text{$\mathcal{A}^{\Stat_\tau(P)}$ $\epsilon$-learns $P$ within $q$ queries}
			\bigr]\ \ge\ \beta.
		\end{equation}
		Among all algorithms satisfying the above condition, the minimal such $q$ is called the deterministic average-case query complexity of Problem~\ref{def:SQ_learning_problem}.
		
		\item Randomized average-case complexity: If $\mathcal{A}$ is randomized and $\Pr_{\mathcal{A}}[\cdot]$ denotes the probability over its internal randomness, we require
		\begin{equation}
			\Pr_{P\sim\mu}\Bigl[
			\Pr_{\mathcal{A}}\bigl[
			\text{$\mathcal{A}^{\Stat_\tau(P)}$ $\epsilon$-learns $P$ within $q$ queries}
			\bigr] \ \ge\ \alpha
			\Bigr]\ \ge\ \beta .
		\end{equation}
		The minimal $q$ for which this condition holds is called the randomized average-case query complexity.
	\end{enumerate}
\end{definition}

\begin{remark}[Relationship between deterministic and randomized complexities]\label{rem:det_vs_rand}
	In the SQ framework, deterministic and randomized average-case complexities are related via the standard min-max principle: for given $(\alpha,\beta)$, a deterministic average-case lower bound can be converted into a randomized average-case lower bound, at the cost of a constant factor depending on $\alpha$ (for example, $2(\alpha-1/2)$). For this reason, the main text of this paper focuses primarily on deterministic average-case lower bounds, while the extension to the randomized setting can be obtained via standard arguments outlined in the appendix.
\end{remark}

These preliminaries allow us to state our SQ lower-bound results in a unified way: given a distribution class $\mathcal{D}$ and a prior measure $\mu$ over it, any SQ algorithm that attempts to $\epsilon$-learn $\mathcal{D}$ in the average-case sense (for parameters $(\alpha,\beta)$) must use at least $q$ statistical queries, where $q$ is lower-bounded in terms of certain information-theoretic quantities (such as the SQ dimension or related average-case hardness parameters).

\subsection{Reduction from learning to decision and average-case lower bounds}\label{subsec:SQ_decide}

In this subsection, we review the standard learning and decision framework used to establish SQ lower bounds due to Feldman~\cite{feldman2017general}. Our notation, however, mostly follow Nietner~\cite{nietner2025average}, and we adopt from their work a convenient average-case formulation.

\begin{definition}[Decision problem: $\mathcal{D}$ vs.\ $Q$, cf.~\cite{nietner2025average}]\label{def:decide_D_vs_Q}
	Let $\mathcal{D}\subseteq\mathcal{D}_X$ be a class of distributions and $Q\in\mathcal{D}_X$ a fixed reference distribution.
	The decision problem $\mathcal{D}$ vs.\ $Q$ is defined as follows:
	
	Given SQ oracle access to an unknown distribution $P\in\mathcal{D}\cup\{Q\}$, the task is to output, within finitely many queries, a correct judgment:
	\[
	\text{``$P=Q$'' or ``$P\in\mathcal{D}$''.}
	\]
	
	We continue to work with the statistical query oracle $\Stat_\tau(P)$ with tolerance $\tau$, as defined in
	Definition~\ref{def:SQ_oracle}.
\end{definition}

\begin{lemma}[Learning is no easier than decision, cf.~\cite{nietner2025average}]\label{lem:learning_hard_as_deciding}
	Let $\mathcal{D}\subseteq\mathcal{D}_X$ and $Q\in\mathcal{D}_X$ satisfy
	\begin{equation}\label{eq:TV_far_assumption}
		d_{\mathrm{TV}}(P,Q) > \epsilon+\tau,\qquad\forall\,P\in\mathcal{D},
	\end{equation}
	where $0<\tau\le\epsilon\le1$. If there exists a deterministic algorithm $\mathcal{A}$ that $\epsilon$-learns $\mathcal{D}$ with at most $q$
	$\tau$-accurate statistical queries (see Definition~\ref{def:SQ_learning_problem}),
	then there exists a deterministic algorithm $\mathcal{B}$ that solves the decision problem
	$\mathcal{D}$ vs.\ $Q$ (Definition~\ref{def:decide_D_vs_Q}) using at most $q+1$ $\tau$-accurate statistical queries, and is always correct.
\end{lemma}

\begin{lemma}[Deterministic lower bound for decision problems, cf.~\cite{nietner2025average}]\label{lem:hardness_deciding_det}
	Let $\mathcal{D}\subseteq\mathcal{D}_X$, $Q\in\mathcal{D}_X$, and let $\mu$ be a probability measure on $\mathcal{D}$.
	If there exists a deterministic algorithm $\mathcal{A}$ that correctly decides
	$\mathcal{D}$ vs.\ $Q$ using at most $q$ $\tau$-accurate statistical queries,
	then necessarily
	\begin{equation}\label{eq:decide_lower_bound}
		q \ \ge\ \left(
		\max_{\phi:X\to[-1,1]}
		\Pr_{P\sim\mu}\bigl[\,|P[\phi]-Q[\phi]|>\tau\,\bigr]
		\right)^{-1}.
	\end{equation}
\end{lemma}

\begin{theorem}[Deterministic average-case SQ query complexity lower bound, cf.~\cite{nietner2025average}]\label{thm:avg_case_SQ_lower}
	Let $\mathcal{D}\subseteq\mathcal{D}_X$, let $\mu$ be a probability measure on $\mathcal{D}$, and let $\epsilon,\tau\in(0,1)$, $\beta\in(0,1)$.
	Suppose there exists a deterministic algorithm $\mathcal{A}$ that, using at most $q$ $\tau$-accurate statistical queries, $\epsilon$-learns $P$ on a subset $\mathcal{D}'\subseteq\mathcal{D}$ of $\mu$-measure $\beta$ (that is,
	\(
	\Pr_{P\sim\mu}[\mathcal{A}\ \text{$\epsilon$-learns $P$}]\ge\beta
	\)).
	Then for any reference distribution $Q\in\mathcal{D}_X$ one has
	\begin{equation}\label{eq:avg_case_lower_bound}
		q+1 \ \ge\
		\frac{
			\beta - \Pr_{P\sim\mu}\bigl[d_{\mathrm{TV}}(P,Q)\le\epsilon+\tau\bigr]
		}{
			\displaystyle\max_{\phi:X\to[-1,1]}
			\Pr_{P\sim\mu}\bigl[\,|P[\phi]-Q[\phi]|>\tau\,\bigr]
		}.
	\end{equation}
\end{theorem}

\subsection{Random statistical query algorithms vs.\ deterministic algorithms}
\label{subsec:random_SQ}

In this subsection, we summarize the relation between random (including classical randomized and quantum) algorithms and deterministic algorithms in terms of average-case query complexity in the statistical query learning framework. The presentation is essentially adapted from Appendix~F of~\cite{nietner2025average}.

Recall from the previous subsection that we have already established an average-case lower bound on the query complexity in the deterministic setting (corresponding to Theorem~30 in~\cite{nietner2025average}), whose core idea is to reduce the learning problem to a decision problem and then lower-bound the SQ complexity of that decision problem. The strategy in the randomized setting is completely analogous, except that we additionally need to handle the possibility that a learning algorithm may guess the correct answer using its internal randomness.

\begin{definition}[Random average-case statistical query complexity]
	Let $\mathcal D$ be a class of distributions and $\mu$ a probability measure on $\mathcal D$, and let $\epsilon,\tau\in(0,1)$.
	Given a statistical query oracle $\Stat_\tau(P)$ (as defined earlier), a randomized learning algorithm $\mathcal A$ (which may be a classical randomized or a quantum algorithm) is said to \emph{randomly $\epsilon$-learn $\mathcal D$ in the average case at parameters $(\alpha,\beta)\in(0,1)^2$} if there exists a query bound $q$ such that
	\[
	\Pr_{P\sim\mu}\!\left[
	\Pr_{\mathcal A}\!\left(
	\text{$\mathcal A^{\Stat_\tau(P)}$ $\epsilon$-learns $P$ using at most $q$ queries}
	\right)\ \ge \ \alpha
	\right] \ \ge\ \beta,
	\]
	where the outer probability is taken over $P\sim\mu$, and the inner probability is taken over the internal randomness of $\mathcal A$.
	The smallest such $q$ is called the random average-case statistical query complexity of $\mathcal D$ under $\mu$ and parameters $(\alpha,\beta)$.
\end{definition}

\begin{lemma}[Learning is at least as hard as decision (random setting), cf.~\cite{nietner2025average}]
	\label{lem:random_learn_as_hard_as_decide}
	Let $\mathcal D$ be a class of distributions and $Q$ a reference distribution such that for all $P\in\mathcal D$ one has
	\[
	d_{\mathrm{TV}}(P,Q)>\epsilon+\tau,
	\]
	where $0<\tau\le\epsilon\le1$. If there exists a randomized algorithm $\mathcal A$ which, with internal success probability at least $\alpha$, can $\epsilon$-learn $\mathcal D$ using at most $q$ $\tau$-accurate statistical queries, then there exists a randomized algorithm that, with the same success probability $\alpha$, can decide between
	``$P=Q$'' and ``$P\in\mathcal D$'' using only $q+1$ $\tau$-accurate statistical queries.
\end{lemma}

\begin{lemma}[SQ lower bound for random algorithms deciding $\mathcal D$ vs.\ $Q$, cf.~\cite{nietner2025average}]
	\label{lem:random_decision_hard}
	Let $\mathcal A$ be a randomized algorithm that uses at most $q$ $\tau$-accurate statistical queries to decide
	between $\mathcal D$ and $Q$, and outputs the correct answer with internal success probability at least $\alpha>1/2$.
	Then for any probability measure $\mu$ on $\mathcal D$,
	\[
	q \ \ge\
	\frac{2(\alpha-\tfrac12)}
	{\displaystyle\max_{\phi:X\to[-1,1]}
		\Pr_{P\sim\mu}\bigl[\bigl|P[\phi]-Q[\phi]\bigr|>\tau\bigr]}.
	\]
\end{lemma}

\begin{theorem}[Random average-case SQ query complexity lower bound, cf.~\cite{nietner2025average}]
	\label{thm:random_avg_SQ_complexity}
	Let $\mathcal A$ be a randomized learning algorithm which, under measure $\mu$ and parameters $(\alpha,\beta)$, solves the $\epsilon$-learning problem using at most $q$ $\tau$-accurate statistical queries (i.e., it succeeds on a subset of distributions of $\mu$-measure $\beta$, and its internal success probability is at least $\alpha$).
	Then for any reference distribution $Q$,
	\[
	q+1
	\ \ge\
	2\cdot
	\frac{\bigl(\alpha-\tfrac12\bigr)\cdot
		\Bigl(\beta-\Pr_{P\sim\mu}\bigl[d_{\mathrm{TV}}(P,Q)\le\epsilon+\tau\bigr]\Bigr)}
	{\displaystyle\max_{\phi:X\to[-1,1]}
		\Pr_{P\sim\mu}\bigl[\bigl|P[\phi]-Q[\phi]\bigr|>\tau\bigr]}.
	\]
\end{theorem}

\subsection{QPStat model and general lower-bound}

In this subsection we recall the QPStat model and the general lower-bound template we need, mainly following~\cite{wadhwa2024noise}.

\begin{definition}[Quantum process statistical query (QPStat) oracle]\label{def:QPStat_oracle}
	Let $E:\mathcal B(\mathcal H)\to\mathcal B(\mathcal H)$ be a CPTP channel, and let $\mathsf S(\mathcal H)$ denote the state space.
	For any state $\rho\in\mathsf S(\mathcal H)$ and operator $O\in\mathcal B(\mathcal H)$ with $\|\rho\|_1\le1$ and $\|O\|_\infty\le1$, and for any tolerance parameter $\tau>0$, the QPStat oracle is defined as
	\[
	QPStat_E:\ (\rho,O,\tau)\ \longmapsto\ \alpha\in\mathbb R,
	\]
	and is required to output $\alpha$ satisfying
	\begin{equation}\label{eq:QPStat_def}
		\bigl|\alpha-\operatorname{Tr}\bigl[O\,E(\rho)\bigr]\bigr|\le\tau .
	\end{equation}
	A QPStat learning algorithm interacting with an unknown channel $E$ may adaptively query this oracle at most $q$ times; we say it uses $q$ QPStat queries with accuracy~$\tau$.
\end{definition}

In this work we are interested in the average-case QPStat learning complexity:
given a channel ensemble $(\mathcal E,\mu)$ and error parameters $(\varepsilon,\tau)$, we would like to use as few QPStat queries as possible so that, on a $\mu$-measure $\beta$ subset of channels, the algorithm outputs a hypothesis channel whose diamond distance from the true channel is at most $\varepsilon$, with internal success probability at least $\alpha$.
When we take the special case $\beta=1$, we recover the worst-case setting where the learner must succeed on all channels in the support of~$\mu$; in this case the lower bound in Lemma~\ref{lem:qpstat_many_vs_one} naturally simplifies to $q+1\ge (2\alpha-1)/\Gamma$.

Formally, for a family of CPTP channels $\mathcal E$ equipped with a measure $\mu$, we say a learning algorithm $\mathcal A$ satisfies:
\begin{itemize}
	\item[1.] It uses at most $q$ QPStat queries with accuracy~$\tau$;
	\item[2.] There exist constants $\alpha>1/2$ and $\beta>0$ such that
	\[
	\mu\Bigl(
	\Bigl\{E\in\mathcal E:\ 
	\Pr_{\text{internal randomness}}\bigl[d_\diamond(\widehat E,E)\le\varepsilon\bigr]
	\ge\alpha
	\Bigr\}
	\Bigr)\ \ge\ \beta ,
	\]
\end{itemize}
in which case we say that $\mathcal A$ QPStat-learns the channel ensemble $(\mathcal E,\mu)$ with parameters $(\varepsilon,\tau,\alpha,\beta)$ in the average-case sense.

\begin{lemma}[QPStat many-vs-one lower-bound template, cf.~\cite{wadhwa2024noise}]\label{lem:qpstat_many_vs_one}
	Let $(\mathcal E,\mu)$ be a family of CPTP channels acting on a fixed Hilbert space $\mathcal H$, together with a probability measure~$\mu$.
	Let $\bar E$ be some reference channel (in applications we will take $\bar E=\mathbb E_{E\sim\mu}E$).
	
	Consider the following QPStat learning problem:
	a (possibly randomized) learning algorithm $\mathcal A$ receives QPStat access to a random channel $E\sim\mu$, i.e., to the oracle $QPStat_E$, and may use at most $q$ QPStat queries with accuracy~$\tau$ before outputting a hypothesis channel~$\widehat E$.
	
	If there exist constants $\alpha>1/2$ and $\beta>0$ such that, for
	\[
	\mathcal G
	:=\{E\in\mathcal E:\ 
	\Pr_{\text{internal randomness}}\bigl[d_\diamond(\widehat E,E)\le\varepsilon\bigr]
	\ge\alpha
	\},
	\]
	one has $\mu(\mathcal G)\ge\beta$, then necessarily
	\begin{equation}\label{eq:qpstat_many_vs_one}
		q+1
		\ \ge\
		\frac{(2\alpha-1)\,\beta}{
			\displaystyle
			\max_{\rho,O}
			\Pr_{E\sim\mu}
			\Bigl(
			\bigl|\operatorname{Tr}[O E(\rho)]
			-\operatorname{Tr}[O\bar E(\rho)]\bigr|
			>\tau\Bigr)
		} ,
	\end{equation}
	where the maximum is taken over all pairs $(\rho,O)$ satisfying
	$\|\rho\|_1\le1$ and $\|O\|_\infty\le1$.
\end{lemma}

\begin{remark}[Without loss of generality we assume $\tau \leq \varepsilon$]
Throughout this work we restrict to the regime $0<\tau \leq \varepsilon$. If $\tau>\varepsilon$, then there can exist channels (or distributions) with $\varepsilon<d(\cdot, \cdot)<\tau$ that are indistinguishable using $\tau$-accurate query answers, so no $\varepsilon$-learner can exist in that parameter regime. Thus assuming $\tau \leq \varepsilon$ is without loss of generality for our lower-bound statements~\cite{nietner2025average}.
\end{remark}

\subsection{Fréchet derivatives and perturbation formulas}\label{subsec:frechet_duhamel}

In this paper we need to relate small changes in a generator to the corresponding change in the semigroup $e^{t\mathcal L}$.
The natural language for this is the Fréchet derivative on Banach spaces and the associated perturbation-type formulas.

\begin{definition}[Fréchet derivative on Banach spaces\cite{yosida2012functional}]\label{def:frechet}
	Let $(E,\|\cdot\|_E),(F,\|\cdot\|_F)$ be Banach spaces, and $U\subset E$ an open set.
	A map $F:U\to F$ is said to be \emph{Fréchet differentiable} at a point $x\in U$ if there exists a bounded linear operator $DF(x):E\to F$ such that
	\begin{equation}
		\lim_{\|h\|_E\to0}
		\frac{\|F(x+h)-F(x)-DF(x)[h]\|_F}{\|h\|_E} = 0.
	\end{equation}
	In this case $DF(x)$ is unique and is called the Fréchet derivative of $F$ at $x$.
	If $F$ is Fréchet differentiable everywhere on $U$ and the map $x\mapsto DF(x)$ is continuous, we say $F\in C^1(U;F)$.
\end{definition}

In the finite-dimensional case (e.g., $E=\mathbb R^n,F=\mathbb R^m$), the Fréchet derivative reduces to the familiar Jacobian matrix.
In the infinite-dimensional case, it provides a canonical version of the first-order linear approximation,
\(
F(x+h)\approx F(x)+DF(x)[h],
\)
see any functional analysis textbook, e.g. ~\cite{yosida2012functional}.

The main Banach space of interest in this paper is
\[
E = \mathsf B(\mathcal B(\mathcal H)),
\]
the space of all bounded superoperators $\mathcal L:\mathcal B(\mathcal H)\to\mathcal B(\mathcal H)$, equipped with the induced $1\to1$ norm
\(
\|\mathcal L\|_{1\to1}
:= \sup_{X\neq0}\frac{\|\mathcal L(X)\|_1}{\|X\|_1}
\)
(see Definition~\ref{def:induced_norms}).
On this space we consider the exponential map
\[
\Exp_t:\ \mathcal L\ \mapsto\ e^{t\mathcal L}
:= \sum_{n=0}^\infty \frac{t^n}{n!}\,\mathcal L^n,
\qquad t\ge0,
\]
whose values are again bounded superoperators $e^{t\mathcal L}:\mathcal B(\mathcal H)\to\mathcal B(\mathcal H)$.

\begin{proposition}[Fréchet differentiability of the operator exponential]\label{prop:exp_frechet}
	Let $E$ be an arbitrary Banach space and $\mathsf B(E)$ the space of its bounded linear operators.
	For any fixed $t\ge0$, the map
	\[
	\Exp_t:\mathsf B(E)\to\mathsf B(E),\qquad
	A\mapsto e^{tA}
	\]
	is a $C^\infty$ (indeed analytic) map on the Banach space $\mathsf B(E)$.
	Its Fréchet derivative at $A\in\mathsf B(E)$, $D\Exp_t(A):\mathsf B(E)\to\mathsf B(E)$, acts on any perturbation $\Delta A\in\mathsf B(E)$ as
	\begin{equation}\label{eq:frechet_exp_general}
		D\Exp_t(A)[\Delta A]
		= \int_0^t e^{(t-s)A}\,\Delta A\,e^{sA}\,\mathrm ds.
	\end{equation}
\end{proposition}

\begin{proof}[Proof sketch]
	First, using the power series
	$e^{tA}=\sum_{n\ge0}t^nA^n/n!$, which converges uniformly in the operator norm, one sees that $\Exp_t$ is analytic on $\mathsf B(E)$.
	Let $A(\alpha)=A+\alpha\Delta A$.
	Applying the perturbation formula to $e^{tA(\alpha)}$ (see Lemma~\ref{lem:duhamel}), and considering
	\[
	e^{t(A+\delta\Delta A)} - e^{tA}
	= \int_0^t e^{(t-s)(A+\delta\Delta A)}\,\delta\Delta A\,e^{sA}\,\mathrm ds,
	\]
	dividing both sides by $\delta$ and letting $\delta\to0$, using the continuity of $e^{(t-s)(A+\delta\Delta A)}\to e^{(t-s)A}$ in the operator norm, one obtains the limit form \eqref{eq:frechet_exp_general}.
	A fully rigorous Banach-space proof can be found in~\cite{engel2000one}.
\end{proof}

We will use the above Fréchet derivative of the exponential map given by the perturbation formula in the main text, i.e., for each parameter direction $\Delta A$, the derivative $D\Exp_t(A)[\Delta A]$ is given by an operator-valued integral.

\medskip

\begin{lemma}[Perturbation formula]\label{lem:duhamel}
	Let $A,B\in\mathsf B(E)$ and $t\ge0$. Then
	\begin{equation}\label{eq:duhamel_basic}
		e^{t(A+B)} - e^{tA}
		= \int_0^t e^{(t-s)(A+B)}\,B\,e^{sA}\,\mathrm ds.
	\end{equation}
	More generally, if $A(\alpha)$ is a family of bounded operators that is Fréchet differentiable, with derivative $A'(\alpha)$, then
	\begin{equation}\label{eq:duhamel_derivative}
		\frac{\mathrm d}{\mathrm d\alpha}e^{tA(\alpha)}
		= \int_0^t e^{(t-s)A(\alpha)}\,A'(\alpha)\,e^{sA(\alpha)}\,\mathrm ds.
	\end{equation}
\end{lemma}

\begin{proof}[Proof sketch]
	The basic form \eqref{eq:duhamel_basic} can be obtained by considering
	\(
	U(s):=e^{(t-s)(A+B)}e^{sA}.
	\)
	Differentiating with respect to $s$ gives
	\[
	\frac{\mathrm d}{\mathrm ds}U(s)
	= -e^{(t-s)(A+B)}(A+B)e^{sA}
	+e^{(t-s)(A+B)}Ae^{sA}
	= -e^{(t-s)(A+B)}Be^{sA}.
	\]
	Integrating $s\in[0,t]$ yields
	\(
	U(t)-U(0)=-\int_0^t e^{(t-s)(A+B)}Be^{sA}\,\mathrm ds,
	\)
	i.e.,
	\(
	e^{tA}-e^{t(A+B)}
	=-\int_0^t \cdots,
	\)
	which rearranges to \eqref{eq:duhamel_basic}.
	
	For \eqref{eq:duhamel_derivative}, replace $B$ by a small perturbation $\delta A'(\alpha)$, apply \eqref{eq:duhamel_basic}, divide by $\delta$, and then let $\delta\to0$.
	Details can be found in~\cite{najfeld1995derivatives}; perturbations of semigroups generated by operators in Banach spaces are treated in~\cite{kato2013perturbation}.
\end{proof}

\begin{remark}[Specialized form used in this paper]\label{rem:duhamel_in_paper}
	In this paper we take
	\(
	E=\mathcal B(\mathcal H),
	\)
	and $A(\alpha)=\mathcal L_\alpha$ to be a straight-line homotopy connecting two points $\mathcal L,\mathcal L'$:
	\(
	\mathcal L_\alpha=\mathcal L+\alpha(\mathcal L'-\mathcal L),
	\)
	so that $A'(\alpha)=\mathcal L'-\mathcal L$.
	Substituting into \eqref{eq:duhamel_derivative} gives exactly
	\begin{equation}
		\frac{\mathrm d}{\mathrm d\alpha}e^{t\mathcal L_\alpha}
		= \int_0^t e^{(t-s)\mathcal L_\alpha}\,(\mathcal L'-\mathcal L)\,e^{s\mathcal L_\alpha}\,\mathrm ds,
	\end{equation}
	which is precisely the Fréchet derivative representation used in Lemma~\ref{lem:L_to_Lambda_Lipschitz}.
	Estimating this integral in the induced $1\to1$ norm then yields a Lipschitz bound of the form
	\(
	\|\Lambda_t^{(\mathcal L')}-\Lambda_t^{(\mathcal L)}\|_{1\to1}
	\le t\|\mathcal L'-\mathcal L\|_{1\to1}.
	\)
\end{remark}

\subsection{Haar measure, concentration, and $2$-norm Lipschitz constants}\label{subsec:haar_concentration}

This subsection, following~\cite{ledoux2001concentration}, reviews Haar measure on compact groups, explains the Lévy concentration on the parameter sphere, and provides the $2$-norm Lipschitz constant estimates that we need later on.

\begin{definition}[Haar measure on a compact group]
	Let $G$ be a second countable compact topological group and $\mathcal B(G)$ its Borel $\sigma$-algebra.
	A probability measure
	$\mu_{\mathrm{Haar}}$ is called a (normalized) Haar measure on $G$ if
	\[
	\mu_{\mathrm{Haar}}(gA)=\mu_{\mathrm{Haar}}(Ag)=\mu_{\mathrm{Haar}}(A),
	\quad \forall\,g\in G,\ A\in\mathcal B(G).
	\]
	Haar measure exists and is unique.
	The most important examples for this paper are the unitary groups
	$G=\mathrm U(d)$ and $\mathrm{SU}(d)$; in these cases we use $\mu_{\mathrm{Haar}}$ to denote the normalized Haar probability measure.
\end{definition}

On the matrix group $\mathrm U(d)$ we choose the metric induced by the Hilbert-Schmidt norm as our geometric background. A discussion of quantum case can be found in ~\cite{brandao2021models}.

\begin{definition}[Hilbert-Schmidt norm and $2$-norm metric]\label{def:HS_norm}
	For any matrix $A\in\mathbb C^{d\times d}$, define the Hilbert-Schmidt norm by
	\[
	\|A\|_2
	:=\bigl(\operatorname{Tr}(A^\dagger A)\bigr)^{1/2}.
	\]
	On the unitary group $G=\mathrm U(d)$, we use
	\[
	d_2(U,V):=\|U-V\|_2,\qquad U,V\in\mathrm U(d)
	\]
	as the metric and call it the $2$-norm metric.
	For a real-valued function
	$F:G\to\mathbb R$, we denote its Lipschitz constant with respect to $d_2$ by
	\[
	\operatorname{Lip}_2(F)
	:=\sup_{U\neq V}\frac{|F(U)-F(V)|}{d_2(U,V)}.
	\]
	When $\operatorname{Lip}_2(F)\le1$ we say that $F$ is $1$-Lipschitz with respect to the $2$-norm.
\end{definition}

In high dimensions, $(G,d_2,\mu_{\mathrm{Haar}})$ forms a typical Lévy family, meaning that any $1$-Lipschitz function exhibits strong concentration under Haar measure.
Systematic treatments can be found in Ledoux's monograph on concentration inequalities~\cite{ledoux2001concentration}.

\begin{definition}[Lévy concentration under Haar measure]\label{def:haar_levy}
	Let $(G,d_2,\mu_{\mathrm{Haar}})$ be the metric probability space as above.
	For any Borel measurable function $F:G\to\mathbb R$, a number $m_F\in\mathbb R$ is called a
	median of $F$ with respect to the Haar measure if
	\[
	\mu_{\mathrm{Haar}}\bigl(\{F\le m_F\}\bigr)\ge\frac12,\qquad
	\mu_{\mathrm{Haar}}\bigl(\{F\ge m_F\}\bigr)\ge\frac12.
	\]
	We say that $(G,d_2,\mu_{\mathrm{Haar}})$ has sub-Gaussian concentration if there exist constants
	$c_G,C_G>0$ such that for any function $F$ that is $L$-Lipschitz with respect to the $2$-norm, and any $r>0$, we have
	\begin{equation}\label{eq:haar_levy_concentration}
		\mu_{\mathrm{Haar}}\bigl(\{|F-m_F|\ge r\}\bigr)
		\;\le\;2\exp\!\Bigl(-c_G\,\frac{r^2}{L^2}\Bigr).
	\end{equation}
	In particular, when $G=\mathrm U(d)$ or $\mathrm{SU}(d)$, one can take
	$c_G=\Theta(d^2)$, reflecting the high-dimensional effect coming from the manifold dimension $\dim G=\Theta(d^2)$.
\end{definition}

\begin{remark}[From the sphere to the unitary group]
	A standard proof strategy is to view $G=\mathrm U(d)$ as a smooth submanifold embedded into
	$\mathbb C^{d\times d}\cong\mathbb R^{2d^2}$, whose Riemannian metric is equivalent to the Hilbert-Schmidt norm.
	Consider a large sphere
	$S^{N-1}\subset\mathbb R^N$ (with $N\sim2d^2$) containing $G$, equipped with the uniform measure.
	Using the spherical Lévy lemma: for any $1$-Lipschitz function
	$f:S^{N-1}\to\mathbb R$,
	\[
	\Pr\bigl(|f-m_f|\ge r\bigr)
	\;\le\;2\exp\bigl(-cN r^2\bigr),
	\]
	and combining it with an equivariant embedding between Haar measure on $G$ and the spherical measure and the corresponding Lipschitz contraction, one obtains a concentration bound of the form \eqref{eq:haar_levy_concentration} (with constants changed only numerically).
\end{remark}

In this paper when we analyze the parameter sphere, we mainly use the following lemma.

\begin{lemma}[Lévy concentration on the parameter sphere, cf.~\cite{ledoux2001concentration}]\label{lem:levy_sphere}
	Let $S^{M-1}$ be the real sphere and $\mu_S$ is standard uniform measure.
	For any $L$-Lipschitz function
	\(
	f:S^{M-1}\to\mathbb R
	\)
	and any $\tau>0$, there exists a constant $c_{\mathrm{par}}>0$ (independent of $M$, $f$, and $\tau$) such that
	\[
	\Pr_{\theta\sim\mu_S}\Bigl(|f(\theta)-\mathbb E_{\mu_S}f|\ge\tau\Bigr)
	\le 2\exp\Bigl(-c_{\mathrm{par}}\,\frac{M\,\tau^2}{L^2}\Bigr).
	\]
\end{lemma}

\section{Random Lindbladian ensembles}\label{zhang1}

\subsection{Linear parametrization and concentration for random Lindbladian ensembles}

In this subsection we formalize several points, first, for any fixed local dimension the family of Lindbladians can be embedded into a finite-dimensional real vector space; moreover, by choosing a bounded basis $\{G_j\}$ on this vector space, one can linearly parametrize Lindbladians; in turn, by endowing the parameter space with a natural random ensemble , we obtain concentration for SQ test functions
\(
F_\varphi(\mathcal L):=P_{\mathcal L}[\varphi]
\),
which yields an upper bound on
\(
\operatorname{frac}(\mu_{\mathcal L},Q,\tau)
\)
in the open-system SQ framework; finally, we discuss how to interpret this parametrization in terms of whether it ``covers all Lindbladians'' and in terms of universality for learning.

Here we assume the underlying Hilbert space is finite dimensional, $\mathcal H\simeq\mathbb C^d$. Denote by $\mathcal B(\mathcal H)$ the space of bounded operators and by $\mathrm{End}(\mathcal B(\mathcal H))$ the space of super-operators.

\begin{definition}[GKSL cone and affine hull]
	Let
	\[
	\mathsf{GKSL}
	\subset \mathrm{End}(\mathcal B(\mathcal H))
	\]
	denote the set of all Lindbladian generators satisfying the Gorini-Kossakowski-Sudarshan-Lindblad (GKSL) conditions, i.e.,
	\begin{itemize}
		\item[i.] $\mathcal L$ is Hermitian-preserving and trace-preserving;
		\item[ii.] $e^{t\mathcal L}$ generates a family of CPTP maps (a quantum Markov semigroup) for all $t\ge0$.
	\end{itemize}
	Let $\mathrm{Aff}(\mathsf{GKSL})$ be its affine hull (the smallest affine subspace containing $\mathsf{GKSL}$), and define
	\[
	\mathsf V := \mathrm{Aff}(\mathsf{GKSL}) - \mathcal L_{\mathrm{ref}}
	\]
	to be the linear space obtained by translating by some reference Lindbladian $\mathcal L_{\mathrm{ref}}\in\mathsf{GKSL}$ and taking this as the origin.
\end{definition}

Since $\mathrm{End}(\mathcal B(\mathcal H))$ is a finite-dimensional real vector space in the finite-dimensional case (with dimension $\mathcal O(d^4)$), we have:

\begin{lemma}[Completeness of linear parametrization]\label{lem:affine_span_L}
	For fixed local dimension $d$, there exist finitely many super-operators
	\(
	G_1,\dots,G_M\in \mathrm{End}(\mathcal B(\mathcal H))
	\)
	forming a linear basis of $\mathsf V$ such that:
	\begin{enumerate}
		\item For any $\mathcal L\in \mathsf{GKSL}$, there exists a unique coefficient vector
		\(
		\theta(\mathcal L)=(\theta_1,\dots,\theta_M)\in\mathbb R^M
		\)
		with
		\[
		\mathcal L
		= \mathcal L_{\mathrm{ref}}
		+ \sum_{j=1}^M \theta_j(\mathcal L)\,G_j.
		\]
		\item There exists a convex cone
		\(
		\mathcal C_{\mathrm{GKSL}}\subset\mathbb R^M
		\)
		such that
		\(
		\mathcal L\in\mathsf{GKSL}
		\)
		if and only if
		\(
		\theta(\mathcal L)\in\mathcal C_{\mathrm{GKSL}}.
		\)
	\end{enumerate}
\end{lemma}

\begin{proof}
	$\mathrm{End}(\mathcal B(\mathcal H))$ is a finite-dimensional real vector space and $\mathrm{Aff}(\mathsf{GKSL})$ is an affine subspace. Fixing a reference point $\mathcal L_{\mathrm{ref}}\in\mathsf{GKSL}$ and translating, we obtain the linear subspace
	\(
	\mathsf V=\mathrm{Aff}(\mathsf{GKSL})-\mathcal L_{\mathrm{ref}}.
	\)
	As $\mathsf V$ is finite-dimensional, there exists a finite basis $\{G_j\}_{j=1}^M$. For any
	$\mathcal L\in\mathsf{GKSL}$ we have
	\(
	\mathcal L-\mathcal L_{\mathrm{ref}}\in\mathsf V
	\), hence there exists a unique coefficient vector $\theta(\mathcal L)$ such that
	\[
	\mathcal L - \mathcal L_{\mathrm{ref}} = \sum_j \theta_j(\mathcal L) G_j.
	\]
	Let the image of the linear coordinate map
	\(
	\mathcal L\mapsto\theta(\mathcal L)
	\)
	be denoted by $\mathcal C_{\mathrm{GKSL}}$. By convexity of the GKSL conditions this image is a convex cone.
\end{proof}

\begin{remark}[On covering all Lindbladians]
	The lemma above shows that for fixed $d$,
	in a suitable linear coordinate system every Lindbladian can be written as
	\[
	\mathcal L = \mathcal L_{\mathrm{ref}}+\sum_{j=1}^M \theta_j G_j,
	\]
	with the caveat that admissible parameters $\theta$ must lie in a convex cone $\mathcal C_{\mathrm{GKSL}}$. In other words, from a purely linear-algebraic point of view, the linear parametrization is complete on the affine hull; but once we take CPTP constraints into account, only those coefficients inside the cone correspond to genuine Lindbladians.
\end{remark}

Next, in order to apply Lévy concentration(Lemma~\ref{lem:levy_sphere}) on a high-dimensional parameter sphere, we select from $\mathcal C_{\mathrm{GKSL}}$ a geometrically ``round'' high-dimensional subset.

\begin{definition}[Parameter-sphere ensemble]\label{def:param_sphere_ensemble}
	Let
	\[
	S^{M-1}:=\{\theta\in\mathbb R^M:\ \|\theta\|_2=1\}
	\]
	be the unit sphere. Given a small parameter $\delta>0$, define
	\[
	\mathcal L(\theta)
	:= \mathcal L_{\mathrm{ref}}
	+ \frac{\delta}{\sqrt M}\sum_{j=1}^M \theta_j G_j.
	\]
	Suppose there exists a nonempty subset
	\(
	\Theta\subseteq S^{M-1}
	\)
	such that for all $\theta\in\Theta$,
	\(
	\mathcal L(\theta)\in\mathsf{GKSL}
	\)
	(for instance, when $\delta$ is small enough that $\mathcal L(\theta)$ stays in a small GKSL neighbourhood of $\mathcal L_{\mathrm{ref}}$).
	Define the parameter measure $\mu_\Theta$ to be the normalized restriction of the spherical measure to $\Theta$. The random Lindbladian ensemble is then defined as the pushforward measure
	\[
	\mu_{\mathcal L} := \mathcal L_\#(\mu_\Theta),
	\]
	i.e., drawing
	\(
	\mathcal L\sim\mu_{\mathcal L}
	\)
	is equivalent to sampling
	\(
	\theta\sim\mu_\Theta
	\)
	and setting
	\(
	\mathcal L=\mathcal L(\theta).
	\)
\end{definition}

\begin{remark}[Local slice and high-dimensional neighbourhood]
	The construction above selects only a high-dimensional neighbourhood of $\mathcal L_{\mathrm{ref}}$ inside the GKSL cone, rather than the entire set $\mathsf{GKSL}$. Thus the random ensemble $\mu_{\mathcal L}$ is distribution-specific and it describes high-dimensional random perturbations around a physically natural Lindbladian model family, not all Lindbladians.
\end{remark}

\begin{remark}[Spherical random-perturbation Lindbladian ensemble]\label{ball lin}
	In the later discussion of random Lindbladian ensembles we will frequently use the parametrization
	\begin{equation}\label{eq:spherical_random_lindbladian_ensemble}
		\mathcal{L}(\theta)
		:= \mathcal{L}_{\mathrm{ref}} + \frac{\delta}{\sqrt{M}} \sum_{j=1}^M \theta_j G_j,
		\qquad \theta \in S^{M-1} \subset \mathbb{R}^M,
	\end{equation}
	where $\mathcal{L}_{\mathrm{ref}}$ is a fixed GKSL generator (the reference Lindbladian), $\{G_j\}_{j=1}^M$ are bounded super-operator directions (for instance those induced by local jump operators, local noise channels, or a basis of Kossakowski matrices; concrete examples will be given in Sec.~\ref{zhang44}), $\theta$ is uniformly distributed on the unit sphere $S^{M-1}$, and $\delta>0$ controls the perturbation strength(For making the line segments staying inside GKSL, see Lemma ~\ref{lem:convexity_GKSL_slice} for proof). We refer to the random Lindbladian family defined by \eqref{eq:spherical_random_lindbladian_ensemble} as the
	spherical random-perturbation Lindbladian ensemble.
	
	From a physical viewpoint, this parametrization can be understood as follows, in a high-dimensional parameter space we add an isotropic small perturbation around some ``average dynamics'' $\mathcal{L}_{\mathrm{ref}}$, where
	$\theta_j$ encodes the random amplitude along different noise/disorder directions, and the $1/\sqrt{M}$ rescaling ensures that the total squared magnitude of the perturbation remains of order $O(\delta^2)$ as $M\to\infty$, analogous to the $1/\sqrt{N}$ normalization in Wigner random matrix models. This choice is also consistent with the spectral rescaling in Denisov’s work~\cite{denisov2019universal}; here, however, our motivation is to ensure that the probabilistic bounds we derive below exhibit effective concentration.
	
	Such ensembles naturally arise in the following physical models:
	\begin{enumerate}
		\item Open-system dynamics in random spin chains:
		For example, in XXZ/XYZ chain models with random local fields and dephasing, sample-to-sample variations in the random field strengths $h_i$, couplings $J_i$, and Lindblad jump rates can all be viewed as different components of the parameter vector $\theta$, corresponding to a set of local super-operator directions $\{G_j\}$, and thus yield a multi-parameter static-disorder ensemble around some average Lindbladian~\cite{monthus2017dissipative}.
		
		\item Structured perturbations of random Kossakowski matrices:
		In GKSL form the noise part is determined by a Kossakowski matrix $K$. Expanding $K$ in a Hermitian basis $\{E_j\}$,
		\(
		K(\theta) = K_{\mathrm{ref}} + \frac{\delta}{\sqrt{M}} \sum_j \theta_j E_j,
		\)
		and mapping the basis $\{E_j\}$ to super-operators $\{G_j\}$ via the Lindblad formula, one obtains an isotropic random-perturbation ensemble around a reference matrix $K_{\mathrm{ref}}$. This may be viewed as a structured version of the fully random Kossakowski ensembles considered in the random Lindblad operator literature~\cite{denisov2019universal}.
		
		\item Multi-channel noise disorder in large-scale quantum processors:
		On platforms such as superconducting qubits or trapped ions, the noise strengths of different physical qubits and coupling channels (e.g., $T_1/T_2$ times~\cite{arute2019quantum,carroll2022dynamics}, dephasing rates~\cite{bruzewicz2019trapped}, leakage rates~\cite{wu2022erasure}) typically exhibit static disorder across samples and within a single chip.
		Incorporating the average noise model into $\mathcal{L}_{\mathrm{ref}}$ and collecting the deviations from the average on each qubit/channel into the directions $\{G_j\}$ yields a high-dimensional random Lindbladian ensemble of the form~\eqref{eq:spherical_random_lindbladian_ensemble}, which captures device-level multi-parameter noise mismatches.
	\end{enumerate}
	Thus \eqref{eq:spherical_random_lindbladian_ensemble} is not merely a toy model, rather, it provides an abstract description of a broad class of multi-parameter static-disorder Lindbladian ensembles widely encountered in open quantum many-body systems and quantum-information hardware. In our later analysis of Lipschitz properties across $\mathcal{L}$ and high-dimensional concentration, this ensemble will serve as our canonical example.
\end{remark}

\begin{lemma}[Convex GKSL slice and line segments staying inside GKSL]
	\label{lem:convexity_GKSL_slice}
	In the construction above there exists $\delta_0>0$ such that for any $0<\delta\le\delta_0$:
	\begin{enumerate}
		\item For all $\theta\in S^{M-1}$, $\mathcal L(\theta)$ is a GKSL generator.
		\item Let
		\(
		\mathcal F := \{\mathcal L(\theta):\ \theta\in S^{M-1}\}.
		\)
		Then for any $\theta,\theta'\in S^{M-1}$ and any $\alpha\in[0,1]$, the convex combination
		\[
		\mathcal L_\alpha := (1-\alpha)\,\mathcal L(\theta) + \alpha\,\mathcal L(\theta')
		\]
		is still a GKSL generator. 
	\end{enumerate}
\end{lemma}

\begin{proof}
	(1) GKSL generators form a convex cone and $\mathcal L_{\mathrm{ref}}$ is an interior point of this cone. Thus there exists $r>0$ such that whenever $\|\mathcal L-\mathcal L_{\mathrm{ref}}\|_{1\to1}\le r$, the operator $\mathcal L$ is still a GKSL generator. For $\mathcal L(\theta)$ we estimate
	\begin{align*}
		\|\mathcal L(\theta)-\mathcal L_{\mathrm{ref}}\|_{1\to1}
		&\le \frac{\delta}{\sqrt M}\sum_{j=1}^M|\theta_j|\,\|G_j\|_{1\to1} \\
		&\le \frac{\delta C_G}{\sqrt M}\|\theta\|_1 \\
		&\le \delta C_G,
	\end{align*}
	where in the last step we used $\|\theta\|_1\le\sqrt M\|\theta\|_2=\sqrt M$ for $\theta\in S^{M-1}$. Thus, as long as $\delta C_G\le r$ (for instance, taking $\delta_0:=r/C_G$ and requiring $0<\delta\le\delta_0$), this bound guarantees that $\mathcal L(\theta)$ is GKSL for all $\theta\in S^{M-1}$.
	
	(2) For $\theta,\theta'\in S^{M-1}$ and $\alpha\in[0,1]$, define
	\[
	\mathcal L_\alpha
	:= (1-\alpha)\,\mathcal L(\theta)+\alpha\,\mathcal L(\theta').
	\]
	Then
	\[
	\mathcal L_\alpha - \mathcal L_{\mathrm{ref}}
	= \frac{\delta}{\sqrt M}
	\sum_{j=1}^M \bigl[(1-\alpha)\theta_j+\alpha\theta'_j\bigr] G_j,
	\]
	and hence
	\begin{align*}
		\|\mathcal L_\alpha-\mathcal L_{\mathrm{ref}}\|_{1\to1}
		&\le \frac{\delta C_G}{\sqrt M}
		\sum_{j=1}^M\bigl|(1-\alpha)\theta_j+\alpha\theta'_j\bigr| \\
		&\le \frac{\delta C_G}{\sqrt M}\bigl((1-\alpha)\|\theta\|_1+\alpha\|\theta'\|_1\bigr) \\
		&\le \frac{\delta C_G}{\sqrt M}\bigl((1-\alpha)\sqrt M+\alpha\sqrt M\bigr) \\
		& = \delta C_G
		\le r.
	\end{align*}
	Therefore, for all $\alpha\in[0,1]$ and all $0<\delta\le\delta_0$, $\mathcal L_\alpha$ remains in the GKSL neighbourhood of $\mathcal L_{\mathrm{ref}}$ and is thus a GKSL generator. 
\end{proof}

\subsection{Interface and Lipschitz continuity of $P_{\mathcal L}[\varphi]$}\label{zming}

We now fix the open-system SQ interface: for a fixed time $t>0$, fixed input state $\rho_{\mathrm{in}}$ and POVM $\{M_x\}_{x\in X}$, we define for each Lindbladian $\mathcal L$ the QMS element
\(
\Lambda_t^{(\mathcal L)} := e^{t\mathcal L}
\)
and the associated classical output distribution
\[
P_{\mathcal L}(x)
:= \operatorname{Tr}\!\bigl(M_x \Lambda_t^{(\mathcal L)}(\rho_{\mathrm{in}})\bigr),\qquad x\in X.
\]
For any test function $\varphi:X\to[-1,1]$, we define the corresponding statistical query expectation
\[
P_{\mathcal L}[\varphi]
:= \sum_{x\in X}\varphi(x)P_{\mathcal L}(x)
= \operatorname{Tr}\!\Bigl(O_\varphi\,\Lambda_t^{(\mathcal L)}(\rho_{\mathrm{in}})\Bigr),
\]
where
\(
O_\varphi:=\sum_{x\in X}\varphi(x)M_x.
\)

\begin{lemma}[Norm bound for the observable $O_\varphi$]\label{lem:Ophi_norm}
	For any $\varphi:X\to[-1,1]$ we have
	\(
	\|O_\varphi\|_\infty \le 1.
	\)
\end{lemma}

\begin{proof}
	For any state $\rho\ge0$ with $\operatorname{Tr}\rho=1$,
	\[
	\operatorname{Tr}(O_\varphi\rho)
	= \sum_x \varphi(x)\operatorname{Tr}(M_x\rho),
	\]
	where $\operatorname{Tr}(M_x\rho)\ge0$ and $\sum_x\operatorname{Tr}(M_x\rho)=1$. Hence
	\(
	\operatorname{Tr}(O_\varphi\rho)
	\)
	is a convex combination of values $\varphi(x)\in[-1,1]$, and thus
	\(
	|\operatorname{Tr}(O_\varphi\rho)|\le1
	\)
	for all $\rho$. Since
	\(
	\|O_\varphi\|_\infty = \sup_{|\psi\rangle}\bigl|\langle\psi|O_\varphi|\psi\rangle\bigr|
	\)
	and the expectation value with respect to any pure state can be written via a density matrix expectation, it follows that
	\(
	\|O_\varphi\|_\infty\le1.
	\)
\end{proof}

\begin{lemma}[Lipschitz continuity with respect to the generator]\label{lem:L_to_Lambda_Lipschitz}
	Let $\mathcal L,\mathcal L'\in\mathsf{GKSL}$ and assume that for all $\alpha\in[0,1]$ the straight-line homotopy
	\(
	\mathcal L_\alpha := \mathcal L + \alpha(\mathcal L'-\mathcal L)
	\)
	stays within $\mathsf{GKSL}$. Then for any $t\ge0$ we have
	\[
	\bigl\|\Lambda_t^{(\mathcal L')}-\Lambda_t^{(\mathcal L)}\bigr\|_{1\to1}
	\le t\,\|\mathcal L'-\mathcal L\|_{1\to1},
	\]
	where $\|\cdot\|_{1\to1}$ is the induced $1\to1$ super-operator norm.
\end{lemma}

\begin{proof}
	Consider the map $\alpha\mapsto e^{t\mathcal L_\alpha}$. In finite dimensions the exponential map is differentiable in the operator-norm topology, and its Fréchet derivative is given by the perturbation formula (Definition~\ref{def:frechet}):
	\[
	\frac{\mathrm d}{\mathrm d\alpha}e^{t\mathcal L_\alpha}
	= \int_0^t e^{(t-s)\mathcal L_\alpha}\,(\mathcal L'-\mathcal L)\,e^{s\mathcal L_\alpha}\,\mathrm ds.
	\]
	Taking the induced $1\to1$ norm, and using submultiplicativity together with the fact that the QMS generated by each $\mathcal L_\alpha$ is CPTP (hence
	$\|e^{r\mathcal L_\alpha}\|_{1\to1}\le1$ for all $r\ge0$), we obtain
	\begin{align*}
		\Bigl\|\frac{\mathrm d}{\mathrm d\alpha}e^{t\mathcal L_\alpha}\Bigr\|_{1\to1}
		&\le \int_0^t \bigl\| e^{(t-s)\mathcal L_\alpha}\,(\mathcal L'-\mathcal L)\,e^{s\mathcal L_\alpha}\bigr\|_{1\to1}\,\mathrm ds \\
		&\le \int_0^t \|e^{(t-s)\mathcal L_\alpha}\|_{1\to1}\,\|\mathcal L'-\mathcal L\|_{1\to1}\,\|e^{s\mathcal L_\alpha}\|_{1\to1}\,\mathrm ds \\
		&\le \int_0^t \|\mathcal L'-\mathcal L\|_{1\to1}\,\mathrm ds \\
		&= t\,\|\mathcal L'-\mathcal L\|_{1\to1}.
	\end{align*}
	Integrating this bound over $\alpha\in[0,1]$ gives
	\begin{align*}
		\bigl\|e^{t\mathcal L'}-e^{t\mathcal L}\bigr\|_{1\to1}
		&= \left\|\int_0^1\frac{\mathrm d}{\mathrm d\alpha}e^{t\mathcal L_\alpha}\,\mathrm d\alpha\right\|_{1\to1} \\
		&\le \int_0^1
		\Bigl\|\frac{\mathrm d}{\mathrm d\alpha}e^{t\mathcal L_\alpha}\Bigr\|_{1\to1}\,\mathrm d\alpha \\
		&\le \int_0^1 t\,\|\mathcal L'-\mathcal L\|_{1\to1}\,\mathrm d\alpha \\
		&= t\,\|\mathcal L'-\mathcal L\|_{1\to1}.
	\end{align*}
	This is exactly the claimed Lipschitz bound with respect to the generator.
\end{proof}

\begin{corollary}[Lipschitz continuity of SQ expectations with respect to the generator]\label{cor:PL_Lipschitz_L}
	Under the assumptions of Lemma~\ref{lem:L_to_Lambda_Lipschitz}, for any test function $\varphi:X\to[-1,1]$ we have
	\[
	\bigl|P_{\mathcal L}[\varphi]-P_{\mathcal L'}[\varphi]\bigr|
	\le t\,\|\mathcal L'-\mathcal L\|_{1\to1}.
	\]
\end{corollary}

\begin{proof}
	By definition,
	\[
	P_{\mathcal L}[\varphi]-P_{\mathcal L'}[\varphi]
	= \operatorname{Tr}\Bigl(O_\varphi\bigl(\Lambda_t^{(\mathcal L)}-\Lambda_t^{(\mathcal L')}\bigr)(\rho_{\mathrm{in}})\Bigr).
	\]
	Using matrix Hölder's inequality (see Remark~\ref{rem:schatten_relations}) together with Lemma~\ref{lem:Ophi_norm} and the definition of the induced norm (Definition~\ref{def:induced_norms}), we obtain
	\begin{align*}
		\bigl|P_{\mathcal L}[\varphi]-P_{\mathcal L'}[\varphi]\bigr|
		&= \left|\operatorname{Tr}\Bigl(O_\varphi\bigl(\Lambda_t^{(\mathcal L)}-\Lambda_t^{(\mathcal L')}\bigr)(\rho_{\mathrm{in}})\Bigr)\right| \\
		&\le \|O_\varphi\|_\infty\,
		\bigl\|\bigl(\Lambda_t^{(\mathcal L)}-\Lambda_t^{(\mathcal L')}\bigr)(\rho_{\mathrm{in}})\bigr\|_1 \\
		&\le \|O_\varphi\|_\infty\,
		\bigl\|\Lambda_t^{(\mathcal L)}-\Lambda_t^{(\mathcal L')}\bigr\|_{1\to1}\,\|\rho_{\mathrm{in}}\|_1 \\
		&\le \bigl\|\Lambda_t^{(\mathcal L)}-\Lambda_t^{(\mathcal L')}\bigr\|_{1\to1},
	\end{align*}
	where we used $\|O_\varphi\|_\infty\le1$ (Lemma~\ref{lem:Ophi_norm}) and $\|\rho_{\mathrm{in}}\|_1=1$ for a density operator. Applying Lemma~\ref{lem:L_to_Lambda_Lipschitz} yields the desired bound
	\[
	\bigl|P_{\mathcal L}[\varphi]-P_{\mathcal L'}[\varphi]\bigr|
	\le t\,\|\mathcal L'-\mathcal L\|_{1\to1}.
	\]
\end{proof}

\subsection{Lipschitz constants and Lévy concentration on the parameter sphere}

We now work with the parametrization
\(
\mathcal L(\theta)
\)
introduced in Definition~\ref{def:param_sphere_ensemble}.

\begin{lemma}[Lipschitz constant on the parameter sphere]\label{lem:theta_Lipschitz}
	Suppose there exists a constant $C_G>0$ such that
	\(
	\|G_j\|_{1\to1}\le C_G
	\)
	for all $1\le j\le M$. Define
	\(
	F_\varphi(\theta)
	:= P_{\mathcal L(\theta)}[\varphi]
	\)
	for $\theta\in\Theta$. Then for any $\theta,\theta'\in\Theta$,
	\[
	|F_\varphi(\theta)-F_\varphi(\theta')|
	\le L_0\,\|\theta-\theta'\|_2,\qquad
	L_0:=t\,\delta\,C_G,
	\]
	and $L_0$ is independent of the test function $\varphi$ and of the dimension $M$.
\end{lemma}

\begin{proof}
	By definition of the parametrization,
	\[
	\mathcal L(\theta')-\mathcal L(\theta)
	= \frac{\delta}{\sqrt M}\sum_{j=1}^M (\theta'_j-\theta_j)\,G_j.
	\]
	Taking the $1\to1$ norm gives
	\begin{align*}
		\bigl\|\mathcal L(\theta')-\mathcal L(\theta)\bigr\|_{1\to1}
		&\le \frac{\delta}{\sqrt M}
		\sum_{j=1}^M |\theta'_j-\theta_j|\,\|G_j\|_{1\to1} \\
		&\le \frac{\delta C_G}{\sqrt M}
		\sum_{j=1}^M |\theta'_j-\theta_j| \\
		&\le \frac{\delta C_G}{\sqrt M}\,\sqrt M\,\|\theta'-\theta\|_2 \\
		&= \delta C_G \,\|\theta'-\theta\|_2,
	\end{align*}
	where in the last step we used the Cauchy-Schwarz inequality.
	$\sum_{j=1}^M |a_j|\le\sqrt M\,\|a\|_2$.
	Then by Corollary~\ref{cor:PL_Lipschitz_L},
	\begin{align*}
		|F_\varphi(\theta)-F_\varphi(\theta')|
		&= \bigl|P_{\mathcal L(\theta)}[\varphi]-P_{\mathcal L(\theta')}[\varphi]\bigr| \\
		&\le t\,\bigl\|\mathcal L(\theta')-\mathcal L(\theta)\bigr\|_{1\to1} \\
		&\le t\,\delta C_G \,\|\theta'-\theta\|_2.
	\end{align*}
	Thus $L_0=t\delta C_G$ is a uniform Lipschitz constant, independent of $\varphi$ and the parameter dimension~$M$.
\end{proof}

\begin{theorem}[Upper bound on the maximally distinguishable fraction in the open-system SQ framework]\label{thm:frac_bound_open_system}
	Let the parameter ensemble $(\Theta,\mu_\Theta,\mathcal L(\theta))$ be as in Definition~\ref{def:param_sphere_ensemble}, and let the Lipschitz constant $L_0$ be as in Lemma~\ref{lem:theta_Lipschitz}. In the SQ interface, for each test function $\varphi:X\to[-1,1]$ define
	\[
	F_\varphi(\theta):=P_{\mathcal L(\theta)}[\varphi],\qquad
	Q[\varphi]:=\mathbb E_{\mathcal L\sim\mu_{\mathcal L}} P_{\mathcal L}[\varphi]
	= \mathbb E_{\theta\sim\mu_\Theta} F_\varphi(\theta).
	\]
	Define
	\[
	\operatorname{frac}(\mu_{\mathcal L},Q,\tau)
	:= \sup_{\varphi:X\to[-1,1]}
	\Pr_{\mathcal L\sim\mu_{\mathcal L}}\bigl(|P_{\mathcal L}[\varphi]-Q[\varphi]|\ge\tau\bigr).
	\]
	Then for any $\tau>0$,
	\[
	\operatorname{frac}(\mu_{\mathcal L},Q,\tau)
	\le 2\exp\!\left(-c_{\mathrm{par}}\,\frac{M\,\tau^2}{L_0^2}\right),
	\]
	where $c_{\mathrm{par}}$ is the constant from Lemma~\ref{lem:levy_sphere} and $M$ is the parameter dimension.
\end{theorem}

\begin{proof}
	For each fixed $\varphi$, Lemma~\ref{lem:theta_Lipschitz} tells us that $F_\varphi$ is $L_0$-Lipschitz on $\Theta\subseteq S^{M-1}$. Viewing $\mu_\Theta$ as the normalized restriction of $\mu_S$ to $\Theta$, we can apply Lemma~\ref{lem:levy_sphere} directly (which yields the same type of concentration for normalized measures on any full-measure subset of the sphere) to obtain
	\[
	\Pr_{\theta\sim\mu_\Theta}\bigl(|F_\varphi(\theta)-\mathbb E_{\theta\sim\mu_\Theta}F_\varphi(\theta)|\ge\tau\bigr)
	\le 2\exp\!\left(-c_{\mathrm{par}}\,\frac{M\,\tau^2}{L_0^2}\right).
	\]
	In terms of the pushforward measure $\mu_{\mathcal L}$ this is equivalent to
	\[
	\Pr_{\mathcal L\sim\mu_{\mathcal L}}\Bigl(|P_{\mathcal L}[\varphi]-Q[\varphi]|\ge\tau\Bigr)
	\le 2\exp\!\left(-c_{\mathrm{par}}\,\frac{M\,\tau^2}{L_0^2}\right),
	\]
	where $Q[\varphi]=\mathbb E_{\mathcal L}P_{\mathcal L}[\varphi]$. Since the right-hand side is independent of $\varphi$, taking the supremum over all $\varphi:X\to[-1,1]$ gives
	\[
	\operatorname{frac}(\mu_{\mathcal L},Q,\tau)
	\le 2\exp\!\left(-c_{\mathrm{par}}\,\frac{M\,\tau^2}{L_0^2}\right).
	\]
\end{proof}

\begin{remark}[Metric, measure, and geometry in random Lindbladian ensembles]\label{rem:metric_measure_ensemble}
	The form of Theorem~\ref{thm:frac_bound_open_system} can naturally be understood within concentration framework~\cite{ledoux2001concentration}, what we are really using is a metric probability space
	\[
	(\Theta,d_2,\mu_\Theta),\qquad
	d_2(\theta,\theta'):=\|\theta-\theta'\|_2,
	\]
	where $\Theta\subseteq S^{M-1}$ is a full-measure subset of the parameter sphere and $\mu_\Theta$ is the probability measure obtained by restricting and renormalizing the rotation-invariant measure (normalized surface measure) on the sphere. For each test function $\varphi:X\to[-1,1]$ we consider
	\[
	F_\varphi:\Theta\to\mathbb R,\qquad
	F_\varphi(\theta):=P_{\mathcal L(\theta)}[\varphi].
	\]
	Lemma~\ref{lem:theta_Lipschitz} states that $F_\varphi$ is $L_0$-Lipschitz with respect to $d_2$, while Lemma~\ref{lem:levy_sphere} is a Lévy-type concentration theorem on $(S^{M-1},d_2,\mu_S)$. Combining the two yields the desired concentration inequality on the parameter space.
	
	This structure is fully analogous to Haar-random states and unitaries in quantum information~\cite{mele2024introduction,brandao2016local,haferkamp2022random}. In the case of random pure states, the space is the $2$-norm sphere $S^{2d-1}$ in Hilbert space, the distance is the Euclidean metric induced by the inner product, and the measure is the unique unitary-invariant probability measure induced by $U(d)$ on the sphere (also called the Haar measure). In this geometry, for any $1$-Lipschitz function $f:S^{2d-1}\to\mathbb R$ one obtains
	\[
	\Pr_{\psi\sim\mu_{\mathrm{Haar}}}\bigl(|f(\psi)-\mathbb Ef|\ge s\bigr)
	\le 2\exp(-c\,d\,s^2),
	\]
	which is precisely Lévy concentration in the language of random pure states.
	
	In our parameter-sphere ensemble, the Euclidean metric $d_2$ and the rotation-invariant measure $\mu_\Theta$ play the same role, they guarantee that all coordinate directions in parameter space are geometrically unbiased, and they ensure that the isoperimetric constant in the sense of Lévy~\cite{ledoux2001concentration} scales linearly with the dimension $M$, leading to a concentration bound with an exponential factor in $M$:
	\[
	\Pr_{\theta\sim\mu_\Theta}\bigl(|F_\varphi(\theta)-\mathbb E_{\mu_\Theta}F_\varphi|\ge\tau\bigr)
	\le 2\exp\!\left(-c_{\mathrm{par}}\,\frac{M\tau^2}{L_0^2}\right).
	\]
	The normalization factor $\delta/\sqrt M$ introduced in the parametrization is exactly what keeps the Lipschitz constant $L_0=t\delta C_G$ from blowing up with the dimension.
	
	More generally, as emphasized in~\cite{ledoux2001concentration}, concentration theory is not restricted to $2$-norm spheres and Haar measure; it applies to arbitrary metric probability spaces $(X,d,\mu)$. If we are willing to change the metric and measure on the parameter space, for example:
	\begin{itemize}
		\item[i.] choose a product measure with independent coordinates (each local coupling $\theta_j$ is an independent random variable), together with an $\ell_1$ or Hamming distance;
		\item[ii.] or choose a symmetric log-concave measure (similar to a high-dimensional Gaussian) together with the Euclidean metric;
	\end{itemize}
	then as long as $(X,d,\mu)$ has good isoperimetric properties (for instance, satisfies a Poincaré or log-Sobolev inequality), one can still obtain concentration inequalities for $F_\varphi$, although the constants and exponential scaling structure will change with the underlying geometry. In that case, the induced $\mu_{\mathcal L}$ describes another random Lindbladian ensemble, corresponding to a different physical modeling assumption (e.g., independently random local jump strengths), rather than the Euclidean-sphere-based ``parameter Haar'' ensemble considered here.
	
	It is worth emphasizing once more that these concentration and SQ-hardness results are inherently distribution-specific, they are always stated relative to a chosen prior $\mu_{\mathcal L}$. This is completely parallel to the situations in random circuits~\cite{mele2024introduction,brandao2016local,haferkamp2022random} and random matrix theory~\cite{mehta2004random,khaymovich2019eigenstate}:
	\begin{itemize}
		\item[i.] In the random-circuit literature, one chooses either Haar-random unitaries or more physically motivated local random gate sets, thereby defining an ensemble $\mu_{\mathcal C}$ of circuits, and then studies concentration and average-case hardness of typical output distributions with respect to that ensemble~\cite{nietner2025average}.
		\item[ii.] In random matrix theory, ensembles such as the Gaussian Unitary Ensemble (GUE)~\cite{khaymovich2019eigenstate}, the Ginibre ensemble~\cite{hamazaki2020universality}, and local random Hamiltonians~\cite{sugimoto2021test,nakata2014thermal} serve as standard random models for structureless chaotic Hamiltonians and non-Hermitian generators of open systems; all statements about spectral statistics and typicality are made relative to these ensembles $\mu_H$.
		\item[iii.] For Lindbladians, our Definition~\ref{def:param_sphere_ensemble} adopts an unbiased Euclidean-sphere parameter ensemble as a baseline example. One can likewise consider other random GKSL ensembles that better reflect specific physical platforms (such as random local jumps plus random Kossakowski matrices), and apply Lévy-type tools on each associated $(X,d,\mu)$ to obtain the corresponding open-system concentration and SQ-hardness results.
	\end{itemize}
	Therefore, Theorem~\ref{thm:frac_bound_open_system} and the preceding lemmas show that once we fix a physically natural and geometrically well-behaved random ensemble $(\mathcal F,\mu_{\mathcal L})$ on the Lindbladian space, the typical SQ expectations $P_{\mathcal L}[\varphi]$ will exhibit strong Lévy-type concentration over the parameter space. In Sec.~\ref{local} we present another example based on an $\ell_1$ metric and a product-measure ensemble of independent local couplings.
\end{remark}

\subsection{Product-measure concentration for independent local couplings}\label{local}

The discussion of the parameter sphere $(\Theta,\mu_\Theta)$ above relies on the rotation-invariant measure and Lévy's lemma on the Euclidean sphere $S^{M-1}$. In this subsection we present a parallel ensemble with independent local couplings, each local coupling constant is a small independent random variable, leading to a product measure on parameter space together with an $\ell_1$ metric. Unlike the spherical case, here one can keep the Lipschitz constant independent of the dimension $M$ without introducing a $\delta/\sqrt M$ normalization, at the expense of having the dependence on $M$ enter the concentration exponent via a McDiarmid-type inequality.

\begin{definition}[Random Lindbladian ensemble with independent local couplings]\label{def:prod_measure_ensemble}
	Let $M\in\mathbb N$ be the parameter dimension.
	Let the single-coordinate distribution $\mu_1$ be supported on the interval $[-\delta,\delta]\subset\mathbb R$ (where $\delta>0$ sets the typical scale of the couplings), and define
	\[
	\Theta_{\mathrm{prod}}
	:= [-\delta,\delta]^M,\qquad
	\mu_{\mathrm{prod}}
	:= \mu_1^{\otimes M}.
	\]
	For $\theta=(\theta_1,\dots,\theta_M)\in\Theta_{\mathrm{prod}}$, define the linear parametrization
	\[
	\mathcal L(\theta)
	:= \mathcal L_{\mathrm{ref}} + \sum_{j=1}^M \theta_j G_j,
	\]
	where $\mathcal L_{\mathrm{ref}}$ is a fixed Lindbladian and $\{G_j\}_{j=1}^M$ is a family of linear maps with $\|G_j\|_{1\to1}\le C_G$ for some constant $C_G>0$. Define the pushforward measure
	\[
	\mu_{\mathcal L}^{\mathrm{prod}}
	:= (\mathcal L)_\#\mu_{\mathrm{prod}},
	\]
	i.e., for any measurable set $A$,
	$\mu_{\mathcal L}^{\mathrm{prod}}(A)
	:= \mu_{\mathrm{prod}}\{\theta:\mathcal L(\theta)\in A\}$.
	Equip the parameter space with the $\ell_1$ metric
	\[
	d_1(\theta,\theta')
	:= \|\theta-\theta'\|_1
	= \sum_{j=1}^M|\theta_j-\theta'_j|.
	\]
\end{definition}

\begin{lemma}[Lipschitz constant on the product-measure parameter space]\label{lem:prod_theta_Lipschitz}
	In the setting of Definition~\ref{def:prod_measure_ensemble}, let
	\[
	F_\varphi(\theta)
	:= P_{\mathcal L(\theta)}[\varphi],
	\qquad
	\theta\in\Theta_{\mathrm{prod}},
	\]
	for an arbitrary test function $\varphi:X\to[-1,1]$.
	Then for all $\theta,\theta'\in\Theta_{\mathrm{prod}}$,
	\[
	|F_\varphi(\theta)-F_\varphi(\theta')|
	\le L_1\,\|\theta-\theta'\|_1,\qquad
	L_1:=t\,C_G,
	\]
	and $L_1$ is independent of $\varphi$ and of the dimension $M$.
\end{lemma}

\begin{proof}
	By linearity of the parametrization,
	\[
	\mathcal L(\theta')-\mathcal L(\theta)
	= \sum_{j=1}^M (\theta'_j-\theta_j)\,G_j.
	\]
	Taking the $1\to1$ norm and using $\|G_j\|_{1\to1}\le C_G$, we obtain
	\begin{align*}
		\bigl\|\mathcal L(\theta')-\mathcal L(\theta)\bigr\|_{1\to1}
		&\le \sum_{j=1}^M|\theta'_j-\theta_j|\,\|G_j\|_{1\to1} \\
		&\le C_G\sum_{j=1}^M|\theta'_j-\theta_j| \\
		&= C_G\,\|\theta'-\theta\|_1.
	\end{align*}
	By Corollary~\ref{cor:PL_Lipschitz_L} (Lipschitz continuity of $P_{\mathcal L}[\varphi]$ with respect to $\mathcal L$), we obtain
	\begin{align*}
		|F_\varphi(\theta)-F_\varphi(\theta')|
		&= \bigl|P_{\mathcal L(\theta)}[\varphi]-P_{\mathcal L(\theta')}[\varphi]\bigr| \\
		&\le t\,\bigl\|\mathcal L(\theta')-\mathcal L(\theta)\bigr\|_{1\to1} \\
		&\le t\,C_G\,\|\theta'-\theta\|_1.
	\end{align*}
	Thus $L_1=tC_G$ is a uniform Lipschitz constant with respect to the $\ell_1$ metric $d_1$, independent of $\varphi$ and of $M$.
\end{proof}

We now use a McDiarmid-type bounded-differences inequality to obtain concentration under the product measure.
Since $\mu_{\mathrm{prod}}$ is a product measure with coordinates supported on a bounded interval $[-\delta,\delta]$, it is well suited to this classical tool.

\begin{lemma}[McDiarmid bounded-differences inequality, cf.~\cite{ledoux2001concentration}]\label{lem:mcdiarmid}
	Let $\theta=(\theta_1,\dots,\theta_M)$ have independent coordinates,
	$\theta_j\sim\mu_1$ supported on $[-\delta,\delta]$.
	Let $F:\Theta_{\mathrm{prod}}\to\mathbb R$ be such that for each $j$ there exists $c_j\ge0$ with the property that for any $\theta$ and any $a,a'\in[-\delta,\delta]$,
	\[
	\bigl|F(\theta_1,\dots,\theta_{j-1},a,\theta_{j+1},\dots,\theta_M)
	- F(\theta_1,\dots,\theta_{j-1},a',\theta_{j+1},\dots,\theta_M)\bigr|
	\le c_j.
	\]
	Then for all $\tau>0$,
	\[
	\Pr_{\theta\sim\mu_{\mathrm{prod}}}\bigl(|F(\theta)-\mathbb E F|\ge\tau\bigr)
	\le 2\exp\!\left(-\frac{2\tau^2}{\sum_{j=1}^M c_j^2}\right).
	\]
\end{lemma}

Combining Lemma~\ref{lem:prod_theta_Lipschitz} with Lemma~\ref{lem:mcdiarmid} yields a concentration theorem for the ensemble with independent local couplings.

\begin{theorem}[Upper bound on the maximally distinguishable fraction for the independent local-coupling ensemble]\label{thm:frac_bound_open_system_prod}
	Let $(\Theta_{\mathrm{prod}},\mu_{\mathrm{prod}},\mathcal L(\theta))$
	be as in Definition~\ref{def:prod_measure_ensemble}, and let the Lipschitz constant $L_1$ be as in Lemma~\ref{lem:prod_theta_Lipschitz}.
	In the SQ interface, for each test function $\varphi:X\to[-1,1]$ define
	\[
	F_\varphi(\theta):=P_{\mathcal L(\theta)}[\varphi],\qquad
	Q[\varphi]:=\mathbb E_{\mathcal L\sim\mu_{\mathcal L}^{\mathrm{prod}}} P_{\mathcal L}[\varphi]
	= \mathbb E_{\theta\sim\mu_{\mathrm{prod}}} F_\varphi(\theta).
	\]
	Define the maximally distinguishable fraction
	\[
	\operatorname{frac}(\mu_{\mathcal L}^{\mathrm{prod}},Q,\tau)
	:= \sup_{\varphi:X\to[-1,1]}
	\Pr_{\mathcal L\sim\mu_{\mathcal L}^{\mathrm{prod}}}\bigl(|P_{\mathcal L}[\varphi]-Q[\varphi]|\ge\tau\bigr).
	\]
	Then for any $\tau>0$,
	\[
	\operatorname{frac}(\mu_{\mathcal L}^{\mathrm{prod}},Q,\tau)
	\;\le\; 2\exp\!\left(-c_{\mathrm{prod}}\,
	\frac{\tau^2}{M\,\delta^2\,L_1^2}\right),
	\qquad
	L_1=tC_G,
	\]
	where $c_{\mathrm{prod}}>0$ is an absolute constant (for instance, one may take $c_{\mathrm{prod}}=\tfrac12$ as suggested by McDiarmid's inequality).
\end{theorem}

\begin{proof}
	Fix a test function $\varphi$.
	For any $j$, consider changing only the $j$-th coordinate:
	given $\theta\in\Theta_{\mathrm{prod}}$ and $a,a'\in[-\delta,\delta]$, define
	\[
	\theta^{(j,a)}:=(\theta_1,\dots,\theta_{j-1},a,\theta_{j+1},\dots,\theta_M),
	\quad
	\theta^{(j,a')}:=(\theta_1,\dots,\theta_{j-1},a',\theta_{j+1},\dots,\theta_M).
	\]
	By the $L_1$-Lipschitz property of Lemma~\ref{lem:prod_theta_Lipschitz},
	\begin{align*}
		\bigl|F_\varphi(\theta^{(j,a)})-F_\varphi(\theta^{(j,a')})\bigr|
		&\le L_1\,\|\theta^{(j,a)}-\theta^{(j,a')}\|_1 \\
		&= L_1\,|a-a'| \\
		&\le 2\delta\,L_1.
	\end{align*}
	Thus we may take $c_j:=2\delta L_1$ for all $1\le j\le M$.
	Inserting this into Lemma~\ref{lem:mcdiarmid} yields
	\begin{align*}
		\Pr_{\theta\sim\mu_{\mathrm{prod}}}\bigl(|F_\varphi(\theta)-\mathbb E_{\theta\sim\mu_{\mathrm{prod}}}F_\varphi(\theta)|\ge\tau\bigr)
		&\le 2\exp\!\left(
		-\frac{2\tau^2}{\sum_{j=1}^M c_j^2}
		\right) \\
		&\le 2\exp\!\left(
		-\frac{2\tau^2}{M\cdot (2\delta L_1)^2}
		\right) \\
		&= 2\exp\!\left(
		-\,\frac{\tau^2}{2M\delta^2 L_1^2}
		\right).
	\end{align*}
	In terms of the pushforward measure $\mu_{\mathcal L}^{\mathrm{prod}}$, this is equivalent to
	\[
	\Pr_{\mathcal L\sim\mu_{\mathcal L}^{\mathrm{prod}}}
	\Bigl(|P_{\mathcal L}[\varphi]-Q[\varphi]|\ge\tau\Bigr)
	\le 2\exp\!\left(
	-\,\frac{\tau^2}{2M\delta^2 L_1^2}
	\right),
	\]
	where $Q[\varphi]=\mathbb E_{\mathcal L}P_{\mathcal L}[\varphi]$.
	Since the right-hand side is independent of $\varphi$, taking the supremum over all $\varphi:X\to[-1,1]$ yields
	\[
	\operatorname{frac}(\mu_{\mathcal L}^{\mathrm{prod}},Q,\tau)
	\le 2\exp\!\left(-c_{\mathrm{prod}}\,
	\frac{\tau^2}{M\,\delta^2\,L_1^2}\right),
	\]
	for some absolute constant $c_{\mathrm{prod}}$ (e.g., $c_{\mathrm{prod}}=\tfrac12$).
\end{proof}

\begin{remark}[recover $\exp(-\Omega(M))$ scaling]\label{rem:prod_vs_sphere}
	
	If one wishes to recover an $\exp(-\Omega(M))$ scaling in the product-measure model, there are various ways to proceed. For example, one can strengthen the small-coupling assumption by shrinking the range of each local parameter with $M$, i.e., let
	\[
	\theta_j\in\Bigl[-\frac{\delta}{M},\frac{\delta}{M}\Bigr].
	\]
	Then
	\(
	\|\mathcal L_{\mathrm{prod}}(\theta)-\mathcal L_{\mathrm{ref}}\|_{1\to1}
	\le \sum_j|\theta_j|\|G_j\|_{1\to1}
	\le M\cdot \frac{\delta}{M}C_G
	=\delta C_G,
	\)
	so the overall scale of the perturbation remains $O(\delta)$, while the bounded-differences constants become $c_j\le 2(\delta/M)L_1$, and
	\[
	\sum_{j=1}^M c_j^2
	\le M\cdot \bigl(2\delta L_1/M\bigr)^2
	= \mathcal O\!\bigl(\delta^2L_1^2/M\bigr).
	\]
	Substituting into McDiarmid gives
	\[
	\Pr(|F_\varphi-\mathbb EF_\varphi|\ge\tau)
	\lesssim \exp\!\left(-\Omega\Bigl(M\,\frac{\tau^2}{\delta^2L_1^2}\Bigr)\right),
	\]
	which mirrors the $\exp(-\Omega(M))$ scaling in the spherical case.
	From the viewpoint of parametrization, this is equivalent to ``moving'' the overall normalization factor $\delta/\sqrt M$ (used in the spherical ensemble) into the coordinate ranges: one formulation is
	\(
	\mathcal L=\mathcal L_{\mathrm{ref}} + \frac{\delta}{\sqrt M}\sum_j \theta_j G_j,\ \|\theta\|_2=1;
	\)
	another is
	\(
	\mathcal L=\mathcal L_{\mathrm{ref}} + \sum_j \theta_j G_j,\ |\theta_j|\lesssim 1/M;
	\)
	in both cases the perturbation scale in $\mathcal L$ space is $O(\delta)$, but the metric and measure on parameter space distribute this scale differently.

\end{remark}

So, our Lipschitz and Lévy machinery provides an exponentially small upper bound for the denominator
\(
\max_\varphi\Pr[|P[\varphi]-Q[\varphi]|>\tau]
\),
 whereas for average-case SQ learning, we additionally need to control the ``far-from-$Q$'' volume in the numerator (i.e., the behaviour of $\beta-\Pr[d_{\mathrm{TV}}(P,Q)\le\epsilon+\tau]$), which will be discussed in detail in Sec.~\ref{zhang3}.

\subsection{QPStat hardness of learning Random Lindbladian channel in the average-case}\label{zhang8}

Provisionally, we now turn to the hardness of learning random Lindbladian channel ensembles.

\begin{lemma}[Lipschitz continuity of QPStat expectations in generator parameters]\label{lem:qpstat_Lipschitz}
	Using the notation of Definition~\ref{def:param_sphere_ensemble},
	let $\mathcal L(\theta)$ be a family of Lindbladians indexed by a parameter space
	$(\Theta,d,\mu_\Theta)$, and assume that:
	\begin{enumerate}
		\item Each $\mathcal L(\theta)$, $\theta\in\Theta$, is a GKSL generator, and the line-segment interpolation
		\(
		\mathcal L_\alpha
		:= (1-\alpha)\mathcal L(\theta)
		+\alpha\mathcal L(\theta')
		\)
		remains in the GKSL cone for all $\alpha\in[0,1]$
		(by Lemma~\ref{lem:convexity_GKSL_slice});
		\item There exists a constant $L_{\mathcal L}>0$, independent of the parameter dimension $M$, such that
		\begin{equation}\label{eq:L_generator_Lipschitz}
			\|\mathcal L(\theta)-\mathcal L(\theta')\|_{1\to1}
			\le L_{\mathcal L}\,\|\theta-\theta'\|_2,
			\qquad\forall\,\theta,\theta'\in\Theta .
		\end{equation}
	\end{enumerate}
	Let $E_\theta := e^{t\mathcal L(\theta)}$ be the one-step channel of the QMS at fixed time $t>0$.
	For any state $\rho$ and operator $O$ with $\|\rho\|_1\le1$ and $\|O\|_\infty\le1$, define
	\begin{equation}
		f_{\rho,O}(\theta)
		:= \operatorname{Tr}\bigl[O\,E_\theta(\rho)\bigr]
		= \operatorname{Tr}\bigl[O\,e^{t\mathcal L(\theta)}(\rho)\bigr] .
	\end{equation}
	Then there exists a constant
	\begin{equation}
		L_{\mathrm{QP}} := t\,L_{\mathcal L} ,
	\end{equation}
	such that for all $\theta,\theta'\in\Theta$,
	\begin{equation}\label{eq:f_rho_O_Lipschitz}
		\bigl|f_{\rho,O}(\theta)-f_{\rho,O}(\theta')\bigr|
		\le L_{\mathrm{QP}}\;\|\theta-\theta'\|_2 .
	\end{equation}
	In other words, $f_{\rho,O}$ is $L_{\mathrm{QP}}$-Lipschitz on $(\Theta,\|\cdot\|_2)$, and the Lipschitz constant is independent of $(\rho,O)$.
\end{lemma}

\begin{proof}
	The proof is almost identical to Lemma~\ref{lem:theta_Lipschitz} and Theorem~\ref{thm:frac_bound_open_system}.
	
	\emph{Step 1: From channel difference to expectation difference.}
	For any $\theta,\theta'$ we have
	\begin{align*}
		\bigl|f_{\rho,O}(\theta)-f_{\rho,O}(\theta')\bigr|
		&= \bigl|\operatorname{Tr}\bigl[O\,\bigl(E_\theta-E_{\theta'}\bigr)(\rho)\bigr]\bigr| \\
		&\le \|O\|_\infty\,
		\bigl\|\bigl(E_\theta-E_{\theta'}\bigr)(\rho)\bigr\|_1 \\
		&\le \|O\|_\infty\,
		\|E_\theta-E_{\theta'}\|_{1\to1}\,\|\rho\|_1 \\
		&\le \|E_\theta-E_{\theta'}\|_{1\to1},
	\end{align*}
	where we used the matrix Hölder inequality, the definition of the induced $1\to1$ norm, and the bounds
	$\|O\|_\infty\le1$, $\|\rho\|_1\le1$.
	
	\emph{Step 2: Bounding the channel difference via the generator difference.}
	Define along the line segment
	\(
	\mathcal L_\alpha
	:= (1-\alpha)\mathcal L(\theta)
	+\alpha\mathcal L(\theta')
	\)
	and set $E_\alpha:=e^{t\mathcal L_\alpha}$.
	By Lemma~\ref{lem:convexity_GKSL_slice}, $\mathcal L_\alpha$ is GKSL for all $\alpha\in[0,1]$, so $E_\alpha$ is CPTP and satisfies $\|E_\alpha\|_{1\to1}\le1$.
	
	By the Fréchet derivative formula (Lemma~\ref{lem:duhamel}), we have
	\begin{equation}\label{eq:Frechet_QPStat}
		\frac{d}{d\alpha}E_\alpha
		= \int_0^t
		e^{(t-s)\mathcal L_\alpha}\,
		\bigl(\mathcal L(\theta')-\mathcal L(\theta)\bigr)\,
		e^{s\mathcal L_\alpha}\,ds .
	\end{equation}
	Taking the induced $1\to1$ norm and using CPTP contractivity yields
	\begin{align*}
		\Bigl\|\frac{d}{d\alpha}E_\alpha\Bigr\|_{1\to1}
		&\le \int_0^t
		\bigl\|e^{(t-s)\mathcal L_\alpha}\bigr\|_{1\to1}\,
		\bigl\|\mathcal L(\theta')-\mathcal L(\theta)\bigr\|_{1\to1}\,
		\bigl\|e^{s\mathcal L_\alpha}\bigr\|_{1\to1}\,ds \\
		&\le \int_0^t
		\bigl\|\mathcal L(\theta')-\mathcal L(\theta)\bigr\|_{1\to1}\,ds \\
		&= t\,\bigl\|\mathcal L(\theta')-\mathcal L(\theta)\bigr\|_{1\to1} .
	\end{align*}
	
	\emph{Step 3: Integrating along the line segment in parameter space.}
	By the fundamental theorem of calculus,
	\begin{equation*}
		E_\theta-E_{\theta'}
		= E_1-E_0
		= \int_0^1 \frac{d}{d\alpha}E_\alpha\,d\alpha ,
	\end{equation*}
	and hence
	\begin{align*}
		\|E_\theta-E_{\theta'}\|_{1\to1}
		&\le \int_0^1
		\Bigl\|\frac{d}{d\alpha}E_\alpha\Bigr\|_{1\to1}\,d\alpha \\
		&\le \int_0^1 t\,\bigl\|\mathcal L(\theta')-\mathcal L(\theta)\bigr\|_{1\to1}\,d\alpha \\
		&= t\,\bigl\|\mathcal L(\theta')-\mathcal L(\theta)\bigr\|_{1\to1}.
	\end{align*}
	Combining this with the generator Lipschitz condition~\eqref{eq:L_generator_Lipschitz} gives
	\begin{equation*}
		\|E_\theta-E_{\theta'}\|_{1\to1}
		\le t\,L_{\mathcal L}\,\|\theta-\theta'\|_2 .
	\end{equation*}
	Substituting this bound into the estimate from Step~1 yields
	\[
	\bigl|f_{\rho,O}(\theta)-f_{\rho,O}(\theta')\bigr|
	\le t\,L_{\mathcal L}\,\|\theta-\theta'\|_2,
	\]
	which is exactly~\eqref{eq:f_rho_O_Lipschitz} with $L_{\mathrm{QP}}:=tL_{\mathcal L}$. This Lipschitz constant is uniform over all $\rho$ and $O$ satisfying $\|\rho\|_1\le1$, $\|O\|_\infty\le1$.
\end{proof}

\begin{remark}
	In the concrete random local Lindbladian ensemble of Definition~\ref{def:random_local_L}, combining \eqref{eq:Gj_norm_bound} and \eqref{eq:L_theta_def} gives
	\(
	L_{\mathcal L}\lesssim \delta C_G
	\),
	and hence $L_{\mathrm{QP}}=tL_{\mathcal L}$ is of the same order as $L_0=t\delta C_G$ from Lemma~\ref{lem:F_phi_Lipschitz}.
	This shows that the classical SQ and QPStat access models are naturally compatible at the level of the underlying Lipschitz structure.
\end{remark}

Next we rewrite the Lévy/McDiarmid-type concentration results proved earlier on parameter space into a form that acts directly on QPStat expectations.

\begin{lemma}[Parameter-space concentration for QPStat deviations]\label{lem:qpstat_concentration}
	Let $(\Theta,d,\mu_\Theta)$ be one of the random Lindbladian parameter ensembles considered in this work, with parameter dimension $M$, and assume that it satisfies the concentration statements established for the spherical ensemble (Theorem~\ref{thm:frac_bound_open_system}) and for the product--measure ensemble (Theorem~\ref{thm:frac_bound_open_system_prod} and Remark~\ref{rem:prod_vs_sphere}):
	namely, there exist constants $c_{\mathrm{par}}>0$ and $L>0$ such that for any measurable $L$-Lipschitz function
	$F:\Theta\to\mathbb R$ and any $\tau>0$,
	\begin{equation}\label{eq:param_concentration}
		\Pr_{\theta\sim\mu_\Theta}
		\bigl(|F(\theta)-\mathbb E_\Theta F|\ge\tau\bigr)
		\le 2\exp\Bigl(-c_{\mathrm{par}}\frac{M\tau^2}{L^2}\Bigr) .
	\end{equation}
	
	Let $E_\theta:=e^{t\mathcal L(\theta)}$, and define the average channel
	\[
	\bar E := \mathbb E_{\theta\sim\mu_\Theta} E_\theta .
	\]
	For any pair $(\rho,O)$ with $\|\rho\|_1\le1$ and $\|O\|_\infty\le1$, set
	\[
	F_{\rho,O}(\theta)
	:= \operatorname{Tr}\bigl[O\,E_\theta(\rho)\bigr] .
	\]
	If $\mathcal L(\theta)$ satisfies the generator Lipschitz condition of Lemma~\ref{lem:qpstat_Lipschitz}, then $F_{\rho,O}$ is $L_{\mathrm{QP}}$-Lipschitz, and for all $\tau>0$,
	\begin{equation}\label{eq:qpstat_concentration}
		\Pr_{\theta\sim\mu_\Theta}
		\Bigl(
		\bigl|\operatorname{Tr}[O E_\theta(\rho)]
		-\operatorname{Tr}[O\bar E(\rho)]\bigr| \geq \tau
		\Bigr)
		\le 2\exp\Bigl(-c_{\mathrm{par}}\frac{M\tau^2}{L_{\mathrm{QP}}^2}\Bigr) ,
	\end{equation}
	where $L_{\mathrm{QP}}=tL_{\mathcal L}$.
\end{lemma}

\begin{proof}
	By Lemma~\ref{lem:qpstat_Lipschitz}, $F_{\rho,O}$ is $L_{\mathrm{QP}}$-Lipschitz on $(\Theta,\|\cdot\|_2)$, with $L_{\mathrm{QP}}=tL_{\mathcal L}$.
	On the other hand,
	\begin{align*}
		\mathbb E_{\theta\sim\mu_\Theta}F_{\rho,O}(\theta)
		&= \mathbb E_{\theta\sim\mu_\Theta}\operatorname{Tr}[O E_\theta(\rho)] \\
		&= \operatorname{Tr}\bigl[O\,(\mathbb E_{\theta\sim\mu_\Theta} E_\theta)(\rho)\bigr] \\
		&= \operatorname{Tr}[O\bar E(\rho)] .
	\end{align*}
	Thus $F:=F_{\rho,O}$ has Lipschitz constant $L_{\mathrm{QP}}$ and mean
	$\mathbb E_\Theta F = \operatorname{Tr}[O\bar E(\rho)]$.
	Applying~\eqref{eq:param_concentration} with $L:=L_{\mathrm{QP}}$ immediately gives~\eqref{eq:qpstat_concentration}.
\end{proof}

\begin{definition}[Random Lindbladian channel ensembles and the QPStat model]\label{def:lindblad_qpstat_ensemble}
	Let $(\Theta,d,\mu_\Theta)$ be one of the Lindbladian parameter spaces defined earlier
	(for example, the spherical random-perturbation ensemble or the independent local-coupling product--measure ensemble in Definition~\ref{def:random_local_L}), with parameter dimension $M$.
	
	For each parameter $\theta\in\Theta$, let
	\[
	\mathcal L(\theta):\ \mathcal B(\mathcal H)\to \mathcal B(\mathcal H)
	\]
	be a Lindbladian generator.
	Let the corresponding quantum Markov semigroup be $\Lambda_t^{(\theta)}:=e^{t\mathcal L(\theta)}$, assumed CPTP for all $t\ge0$.
	At a fixed evolution time $t>0$, we obtain the CPTP channel ensemble
	\[
	E_\theta := \Lambda_t^{(\theta)} = e^{t\mathcal L(\theta)},
	\qquad \theta\sim\mu_\Theta ,
	\]
	and the average channel
	\[
	\bar E := \mathbb E_{\theta\sim\mu_\Theta} E_\theta .
	\]
	
	For each channel $E$, the QPStat oracle $QPStat_E$ is defined as in Definition~\ref{def:QPStat_oracle}.
	A QPStat learning algorithm $\mathcal A$ interacting with an unknown parameter $\theta\sim\mu_\Theta$ and its associated channel $E_\theta$ may adaptively call this oracle at most $q$ times; we say it uses $q$ QPStat queries with accuracy~$\tau$.
\end{definition}

Note that under the linear parametrization~\eqref{eq:L_theta_def} in Definition~\ref{def:random_local_L},
Lemma~\ref{lem:convexity_GKSL_slice} ensures that $\mathcal L(\theta)$ forms a GKSL slice, while Lemma~\ref{lem:F_phi_Lipschitz} and Lemma~\ref{lem:qpstat_Lipschitz} give a uniform Lipschitz constant:
\[
\|\mathcal L(\theta)-\mathcal L(\theta')\|_{1\to1}
\le \delta C_G\,\|\theta-\theta'\|_2,
\qquad
L_{\mathrm{QP}} = t\delta C_G .
\]
Thus, for these concrete random Lindbladian ensembles, the assumptions of Lemma~\ref{lem:qpstat_concentration} are automatically satisfied.

\begin{theorem}[Exponential average-case QPStat lower bound for random Lindbladian channels]\label{thm:lindblad_qpstat_lower_bound}
	Let $(\Theta,d,\mu_\Theta)$ be one of the random Lindbladian parameter ensembles considered in this work, with parameter dimension $M$, and let $\mathcal L(\theta)$, $E_\theta$ and the average channel $\bar E$ be defined as in Definition~\ref{def:lindblad_qpstat_ensemble}.
	
	Assume that:
	\begin{enumerate}
		\item[\emph{(i)}] (Generator Lipschitz property)
		There exists a constant $L_{\mathcal L}>0$, independent of $M$, such that
		\[
		\|\mathcal L(\theta)-\mathcal L(\theta')\|_{1\to1}
		\le L_{\mathcal L}\,\|\theta-\theta'\|_2,
		\qquad\forall\,\theta,\theta'\in\Theta ,
		\]
		and the line-segment interpolation $\mathcal L_\alpha$ remains in the GKSL cone
		(as guaranteed by Lemma~\ref{lem:convexity_GKSL_slice});
		\item[\emph{(ii)}] (Concentration on parameter space)
		The parameter space $(\Theta,\mu_\Theta)$ satisfies the concentration inequality
		\eqref{eq:qpstat_concentration} in Lemma~\ref{lem:qpstat_concentration}, i.e., for any
		$\|\rho\|_1\le1$, $\|O\|_\infty\le1$ and any $\tau>0$,
		\[
		\Pr_{\theta\sim\mu_\Theta}
		\Bigl(
		\bigl|\operatorname{Tr}[O E_\theta(\rho)]
		-\operatorname{Tr}[O\bar E(\rho)]\bigr| \geq \tau
		\Bigr)
		\le 2\exp\Bigl(-c_{\mathrm{par}}\frac{M\tau^2}{L_{\mathrm{QP}}^2}\Bigr) ,
		\]
		where $L_{\mathrm{QP}}=tL_{\mathcal L}$ is independent of $M$.
	\end{enumerate}
	
	Consider any QPStat learning algorithm $\mathcal A$ for the random channels $E_\theta$ with $\theta\sim\mu_\Theta$, such that:
	\begin{itemize}
		\item[\emph{i.}] $\mathcal A$ uses at most $q$ QPStat queries with accuracy~$\tau$;
		\item[\emph{ii.}] There exist constants $\alpha>1/2$ and $\beta>0$ such that, for
		\[
		\Theta_{\mathrm{good}}
		:=\Bigl\{\theta\in\Theta:\ 
		\Pr_{\text{internal randomness}}\bigl[d_\diamond(\widehat E,E_\theta)\le\varepsilon\bigr]
		\ge\alpha
		\Bigr\},
		\]
		one has $\mu_\Theta(\Theta_{\mathrm{good}})\ge\beta$.
	\end{itemize}
	Then there exists a constant $c>0$, independent of $M$, such that for all sufficiently large $M$ (and in particular in the worst-case setting where one may take $\beta=1$),
	\begin{equation}\label{eq:q_exp_lower_bound}
		q
		\ \ge\ c\,(2\alpha-1)\,\beta\,
		\exp\bigl(c\,M\bigr) .
	\end{equation}
	In particular, $q$ must grow at least exponentially in the parameter dimension $M$, that is,
	\begin{equation}
		q \ge \exp\bigl(\Omega(M)\bigr) .
	\end{equation}
	In other words, under the QPStat access model, the random Lindbladian channel ensemble is exponentially hard to learn on average.
\end{theorem}

\begin{proof}
	By Lemma~\ref{lem:qpstat_concentration}, for any pair $(\rho,O)$ with $\|\rho\|_1\le1$ and $\|O\|_\infty\le1$ we have
	\begin{equation}\label{eq:Gamma_bound_step}
		\Pr_{\theta\sim\mu_\Theta}
		\Bigl(
		\bigl|\operatorname{Tr}[O E_\theta(\rho)]
		-\operatorname{Tr}[O\bar E(\rho)]\bigr| \geq \tau
		\Bigr)
		\le 2\exp\Bigl(-c_{\mathrm{par}}\frac{M\tau^2}{L_{\mathrm{QP}}^2}\Bigr) .
	\end{equation}
	Hence
	\begin{equation}\label{eq:Gamma_bound}
		\Gamma
		:=\max_{\rho,O}
		\Pr_{\theta\sim\mu_\Theta}
		\Bigl(
		\bigl|\operatorname{Tr}[O E_\theta(\rho)]
		-\operatorname{Tr}[O\bar E(\rho)]\bigr| \geq \tau
		\Bigr)
		\le 2\exp\Bigl(-c_{\mathrm{par}}\frac{M\tau^2}{L_{\mathrm{QP}}^2}\Bigr) ,
	\end{equation}
	where the maximum is over all $\|\rho\|_1\le1$, $\|O\|_\infty\le1$.
	
	On the other hand, viewing the channel ensemble $\{E_\theta\}$ with measure $\mu_\Theta$ as $(\mathcal E,\mu)$ in Lemma~\ref{lem:qpstat_many_vs_one}, and taking the reference channel to be $\bar E$, we obtain
	\begin{equation}\label{eq:qpstat_general_lb}
		q+1
		\ge \frac{(2\alpha-1)\,\beta}{\Gamma} .
	\end{equation}
	Combining~\eqref{eq:Gamma_bound} and~\eqref{eq:qpstat_general_lb}, we arrive at
	\begin{align*}
		q+1
		&\ge \frac{(2\alpha-1)\,\beta}{\Gamma} \\
		&\ge \frac{(2\alpha-1)\,\beta}{2}
		\exp\Bigl(\frac{c_{\mathrm{par}} M\tau^2}{L_{\mathrm{QP}}^2}\Bigr) .
	\end{align*}
	Define
	\begin{equation}\label{eq:def_c_constant}
		c := \frac{c_{\mathrm{par}}\tau^2}{2L_{\mathrm{QP}}^2}>0 .
	\end{equation}
	Then
	\begin{equation}\label{eq:q_plus_one_exp}
		q+1
		\;\ge\; \frac{(2\alpha-1)\,\beta}{2}\,\exp(2cM) .
	\end{equation}
	For sufficiently large $M$, the additive constant $1$ can be absorbed into the exponential term: there exists $M_0$ such that for all $M\ge M_0$,
	\[
	q
	\ge \frac{(2\alpha-1)\,\beta}{4}\,\exp(2cM) .
	\]
	Absorbing fixed multiplicative constants into $c$ (and adjusting $M_0$ if necessary), this implies the existence of a constant $c>0$, independent of $M$, such that
	\[
	q
	\ \ge\ c\,(2\alpha-1)\,\beta\,\exp\bigl(cM\bigr),
	\]
	which is exactly the claimed lower bound~\eqref{eq:q_exp_lower_bound}. In particular,
	\[
	q \ge \exp\bigl(\Omega(M)\bigr),
	\]
	i.e., the number of QPStat queries must grow at least exponentially in the parameter dimension~$M$. This completes the proof.
\end{proof}

\begin{remark}[Physical interpretation]
	Compared with our SQ-hardness results for output distributions, Theorem~\ref{thm:lindblad_qpstat_lower_bound} shows that even when the learner is granted the more powerful QPStat access (allowing arbitrary input states and measurements of arbitrary bounded observables), typical Lindbladian channels remain exponentially hard to learn in diamond distance on average, as long as the parameter dimension $M$ is large enough.
	
	From a physical point of view, Theorem~\ref{thm:lindblad_qpstat_lower_bound} says that, for a high-dimensional and physically natural random Lindbladian ensemble, even if the learner can probe the channel globally and in an entangled fashion in each QPStat query, the parameter rigidity and statistical concentration of typical Lindbladian noise still force any procedure that aims to reconstruct the microscopic Lindbladian structure (in the sense of small diamond distance) to use exponentially many statistical queries on average.
	Together with our earlier distribution-level SQ-hardness results, this provides a coherent picture: in the open-system setting, both learning the induced output distributions and learning the underlying channels themselves are generically exponentially hard in the QPStat/SQ framework.
\end{remark}

\section{Physically natural spherical random-perturbation Lindbladian ensembles}\label{zhang44}

In this subsection we present a physically natural model of random Lindbladian ensembles, and, within this model, rigorously realize the abstract hypothesis of ``concentration for random Lindbladian ensembles'' put forward earlier. N amely, on a high-dimensional parameter sphere, for all SQ test functions $\varphi$ the map
\[
\mathcal L \mapsto P_{\mathcal L}[\varphi]
\]
is uniformly Lipschitz, hence satisfies a L\'evy-type concentration inequality on the parameter space, which in turn implies exponential decay of $\operatorname{frac}(\mu_{\mathcal L},Q,\tau)$ in the SQ framework. Based on this model, we also discuss how the parameter dimension $M$ depends on the number of system qubits. Figure~\ref{fig:random_local_L} shows a one-dimensional spin chain with local dephasing and amplitude-damping jump operators on each site and two-body dissipative channels on each bond, which together generate the directions ${G_j}$ used in our spherical embedding.

\begin{figure}[t]
	\centering
	\includegraphics[width=0.7\textwidth]{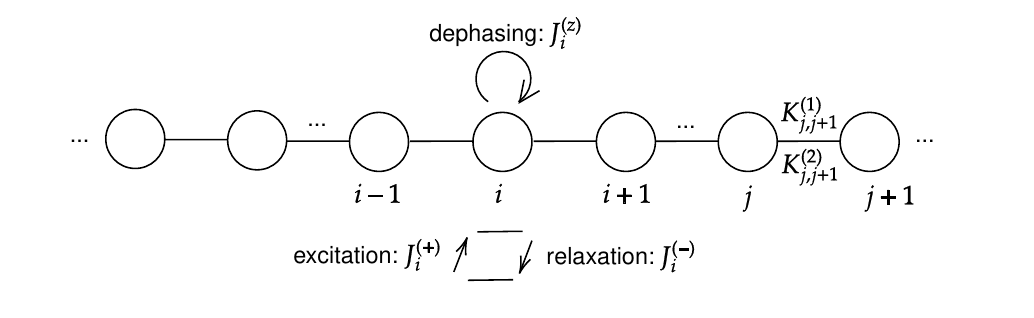}
	\caption{Random local Lindbladian on a one-dimensional spin chain.
		Each circle represents a qubit on the chain. On site $i$, local
		dephasing is described by the jump operator $J_i^{(z)}$, indicated
		by the loop arrow. Excitation and relaxation processes are
		implemented by the amplitude-raising and amplitude-lowering jump
		operators $J_i^{(+)}$ and $J_i^{(-)}$, respectively. On each bond
		$(j,j+1)$ we include two-body dissipative channels
		$K_{j,j+1}^{(1)}$ and $K_{j,j+1}^{(2)}$, representing, for
		example, incoherent spin-exchange and correlated dephasing. In our
		construction, the corresponding dissipators are collected into a
		set of local directions $\{G_j\}$ that enter the spherical
		random-perturbation parametrization
		$\mathcal L(\theta)=\mathcal L_{\mathrm{ref}}
		+ \frac{\delta}{\sqrt{M}}\sum_j \theta_j G_j$.}
	\label{fig:random_local_L}
\end{figure}

\begin{definition}[Random local Lindbladian ensemble]\label{def:random_local_L}
	Consider a one-dimensional chain of $N$ qubits with Hilbert space
	$\mathcal H := (\mathbb C^2)^{\otimes N}$ of dimension $d=2^N$.
	Let $\mathcal B(\mathcal H)$ be the operator algebra on the system.
	
	\begin{enumerate}
		\item 
		Let $\sigma\in\mathcal S_d^+$ be a fixed full-rank steady state (e.g., the Gibbs state of some local Hamiltonian $H$), and let $\mathcal L_{\mathrm{ref}}:\mathcal B(\mathcal H)\to\mathcal B(\mathcal H)$ be a given GKSL generator whose unique fixed point is $\sigma$:
		\[
		\mathcal L_{\mathrm{ref}}(\sigma) = 0,\qquad
		\Lambda_t^{(\mathrm{ref})} := e^{t\mathcal L_{\mathrm{ref}}}
		\text{ is a CPTP QMS}.
		\]
		
		\item 
		On each lattice site $i=1,\dots,N$, choose some local Lindblad jump operators, for example
		\[
		J_i^{(z)} := \sqrt{\gamma_z}\,\sigma_i^z,\quad
		J_i^{(+)} := \sqrt{\gamma_+}\,\sigma_i^+,\quad
		J_i^{(-)} := \sqrt{\gamma_-}\,\sigma_i^-,
		\]
		and on each nearest-neighbour pair $(i,i+1)$ choose some two-body jump operators, for example
		\[
		K_{i,i+1}^{(1)} := \sqrt{\kappa_1}\,\sigma_i^+\sigma_{i+1}^-,
		\qquad
		K_{i,i+1}^{(2)} := \sqrt{\kappa_2}\,\sigma_i^z\sigma_{i+1}^z.
		\]
		For each jump operator $L_\alpha$ (where the index $\alpha$ encodes position and type), define the dissipator
		\[
		\mathcal D_\alpha(\rho)
		:= L_\alpha\rho L_\alpha^\dagger
		-\frac12\{L_\alpha^\dagger L_\alpha,\rho\}.
		\]
		Each $\mathcal D_\alpha$ by itself is a GKSL generator (with the Hamiltonian part ignored).
		
		Now enumerate all these local dissipators as
		\(
		G_1,\dots,G_M
		\), i.e.,
		\[
		\{G_j\}_{j=1}^M
		=\{\mathcal D_\alpha\}_{\alpha\in\mathcal I},
		\]
		where the index set $\mathcal I$ has cardinality $M$ growing with the system size $N$ (e.g., $M=O(N)$ or $O(N^2)$).
		
		\item 
		Assume there exists a constant $C_G>0$ (independent of the parameter dimension $M$) such that
		\begin{equation}\label{eq:Gj_norm_bound}
			\|G_j\|_{1\to1} \le C_G,\qquad \forall j=1,\dots,M.
		\end{equation}
		
		\item 
		Take the parameter space to be the real sphere
		\[
		S^{M-1} := \{\theta\in\mathbb R^M:\ \|\theta\|_2=1\},
		\]
		choose a sufficiently small $\delta>0$ and define the linear parametrization
		\begin{equation}\label{eq:L_theta_def}
			\mathcal L(\theta)
			:=
			\mathcal L_{\mathrm{ref}}
			+ \frac{\delta}{\sqrt M}
			\sum_{j=1}^M \theta_j G_j, \qquad \theta\in S^{M-1}.
		\end{equation}
		
		\item 
		On $S^{M-1}$, take the uniform spherical measure $\mu_{\mathrm{sph}}$ and push it forward through the map $\theta\mapsto\mathcal L(\theta)$ to obtain
		\[
		\mu_{\mathcal L} := (\mathcal L)_\#\mu_{\mathrm{sph}},
		\]
		which is a random local Lindbladian ensemble defined on the space of Lindbladians.
		We denote this family by
		\[
		\mathcal F := \{\mathcal L(\theta):\ \theta\in S^{M-1}\}.
		\]
	\end{enumerate}
\end{definition}

\begin{lemma}[Uniform Lipschitz bound for SQ expectations]\label{lem:F_phi_Lipschitz}
	Maintain the setup above.
	For a fixed evolution time $t>0$, input state $\rho_{\mathrm{in}}\in\mathcal S_d^+$,
	and a POVM $\{M_x\}_{x\in X}$, define for any Lindbladian $\mathcal L$ the measurement distribution
	\[
	P_{\mathcal L}(x)
	:= \operatorname{Tr}\Bigl(M_x\,\Lambda_t^{(\mathcal L)}(\rho_{\mathrm{in}})\Bigr),
	\qquad x\in X,
	\]
	and for any test function $\varphi:X\to[-1,1]$,
	\[
	P_{\mathcal L}[\varphi]
	:= \sum_{x\in X}\varphi(x)P_{\mathcal L}(x)
	= \operatorname{Tr}\Bigl(O_\varphi\,\Lambda_t^{(\mathcal L)}(\rho_{\mathrm{in}})\Bigr),
	\quad
	O_\varphi := \sum_{x\in X}\varphi(x)M_x.
	\]
	
	Let $C_G$ be as in~\eqref{eq:Gj_norm_bound}, let $\mathcal L(\theta)$ for $\theta\in S^{M-1}$ be defined by~\eqref{eq:L_theta_def}, and set
	\[
	F_\varphi(\theta) := P_{\mathcal L(\theta)}[\varphi].
	\]
	Then there exists a constant
	\[
	L_0 := t\,\delta\,C_G,
	\]
	such that for any $\varphi:X\to[-1,1]$ and any $\theta,\theta'\in S^{M-1}$,
	\begin{equation}\label{eq:F_phi_Lipschitz}
		|F_\varphi(\theta)-F_\varphi(\theta')|
		\le L_0\,\|\theta-\theta'\|_2.
	\end{equation}
	In other words, $F_\varphi$ is Lipschitz on the parameter sphere, with a Lipschitz constant independent of $\varphi$.
\end{lemma}

\begin{proof}
	Similar to the proofs in Sec.~\ref{zming}, we proceed in four steps.
	
	\emph{Step 1: Bounding $\|O_\varphi\|_\infty$.}
	Since $\{M_x\}$ is a POVM, we have $M_x\ge0$ and $\sum_x M_x=\mathbb I$.
	For any state $\rho$,
	\[
	\operatorname{Tr}(O_\varphi\rho)
	= \sum_x \varphi(x)\operatorname{Tr}(M_x\rho),
	\]
	where $\operatorname{Tr}(M_x\rho)\ge0$ and $\sum_x\operatorname{Tr}(M_x\rho)=1$.
	Hence the right-hand side is a weighted average of values $\varphi(x)\in[-1,1]$, and therefore
	\(
	|\operatorname{Tr}(O_\varphi\rho)|\le1
	\)
	for all $\rho$.
	As the operator norm is given by $\|O\|_\infty=\sup_{\rho}|\operatorname{Tr}(O\rho)|$, we have $\|O_\varphi\|_\infty\le1$.
	
	\emph{Step 2: From channel difference to expectation difference.}
	For any Lindbladians $\mathcal L,\mathcal L'$, we compute
	\begin{align*}
		\bigl|P_{\mathcal L}[\varphi]-P_{\mathcal L'}[\varphi]\bigr|
		&= \Bigl|
		\operatorname{Tr}\Bigl(
		O_\varphi
		\bigl[\Lambda_t^{(\mathcal L)}-\Lambda_t^{(\mathcal L')}\bigr]
		(\rho_{\mathrm{in}})
		\Bigr)
		\Bigr|\\
		&\le
		\|O_\varphi\|_\infty\,
		\bigl\|\bigl[\Lambda_t^{(\mathcal L)}-\Lambda_t^{(\mathcal L')}\bigr]
		(\rho_{\mathrm{in}})\bigr\|_1\\
		&\le
		\|O_\varphi\|_\infty\,
		\bigl\|\Lambda_t^{(\mathcal L)}-\Lambda_t^{(\mathcal L')}\bigr\|_{1\to1}\,
		\|\rho_{\mathrm{in}}\|_1\\
		&\le
		\bigl\|\Lambda_t^{(\mathcal L)}-\Lambda_t^{(\mathcal L')}\bigr\|_{1\to1},
	\end{align*}
	where we used matrix Hölder's inequality
	$|\operatorname{Tr}(AB)|\le\|A\|_\infty\|B\|_1$, the definition of the induced $1\to1$ norm, and the bounds $\|O_\varphi\|_\infty\le1$, $\|\rho_{\mathrm{in}}\|_1=1$.
	
	\emph{Step 3: Bounding $\bigl\|\Lambda_t^{(\mathcal L)}-\Lambda_t^{(\mathcal L')}\bigr\|_{1\to1}$ via the perturbation formula.}
	For $\alpha\in[0,1]$ define
	\[
	\mathcal L_\alpha := (1-\alpha)\mathcal L+\alpha\mathcal L',
	\]
	and by Lemma~\ref{lem:convexity_GKSL_slice} we know that $\mathcal L_\alpha$ remains GKSL along the line segment, so $\Lambda_t^{(\alpha)}:=e^{t\mathcal L_\alpha}$ is CPTP and
	$\|\Lambda_t^{(\alpha)}\|_{1\to1}\le1$.
	The perturbation (Duhamel) formula gives
	\[
	\frac{\mathrm d}{\mathrm d\alpha}\Lambda_t^{(\alpha)}
	= \int_0^t
	\Lambda_{t-s}^{(\alpha)}\circ(\mathcal L'-\mathcal L)\circ\Lambda_s^{(\alpha)}
	\,\mathrm ds.
	\]
	Taking the $1\to1$ norm and using the contractivity of CPTP maps, we have
	\begin{align*}
		\Bigl\|
		\frac{\mathrm d}{\mathrm d\alpha}\Lambda_t^{(\alpha)}
		\Bigr\|_{1\to1}
		&\le \int_0^t
		\bigl\|\Lambda_{t-s}^{(\alpha)}\bigr\|_{1\to1}\,
		\|\mathcal L'-\mathcal L\|_{1\to1}\,
		\bigl\|\Lambda_s^{(\alpha)}\bigr\|_{1\to1}\,\mathrm ds\\
		&\le \int_0^t \|\mathcal L'-\mathcal L\|_{1\to1}\,\mathrm ds\\
		&= t\|\mathcal L'-\mathcal L\|_{1\to1}.
	\end{align*}
	Integrating over $\alpha\in[0,1]$ yields
	\begin{align*}
		\bigl\|\Lambda_t^{(\mathcal L')}-\Lambda_t^{(\mathcal L)}\bigr\|_{1\to1}
		&= \left\|\int_0^1
		\frac{\mathrm d}{\mathrm d\alpha}\Lambda_t^{(\alpha)}
		\,\mathrm d\alpha\right\|_{1\to1}\\
		&\le \int_0^1
		\Bigl\|
		\frac{\mathrm d}{\mathrm d\alpha}\Lambda_t^{(\alpha)}
		\Bigr\|_{1\to1}\,\mathrm d\alpha\\
		&\le \int_0^1 t\|\mathcal L'-\mathcal L\|_{1\to1}\,\mathrm d\alpha\\
		&= t\|\mathcal L'-\mathcal L\|_{1\to1}.
	\end{align*}
	
	\emph{Step 4: Specializing to $\mathcal L(\theta),\mathcal L(\theta')$.}
	From~\eqref{eq:L_theta_def},
	\[
	\mathcal L(\theta')-\mathcal L(\theta)
	= \frac{\delta}{\sqrt M}
	\sum_{j=1}^M (\theta_j'-\theta_j)G_j,
	\]
	and hence
	\begin{align*}
		\bigl\|\mathcal L(\theta')-\mathcal L(\theta)\bigr\|_{1\to1}
		&\le \frac{\delta}{\sqrt M}
		\sum_{j=1}^M |\theta_j'-\theta_j|\,\|G_j\|_{1\to1}\\
		&\le \frac{\delta C_G}{\sqrt M}
		\sum_{j=1}^M |\theta_j'-\theta_j|\\
		&\le \frac{\delta C_G}{\sqrt M}\,\sqrt M\,\|\theta'-\theta\|_2\\
		&= \delta C_G \,\|\theta'-\theta\|_2,
	\end{align*}
	where we used the Cauchy--Schwarz inequality
	$\sum_j|a_j|\le\sqrt M\|a\|_2$.
	Combining this with the previous bound, we obtain
	\[
	\bigl\|\Lambda_t^{(\mathcal L(\theta'))}-\Lambda_t^{(\mathcal L(\theta))}\bigr\|_{1\to1}
	\le t\,\delta C_G\,\|\theta'-\theta\|_2.
	\]
	Substituting into the estimate from Step~2 gives
	\[
	|F_\varphi(\theta')-F_\varphi(\theta)|
	\le t\,\delta C_G\,\|\theta'-\theta\|_2.
	\]
	Setting $L_0:=t\delta C_G$ yields~\eqref{eq:F_phi_Lipschitz}.
	Because the right-hand side is independent of $\varphi$, this Lipschitz constant applies uniformly to all test functions.
\end{proof}

\begin{remark}
	Lemma~\ref{lem:F_phi_Lipschitz} gives the second part of the abstract assumption:
``For every test function $\varphi$, the map
		$F_\varphi(\mathcal L):=P_{\mathcal L}[\varphi]$ is $L_0$-Lipschitz on $\mathcal F$,
		with $L_0$ independent of $\varphi$.''
\end{remark}

\begin{theorem}[L\'evy concentration and \texorpdfstring{$\operatorname{frac}$}{frac} upper bound for random local Lindbladian ensembles]
	\label{thm:L_Levy_frac}
	Under the assumptions of Definition~\ref{def:random_local_L} and Lemma~\ref{lem:F_phi_Lipschitz}, suppose the parameter dimension $M$ is sufficiently large and the spherical measure $\mu_{\mathrm{sph}}$ satisfies the standard L\'evy concentration inequality: there exists a constant $c_{\mathrm{sph}}>0$ such that for any
	$L_0$-Lipschitz function $f:S^{M-1}\to\mathbb R$ and any $\tau>0$,
	\begin{equation}\label{eq:Levy_on_sphere}
		\Pr_{\theta\sim\mu_{\mathrm{sph}}}
		\bigl(|f(\theta)-\mathbb Ef|\ge\tau\bigr)
		\le 2\exp\Bigl(-c_{\mathrm{sph}}\frac{M\tau^2}{L_0^2}\Bigr).
	\end{equation}
	
	For any test function $\varphi:X\to[-1,1]$, set
	\[
	Q[\varphi]
	:= \mathbb E_{\mathcal L\sim\mu_{\mathcal L}}
	\bigl[P_{\mathcal L}[\varphi]\bigr]
	= \mathbb E_{\theta\sim\mu_{\mathrm{sph}}}
	\bigl[F_\varphi(\theta)\bigr],
	\]
	and define in the SQ framework
	\[
	\operatorname{frac}(\mu_{\mathcal L},Q,\tau)
	:= \sup_{\varphi:X\to[-1,1]}
	\Pr_{\mathcal L\sim\mu_{\mathcal L}}
	\bigl(|P_{\mathcal L}[\varphi]-Q[\varphi]|\ge\tau\bigr).
	\]
	Then for any $\tau>0$ we have the exponentially small bound
	\begin{equation}\label{eq:frac_L_Q_tau_bound}
		\operatorname{frac}(\mu_{\mathcal L},Q,\tau)
		\le 2\exp\Bigl(-c_{\mathrm{sph}}\frac{M\tau^2}{L_0^2}\Bigr),
		\qquad L_0=t\delta C_G.
	\end{equation}
\end{theorem}

\begin{proof}
	By Lemma~\ref{lem:F_phi_Lipschitz}, for each fixed test function $\varphi$ the function $F_\varphi(\theta)=P_{\mathcal L(\theta)}[\varphi]$ is $L_0$-Lipschitz on $S^{M-1}$. Applying the L\'evy inequality~\eqref{eq:Levy_on_sphere}, we obtain for any $\tau>0$,
	\[
	\Pr_{\theta\sim\mu_{\mathrm{sph}}}
	\bigl(|F_\varphi(\theta)-\mathbb E F_\varphi|\ge\tau\bigr)
	\le 2\exp\Bigl(-c_{\mathrm{sph}}\frac{M\tau^2}{L_0^2}\Bigr).
	\]
	Noting that the pushforward of $\mu_{\mathrm{sph}}$ under $\mathcal L(\theta)$ is precisely $\mu_{\mathcal L}$, and that $\mathbb EF_\varphi = Q[\varphi]$, we obtain
	\[
	\Pr_{\mathcal L\sim\mu_{\mathcal L}}
	\bigl(|P_{\mathcal L}[\varphi]-Q[\varphi]|\ge\tau\bigr)
	\le 2\exp\Bigl(-c_{\mathrm{sph}}\frac{M\tau^2}{L_0^2}\Bigr),
	\]
	with the right-hand side independent of $\varphi$. Taking the supremum over $\varphi$ then yields
	\[
	\operatorname{frac}(\mu_{\mathcal L},Q,\tau)
	\le 2\exp\Bigl(-c_{\mathrm{sph}}\frac{M\tau^2}{L_0^2}\Bigr),
	\]
	which is precisely~\eqref{eq:frac_L_Q_tau_bound}.
\end{proof}

Therefore,~\eqref{eq:frac_L_Q_tau_bound} gives a Liouvillian concentration form of $\operatorname{frac}(\mu_{\mathcal L},Q,\tau)$ for the random local Lindbladian ensemble, which can be directly inserted into the general SQ-hardness criteria. QPStat is also a similar proof. To avoid repetition, it will not be elaborated here.

\section{far-from-$Q$}\label{zhang3}

In Sec.~\ref{zhang1} we established, for the random Lindbladian parameter-sphere ensemble, a uniform Lipschitz constant $L_0=t\delta C_G$ for the SQ test functions
\(
F_\varphi(\theta)=P_{\mathcal L(\theta)}[\varphi]
\),
and by L\'evy’s lemma on the sphere obtained an exponential upper bound on
\(
\operatorname{frac}(\mu_{\mathcal L},Q,\tau)
\)
(Theorem~\ref{thm:frac_bound_open_system}). In this section, using a measure-theoretic framework, we further derive:
\begin{itemize}
	\item[1.] concentration of $d_{\mathrm{TV}}(P_{\mathcal L},Q)$ on the parameter sphere (far-from-$Q$);
	\item[2.] a far-from-$D$ result for any fixed $D$.
\end{itemize}

\subsection{Lipschitz continuity of TV distance with respect to parameters}

We first upgrade Lemma~\ref{lem:theta_Lipschitz} to the level of TV distance. For convenience in computing the Lipschitz constant, we give an equivalent variational form of the TV distance.

\begin{lemma}[Variational characterizations of total variation distance]\label{lem:TV_variational}
	Let $X$ be a finite set and $P,Q\in\mathcal D_X$ two probability distributions on $X$. By Definition~\ref{def:TV_and_expectation}, their total variation distance is
	\[
	d_{\mathrm{TV}}(P,Q)
	:=\frac{1}{2}\sum_{x\in X}\bigl|P(x)-Q(x)\bigr|.
	\]
	Then the following equivalent characterizations hold:
	\begin{align}
		d_{\mathrm{TV}}(P,Q)
		&= \max_{S\subseteq X}\bigl|P(S)-Q(S)\bigr|,
		\label{eq:TV_set_variation}\\[0.5em]
		&= \frac12 \sup_{\varphi:X\to[-1,1]}
		\bigl|P[\varphi]-Q[\varphi]\bigr|,
		\label{eq:TV_test_function}
	\end{align}
	where
	\(
	P[\varphi]:=\sum_{x\in X}\varphi(x)P(x)
	\)
	denotes the expectation of $\varphi$ under $P$.
\end{lemma}

\begin{proof}
	Define
	\[
	S_+ := \{x\in X:\ P(x)\ge Q(x)\},\qquad
	S_- := X\setminus S_+.
	\]
	Then
	\begin{align*}
		\sum_{x\in X}\bigl|P(x)-Q(x)\bigr|
		&= \sum_{x\in S_+}\bigl(P(x)-Q(x)\bigr)
		+\sum_{x\in S_-}\bigl(Q(x)-P(x)\bigr).
	\end{align*}
	On the other hand, since $P$ and $Q$ are both probability distributions,
	\[
	\sum_{x\in X}\bigl(P(x)-Q(x)\bigr) = 0,
	\]
	so
	\[
	\sum_{x \in S_{+}}(P(x)-Q(x))=-\sum_{x \in S_{-}}(P(x)-Q(x)) .
	\]
	Hence
	\[
	\sum_{x\in S_+}\bigl(P(x)-Q(x)\bigr)
	=
	\sum_{x\in S_-}\bigl(Q(x)-P(x)\bigr)
	= \frac12\sum_{x\in X}\bigl|P(x)-Q(x)\bigr|.
	\]
	Therefore
	\[
	d_{\mathrm{TV}}(P,Q)
	= \frac12\sum_{x\in X}\bigl|P(x)-Q(x)\bigr|
	= P(S_+)-Q(S_+).
	\]
	
	Next, for any subset $S\subseteq X$, decompose it as $S=(S\cap S_+)\cup(S\cap S_-)$. On $S \cap S_{-}$ we have $P(x)-Q(x)<0$, so discarding this part can only increase the sum; then extending $S \cap S_{+}$ to the whole $S_{+}$ can further increase the sum. Thus
	\begin{align*}
		P(S)-Q(S)
		&= \sum_{x\in S}\bigl(P(x)-Q(x)\bigr)\\
		&= \sum_{x\in S\cap S_+}\bigl(P(x)-Q(x)\bigr)
		+\sum_{x\in S\cap S_-}\bigl(P(x)-Q(x)\bigr)\\
		&\le \sum_{x\in S_+}\bigl(P(x)-Q(x)\bigr)
		= d_{\mathrm{TV}}(P,Q),
	\end{align*}
	and similarly
	\(
	Q(S)-P(S)\le d_{\mathrm{TV}}(P,Q)
	\). Therefore,
	\[
	\bigl|P(S)-Q(S)\bigr|
	\le d_{\mathrm{TV}}(P,Q)
	\quad\text{for all }S\subseteq X.
	\]
	Choosing $S=S_+$ yields
	\(
	\max_S |P(S)-Q(S)|
	= d_{\mathrm{TV}}(P,Q)
	\), which proves~\eqref{eq:TV_set_variation}.
	
	We now prove
	\[
	\sup_{\varphi:X\to[-1,1]}\bigl|P[\varphi]-Q[\varphi]\bigr|
	\le 2\,d_{\mathrm{TV}}(P,Q).
	\]
	For any $\varphi:X\to[-1,1]$,
	\begin{align*}
		P[\varphi]-Q[\varphi]
		&= \sum_{x\in X}\varphi(x)\bigl(P(x)-Q(x)\bigr),
	\end{align*}
	so
	\begin{align*}
		\bigl|P[\varphi]-Q[\varphi]\bigr|
		&\le \sum_{x\in X}|\varphi(x)|\,\bigl|P(x)-Q(x)\bigr|\\
		&\le \sum_{x\in X}\bigl|P(x)-Q(x)\bigr|\\
		&= 2\,d_{\mathrm{TV}}(P,Q),
	\end{align*}
	where we used $|\varphi(x)|\le1$. Taking the supremum over all $\varphi$ gives
	\[
	\sup_{\varphi:X\to[-1,1]}\bigl|P[\varphi]-Q[\varphi]\bigr|
	\le 2\,d_{\mathrm{TV}}(P,Q).
	\]
	
	Next we construct a test function that attains this bound. Define
	\[
	\varphi_*(x)
	:= \begin{cases}
		+1, & P(x)\ge Q(x),\\
		-1, & P(x)< Q(x),
	\end{cases}
	\]
	(when $P(x)=Q(x)$ we may choose any value in $[-1,1]$ without affecting the calculation below).
	Then
	\begin{align*}
		P[\varphi_*]-Q[\varphi_*]
		&= \sum_{x\in X}\varphi_*(x)\bigl(P(x)-Q(x)\bigr)\\
		&= \sum_{x\in S_+} (+1)\bigl(P(x)-Q(x)\bigr)
		+\sum_{x\in S_-} (-1)\bigl(P(x)-Q(x)\bigr)\\
		&= \sum_{x\in S_+}\bigl(P(x)-Q(x)\bigr)
		+\sum_{x\in S_-}\bigl(Q(x)-P(x)\bigr)\\
		&= \sum_{x\in X}\bigl|P(x)-Q(x)\bigr|\\
		&= 2\,d_{\mathrm{TV}}(P,Q).
	\end{align*}
	Hence
	\[
	\sup_{\varphi:X\to[-1,1]}\bigl|P[\varphi]-Q[\varphi]\bigr|
	\ge \bigl|P[\varphi_*]-Q[\varphi_*]\bigr|
	= 2\,d_{\mathrm{TV}}(P,Q).
	\]
	
	Combining the upper and lower bounds, we obtain
	\[
	\sup_{\varphi:X\to[-1,1]}\bigl|P[\varphi]-Q[\varphi]\bigr|
	= 2\,d_{\mathrm{TV}}(P,Q),
	\]
	which is equivalent to
	\[
	d_{\mathrm{TV}}(P,Q)
	= \frac12\sup_{\varphi:X\to[-1,1]}\bigl|P[\varphi]-Q[\varphi]\bigr|.
	\]
	This proves~\eqref{eq:TV_test_function}.
\end{proof}

\begin{lemma}[Lipschitz continuity of TV distance with respect to parameters]\label{lem:TV_Lipschitz_theta}
	Under the assumptions of Lemma~\ref{lem:theta_Lipschitz}, for any reference distribution $Q\in\mathcal D_X$ define, using Lemma~\ref{lem:TV_variational},
	\[
	F_{\mathrm{TV}}(\theta)
	:= d_{\mathrm{TV}}(P_{\mathcal L(\theta)},Q)
	= \frac12 \sup_{\varphi:X\to[-1,1]}
	\bigl|P_{\mathcal L(\theta)}[\varphi]-Q[\varphi]\bigr|.
	\]
	Then for any $\theta,\theta'\in\Theta$,
	\[
	\bigl|F_{\mathrm{TV}}(\theta)-F_{\mathrm{TV}}(\theta')\bigr|
	\le L_{\mathrm{TV}}\,\|\theta-\theta'\|_2,\qquad
	L_{\mathrm{TV}}:=\frac{L_0}{2}=\frac{t\delta C_G}{2}.
	\]
\end{lemma}

\begin{proof}
	Let $P_\theta:=P_{\mathcal L(\theta)}$. By the variational representation of the TV distance,
	\[
	F_{\mathrm{TV}}(\theta)
	= \frac12 \sup_{\varphi:X\to[-1,1]} |P_\theta[\varphi]-Q[\varphi]|.
	\]
	Thus
	\begin{align*}
		\bigl|F_{\mathrm{TV}}(\theta)-F_{\mathrm{TV}}(\theta')\bigr|
		&= \frac12\Bigl|
		\sup_{\varphi} |P_\theta[\varphi]-Q[\varphi]|
		-\sup_{\varphi} |P_{\theta'}[\varphi]-Q[\varphi]|
		\Bigr| \\
		&\le \frac12 \sup_{\varphi}
		\Bigl|
		|P_\theta[\varphi]-Q[\varphi]|
		-|P_{\theta'}[\varphi]-Q[\varphi]|
		\Bigr|\\
		&\le \frac12 \sup_{\varphi}
		|P_\theta[\varphi]-P_{\theta'}[\varphi]|.
	\end{align*}
	In the last inequality we used that for any real numbers $x,y,a$,
	\(
	\bigl||x-a|-|y-a|\bigr|\le|x-y|
	\).
	By Lemma~\ref{lem:theta_Lipschitz}, for any $\varphi:X\to[-1,1]$,
	\[
	|P_\theta[\varphi]-P_{\theta'}[\varphi]|
	= |F_\varphi(\theta)-F_\varphi(\theta')|
	\le L_0\,\|\theta-\theta'\|_2.
	\]
	Substituting into the previous inequality gives
	\[
	|F_{\mathrm{TV}}(\theta)-F_{\mathrm{TV}}(\theta')|
	\le \frac12 L_0\,\|\theta-\theta'\|_2,
	\]
	so $L_{\mathrm{TV}}=L_0/2$.
\end{proof}

\subsection{Mean assumption and far-from-$Q$ concentration}

We now introduce a nonzero-mean assumption. Intuitively this is an ``open-system version of the Porter-Thomas mean assumption''. This assumption is natural, under sufficiently deep random circuits the output probability distribution approaches the Porter-Thomas distribution~\cite{boixo2018characterizing}, so many global indicators (such as cross-entropy difference and heavy-output probability) converge to constants in the high-dimensional limit~\cite{hangleiter2023computational}. In the case of Nietner’s random unitary circuits~\cite{nietner2025average}, this constant is computed by Gaussian integration as
$\mathbb{E} d_{\mathrm{TV}}\left(P_U, U\right) \approx 1 / e$.
This computation relies on the fact that the TV distance depends only on a few independent Gaussian variables. In our setting, however, the exponential map makes the dependence non-linear on the entire vector of parameters, so we can no longer compute the mean exactly in the same way. Nevertheless, in the linear-response regime we can give an explicit formula, as we will show in Sec.~\ref{zhang4}; and in Sec.~\ref{zhang5} we present a simple but concrete example in which this mean takes a simple constant value and perform some numerical simulations, lending support to our assumption.

\begin{assumption}[Open-system Porter-Thomas mean assumption]\label{ass:mean_TV}
	There exists a constant $m_0>0$ and an error term $\epsilon_{\mathrm{mean}}>0$ which tends to $0$ as the system parameters (e.g., the Hilbert space dimension $d$ or the parameter dimension $M$) grow, such that the TV distance
	\[
	F_{\mathrm{TV}}(\theta)
	= d_{\mathrm{TV}}(P_{\mathcal L(\theta)},Q)
	\]
	has mean over the parameter-sphere ensemble $(\Theta,\mu_\Theta)$ satisfying
	\[
	m:=\mathbb E_{\theta\sim\mu_\Theta}F_{\mathrm{TV}}(\theta)
	\in [m_0-\epsilon_{\mathrm{mean}},\ m_0+\epsilon_{\mathrm{mean}}].
	\]
\end{assumption}

Given the Lipschitz constant and the mean assumption, we can use L\'evy’s lemma on the parameter sphere (Lemma~\ref{lem:levy_sphere}) to obtain a concentration result for $d_{\mathrm{TV}}(P_{\mathcal L},Q)$.

\begin{theorem}[Concentration of TV distance with respect to the reference $Q$]\label{thm:far_from_Q_open}
	In the setting of Definition~\ref{def:param_sphere_ensemble}, assume that Lemmas~\ref{lem:theta_Lipschitz} and \ref{lem:TV_Lipschitz_theta} and Assumption~\ref{ass:mean_TV} all hold. Then for any $\xi>0$,
	\[
	\Pr_{\theta\sim\mu_\Theta}
	\Bigl(
	\bigl|d_{\mathrm{TV}}(P_{\mathcal L(\theta)},Q)-m_0\bigr|
	\ge \xi+\epsilon_{\mathrm{mean}}
	\Bigr)
	\le
	2\exp\!\left(
	-c_{\mathrm{par}}\,
	\frac{M\,\xi^2}{L_{\mathrm{TV}}^2}
	\right),
	\]
	where $c_{\mathrm{par}}$ is the constant from Lemma~\ref{lem:levy_sphere}, $M$ is the parameter dimension, and $L_{\mathrm{TV}}=t\delta C_G/2$.
\end{theorem}

\begin{proof}
	By Lemma~\ref{lem:TV_Lipschitz_theta}, $F_{\mathrm{TV}}(\theta)$ is $L_{\mathrm{TV}}$-Lipschitz on $\Theta\subseteq S^{M-1}$.  
	Applying Lemma~\ref{lem:levy_sphere}, we obtain concentration around the mean $m:=\mathbb EF_{\mathrm{TV}}$:
	\[
	\Pr_{\theta\sim\mu_\Theta}
	\bigl(|F_{\mathrm{TV}}(\theta)-m|\ge s\bigr)
	\le 2\exp\!\left(
	-c_{\mathrm{par}}\,
	\frac{M\,s^2}{L_{\mathrm{TV}}^2}
	\right).
	\]
	By Assumption~\ref{ass:mean_TV}, $|m-m_0|\le\epsilon_{\mathrm{mean}}$, hence
	\[
	|F_{\mathrm{TV}}(\theta)-m_0|
	\ge \xi+\epsilon_{\mathrm{mean}}
	\ \Rightarrow\
	|F_{\mathrm{TV}}(\theta)-m|
	\ge \xi.
	\]
	Thus
	\[
	\Pr_{\theta}\Bigl(|F_{\mathrm{TV}}(\theta)-m_0|\ge\xi+\epsilon_{\mathrm{mean}}\Bigr)
	\le
	\Pr_{\theta}\bigl(|F_{\mathrm{TV}}(\theta)-m|\ge\xi\bigr)
	\le
	2\exp\!\left(
	-c_{\mathrm{par}}\,
	\frac{M\,\xi^2}{L_{\mathrm{TV}}^2}
	\right).
	\]
\end{proof}

\begin{corollary}[Typicality of far-from-$Q$]\label{cor:far_from_Q_typical}
	Under the assumptions of Theorem~\ref{thm:far_from_Q_open}, choose $\xi=m_0/4$ and assume the system size is large enough that $\epsilon_{\mathrm{mean}}\le m_0/4$. Moreover, we have
	\[
	\left\{\theta:\ d_{\mathrm{TV}}(P_{\mathcal L(\theta)},Q)\le \frac{m_0}{2}\right\} \subseteq\left\{\theta:\left|F(\theta)-m_0\right| \geq \frac{m_0}{2}\right\},
	\]
	so the probability bound gives
	\[
	\Pr_{\theta\sim\mu_\Theta}
	\Bigl(
	d_{\mathrm{TV}}(P_{\mathcal L(\theta)},Q)\le \frac{m_0}{2}
	\Bigr)
	\le
	2\exp\!\left(
	-c_{\mathrm{par}}\,
	\frac{M\,m_0^2}{16L_{\mathrm{TV}}^2}
	\right).
	\]
	In other words, as the parameter dimension $M$ grows linearly, for the vast majority of Lindbladians the output distribution $P_{\mathcal L}$ is at least a constant distance $m_0/2$ away from the reference distribution $Q$ in TV distance.
\end{corollary}

\subsection{Average-case SQ complexity}\label{result 2}

Based on
\begin{itemize}
	\item[(a)] the Lipschitz and L\'evy concentration on the parameter-sphere ensemble (Theorem~\ref{thm:frac_bound_open_system}), which gives an exponentially small upper bound on
	\(
	\max_{\varphi:X\to[-1,1]}\Pr_{\mathcal L\sim\mu_{\mathcal L}}[|P_{\mathcal L}[\varphi]-Q[\varphi]|>\tau]
	\)
	(the denominato);
	\item[(b)] the open-system Porter-Thomas mean assumption and its induced ``far-from-$Q$/far-from-any-fixed-$D$'' volume estimates (Theorem~\ref{thm:far_from_Q_open} and its Corollary~\ref{cor:far_from_Q_typical}), which provide control of
	\(
	\beta-\Pr_{\mathcal L\sim\mu_{\mathcal L}}[d_{\mathrm{TV}}(P_{\mathcal L},Q)\le\epsilon+\tau]
	\)
	(the numerator).
\end{itemize}
Combining these with the SQ learning and decision framework~\cite{feldman2017general,nietner2025average} (Theorems~\ref{thm:avg_case_SQ_lower}, \ref{thm:random_avg_SQ_complexity} and related lemmas), we can already state at the beginning of the paper the following average-case SQ complexity statements.

\begin{itemize}
	\item[i.] The average-case SQ lower bound for deterministic decision problems is likewise exponential.
	
	Let $\mu_{\mathcal L}$ be the distribution over Lindbladians induced by the parameter-sphere ensemble, and define the reference distribution $Q$ as the ensemble average
	\(
	Q[\varphi]=\mathbb E_{\mathcal L\sim\mu_{\mathcal L}}P_{\mathcal L}[\varphi]
	\).
	From Theorem~\ref{thm:frac_bound_open_system} and TV-Lipschitz continuity, we already know that for any $\tau>0$,
	\[
	\max_{\varphi:X\to[-1,1]}
	\Pr_{\mathcal L\sim\mu_{\mathcal L}}\bigl[|P_{\mathcal L}[\varphi]-Q[\varphi]|>\tau\bigr]
	\ =:\ \operatorname{frac}(\mu_{\mathcal L},Q,\tau)
	\ \le\
	2\exp\!\left(-c_{\mathrm{par}}\frac{M\tau^2}{L_{0}^2}\right),
	\]
	On the other hand, the ``far-from-$Q$'' result (Theorem~\ref{thm:far_from_Q_open} and its corollary) yields, for suitable constant $m_0>0$ and sufficiently large system size,
	\[
	\Pr_{\mathcal L\sim\mu_{\mathcal L}}\bigl[d_{\mathrm{TV}}(P_{\mathcal L},Q)\le\epsilon+\tau\bigr]
	\ \le\ 2\exp\!\left(-c'_{\mathrm{par}}M\right).
	\]
	Hence for any fixed $\beta\in(0,1)$, if we only require an algorithm to succeed on a subset of Lindbladians of $\mu_{\mathcal L}$-measure at least $\beta$, then the ``numerator''
	\[
	\beta-\Pr_{\mathcal L\sim\mu_{\mathcal L}}\bigl[d_{\mathrm{TV}}(P_{\mathcal L},Q)\le\epsilon+\tau\bigr]
	\]
	still remains of order constant for sufficiently large $M$ (for example $\ge\beta/2$). Substituting these two estimates into the average-case decision lower bound (a special case of Theorem~\ref{thm:avg_case_SQ_lower}), we obtain that any deterministic SQ algorithm that attempts to correctly decide
	``$P\in\mathcal D$'' vs.\ ``$P=Q$'' on a subset of $\mu_{\mathcal L}$-measure at least $\beta$ must have average-case statistical query complexity
	\[
	q+1 \ \gtrsim\ \frac{\beta}{4}\,
	\exp\!\left(
	c_{\mathrm{par}}\frac{M\tau^2}{L_{\mathrm{0}}^2}
	\right)
	\;=\;
	\exp\!\Bigl(\Omega(M)\Bigr).
	\]
	That is, the average-case SQ query complexity of the decision problem still grows exponentially in the parameter dimension $M$.
	
	\item[ii.] The average-case SQ lower bound for randomized decision problems carries an $(\alpha,\beta)$ factor but is still exponential.
	
	For randomized (including classical probabilistic or quantum) decision algorithms, Theorem~\ref{thm:random_avg_SQ_complexity} states that if the algorithm, with internal randomness, succeeds with probability at least $\alpha>1/2$ and correctly decides $\mathcal D$ vs.\ $Q$ on a subset of Lindbladians of $\mu_{\mathcal L}$-measure at least $\beta$, then
	\[
	q+1
	\ \ge\
	2\cdot
	\frac{\bigl(\alpha-\tfrac12\bigr)\cdot
		\Bigl(\beta-\Pr_{\mathcal L\sim\mu_{\mathcal L}}[d_{\mathrm{TV}}(P_{\mathcal L},Q)\le\epsilon+\tau]\Bigr)}
	{\displaystyle\max_{\varphi:X\to[-1,1]}
		\Pr_{\mathcal L\sim\mu_{\mathcal L}}[|P_{\mathcal L}[\varphi]-Q[\varphi]|>\tau]}.
	\]
	Again substituting the ``numerator'' and ``denominator'' estimates above, for sufficiently large $M$ we obtain
	\[
	q+1 \ \gtrsim\ (\alpha-\tfrac12)\,\beta\,
	\exp\!\left(
	c_{\mathrm{par}}\frac{M\tau^2}{L_{\mathrm{0}}^2}
	\right).
	\]
	Therefore, as long as we require a fixed constant success probability $\alpha>1/2$ and a fixed constant coverage $\beta>0$ in the average sense, the average-case SQ query complexity of randomized decision algorithms is likewise exponential in $M$, with an additional natural linear prefactor $(\alpha-\tfrac12)\beta$.
	
	\item[iii.] Average-case SQ learning complexity of random Lindbladian microstructure: learning is no easier than decision.
	
	Furthermore, using the general reduction ``learning is no easier than decision'' (Theorem~\ref{thm:avg_case_SQ_lower} itself and its randomized variant~\ref{thm:random_avg_SQ_complexity}), we directly obtain exponential average-case lower bounds for the open-system SQ learning problem. Consider the SQ learning task: given access to a SQ oracle for an unknown $\mathcal L\sim\mu_{\mathcal L}$ (via a fixed time $t$, fixed input state $\rho_{\mathrm{in}}$ and a finite POVM interface), we are required to output a distribution $\widehat P$ such that
	\(
	d_{\mathrm{TV}}(\widehat P,P_{\mathcal L})<\epsilon
	\).
	
	\begin{itemize}
		\item[1.] Deterministic average-case learning algorithms: suppose a deterministic algorithm, with at most $q$ many $\tau$-accurate statistical queries, succeeds in $\epsilon$-learning on a subfamily of Lindbladians of $\mu_{\mathcal L}$-measure at least $\beta$. Then Theorem~\ref{thm:avg_case_SQ_lower} implies
		\[
		q+1 \ \ge\
		\frac{
			\beta - \Pr_{\mathcal L\sim\mu_{\mathcal L}}\bigl[d_{\mathrm{TV}}(P_{\mathcal L},Q)\le\epsilon+\tau\bigr]
		}{
			\displaystyle\max_{\varphi:X\to[-1,1]}
			\Pr_{\mathcal L\sim\mu_{\mathcal L}}[|P_{\mathcal L}[\varphi]-Q[\varphi]|>\tau]
		}.
		\]
		Plugging in the estimates ``numerator $\sim \beta$'' and ``denominator $\sim \exp(-cM)$'' yields, for sufficiently large $M$,
		\[
		q+1 \ \gtrsim\ \beta\,
		\exp\!\left(
		c_{\mathrm{par}}\frac{M\tau^2}{L_{\mathrm{0}}^2}
		\right),
		\qquad
		q_{\mathrm{det}}^{\mathrm{avg\text{-}learn}}(\epsilon,\tau;\beta)
		\ =\ \exp\!\bigl(\Omega(M)\bigr).
		\]
		In other words, deterministic SQ learning of the full open-system dynamical family still requires an exponential number of statistical queries.
		
		\item[2.] Randomized (classical or quantum) average-case learning algorithms: if an algorithm, with internal randomness, succeeds with probability at least $\alpha>1/2$ and completes $\epsilon$-learning on a subfamily of Lindbladians of $\mu_{\mathcal L}$-measure at least $\beta$, then by Theorem~\ref{thm:random_avg_SQ_complexity},
		\[
		q+1
		\ \ge\
		2\cdot
		\frac{\bigl(\alpha-\tfrac12\bigr)\cdot
			\Bigl(\beta-\Pr_{\mathcal L\sim\mu_{\mathcal L}}[d_{\mathrm{TV}}(P_{\mathcal L},Q)\le\epsilon+\tau]\Bigr)}
		{\displaystyle\max_{\varphi:X\to[-1,1]}
			\Pr_{\mathcal L\sim\mu_{\mathcal L}}[|P_{\mathcal L}[\varphi]-Q[\varphi]|>\tau]}.
		\]
		Substituting our estimates for the ``numerator'' and ``denominator'' and taking $M$ sufficiently large, we obtain
		\[
		q+1 \ \gtrsim\ (\alpha-\tfrac12)\,\beta\,
		\exp\!\left(
		c_{\mathrm{par}}\frac{M\tau^2}{L_{\mathrm{0}}^2}
		\right),
		\qquad
		q_{\mathrm{rand}}^{\mathrm{avg\text{-}learn}}(\epsilon,\tau;\alpha,\beta)
		\ =\ \exp\!\bigl(\Omega(M)\bigr).
		\]
		This shows that randomness (including internal randomness of quantum algorithms) cannot significantly reduce the average-case SQ query complexity of learning open-system dynamics.
	\end{itemize}
\end{itemize}

We may conclude that as long as the randomly parametrized Lindbladian family satisfies Lipschitz and L\'evy concentration on a high-dimensional parameter sphere, and the TV distance to the reference distribution $Q$ enjoys a Porter-Thomas-like nonzero-mean property, then in the SQ model—whether for decision or learning problems, deterministic or randomized—the average-case statistical query complexity always grows exponentially with the parameter dimension $M$.

\subsection{From far-from-$Q$ to far-from-any-$D$}

We have an additional interesting corollary: we can derive an open-system version of ``far-from-any-$D$'' from the far-from-$Q$ result, to control the TV distance with respect to any fixed distribution $D$.

\begin{lemma}[far-from-$Q\Rightarrow$ far-from-$D$]\label{lem:OS_far_from_Q_to_D}
	Suppose $Q,D\in\mathcal D_X$ satisfy
	\(
	d_{\mathrm{TV}}(Q,D)>\epsilon+\tau
	\)
	for some $\epsilon,\tau>0$. Let $\mu_{\mathcal L}$ be the random Lindbladian ensemble above. For each $\varphi:X\to[-1,1]$, define
	\(
	P_{\mathcal L}[\varphi]
	\)
	and $Q[\varphi]$ as in Theorem~\ref{thm:frac_bound_open_system}, and set
	\[
	\operatorname{frac}(\mu_{\mathcal L},Q,\tau)
	:= \sup_{\varphi:X\to[-1,1]}
	\Pr_{\mathcal L\sim\mu_{\mathcal L}}
	\bigl(|P_{\mathcal L}[\varphi]-Q[\varphi]|\ge\tau\bigr).
	\]
	Then
	\[
	\Pr_{\mathcal L\sim\mu_{\mathcal L}}
	\bigl[d_{\mathrm{TV}}(P_{\mathcal L},D)<\epsilon\bigr]
	\le \operatorname{frac}(\mu_{\mathcal L},Q,\tau).
	\]
\end{lemma}

\begin{proof}
	By the variational representation of TV distance, there exists a measurable subset $S\subseteq X$ such that
	\[
	d_{\mathrm{TV}}(Q,D)=|Q(S)-D(S)|,
	\]
	and let $\varphi:=\mathbf 1_S: X \rightarrow\{0,1\}$. Then
	\(
	d_{\mathrm{TV}}(Q,D)=|D[\varphi]-Q[\varphi]|
	\).
	For any $\mathcal L$ satisfying $d_{\mathrm{TV}}(P_{\mathcal L},D)<\epsilon$, we have
	\[
	|P_{\mathcal L}[\varphi]-D[\varphi]|
	\le d_{\mathrm{TV}}(P_{\mathcal L},D)
	<\epsilon.
	\]
	Hence
	\begin{align*}
		\bigl|P_{\mathcal L}[\varphi]-Q[\varphi]\bigr|
		&\ge |D[\varphi]-Q[\varphi]|
		- |P_{\mathcal L}[\varphi]-D[\varphi]|\\
		&> d_{\mathrm{TV}}(Q,D)-\epsilon\\
		&> (\epsilon+\tau)-\epsilon = \tau.
	\end{align*}
	Therefore
	\[
	\bigl\{d_{\mathrm{TV}}(P_{\mathcal L},D)<\epsilon\bigr\}
	\subseteq
	\bigl\{|P_{\mathcal L}[\varphi]-Q[\varphi]|>\tau\bigr\}
	\subseteq
	\bigl\{|P_{\mathcal L}[\varphi]-Q[\varphi]|\ge\tau\bigr\},
	\]
	and taking probabilities under $\mu_{\mathcal L}$ and using the definition of $\operatorname{frac}$ gives the claim.
\end{proof}

\begin{lemma}[Controlling the ``heaviest $\epsilon$-ball'']\label{lem:OS_ball_weight}
	In the setting of Lemma~\ref{lem:OS_far_from_Q_to_D}, for any $\epsilon,\tau>0$ and any $D\in\mathcal D_X$,
	\[
	\Pr_{\mathcal L\sim\mu_{\mathcal L}}
	\bigl[d_{\mathrm{TV}}(P_{\mathcal L},D)<\epsilon\bigr]
	\le
	\max\Bigl\{
	\operatorname{frac}(\mu_{\mathcal L},Q,\tau),\ 
	\Pr_{\mathcal L\sim\mu_{\mathcal L}}
	\bigl[d_{\mathrm{TV}}(P_{\mathcal L},Q)\le 2\epsilon+\tau\bigr]
	\Bigr\}.
	\]
\end{lemma}

\begin{proof}
	We split into two cases according to $d_{\mathrm{TV}}(D,Q)$.
	
	\emph{Case 1:} If $d_{\mathrm{TV}}(D,Q)>\epsilon+\tau$, then by Lemma~\ref{lem:OS_far_from_Q_to_D},
	\[
	\Pr_{\mathcal L\sim\mu_{\mathcal L}}\bigl[d_{\mathrm{TV}}(P_{\mathcal L},D)<\epsilon\bigr]
	\le \operatorname{frac}(\mu_{\mathcal L},Q,\tau).
	\]
	
	\emph{Case 2:} If $d_{\mathrm{TV}}(D,Q)\le\epsilon+\tau$, then for any $\mathcal L$ with $d_{\mathrm{TV}}(P_{\mathcal L},D)<\epsilon$, the triangle inequality gives
	\[
	d_{\mathrm{TV}}(P_{\mathcal L},Q)
	\le d_{\mathrm{TV}}(P_{\mathcal L},D)+d_{\mathrm{TV}}(D,Q)
	<\epsilon+(\epsilon+\tau)
	=2\epsilon+\tau,
	\]
	hence
	\[
	\bigl\{d_{\mathrm{TV}}(P_{\mathcal L},D)<\epsilon\bigr\}
	\subseteq
	\bigl\{d_{\mathrm{TV}}(P_{\mathcal L},Q)\le 2\epsilon+\tau\bigr\},
	\]
	and taking probabilities yields
	\[
	\Pr_{\mathcal L\sim\mu_{\mathcal L}}\bigl[d_{\mathrm{TV}}(P_{\mathcal L},D)<\epsilon\bigr]
	\le
	\Pr_{\mathcal L\sim\mu_{\mathcal L}}\bigl[d_{\mathrm{TV}}(P_{\mathcal L},Q)\le 2\epsilon+\tau\bigr].
	\]
	Combining the two cases gives the desired upper bound.
\end{proof}

Combining Theorem~\ref{thm:frac_bound_open_system} (which gives an exponential upper bound on $\operatorname{frac}(\mu_{\mathcal L},Q,\tau)$) with Corollary~\ref{cor:far_from_Q_typical} (which gives a high-probability lower bound on $d_{\mathrm{TV}}(P_{\mathcal L},Q)$ staying above a constant threshold), we obtain the following ``far-from-any-$D$'' result:

\begin{theorem}[Far-from-$D$ for random Lindbladian ensembles with respect to any fixed distribution $D$]\label{thm:far_from_all_D_open}
	Under the assumptions of Theorems~\ref{thm:frac_bound_open_system} and \ref{thm:far_from_Q_open} (and Corollary~\ref{cor:far_from_Q_typical}), there exist constants $\epsilon_*>0$ and $c_*>0$ such that for sufficiently large parameter dimension $M$ and for any $D\in\mathcal D_X$,
	\[
	\Pr_{\mathcal L\sim\mu_{\mathcal L}}
	\bigl[d_{\mathrm{TV}}(P_{\mathcal L},D)\ge \epsilon_*\bigr]
	\;\ge\;
	1 - 2\exp(-c_* M).
	\]
	In other words, for the random Lindbladian parameter-sphere ensemble, for almost all Lindbladians the output distribution $P_{\mathcal L}$ is at least a constant $\epsilon_*$ away in TV distance from any fixed classical distribution $D$, and the measure of the bad set decays exponentially in the parameter dimension $M$.
\end{theorem}

\begin{proof}
	We start from the two existing exponentially small upper bounds and substitute them into Lemma~\ref{lem:OS_ball_weight}.
	
	First, by Theorem~\ref{thm:frac_bound_open_system}, for any fixed $\tau>0$,
	\[
	\operatorname{frac}(\mu_{\mathcal L},Q,\tau)
	:=\sup_{\varphi:X\to[-1,1]}
	\Pr_{\mathcal L\sim\mu_{\mathcal L}}
	\bigl(|P_{\mathcal L}[\varphi]-Q[\varphi]|\ge\tau\bigr)
	\le
	2\exp\!\left(
	-c_{\mathrm{par}}\,
	\frac{M\,\tau^2}{L_0^2}
	\right),
	\]
	where $L_0=t\delta C_G$ is the uniform Lipschitz constant in Lemma~\ref{lem:theta_Lipschitz}.
	
	Second, by Theorem~\ref{thm:far_from_Q_open} and the open-system Porter-Thomas mean assumption, there exist a constant $m_0>0$ and $\epsilon_{\mathrm{mean}}\to0$ such that, for sufficiently large system size satisfying $\epsilon_{\mathrm{mean}}\le m_0/4$, Corollary~\ref{cor:far_from_Q_typical} gives
	\[
	\Pr_{\mathcal L\sim\mu_{\mathcal L}}
	\bigl[d_{\mathrm{TV}}(P_{\mathcal L},Q)\le m_0/2\bigr]
	\le
	2\exp\!\left(
	-c_{\mathrm{par}}\,
	\frac{M\,m_0^2}{4L_0^2}
	\right),
	\]
	where we used $L_{\mathrm{TV}}=L_0/2$.
	
	Now fix the concrete constants
	\[
	\epsilon_*:=\frac{m_0}{8},\qquad
	\tau:=\frac{m_0}{8},
	\]
	so that
	\(
	2\epsilon_*+\tau = \frac{3m_0}{8}<\frac{m_0}{2}
	\).
	By monotonicity of the probability bound we have
	\[
	\Pr_{\mathcal L\sim\mu_{\mathcal L}}
	\bigl[d_{\mathrm{TV}}(P_{\mathcal L},Q)\le 2\epsilon_*+\tau\bigr]
	\le
	\Pr_{\mathcal L\sim\mu_{\mathcal L}}
	\bigl[d_{\mathrm{TV}}(P_{\mathcal L},Q)\le m_0/2\bigr]
	\le
	2\exp\!\left(
	-c_2 M
	\right),
	\]
	where
	\(
	c_2:=c_{\mathrm{par}}\,m_0^2/(4L_0^2)
	\).
	On the other hand, for $\tau=m_0/8$,
	\[
	\operatorname{frac}(\mu_{\mathcal L},Q,\tau)
	\le
	2\exp\!\left(
	-c_1M
	\right),
	\qquad
	c_1:=c_{\mathrm{par}}\,\frac{m_0^2}{64L_0^2}.
	\]
	
	Substituting these two bounds into Lemma~\ref{lem:OS_ball_weight}, we obtain for any $D\in\mathcal D_X$,
	\[
	\Pr_{\mathcal L\sim\mu_{\mathcal L}}
	\bigl[d_{\mathrm{TV}}(P_{\mathcal L},D)<\epsilon_*\bigr]
	\le
	\max\Bigl\{
	2e^{-c_1M},\ 2e^{-c_2M}
	\Bigr\}
	=
	2e^{-c_*M},
	\]
	where $c_*:=\min\{c_1,c_2\}>0$.
	Taking complements yields
	\[
	\Pr_{\mathcal L\sim\mu_{\mathcal L}}
	\bigl[d_{\mathrm{TV}}(P_{\mathcal L},D)\ge\epsilon_*\bigr]
	\ge
	1-2e^{-c_*M},
	\]
	which holds for any fixed $D\in\mathcal D_X$, proving the theorem.
\end{proof}

\section{Linear-response scaling of the mean TV distance in random Lindbladian ensembles}\label{zhang4}

In this section, for the concrete random local Lindbladian ensemble
$\{\mathcal L(\theta)\}_{\theta\in S^{M-1}}$ introduced earlier, we derive the explicit asymptotic scaling of the mean
\[
F_{\rm TV}(\theta)
:= d_{\rm TV}\bigl(P_{\mathcal L(\theta)},Q\bigr),
\qquad
m(\delta):=\mathbb E_{\theta\sim\mu_{\rm sph}}F_{\rm TV}(\theta),
\]
in the linear-response (small-perturbation $\delta\to0$) regime. Here $\delta>0$ is the perturbation amplitude in the random Lindbladian parametrization. The technical inspiration for this analysis comes from Kubo’s linear-response theory~\cite{kubo1957statistical}; see Remark~\ref{kubo} for a more detailed discussion.

We focus on the TV distance
\[
F_{\rm TV}(\theta)
:= d_{\rm TV}\bigl(P_{\mathcal L(\theta)},Q\bigr)
= \frac12\sum_{x\in X}
\bigl|P_{\mathcal L(\theta)}(x)-Q(x)\bigr|.
\]

\medskip

We proceed in three steps, first we derive a linear-response expansion from the perturbation formula; then analyze the absolute-value expectation of a linear functional on the sphere; and finally combine the two to obtain the linear-response scaling of $m(\delta)$.

\subsection{Linear-response expansion from the perturbation formula}

\begin{lemma}[Perturbation formula and first-order response coefficients]\label{lem:perturbation_linear_response}
	In the setting of Definition~\ref{def:random_local_L_recall}, define
	\[
	K(\theta)
	:= \frac{1}{\sqrt M}\sum_{j=1}^M\theta_j G_j,
	\qquad
	\mathcal L(\theta)=\mathcal L_{\rm ref}+\delta K(\theta).
	\]
	Then there exists a constant $C_{\rm rem}(t,C_G)$ independent of $\theta$ and $M$, such that for each $x\in X$ we have the expansion
	\begin{equation}\label{eq:Delta_p_linear_response}
		\Delta p_x(\theta)
		:= P_{\mathcal L(\theta)}(x)-Q(x)
		= \frac{\delta}{\sqrt M}\sum_{j=1}^M \theta_j a_{j,x}
		+ R_x(\theta,\delta),
	\end{equation}
	where
	\begin{equation}\label{eq:ajx_def}
		a_{j,x}
		:= \operatorname{Tr}\biggl(
		M_x\int_0^t
		\Lambda_{t-s}^{({\rm ref})}\circ G_j\circ
		\Lambda_s^{({\rm ref})}(\rho_{\rm in})\,\mathrm ds
		\biggr),
	\end{equation}
	and the remainder satisfies the uniform bound
	\begin{equation}\label{eq:R_x_delta_bound}
		\sup_{\theta\in S^{M-1}}|R_x(\theta,\delta)|
		\le C_{\rm rem}(t,C_G)\,\delta^2.
	\end{equation}
\end{lemma}

\begin{proof}
	Set
	\[
	\Lambda_t^{(\theta)} = e^{t(\mathcal L_{\rm ref}+\delta K(\theta))},
	\]
	and perform a perturbative expansion in $\delta$.
	Define the interpolating generator
	\[
	\mathcal L_\alpha := \mathcal L_{\rm ref} + \alpha\delta K(\theta),
	\]
	with corresponding QMS
	\[
	\Lambda_t^{(\alpha)} := e^{t\mathcal L_\alpha}.
	\]
	The standard perturbation formula gives
	\[
	\Lambda_t^{(\theta)}-\Lambda_t^{({\rm ref})}
	= \int_0^1 \frac{\mathrm d}{\mathrm d\alpha}
	\Lambda_t^{(\alpha)}\,\mathrm d\alpha,
	\]
	with
	\[
	\frac{\mathrm d}{\mathrm d\alpha}\Lambda_t^{(\alpha)}
	= \int_0^t
	\Lambda_{t-s}^{(\alpha)}\circ(\delta K(\theta))\circ
	\Lambda_s^{(\alpha)}\,\mathrm ds.
	\]
	Expanding $\Lambda_t^{(\alpha)}$ itself once more in $\delta$, one obtains the standard second-order formula
	\begin{equation}\label{eq:second_order_perturbation}
		\Lambda_t^{(\theta)}
		= \Lambda_t^{({\rm ref})}
		+ \delta\int_0^t \Lambda_{t-s}^{({\rm ref})}\circ K(\theta)\circ
		\Lambda_s^{({\rm ref})}\,\mathrm ds
		+ \delta^2 R_t(\theta,\delta),
	\end{equation}
	where the remainder $R_t(\theta,\delta)$ can be written as a double integral
	\[
	R_t(\theta,\delta)
	= \int_0^t\!\int_0^{s}
	\Lambda_{t-s}^{({\rm ref})}\circ K(\theta)\circ
	\Lambda_{s-r}^{(\alpha_{s,r})}\circ K(\theta)\circ
	\Lambda_r^{({\rm ref})}\,\mathrm dr\,\mathrm ds,
	\]
	with some interpolation parameter $\alpha_{s,r}\in[0,1]$.
	
	Taking the $1\to1$ norm gives
	\begin{align*}
		\|R_t(\theta,\delta)\|_{1\to1}
		&\le \int_0^t\!\int_0^{s}
		\bigl\|\Lambda_{t-s}^{({\rm ref})}\bigr\|_{1\to1}\,
		\bigl\|K(\theta)\bigr\|_{1\to1}\,
		\bigl\|\Lambda_{s-r}^{(\alpha_{s,r})}\bigr\|_{1\to1}\,
		\bigl\|K(\theta)\bigr\|_{1\to1}\,
		\bigl\|\Lambda_r^{({\rm ref})}\bigr\|_{1\to1}\,\mathrm dr\,\mathrm ds \\
		&\le \|K(\theta)\|_{1\to1}^2
		\int_0^t\!\int_0^{s} \mathrm dr\,\mathrm ds,
	\end{align*}
	where we used that all $\Lambda_\cdot^{(\cdot)}$ are CPTP and hence
	$\|\Lambda_\cdot^{(\cdot)}\|_{1\to1}\le1$.
	On the other hand, by~\eqref{eq:Gj_norm_bound_recall} and
	$
	K(\theta)=\frac1{\sqrt M}\sum_j\theta_j G_j
	$ we obtain
	\begin{align*}
		\|K(\theta)\|_{1\to1}
		&\le \frac{1}{\sqrt M}\sum_{j=1}^M|\theta_j|\,\|G_j\|_{1\to1} \\
		&\le \frac{C_G}{\sqrt M}\sum_{j=1}^M|\theta_j| \\
		&\le \frac{C_G}{\sqrt M}\,\sqrt M\,\|\theta\|_2 \\
		&= C_G,
	\end{align*}
	where we used $\|G_j\|_{1\to1}\le C_G$ and
	$\|\theta\|_1\le\sqrt M\|\theta\|_2=\sqrt M$ for $\theta\in S^{M-1}$.
	Thus
	\[
	\|R_t(\theta,\delta)\|_{1\to1}
	\le C_G^2\int_0^t\!\int_0^{s} \mathrm dr\,\mathrm ds
	= \frac{C_G^2 t^2}{2}
	=: C_{\rm rem}(t,C_G),
	\]
	which is uniform in $\theta$ and $M$.
	
	Applying~\eqref{eq:second_order_perturbation} to $\rho_{\rm in}$, and then taking the trace with $M_x$, we get
	\begin{align*}
		\Delta p_x(\theta)
		&= P_{\mathcal L(\theta)}(x)-Q(x) \\
		&= \delta\,\operatorname{Tr}\biggl(
		M_x\int_0^t
		\Lambda_{t-s}^{({\rm ref})}\circ K(\theta)\circ
		\Lambda_s^{({\rm ref})}(\rho_{\rm in})\,\mathrm ds
		\biggr)
		+ \delta^2\operatorname{Tr}\bigl(M_x R_t(\theta,\delta)(\rho_{\rm in})\bigr).
	\end{align*}
	Using $|\operatorname{Tr}(AB)|\le\|A\|_\infty\|B\|_1$ and $\|M_x\|_\infty\le1$, we obtain
	\begin{align*}
		\bigl|\operatorname{Tr}\bigl(M_x R_t(\theta,\delta)(\rho_{\rm in})\bigr)\bigr|
		&\le \|M_x\|_\infty\,\bigl\|R_t(\theta,\delta)(\rho_{\rm in})\bigr\|_1 \\
		&\le \|R_t(\theta,\delta)\|_{1\to1} \\
		&\le C_{\rm rem}(t,C_G).
	\end{align*}
	Expanding $K(\theta)$ as
	$
	K(\theta)=\frac1{\sqrt M}\sum_{j=1}^M\theta_j G_j
	$,
	we obtain~\eqref{eq:Delta_p_linear_response}--\eqref{eq:R_x_delta_bound}.
\end{proof}

\begin{remark}[Physical meaning of the first-order response coefficients $a_{j,x}$]
	The coefficient $a_{j,x}$ captures the linear response along the $j$-th local dissipative direction $G_j$: it is obtained by inserting $G_j$ into the reference dynamics $\Lambda_t^{({\rm ref})}$, integrating this response over the time window $[0,t]$, and then reading it out in the POVM outcome $x$. The collection of $a_{j,x}$ over all pairs $(j,x)$ characterizes the first-order effect of the random perturbation on the measurement statistics.
\end{remark}

\begin{remark}[Relation to Kubo linear response]\label{kubo}
	It is worth noting that the perturbative expansion used here is structurally analogous to Kubo’s linear-response theory. In the Hamiltonian closed-system case, an external field $f_j(t)$ couples to an operator $B_j$, and the response of some observable $A$ is written as~\cite{kamenev2023field}
	\[
	\delta\langle A(t)\rangle
	= \sum_j\int_{-\infty}^t \chi_{A,B_j}(t-s)\,f_j(s)\,\mathrm ds,
	\]
	where $\chi_{A,B_j}$ is expressed in terms of time-correlation functions of $A$ and $B_j$ under the unperturbed dynamics. In our open-system framework, $\delta\theta_j$ plays the role of a static external field, $G_j$ plays the role of the coupling channel, and $a_{j,x}$ is the response kernel of the measurement probability $P_{\mathcal L}(x)$ with respect to channel $j$. The differences are:
	\begin{enumerate}
		\item the dynamics is governed by a GKSL quantum Markov semigroup $\Lambda_t=e^{t\mathcal L}$ rather than a unitary flow;
		\item the observable we care about is the entire output distribution $P_{\mathcal L}$, measured by its TV distance to a reference distribution $Q$, i.e., a nonlinear distinguishability functional. In Theorem~\ref{thm:mean_TV_linear_response} we further average over the high-dimensional parameter sphere in $\theta$ and obtain the linear-response scaling of the mean TV distance $m(\delta)$.
	\end{enumerate}
	From this perspective, the linear-response analysis in this section can be viewed as a generalization of Kubo theory to the GKSL open-system setting, with a nonstandard observable tailored to learnability.
\end{remark}

\subsection{Absolute-value expectation of linear forms on the sphere}

We now give a purely geometric lemma describing the expectation of $|\langle\theta,b\rangle|$ on a high-dimensional sphere.

\begin{lemma}[Absolute-value expectation of one-dimensional projections on the sphere]\label{lem:sphere_inner_product}
	Let $M\ge2$, let $\theta\sim\mu_{\rm sph}$ be a uniformly random vector on $S^{M-1}\subset\mathbb R^M$, and $b\in\mathbb R^M$ any nonzero vector. Then
	\begin{equation}\label{eq:inner_product_expect_exact}
		\mathbb E_{\theta\sim\mu_{\rm sph}}
		\bigl|\langle\theta,b\rangle\bigr|
		= \|b\|_2\,
		\frac{\Gamma\bigl(\tfrac M2\bigr)}
		{\sqrt\pi\,\Gamma\bigl(\tfrac{M+1}2\bigr)}.
	\end{equation}
	In particular, in the high-dimensional limit $M\to\infty$,
	\begin{equation}\label{eq:inner_product_expect_asymp}
		\mathbb E_{\theta\sim\mu_{\rm sph}}
		\bigl|\langle\theta,b\rangle\bigr|
		= \|b\|_2\,\sqrt{\frac{2}{\pi M}}\,(1+o(1)).
	\end{equation}
\end{lemma}

\begin{proof}
	By spherical symmetry, we may assume $b$ is aligned with the first coordinate direction, i.e., $b=\|b\|_2\,e_1$, where $e_1=(1,0,\dots,0)$. Then
	$
	\langle\theta,b\rangle=\|b\|_2\,\theta_1
	$,
	so
	\[
	\mathbb E_{\theta\sim\mu_{\rm sph}}\bigl|\langle\theta,b\rangle\bigr|
	= \|b\|_2\,\mathbb E_{\theta\sim\mu_{\rm sph}}|\theta_1|.
	\]
	It thus suffices to determine the distribution of $\theta_1$ and its absolute-value expectation.
	
	If $\theta$ is uniform on $S^{M-1}$, the density of $\theta_1$ on $[-1,1]$ is (a Beta-type density)
	\[
	f_M(u)
	= c_M(1-u^2)^{\frac{M-3}{2}},
	\qquad u\in[-1,1],
	\]
	where the normalization constant $c_M$ is given by
	\begin{align*}
		c_M^{-1}
		&= \int_{-1}^1(1-u^2)^{\frac{M-3}{2}}\,\mathrm du \\
		&= 2\int_0^1(1-u^2)^{\frac{M-3}{2}}\,\mathrm du.
	\end{align*}
	Letting $v=u^2$ yields
	\begin{align*}
		c_M^{-1}
		&= 2\int_0^1 (1-u^2)^{\frac{M-3}{2}}\,\mathrm du \\
		&= 2\int_0^1 (1-v)^{\frac{M-3}{2}}\frac{\mathrm dv}{2\sqrt v} \\
		&= \int_0^1 v^{-1/2}(1-v)^{\frac{M-3}{2}}\,\mathrm dv \\
		&= \mathrm{B}\Bigl(\tfrac12,\tfrac{M-1}{2}\Bigr),
	\end{align*}
	where $\mathrm{B}$ is the Beta function.
	Using
	$
	\mathrm{B}(x,y)=\frac{\Gamma(x)\Gamma(y)}{\Gamma(x+y)}
	$
	we obtain
	\[
	c_M
	= \frac{\Gamma\bigl(\tfrac M2\bigr)}
	{\sqrt\pi\,\Gamma\bigl(\tfrac{M-1}2\bigr)}.
	\]
	Now compute
	\begin{align*}
		\mathbb E_{\theta\sim\mu_{\rm sph}}|\theta_1|
		&= 2c_M\int_0^1 u(1-u^2)^{\frac{M-3}{2}}\,\mathrm du.
	\end{align*}
	Again let $v=u^2$, $\mathrm dv=2u\,\mathrm du$, to get
	\begin{align*}
		\mathbb E_{\theta\sim\mu_{\rm sph}}|\theta_1|
		&= 2c_M\int_0^1 u(1-u^2)^{\frac{M-3}{2}}\,\mathrm du \\
		&= c_M\int_0^1 (1-v)^{\frac{M-3}{2}}\,\mathrm dv \\
		&= c_M\,
		\mathrm{B}\Bigl(1,\tfrac{M-1}{2}\Bigr).
	\end{align*}
	Since
	\[
	\mathrm{B}\Bigl(1,\tfrac{M-1}{2}\Bigr)
	=\frac{\Gamma(1)\Gamma\bigl(\tfrac{M-1}{2}\bigr)}
	{\Gamma\bigl(\tfrac{M+1}{2}\bigr)}
	=\frac{\Gamma\bigl(\tfrac{M-1}{2}\bigr)}
	{\Gamma\bigl(\tfrac{M+1}{2}\bigr)},
	\]
	we find
	\begin{align*}
		\mathbb E_{\theta\sim\mu_{\rm sph}}|\theta_1|
		&= c_M\,\frac{\Gamma\bigl(\tfrac{M-1}{2}\bigr)}
		{\Gamma\bigl(\tfrac{M+1}{2}\bigr)} \\
		&= \frac{\Gamma\bigl(\tfrac M2\bigr)}
		{\sqrt\pi\,\Gamma\bigl(\tfrac{M+1}2\bigr)}.
	\end{align*}
	This gives~\eqref{eq:inner_product_expect_exact}.
	
	For the asymptotics~\eqref{eq:inner_product_expect_asymp}, apply Stirling’s formula for the Gamma function,
	\[
	\frac{\Gamma(z+\tfrac12)}{\Gamma(z)}
	\sim z^{1/2}
	\quad (z\to\infty),
	\]
	which implies
	\[
	\frac{\Gamma\bigl(\tfrac M2\bigr)}
	{\Gamma\bigl(\tfrac{M+1}2\bigr)}
	\sim \Bigl(\tfrac M2\Bigr)^{-1/2},
	\]
	and hence
	\[
	\mathbb E_{\theta\sim\mu_{\rm sph}}|\theta_1|
	= \frac{1}{\sqrt\pi}
	\frac{\Gamma\bigl(\tfrac M2\bigr)}
	{\Gamma\bigl(\tfrac{M+1}2\bigr)}
	\sim \frac{1}{\sqrt\pi}\Bigl(\frac{2}{M}\Bigr)^{1/2}
	= \sqrt{\frac{2}{\pi M}}.
	\]
	Substituting into $\mathbb E_{\theta\sim\mu_{\rm sph}}|\langle\theta,b\rangle|$ yields~\eqref{eq:inner_product_expect_asymp}.
\end{proof}

\begin{remark}[Relation to Gaussian approximation]\label{gauss}
	When $M$ is large, the coordinates of $\theta$ may be thought of as ``approximately independent Gaussians'' constrained by $\sum_j\theta_j^2=1$; Lemma~\ref{lem:sphere_inner_product} can be viewed as a precise formulation of this intuition. In the high-dimensional limit, the distribution of $\langle\theta,b\rangle$ approaches that of a Gaussian variable with variance $1/M$, and thus $\mathbb E|\langle\theta,b\rangle|\approx\|b\|_2\sqrt{2/(\pi M)}$.
\end{remark}

For later convenience, define
\begin{equation}\label{eq:kappa_M_def}
	\kappa_M
	:= \sqrt M\,
	\frac{\Gamma\bigl(\tfrac M2\bigr)}
	{\sqrt\pi\,\Gamma\bigl(\tfrac{M+1}2\bigr)},
\end{equation}
so that Lemma~\ref{lem:sphere_inner_product} is equivalent to
\begin{equation}\label{eq:inner_product_expect_kappaM}
	\mathbb E_{\theta\sim\mu_{\rm sph}}
	\bigl|\langle\theta,b\rangle\bigr|
	= \frac{\kappa_M}{\sqrt M}\,\|b\|_2,
	\qquad
	\kappa_M\to\sqrt{\frac{2}{\pi}}\quad(M\to\infty).
\end{equation}

\subsection{Linear-response scaling of the mean TV distance}

We now combine the two ingredients above to derive a linear-response expression for $m(\delta)$.

\begin{theorem}[Asymptotic scaling of the mean TV distance in the linear-response regime]\label{thm:mean_TV_linear_response}
	In the setting of Definition~\ref{def:random_local_L_recall} and
	Lemma~\ref{lem:perturbation_linear_response},
	define
	\[
	F_{\rm TV}(\theta)
	:= d_{\rm TV}\bigl(P_{\mathcal L(\theta)},Q\bigr)
	= \frac12\sum_{x\in X}
	\bigl|P_{\mathcal L(\theta)}(x)-Q(x)\bigr|,
	\quad
	m(\delta)
	:= \mathbb E_{\theta\sim\mu_{\rm sph}}F_{\rm TV}(\theta).
	\]
	Introduce the first-order response vectors
	\[
	a^{(x)} := (a_{1,x},\dots,a_{M,x})\in\mathbb R^M,
	\qquad
	\|a^{(x)}\|_2
	:= \Bigl(\sum_{j=1}^M a_{j,x}^2\Bigr)^{1/2},
	\]
	where $a_{j,x}$ is defined in~\eqref{eq:ajx_def}.
	Then in the linear-response regime $\delta\to0$, we have
	\begin{equation}\label{eq:m_delta_linear_asymp}
		m(\delta)
		= \delta\,m_0(M)
		+ \mathcal O(\delta^2),
	\end{equation}
	where
	\begin{equation}\label{eq:m0_M_def}
		m_0(M)
		:= \frac{\kappa_M}{2M}
		\sum_{x\in X}\|a^{(x)}\|_2,
	\end{equation}
	with $\kappa_M$ defined in~\eqref{eq:kappa_M_def}.
	In particular, in the high-dimensional limit $M\to\infty$,
	\begin{equation}\label{eq:m_delta_linear_asymp_simplified}
		m(\delta)
		= \delta\,\frac{1}{2M}\sqrt{\frac{2}{\pi}}
		\sum_{x\in X}\|a^{(x)}\|_2
		+ \mathcal O(\delta^2)
		+ o(\delta),
	\end{equation}
	where the $o(\delta)$ term comes from the asymptotic error in
	$\kappa_M\to\sqrt{2/\pi}$.
\end{theorem}

\begin{proof}
	By Lemma~\ref{lem:perturbation_linear_response},
	\[
	\Delta p_x(\theta)
	= \frac{\delta}{\sqrt M}\sum_{j=1}^M\theta_j a_{j,x}
	+ R_x(\theta,\delta)
	= \frac{\delta}{\sqrt M}\,\langle\theta,a^{(x)}\rangle
	+ R_x(\theta,\delta),
	\]
	with
	\(
	\sup_{\theta}|R_x(\theta,\delta)|
	\le C_{\rm rem}\delta^2
	\).
	Hence
	\[
	\bigl|\Delta p_x(\theta)\bigr|
	\le \frac{\delta}{\sqrt M}
	\bigl|\langle\theta,a^{(x)}\rangle\bigr|
	+ |R_x(\theta,\delta)|.
	\]
	By the triangle inequality,
	\[
	\Bigl|\,
	\bigl|\Delta p_x(\theta)\bigr|
	- \frac{\delta}{\sqrt M}
	\bigl|\langle\theta,a^{(x)}\rangle\bigr|
	\Bigr|
	\le |R_x(\theta,\delta)|
	\le C_{\rm rem}\delta^2.
	\]
	Taking expectation over $\theta$ and applying Lemma~\ref{lem:sphere_inner_product} yields
	\begin{align*}
		\mathbb E_\theta\bigl|\Delta p_x(\theta)\bigr|
		&= \frac{\delta}{\sqrt M}\,
		\mathbb E_\theta
		\bigl|\langle\theta,a^{(x)}\rangle\bigr|
		+ \mathcal O(\delta^2) \\
		&= \frac{\delta}{\sqrt M}\,
		\frac{\kappa_M}{\sqrt M}\,\|a^{(x)}\|_2
		+ \mathcal O(\delta^2) \\
		&= \delta\,\frac{\kappa_M}{M}\,\|a^{(x)}\|_2
		+ \mathcal O(\delta^2),
	\end{align*}
	where the constant hidden in $\mathcal O(\delta^2)$ can be controlled by
	$C_{\rm rem}(t,C_G)$ and $|X|$, and is independent of $M$ and $\theta$.
	
	By the definition of TV distance,
	\[
	F_{\rm TV}(\theta)
	= \frac12\sum_{x\in X}|\Delta p_x(\theta)|,
	\]
	and taking expectation over $\theta$ gives
	\begin{align*}
		m(\delta)
		&= \mathbb E_\theta F_{\rm TV}(\theta)
		= \frac12\sum_{x\in X}
		\mathbb E_\theta\bigl|\Delta p_x(\theta)\bigr| \\
		&= \frac12\sum_{x\in X}
		\Bigl[\delta\,\frac{\kappa_M}{M}\,\|a^{(x)}\|_2
		+ \mathcal O(\delta^2)\Bigr].
	\end{align*}
	Separating the first-order and $\mathcal O(\delta^2)$ terms,
	\[
	m(\delta)
	= \delta\,\frac{\kappa_M}{2M}
	\sum_{x\in X}\|a^{(x)}\|_2
	+ \mathcal O(\delta^2),
	\]
	which is exactly~\eqref{eq:m_delta_linear_asymp}--\eqref{eq:m0_M_def}.
	
	Finally, using the asymptotic behaviour
	\(
	\kappa_M\to\sqrt{2/\pi}
	\)
	from Lemma~\ref{lem:sphere_inner_product}, we may write
	\(
	\kappa_M
	= \sqrt{2/\pi}+\mathcal O(M^{-1})
	\),
	and substituting this into $m_0(M)$ yields~\eqref{eq:m_delta_linear_asymp_simplified}: the leading linear term in $\delta$ has coefficient
	\(
	\frac{1}{2M}\sqrt{2/\pi}\sum_{x\in X}\|a^{(x)}\|_2
	\),
	while the deviation of $\kappa_M$ from $\sqrt{2/\pi}$ contributes an additional $o(\delta)$ term.
\end{proof}

\begin{remark}[Relation to the open-system Porter-Thomas mean assumption]\label{rmk:mean_TV_PT}
	The expression for
	$m(\delta)=\mathbb E d_{\rm TV}(P_{\mathcal L(\theta)},Q)$
	in Theorem~\ref{thm:mean_TV_linear_response} provides a linear-response, analytic incarnation of the ``open-system Porter-Thomas mean assumption''. We can distinguish the following cases:
	\begin{enumerate}
		\item 
		If the scaling of $\sum_x\|a^{(x)}\|_2$ remains extensive as $M\to\infty$, for example
		\[
		\frac{1}{M}\sum_{x\in X}\|a^{(x)}\|_2
		\to c_a>0,
		\]
		then
		\[
		m(\delta)
		\approx \delta\cdot \frac{\sqrt{2/\pi}}{2}\,c_a,
		\]
		i.e., in the large-system limit the mean TV distance maintains an $\mathcal O(\delta)$ constant scaling.
		
		\item If in some concrete models
		$\sum_x\|a^{(x)}\|_2$ grows slower than $M$, e.g.,
		$
		\sum_x\|a^{(x)}\|_2 = o(M)
		$,
		then $m_0(M)\to0$ and the mean TV distance shrinks with system size. This means that in such an ensemble and for the given measurement interface, most
		$P_{\mathcal L(\theta)}$ become TV-close to the reference $Q$; in this regime, the SQ-hardness criterion based on
		$d_{\rm TV}(P_{\mathcal L},Q)$ cannot yield a nontrivial lower bound.
	\end{enumerate}
\end{remark}

\begin{remark}[Parameter dimension $M$ and validity of linear response]
	It is important to emphasize that in the linear-response expansion~\eqref{eq:Delta_p_linear_response},
	$\delta$ is the small parameter, and the growth of the parameter dimension $M$ is explicitly controlled via
	$\|K(\theta)\|_{1\to1}\le C_G$ and the $1/\sqrt M$ scaling in Lemma~\ref{lem:sphere_inner_product}.
	As long as $\delta$ is small enough so that the second-order remainder $\mathcal O(\delta^2)$ is negligible compared to the first-order term, linear response remains valid even for large $M$. High dimension does not break linear response; it merely enters into the definition of $m_0(M)$ through the $1/\sqrt M$ scaling of
	$\mathbb E|\langle\theta, a^{(x)}\rangle|$.
\end{remark}

\section{Linear response mean scaling and numerical simulation in a random local amplitude-damping Lindbladian model}\label{zhang5}

In this section, within the framework of the random local Lindbladian ensemble
(Definition~\ref{def:random_local_L}),
we select a concrete and explicitly computable one-dimensional amplitude-damping model. In this model we compute explicitly the first-order linear-response coefficients
\(\{a_{j,x}\}\) and their averaged norms
\(\frac{1}{M}\sum_x \|a^{(x)}\|_2\), and thereby obtain the scaling constant \(m_0>0\) for the mean TV distance
\(\mathbb E_\theta d_{\mathrm{TV}}(P_{\mathcal L(\theta)},Q)\)
in the small-noise limit. We also perform numerical simulations on the one-dimensional local amplitude-damping model in
Definition~\ref{def:local_amp_damp_model},
in order to test
Theorem~\ref{thm:TV_linear_response_mean}
and provide numerical support for Assumption~\ref{ass:mean_TV}.

\subsection{scaling constant \(m_0>0\) for the mean TV distance}\label{zhang15}

\begin{definition}[Random local amplitude-damping Lindbladian model]
	\label{def:local_amp_damp_model}
	Building on the setting of Definition~\ref{def:random_local_L},
	we make the following concrete choices:
	
	\begin{enumerate}
		\item 
		The system is a one-dimensional chain of $N$ qubits, with Hilbert space
		\[
		\mathcal H = (\mathbb C^2)^{\otimes N},\qquad d = 2^N.
		\]
		
		\item 
		For convenience in the first-order linear-response analysis, in this toy model we set
		\[
		\mathcal L_{\mathrm{ref}} = 0,
		\]
		so that the reference evolution is the identity
		\(\Lambda_t^{(\mathrm{ref})} = e^{t\mathcal L_{\mathrm{ref}}}=\mathrm{id}\).
		The input state is chosen as the all-$1$ computational-basis state
		\[
		\rho_{\mathrm{in}} := \ket{1\dots 1}\!\bra{1\dots 1}.
		\]
		
		\item 
		On each site $j=1,\dots,N$ we introduce a local amplitude-damping jump operator
		\[
		J_j^{(-)} := \sqrt{\gamma}\,\sigma_j^-,
		\]
		where $\sigma_j^-$ is the lowering operator acting on the $j$-th qubit and $\gamma>0$ is the single-qubit decay rate.
		For each $J_j^{(-)}$ we define the corresponding dissipator
		\[
		G_j := \mathcal D_j^{(-)},\qquad
		\mathcal D_j^{(-)}(\rho)
		:= J_j^{(-)}\rho J_j^{(-)\dagger}
		-\frac12\{J_j^{(-)\dagger}J_j^{(-)},\rho\}.
		\]
		The parameter dimension is thus
		\[
		M := N.
		\]
		
		\item 
		The measurement is the global computational-basis POVM
		\[
		\{M_x\}_{x\in\{0,1\}^N},\qquad
		M_x := \ket{x}\!\bra{x},\quad
		X := \{0,1\}^N.
		\]
		We write
		\(\ket{1^N} := \ket{1\dots 1}\),
		and for each $j$ define the bitstring
		\[
		x_j^{(0)} \in \{0,1\}^N
		\]
		to be the string with the $j$-th bit equal to $0$ and all others equal to $1$.
		Evidently,
		\(\{1^N\}\cup\{x_j^{(0)}\}_{j=1}^N\subset X\).
		
		\item 
		The parameter space is the real sphere
		\[
		S^{M-1}
		= \{\theta\in\mathbb R^M:\ \|\theta\|_2=1\},\qquad M=N.
		\]
		For a small parameter \(\delta>0\),
		we define the linear parametrization
		\[
		\mathcal L(\theta)
		:= \mathcal L_{\mathrm{ref}}
		+ \frac{\delta}{\sqrt M}\sum_{j=1}^M\theta_j G_j
		= \frac{\delta}{\sqrt N}\sum_{j=1}^N\theta_j G_j,\qquad
		\theta\in S^{M-1}.
		\]
		We take the rotationally invariant measure $\mu_{\mathrm{sph}}$ on $S^{M-1}$, and push it forward via $\theta\mapsto\mathcal L(\theta)$ to obtain
		\[
		\mu_{\mathcal L}
		:= (\mathcal L)_\#\mu_{\mathrm{sph}},
		\]
		which defines a random local amplitude-damping Lindbladian ensemble
		\(
		\mathcal F := \{\mathcal L(\theta):\ \theta\in S^{M-1}\}.
		\)
		
		\item 
		In this model, it is natural to take the reference distribution $Q$ as the measurement distribution under the reference dynamics (i.e., $\mathcal L_{\mathrm{ref}}=0$):
		\[
		Q(x)
		:= \operatorname{Tr}(M_x\rho_{\mathrm{in}})
		= \delta_{x,1^N},\qquad x\in X.
		\]
	\end{enumerate}
\end{definition}

\begin{remark}
	In more general settings,
	$\mathcal L_{\mathrm{ref}}$ may be any GKSL generator of a QMS, e.g., a local amplitude-damping term plus a local Hamiltonian,
	and the reference distribution $Q$ may be chosen as the ensemble average
	\(
	Q[\varphi]
	= \mathbb E_{\mathcal L\sim\mu_{\mathcal L}}P_{\mathcal L}[\varphi].
	\)
	In the small-noise limit \(\delta\to 0\), different choices of reference only change higher-order corrections and do not affect the first-order linear-response scaling constant.
	In this section we set $\mathcal L_{\mathrm{ref}}=0$ purely to make $a_{j,x}$ explicitly computable.
\end{remark}

\begin{lemma}[Action of the local amplitude-damping dissipator on the input state]
	\label{lem:local_amp_damp_action}
	In the setting of Definition~\ref{def:local_amp_damp_model},
	for each $j=1,\dots,N$,
	\[
	G_j(\rho_{\mathrm{in}})
	= \gamma\Bigl(
	\ket{x_j^{(0)}}\!\bra{x_j^{(0)}}
	- \ket{1^N}\!\bra{1^N}
	\Bigr),
	\]
	where $x_j^{(0)}$ is the bitstring with the $j$-th bit equal to $0$ and the others $1$, and \(\ket{1^N}=\ket{1\dots 1}\).
\end{lemma}

\begin{proof}
	We first examine the single-qubit amplitude-damping dissipator
	\(
	\mathcal D^{(-)}(\rho)
	:= \gamma(\sigma^-\rho\sigma^+ - \tfrac12\{\sigma^+\sigma^-,\rho\})
	\),
	where
	\(
	\sigma^- = \ket0\!\bra1
	\),
	\(
	\sigma^+ = \ket1\!\bra0
	\).
	For the single-qubit input \(\rho = \ket1\!\bra1\), we have
	\begin{align*}
		\sigma^-\rho\sigma^+
		&= \ket0\!\bra1\,\ket1\!\bra1\,\ket1\!\bra0
		= \ket0\!\bra0,\\
		\sigma^+\sigma^-
		&= \ket1\!\bra1,\\
		\frac12\{\sigma^+\sigma^-,\rho\}
		&= \frac12(\ket1\!\bra1\,\rho + \rho\ket1\!\bra1)
		= \ket1\!\bra1.
	\end{align*}
	Therefore
	\[
	\mathcal D^{(-)}(\ket1\!\bra1)
	= \gamma(\ket0\!\bra0 - \ket1\!\bra1).
	\]
	
	On the many-body system,
	$J_j^{(-)} = \sqrt{\gamma}\,\sigma_j^-$
	acts only on the $j$-th qubit, with identity on all others.
	Writing
	\(\ket{1^N} = \ket{1}\otimes\dots\otimes\ket1\),
	we have
	\[
	J_j^{(-)}\ket{1^N}
	= \sqrt{\gamma}\,\ket1\otimes\dots\otimes\ket0_j\otimes\dots\otimes\ket1
	= \sqrt{\gamma}\,\ket{x_j^{(0)}},
	\]
	and hence
	\[
	J_j^{(-)}\rho_{\mathrm{in}}J_j^{(-)\dagger}
	= \gamma\,\ket{x_j^{(0)}}\!\bra{x_j^{(0)}}.
	\]
	On the other hand,
	\[
	J_j^{(-)\dagger}J_j^{(-)}
	= \gamma\,\bigl(\mathbb I^{\otimes(j-1)}
	\otimes \ket1\!\bra1{}_j
	\otimes \mathbb I^{\otimes(N-j)}\bigr),
	\]
	so for the input state $\rho_{\mathrm{in}}=\ket{1^N}\!\bra{1^N}$,
	\[
	\frac12\{J_j^{(-)\dagger}J_j^{(-)},\rho_{\mathrm{in}}\}
	= \gamma\,\ket{1^N}\!\bra{1^N}.
	\]
	Substituting these into the definition of $G_j=\mathcal D_j^{(-)}$, we obtain
	\[
	G_j(\rho_{\mathrm{in}})
	= \gamma\bigl(\ket{x_j^{(0)}}\!\bra{x_j^{(0)}}
	- \ket{1^N}\!\bra{1^N}\bigr),
	\]
	as claimed.
\end{proof}

\begin{lemma}[Explicit expression for the first-order linear-response coefficients $a_{j,x}$]
	\label{lem:a_jx_explicit}
	In the setting of Definition~\ref{def:local_amp_damp_model},
	let
	\[
	\Lambda_t^{(\theta)}
	:= e^{t\mathcal L(\theta)},\qquad
	P_{\mathcal L(\theta)}(x)
	:= \operatorname{Tr}\bigl(M_x\,\Lambda_t^{(\theta)}(\rho_{\mathrm{in}})\bigr).
	\]
	In the small-noise limit $\delta\to 0$,
	there exist coefficients $a_{j,x}\in\mathbb R$ such that
	\begin{equation}\label{eq:linear_response_P}
		P_{\mathcal L(\theta)}(x)
		= Q(x)
		+ \frac{\delta}{\sqrt M}\sum_{j=1}^M\theta_j a_{j,x}
		+ \mathcal O(\delta^2),
	\end{equation}
	where $Q(x)=\delta_{x,1^N}$,
	and
	\begin{equation}\label{eq:a_jx_explicit}
		a_{j,x}
		= t\gamma\bigl(\delta_{x,x_j^{(0)}} - \delta_{x,1^N}\bigr),
		\qquad j=1,\dots,M,\quad x\in X,
	\end{equation}
	with $M=N$.
\end{lemma}

\begin{proof}
	Since \(\mathcal L(\theta)\) depends linearly on $\delta$, we write
	\[
	\mathcal L(\theta)
	= \frac{\delta}{\sqrt M}\sum_{j=1}^M\theta_j G_j,
	\qquad
	\mathcal L_{\mathrm{ref}}=0.
	\]
	Expanding $e^{t\mathcal L(\theta)}$ around $\delta=0$ to first order gives
	\[
	\Lambda_t^{(\theta)}
	= e^{t\mathcal L(\theta)}
	= \mathbb I + t\mathcal L(\theta) + \mathcal O(\delta^2)
	= \mathbb I
	+ t\frac{\delta}{\sqrt M}\sum_{j=1}^M\theta_j G_j
	+ \mathcal O(\delta^2).
	\]
	Acting on the input state $\rho_{\mathrm{in}}$,
	\[
	\Lambda_t^{(\theta)}(\rho_{\mathrm{in}})
	= \rho_{\mathrm{in}}
	+ t\frac{\delta}{\sqrt M}
	\sum_{j=1}^M\theta_j G_j(\rho_{\mathrm{in}})
	+ \mathcal O(\delta^2).
	\]
	For each measurement outcome $x\in X$,
	\begin{align*}
		P_{\mathcal L(\theta)}(x)
		&= \operatorname{Tr}\bigl(M_x\Lambda_t^{(\theta)}(\rho_{\mathrm{in}})\bigr) \\
		&= \operatorname{Tr}(M_x\rho_{\mathrm{in}})
		+ t\frac{\delta}{\sqrt M}
		\sum_{j=1}^M\theta_j\operatorname{Tr}\bigl(M_x G_j(\rho_{\mathrm{in}})\bigr)
		+ \mathcal O(\delta^2).
	\end{align*}
	Since $\rho_{\mathrm{in}}=\ket{1^N}\!\bra{1^N}$ and $M_x=\ket{x}\!\bra{x}$,
	\[
	\operatorname{Tr}(M_x\rho_{\mathrm{in}})
	= \delta_{x,1^N}
	= Q(x).
	\]
	Using Lemma~\ref{lem:local_amp_damp_action},
	\[
	\operatorname{Tr}\bigl(M_x G_j(\rho_{\mathrm{in}})\bigr)
	= \gamma(\delta_{x,x_j^{(0)}} - \delta_{x,1^N}).
	\]
	It follows that
	\[
	P_{\mathcal L(\theta)}(x)
	= Q(x)
	+ \frac{\delta}{\sqrt M}
	\sum_{j=1}^M\theta_j
	\bigl[t\gamma(\delta_{x,x_j^{(0)}} - \delta_{x,1^N})\bigr]
	+ \mathcal O(\delta^2),
	\]
	so the first-order coefficients are
	\[
	a_{j,x}
	:= t\gamma(\delta_{x,x_j^{(0)}} - \delta_{x,1^N}),
	\]
	which yields~\eqref{eq:linear_response_P} and
	\eqref{eq:a_jx_explicit}.
\end{proof}

\begin{lemma}[Norms and mean scaling of the first-order coefficient vectors $a^{(x)}$]
	\label{lem:a_x_norms}
	In the setting of the previous lemma,
	associate to each outcome $x\in X$ the first-order coefficient vector
	\[
	a^{(x)}
	:= (a_{1,x},\dots,a_{M,x})\in\mathbb R^M,\qquad
	\|a^{(x)}\|_2
	:= \Bigl(\sum_{j=1}^M a_{j,x}^2\Bigr)^{1/2}.
	\]
	Then:
	\begin{enumerate}
		\item For the all-$1$ outcome $x=1^N$,
		\[
		\|a^{(1^N)}\|_2
		= \gamma t\sqrt{M}.
		\]
		\item For each $j=1,\dots,M$,
		the outcome $x_j^{(0)}$ satisfies
		\[
		\|a^{(x_j^{(0)})}\|_2
		= \gamma t.
		\]
		\item For all other outcomes (bitstrings containing at least two zeros),
		we have
		\(
		\|a^{(x)}\|_2 = 0
		\).
		\item Consequently, the averaged norm satisfies
		\begin{equation}\label{eq:a_x_norm_average}
			\frac{1}{M}\sum_{x\in X}\|a^{(x)}\|_2
			= \gamma t\Bigl(1 + \frac{1}{\sqrt{M}}\Bigr)
			\xrightarrow[M\to\infty]{}\gamma t.
		\end{equation}
	\end{enumerate}
\end{lemma}

\begin{proof}
	From~\eqref{eq:a_jx_explicit},
	for fixed $x$ the components are
	\[
	a_{j,x} = t\gamma(\delta_{x,x_j^{(0)}} - \delta_{x,1^N}).
	\]
	
	\medskip
	(1) For $x=1^N$,
	\(\delta_{x,x_j^{(0)}} = 0\) and $\delta_{x,1^N}=1$, so
	\[
	a_{j,1^N} = -t\gamma,\qquad j=1,\dots,M.
	\]
	Hence
	\[
	\|a^{(1^N)}\|_2
	= \Bigl(\sum_{j=1}^M (t\gamma)^2\Bigr)^{1/2}
	= \gamma t\sqrt{M}.
	\]
	
	\medskip
	(2) For $x=x_k^{(0)}$ (the bitstring with the $k$-th bit $0$ and others $1$),
	\[
	\delta_{x_k^{(0)},x_j^{(0)}} = \delta_{j,k},\qquad
	\delta_{x_k^{(0)},1^N} = 0,
	\]
	so
	\[
	a_{j,x_k^{(0)}}
	= t\gamma(\delta_{j,k} - 0)
	= \begin{cases}
		t\gamma, & j=k,\\
		0, & j\neq k.
	\end{cases}
	\]
	Thus
	\[
	\|a^{(x_k^{(0)})}\|_2
	= \bigl((t\gamma)^2 + 0 + \dots + 0\bigr)^{1/2}
	= \gamma t.
	\]
	
	\medskip
	(3) For any other $x$ (with at least two zeros), $x$ is neither $1^N$ nor any $x_j^{(0)}$, so
	\(
	\delta_{x,x_j^{(0)}} = 0 = \delta_{x,1^N}
	\)
	for all $j$, whence $a_{j,x}=0$ and the norm is zero.
	
	\medskip
	(4) Summing over all $x$, only $x=1^N$ and the $M$ single-flip strings $x_j^{(0)}$ contribute, hence
	\begin{align*}
		\sum_{x\in X}\|a^{(x)}\|_2
		&= \|a^{(1^N)}\|_2 + \sum_{j=1}^M \|a^{(x_j^{(0)})}\|_2
		+ \sum_{\text{other }x}\|a^{(x)}\|_2 \\
		&= \gamma t\sqrt{M} + M\gamma t.
	\end{align*}
	Dividing by $M$ yields
	\begin{align*}
		\frac{1}{M}\sum_{x\in X}\|a^{(x)}\|_2
		&= \frac{\gamma t\sqrt{M} + M\gamma t}{M} \\
		&= \gamma t\Bigl(1 + \frac{1}{\sqrt{M}}\Bigr),
	\end{align*}
	which converges to $\gamma t$ as $M\to\infty$.
\end{proof}

\begin{theorem}[Mean linear-response TV distance in the random local amplitude-damping model]\label{thm:TV_linear_response_mean}
	In the setting of Definition~\ref{def:local_amp_damp_model},
	Lemmas~\ref{lem:a_jx_explicit} and
	\ref{lem:a_x_norms},
	define
	\[
	F_{\mathrm{TV}}(\theta)
	:= d_{\mathrm{TV}}(P_{\mathcal L(\theta)},Q)
	= \frac12\sum_{x\in X}
	\bigl|P_{\mathcal L(\theta)}(x)-Q(x)\bigr|,
	\]
	and
	\[
	m(\delta)
	:= \mathbb E_{\theta\sim\mu_{\mathrm{sph}}}
	\bigl[F_{\mathrm{TV}}(\theta)\bigr].
	\]
	Then in the small-noise limit $\delta\to 0$,
	\begin{equation}\label{eq:m_epsilon_general}
		m(\delta)
		= \delta\,
		\frac{\kappa_M}{2M}
		\sum_{x\in X}\|a^{(x)}\|_2
		+ \mathcal O(\delta^2),
	\end{equation}
	where $\kappa_M$ is as in Lemma~\ref{lem:sphere_inner_product} and Remark~\ref{gauss}.
	Using Lemma~\ref{lem:a_x_norms} we obtain in particular
	\begin{equation}\label{eq:m_epsilon_toy_model}
		m(\delta)
		= \delta\,
		\frac{\kappa_M}{2}\,\gamma t
		\Bigl(1+\frac{1}{\sqrt{M}}\Bigr)
		+ \mathcal O(\delta^2),
	\end{equation}
	and in the limit $M\to\infty$,
	\begin{equation}\label{eq:m_epsilon_limit}
		m(\delta)
		\xrightarrow[M\to\infty]{}
		\delta\,
		\gamma t\,
		\sqrt{\frac{1}{2\pi}}
		+ \mathcal O(\delta^2).
	\end{equation}
	In other words, in this concrete random local amplitude-damping model, the first-order linear-response scaling constant of the mean TV distance is
	\[
	m_l := \gamma t\sqrt{\frac{1}{2\pi}}>0,
	\]
	so that
	\[
	m(\delta)
	= m_l\,\delta + \mathcal O(\delta^2)
	\quad(M\to\infty).
	\]
\end{theorem}

\begin{proof}
	By Lemma~\ref{lem:a_jx_explicit},
	\[
	P_{\mathcal L(\theta)}(x)-Q(x)
	= \frac{\delta}{\sqrt{M}}
	\sum_{j=1}^M \theta_j a_{j,x}
	+ \mathcal O(\delta^2)
	= \frac{\delta}{\sqrt{M}}\,
	\langle\theta,a^{(x)}\rangle
	+ \mathcal O(\delta^2),
	\]
	where $a^{(x)}=(a_{1,x},\dots,a_{M,x})$.
	Thus the TV distance can be written as
	\[
	F_{\mathrm{TV}}(\theta)
	= \frac12\sum_{x\in X}
	\bigl|P_{\mathcal L(\theta)}(x)-Q(x)\bigr|
	= \frac{\delta}{2\sqrt{M}}
	\sum_{x\in X}
	\bigl|\langle\theta,a^{(x)}\rangle\bigr|
	+ \mathcal O(\delta^2).
	\]
	Taking expectation over $\theta\sim\mu_{\mathrm{sph}}$ and using Lemma~\ref{lem:sphere_inner_product},
	\[
	\mathbb E_\theta
	\bigl[|\langle\theta,a^{(x)}\rangle|\bigr]
	= \frac{\kappa_M}{\sqrt{M}}\,
	\|a^{(x)}\|_2.
	\]
	Hence
	\begin{align*}
		m(\delta)
		&= \mathbb E_\theta F_{\mathrm{TV}}(\theta) \\
		&= \frac{\delta}{2\sqrt{M}}
		\sum_{x\in X}
		\mathbb E_\theta
		\bigl[|\langle\theta,a^{(x)}\rangle|\bigr]
		+ \mathcal O(\delta^2) \\
		&= \frac{\delta}{2\sqrt{M}}
		\sum_{x\in X}
		\frac{\kappa_M}{\sqrt{M}}\,
		\|a^{(x)}\|_2
		+ \mathcal O(\delta^2) \\
		&= \delta\,
		\frac{\kappa_M}{2M}
		\sum_{x\in X}\|a^{(x)}\|_2
		+ \mathcal O(\delta^2),
	\end{align*}
	which is~\eqref{eq:m_epsilon_general}.
	
	For this concrete model,
	Lemma~\ref{lem:a_x_norms} and~\eqref{eq:a_x_norm_average} give
	\[
	\frac{1}{M}\sum_{x\in X}\|a^{(x)}\|_2
	= \gamma t\Bigl(1 + \frac{1}{\sqrt{M}}\Bigr),
	\]
	so
	\begin{align*}
		m(\delta)
		&= \delta\,\frac{\kappa_M}{2M}
		\sum_{x\in X}\|a^{(x)}\|_2
		+ \mathcal O(\delta^2) \\
		&= \delta\,\frac{\kappa_M}{2}\,\gamma t
		\Bigl(1 + \frac{1}{\sqrt{M}}\Bigr)
		+ \mathcal O(\delta^2),
	\end{align*}
	which is~\eqref{eq:m_epsilon_toy_model}.
	
	Using
	\(\kappa_M\to\sqrt{2/\pi}\) (Remark~\ref{gauss}) and
	\(
	1+1/\sqrt{M}\to 1
	\),
	we obtain
	\begin{align*}
		m(\delta)
		&\xrightarrow[M\to\infty]{}
		\delta\,
		\frac{1}{2}\sqrt{\frac{2}{\pi}}\gamma t
		+ \mathcal O(\delta^2) \\
		&= \delta\,\gamma t
		\sqrt{\frac{1}{2\pi}}
		+ \mathcal O(\delta^2),
	\end{align*}
	i.e., the first-order linear-response constant is
	\(m_l=\gamma t\sqrt{1/(2\pi)}\), as claimed.
\end{proof}

\begin{remark}[Physical interpretation and scaling summary]
	Theorem~\ref{thm:TV_linear_response_mean} shows that in this random local amplitude-damping Lindbladian model, although the parameter dimension $M=N$ grows with system size, in the small-noise limit the linear-response coefficient
	\(\frac{1}{M}\sum_x\|a^{(x)}\|_2\)
	of the mean TV distance does not decay with $M$, but instead converges to a strictly positive constant
	\(\gamma t\). Multiplying this by the spherical geometric factor $\kappa_M/2$ gives the macroscopic constant
	\(m_0=\gamma t\sqrt{1/(2\pi)}\).
	This indicates that in this class of ``local Lindbladian + random parameter sphere'' models, as long as the noise strength $\delta$ is fixed, the mean TV distance does not shrink with increasing parameter dimension. This supports the modeling expectation in the ``open-system Porter-Thomas mean assumption'' that $m_0>0$ does not depend on $M$ (or at least does not decay with $M$).
\end{remark}

\subsection{Numerical simulations: validating the open-system Porter--Thomas mean assumption}
\label{subsec:numerics}

In this subsection we perform numerical simulations on the one-dimensional local amplitude-damping model in
Definition~\ref{def:local_amp_damp_model},
in order to test
Theorem~\ref{thm:TV_linear_response_mean}
and provide numerical support for Assumption~\ref{ass:mean_TV}.
By Theorem~\ref{thm:TV_linear_response_mean},
in the limit $M \rightarrow \infty$ we have
\[
m(\delta)
:= \mathbb{E}_{\theta\sim\mu_{\mathrm{sph}}}
d_{\mathrm{TV}}\bigl(P_{\mathcal L(\theta)},Q\bigr)
= m_l\,\delta + \mathcal O(\delta^2),\qquad
m_l = \gamma t\sqrt{\tfrac{1}{2\pi}} ,
\]
where $m_l$ is independent of the parameter dimension $M$.
In the numerics we set
$\gamma=t=1$, so that the theoretical slope is
$m_l=\sqrt{1/(2\pi)}\approx 0.399$.

The numerical model is exactly the same as in
Theorem~\ref{thm:TV_linear_response_mean}:
the system is an $N$-qubit chain with Hilbert space dimension $d=2^N$,
the initial state is
$\rho_{\mathrm{in}}=\ket{1^N}\!\bra{1^N}$,
the measurement is the global computational-basis POVM,
and the reference distribution is $Q(x)=\delta_{x,1^N}$.
In the Liouville representation,
for each local dissipator $G_j$ we construct the $d^2\times d^2$ super-operator matrix using the column-major vectorization convention
\(
\mathrm{vec}(A\rho B^\dagger)
=(B^\mathsf{T}\otimes A)\mathrm{vec}(\rho)
\).
Given $\theta\in S^{M-1}$ we form
\[
L(\theta)
=
\frac{\delta}{\sqrt{M}}\sum_{j=1}^M \theta_j G_j,
\qquad
\Lambda_t(\theta)
=
\exp\bigl(t L(\theta)\bigr),
\]
and then obtain $\rho_{t,\theta}$ from
\(
\mathrm{vec}(\rho_{t,\theta})
=
\Lambda_t(\theta)\,\mathrm{vec}(\rho_{\mathrm{in}})
\).
Reading off the diagonal entries of $\rho_{t,\theta}$ yields the output distribution $P_\theta$.
For each pair $(N,\delta)$,
we independently sample
$n_{\mathrm{samp}}$ parameters
$\theta^{(\ell)}$ from the sphere $S^{M-1}$,
and estimate
\[
\widehat m(\delta;N)
=
\frac{1}{n_{\mathrm{samp}}}
\sum_{\ell=1}^{n_{\mathrm{samp}}}
d_{\mathrm{TV}}\bigl(P_{\theta^{(\ell)}},Q\bigr)
\]
via Monte Carlo,
where $M=N$.
The sample standard deviation $\widehat\sigma(\delta;N)$
yields the standard error
$\widehat\sigma(\delta;N)/\sqrt{n_{\mathrm{samp}}}$,
which we use as the error bars.

Fig.~\ref{fig:m_delta_small} shows the numerical results for
$N=2,3,4,5,6$
in the linear-response window
$\delta\in[0.1,0.2]$.
Within this interval we take $10$ equally spaced points,
and at each point we estimate
$\widehat m(\delta;N)$ with
$n_{\mathrm{samp}}=200$ samples.
For each $N$ we perform a linear fit
\(
\widehat m(\delta;N)
\approx m_{mc}^{(N)}\,\delta + b_N .
\)
The fitted slopes
$m_{mc}^{(N)}$
are concentrated in the interval
$0.39\text{--}0.42$;
for example,
$m_{mc}^{(2)}\approx 0.423$,
$m_{mc}^{(4)}\approx 0.403$,
$m_{mc}^{(6)}\approx 0.391$.
These agree with the theoretical value
$m_l\approx0.399$
within a relative error of about $5\%\text{--}10\%$,
while the intercepts $b_N$ are much smaller than the data itself and numerically close to $0$.
As $N$ increases (and in this model $M=N$),
the fitted lines gradually approach the common theoretical line
$m(\delta)=m_l\delta$,
indicating that the finite-size corrections decrease with growing $M$,
in agreement with the $M\to\infty$ limit in
Theorem~\ref{thm:TV_linear_response_mean}.
This provides a direct numerical validation of our analytic derivation in Sec.~\ref{zhang15},
and supports the physical interpretation in Assumption~\ref{ass:mean_TV}
that “$m>0$ is a constant independent of $M$, rather than decaying with $M$”.

\begin{figure}[t]
	\centering
	\includegraphics[width=0.48\textwidth]{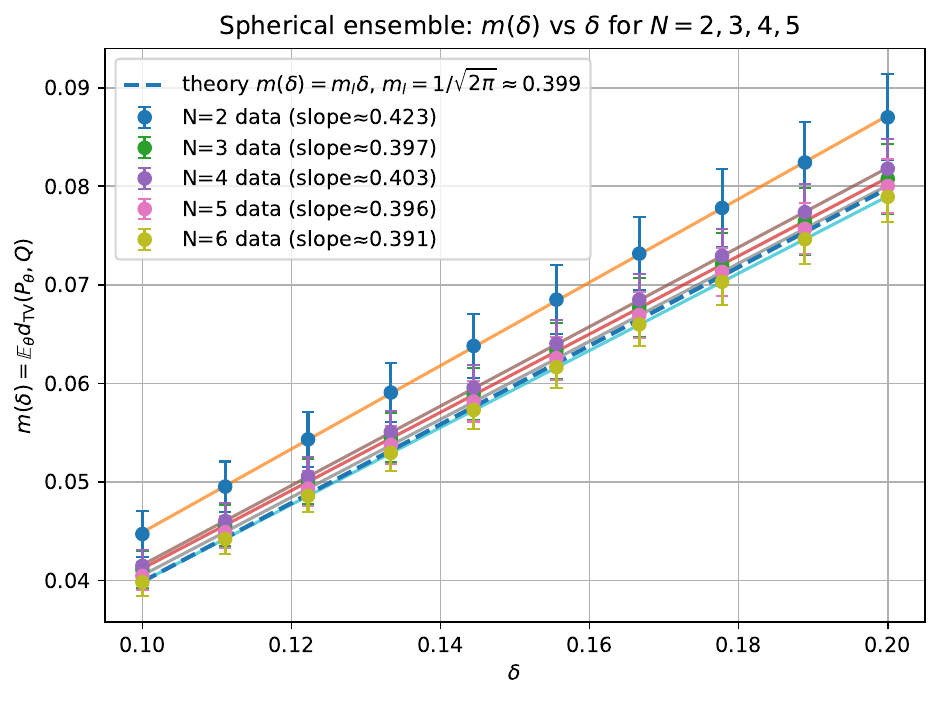}
	\caption{%
		Average TV distance
		$m(\delta;N)
		=\mathbb{E}_\theta d_{\mathrm{TV}}(P_\theta,Q)$
		in the spherical random local Lindbladian model,
		for $N=2,\dots,6$ in the small-noise window
		$\delta\in[0.1,0.2]$.
		Solid dots are Monte Carlo data with error bars indicating the standard error.
		Solid lines are linear fits for each $N$,
		while the dashed line shows the theoretical linear response
		$m(\delta)=m_l\delta$ with
		$m_l=\sqrt{1/(2\pi)}\approx0.399$.}
	\label{fig:m_delta_small}
\end{figure}

\begin{figure}[t]
	\centering
	\includegraphics[width=0.48\textwidth]{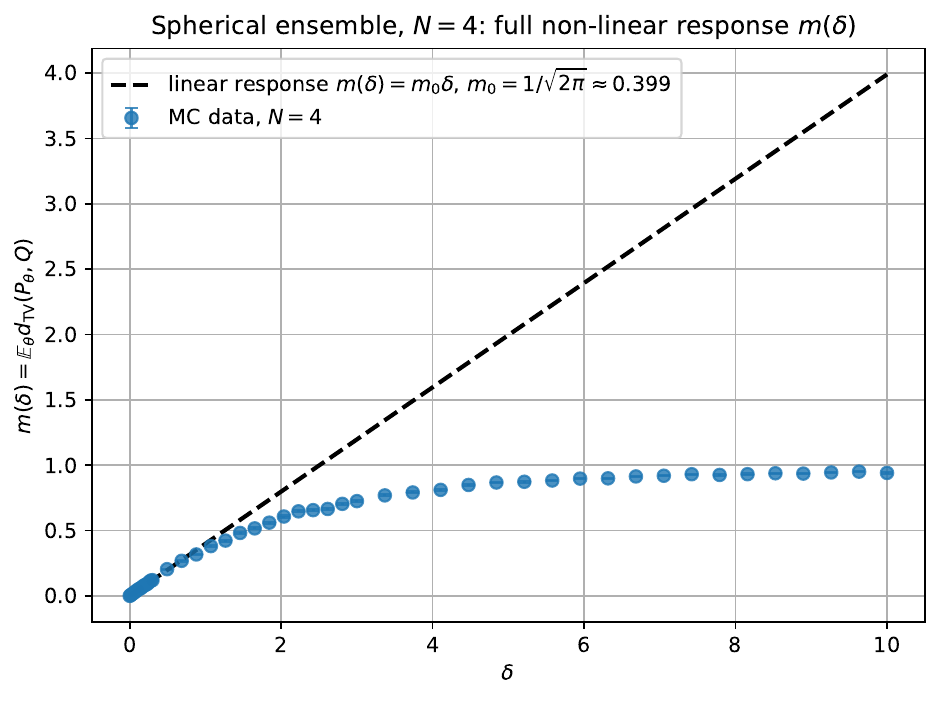}
	\caption{%
		Average TV distance
		$m(\delta;4)$
		in the spherical random local Lindbladian model
		for a fixed system size $N=4$
		over a wide noise range
		$\delta\in[0,10]$.
		Solid dots are Monte Carlo data with error bars,
		and the dashed line shows the linear-response prediction
		$m(\delta)=m_l\delta$.
		The two coincide very well in the small-noise regime,
		while for larger $\delta$
		the curve $m(\delta)$ gradually bends away from the straight line and saturates,
		illustrating the nonlinear response of the Lindbladian dynamics.
		In particular, as long as $\delta>0$,
		$m(\delta;4)$ remains strictly positive and of order $\mathcal O(1)$.}
	\label{fig:m_delta_full}
\end{figure}

To probe the breakdown of linear response and the strong-noise behavior,
we further fix $N=4$
and extend the noise strength to
$\delta\in[0,10]$.
We use non-uniform step sizes in the three intervals
$\delta\in[0.01,0.3]$,
$[0.3,3.0]$, and $[3.0,10.0]$
(with denser sampling at small $\delta$).
At each point we estimate
$\widehat m(\delta;4)$ with
$n_{\mathrm{samp}}=400$ samples.
The results are shown in Fig.~\ref{fig:m_delta_full}.
In the small-noise regime $\delta\lesssim0.2$,
the numerical data almost perfectly follows the straight line $m_l\delta$,
again confirming the expansion
$m(\delta)
= m_l\delta+\mathcal O(\delta^2)$
from Theorem~\ref{thm:TV_linear_response_mean}.
When $\delta$ increases to $O(1)$ and beyond,
$m(\delta)$ continues to grow monotonically,
but clearly lies below the linear extrapolation,
and after $\delta\sim 3$
gradually saturates to an $\mathcal O(1)$ constant.
This can be understood as follows:
once the local dissipation strength is sufficiently large,
the system almost completely relaxes within the fixed evolution time $t$
to a steady state that is essentially independent of $\theta$
(near the all-zero state in this model).
In this regime,
the TV distance between $P_\theta$ and the reference distribution $Q$ is bounded by the structure of the steady state itself and no longer grows linearly with $\delta$,
leading to a characteristic nonlinear response and saturation in the large-noise region.

To more directly characterize the behavior of the finite-size error term
$\epsilon_{\mathrm{mean}}>0$
in Assumption~\ref{ass:mean_TV}
as the system parameters increase,
we further fix two representative noise strengths:
one in the linear-response regime,
$\delta=0.15$,
and one in the nonlinear regime,
$\delta=1.0$.
We then study
\[
m(\delta;N)
:=
\mathbb{E}_{\theta} d_{\mathrm{TV}}\bigl(P_{\mathcal L(\theta)},Q\bigr)
\]
as a function of the number of qubits $N$.
Numerically we take $N=1,2,3,4,5,6$,
and for each $(\delta,N)$ we estimate
$\widehat m(\delta;N)$ using
$n_{\mathrm{samp}}=400$ samples.
For a fixed $\delta$,
we regard the value at the largest system size $N_{\max}=6$,
\(
\widehat m(\delta;N_{\max}),
\)
as an ``approximate asymptotic value''
$\widehat m_\infty(\delta)$,
and plot only the data for $N=1,\dots,5$.
The left and right panels of Fig.~\ref{fig:m_vs_N} show the results for
$\delta=0.15$ (linear-response regime) and
$\delta=1.0$ (nonlinear regime), respectively;
the horizontal dashed lines indicate the numerical reference values
$\widehat m_\infty(\delta)$.

We observe that,
both in the small-noise and in the nonlinear regime,
$\widehat m(\delta;N)$
converges steadily towards the same horizontal dashed line as $N$ increases.
For $\delta=0.15$,
as $N$ grows from $1$ to $5$,
the average TV distance decreases from about $7\times10^{-2}$
and approaches $\widehat m_\infty(0.15)$.
For $\delta=1.0$,
$m(\delta;N)$ in the nonlinear regime exhibits a mild increase followed by a rapid saturation,
and approaches the reference value at $N=6$.
For a given noise strength $\delta$,
the average TV distance shows only a weak dependence on $N$ at finite size,
and quickly converges to a constant that is independent of $N$ (and hence $M$),
in line with the assumption in Assumption~\ref{ass:mean_TV}
that there exists a limiting mean $m_0(\delta)$.

\begin{figure}[t]
	\centering
	\includegraphics[width=0.9\textwidth]{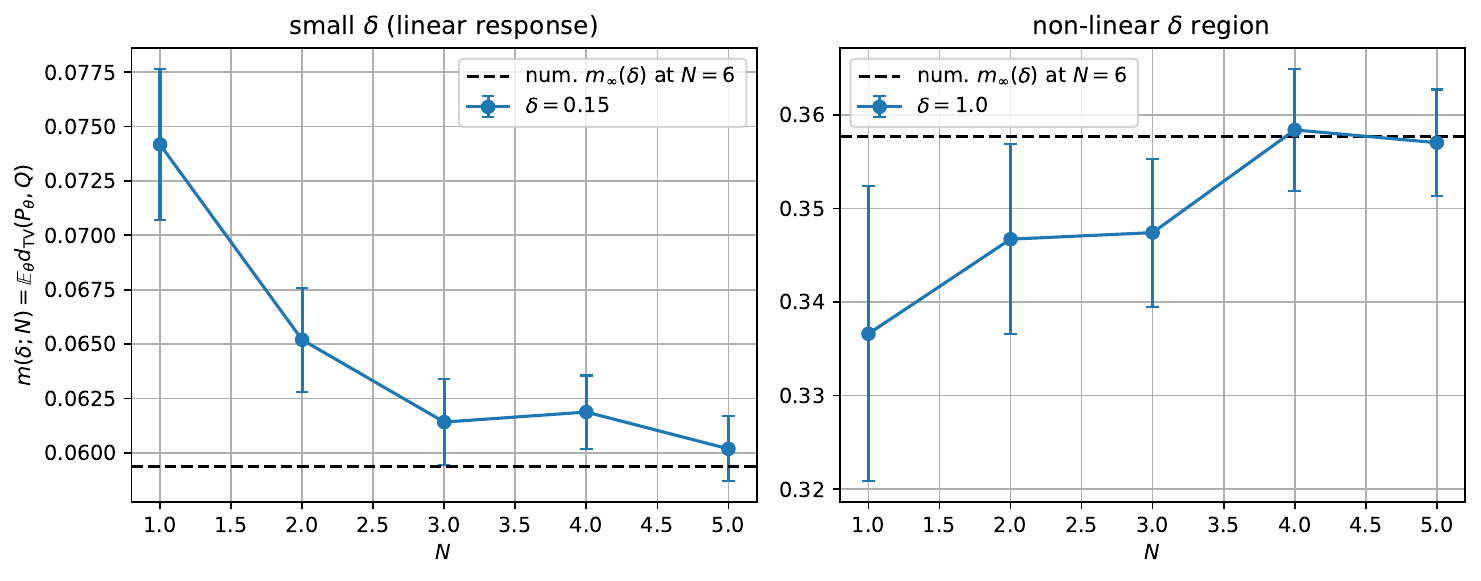}
	\caption{%
		Average TV distance
		$m(\delta;N)
		=\mathbb{E}_\theta d_{\mathrm{TV}}(P_\theta,Q)$
		as a function of the number of qubits $N$
		in the spherical random local Lindbladian model.
		The left panel shows the representative point $\delta=0.15$ in the linear-response regime,
		and the right panel shows the representative point $\delta=1.0$ in the nonlinear regime.
		Solid dots are Monte Carlo data with error bars indicating the standard error,
		and the dashed lines indicate the numerical reference values at $N=6$,
		$\widehat m_\infty(\delta)
		:=\widehat m(\delta;N=6)$.
		As $N$ increases from $1$ to $5$,
		$m(\delta;N)$ quickly converges to a common macroscopic constant at both noise strengths.}
	\label{fig:m_vs_N}
\end{figure}

\begin{figure}[t]
	\centering
	\includegraphics[width=0.6\textwidth]{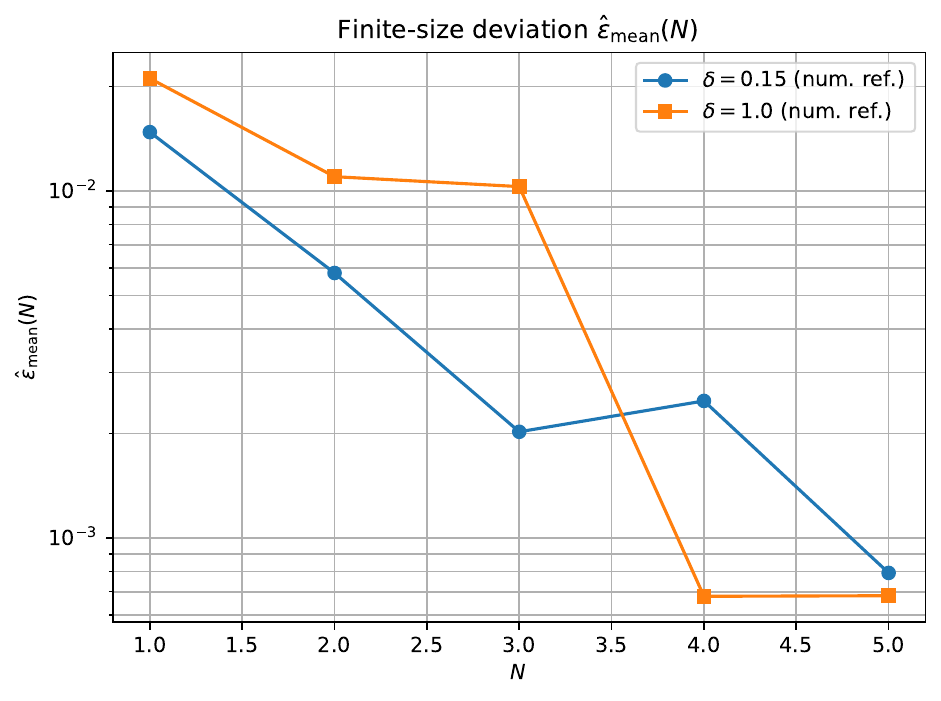}
	\caption{%
		Finite-size deviation
		$\widehat\epsilon_{\mathrm{mean}}(\delta;N)
		=
		|\widehat m(\delta;N)-\widehat m_\infty(\delta)|$
		as a function of the number of qubits $N$,
		where $\widehat m_\infty(\delta)$
		is taken from the numerical value at $N=6$.
		Blue circles and orange squares correspond to
		$\delta=0.15$ in the linear-response regime
		and $\delta=1.0$ in the nonlinear regime, respectively.
		The vertical axis is on a logarithmic scale.
		Both curves exhibit a clear trend towards $0$,
		showing that the finite-size corrections to the average TV distance
		decay rapidly in this model,
		and thus providing direct numerical evidence for
		the condition $\epsilon_{\mathrm{mean}}\to0$ in
		Assumption~\ref{ass:mean_TV}.}
	\label{fig:eps_mean}
\end{figure}

To quantify the finite-size bias,
we define a numerical proxy for the error term
$\epsilon_{\mathrm{mean}}$
in Assumption~\ref{ass:mean_TV} in this model as
\[
\widehat\epsilon_{\mathrm{mean}}(\delta;N)
:=
\bigl|\widehat m(\delta;N)-\widehat m_\infty(\delta)\bigr|,
\qquad
\widehat m_\infty(\delta)
=
\widehat m(\delta;N_{\max}),\ N_{\max}=6,
\]
and plot
$\widehat\epsilon_{\mathrm{mean}}(\delta;N)$
on a semi-logarithmic scale in Fig.~\ref{fig:eps_mean}.
We see that,
for both $\delta=0.15$ and $\delta=1.0$,
$\widehat\epsilon_{\mathrm{mean}}(\delta;N)$
decreases approximately monotonically,
from order $10^{-2}$ at $N=1$
down to $10^{-3}$ or smaller at $N=5$.
This shows that in this local Lindbladian model
the finite-size deviation
$\epsilon_{\mathrm{mean}}$
of the average TV distance indeed decreases rapidly as the system size (number of qubits $N$ and parameter dimension $M$) grows,
thereby providing numerical evidence for the condition
$\epsilon_{\mathrm{mean}}\to 0$
in Assumption~\ref{ass:mean_TV}.

\section{PUFs from random Lindbladians: distribution-level and tomography-based verification}  
\label{sec:crqpuf_applications}

In this section we turn the SQ/QPStat learning-hardness results for random Lindbladian ensembles obtained above into two concrete cryptographic protocols.

Physical unclonable functions (PUFs) were originally proposed as hardware primitives that rely on microscopic fabrication noise, and are used to implement cryptographic tasks such as authentication and fingerprinting~\cite{chang2017retrospective,delavar2017puf,pappu2002physical}. In the classical setting, PUFs are typically realized by specific circuit architectures or optical media, where uncontrollable imperfections introduced during fabrication generate an unclonable challenge--response relation. However, subsequent works have shown that such classical PUFs are vulnerable to side-channel attacks and machine-learning attacks: once sufficiently many challenge--response pairs have been observed, an adversary can often construct an efficient classical model that reproduces the PUF’s behavior at the statistical level~\cite{ruhrmair2010modeling,ganji2016strong,tebelmann2019side,khalafalla2019pufs}. In order to bypass, at the level of principle, such learning-based attacks, the notion of a quantum PUF was introduced, using the no-cloning of quantum states and Haar-random unitary evolutions to achieve information-theoretic authentication~\cite{arapinis2021quantum,phalak2021quantum}. Existing secure constructions, however, typically require approximate Haar-random unitaries as well as quantum communication and quantum memory, which makes their implementation on foreseeable near- and medium-term platforms extremely costly~\cite{arapinis2021quantum}. As a compromise, classical-readout quantum physically unclonable functions (CR-QPUFs) were proposed, aiming to retain part of the security of quantum PUFs while relying only on classical communication and storage~\cite{pirnay2022learning,phalak2021quantum}.

Our goal in this section is to show that the learning-hardness results for random Lindbladian ensembles can be turned into explicit Lindbladian-PUF protocols in which
(i) the verifier’s and honest prover’s costs remain polynomial in the relevant system parameters, while
(ii) any adversary that passes verification with non-negligible probability must, in the SQ or QPStat sense, solve an exponentially hard learning problem.

We will consider two verification paradigms:

\begin{itemize}
	\item[(A)] \textbf{Distribution-level verification via an SQ interface and a huge test-function family.}  
	In the factory (setup) phase, the manufacturer has sample access to the device implementing the Lindbladian channel $E_\theta$ followed by a fixed measurement $\{M_x\}_{x\in X}$, and uses this sample access to learn the entire classical output distribution
	$P_\theta(x) = \operatorname{Tr}(M_x E_\theta(\rho_{\mathrm{in}}))$
	up to small coordinate-wise error, storing the fingerprint
	$\{P_\theta(x)\}_{x\in X}$ as a classical secret.
	After deployment, however, the physical device only exposes an SQ interface $\operatorname{Stat}_\tau(P_\theta)$: external users, including any adversary, can no longer see raw samples, but can only obtain approximate expectation values $P_\theta[\varphi]$ of test functions $\varphi:X\to[-1,1]$.
	
	The public challenge space is taken to be an exponentially large family of Hadamard-type test functions
	$\mathcal F_{\mathrm{Had}}$ on $X$ (Definition~\ref{def:hadamard_tests}), indexed by bit strings $\mathrm{BIT}\in\{0,1\}^L$.
	In each authentication session the verifier draws $N_{\mathrm{chal}}$ test functions uniformly at random from $\mathcal F_{\mathrm{Had}}$, and checks whether the SQ oracle’s answers are consistent with the precomputed probability fingerprint. An honest prover, who has physical access to the genuine device, only needs a polynomial number of SQ calls to pass these tests , so verification is efficient.
	
	On the other hand, because the challenge space has size $|\mathcal F_{\mathrm{Had}}|=2^L$, any table-lookup attack that attempts to precompute and store answers for all possible challenges would require exponentially many SQ queries in $L$. Our design therefore rules out such non-learning attacks by construction. Moreover, any adversary that does not rely on table lookup, but still wishes to answer a random subset of Hadamard tests function correctly with non-negligible probability, is effectively forced to SQ-learn the whole output distribution in total variation distance. By the SQ-hardness theorems in Sec.~\ref{result 2}, this requires an exponential number of SQ queries in the Lindbladian parameter dimension $M$, while the verifier’s cost remains polynomial in $|X|$ and independent of $M$. We therefore regard Scheme~A as a genuinely positive cryptographic construction: verification is easy, whereas successful attacks are SQ-exponentially hard.
	
	\item[(B)] \textbf{Tomography-based verification via an extended QPStat interface.}  
	At the channel level we consider an extended QPStat oracle (Definition~\ref{def:extended_qpstat}) and construct tomographic bases on the operator space and observable space (Lemma~\ref{lem:tomographic_bases}). In the setup phase, the verifier uses QPStat (or experimental) tomography~\cite{bengtsson2017geometry,nielsen2010quantum} to estimate a tomographically complete family of expectation values $\operatorname{Tr}[G_k E_\theta(F_j)]$, and stores the resulting tomographic fingerprint matrix $\mathcal T(E_\theta)$ (Definition~\ref{def:tomographic_fingerprint}) as a classical secret. In the authentication phase, the prover is required to reproduce these entries entrywise. By Lemma~\ref{lem:fingerprint_norm_equivalence}, passing verification implies that the prover has learned the entire channel $E_\theta$ in diamond distance, and is therefore ruled out by the QPStat hardness theorem for random Lindbladians, Theorem~\ref{thm:lindblad_qpstat_lower_bound} in Section~\ref{zhang8}.
	
	However, the tomographic fingerprint has $D^2\sim d^4$ entries when the Hilbert-space dimension is $d$, so even with constant accuracy per entry the verifier’s initialization cost $q_{\mathrm{tom}}$ scales at least like a fixed polynomial in $d$, and can be as large as order $d^4$. In typical local Lindbladian models the parameter dimension $M$ is at most polynomial in $d$, so the verification cost is of essentially the same order as the information content of the channel itself. From a cryptographic point of view this falls into the usual dilemma that verification is ``as hard as'' the learning task~\cite{goldreich2004foundations}, and we therefore view Scheme~B mainly as an information-theoretic benchmark rather than a practically efficient Lindbladian-PUF.
\end{itemize}

We will first present the SQ-based distribution-level Lindbladian-PUF together with its security theorem, and regard it as the main positive cryptographic construction of this section. We then introduce the extended QPStat model and tomographic bases, and on this basis discuss the QPStat security and verification cost of the tomography-based Lindbladian-PUF.

\subsection{Scheme A: Lindbladian-PUF based on output distributions}
\label{subsec:crqpuf_distribution}

In this subsection we present a Lindbladian-PUF protocol based on an SQ interface and a family of Hadamard-type test functions.
In the security analysis, we do not impose any structural restriction on the internal attack strategy of the adversary $A$:
$A$ may be any (possibly randomized) algorithm, as long as during the initialization phase its access to the external device
is strictly mediated through the SQ interface $\operatorname{Stat}_\tau(P_\theta)$, and the total number of queries does not exceed a given budget $q$.
In this black-box SQ attack model, the core design ideas of Scheme A can be summarized as follows:

\begin{itemize}
	\item[1.] In the manufacturing phase, the vendor efficiently learns the full output distribution $P_\theta$ in the laboratory via a sample interface, 
	and stores it as a secret probability fingerprint.
	
	\item[2.] After the device leaves the factory, the physical device only exposes the SQ interface
	$\operatorname{Stat}_\tau(P_\theta)$ to the outside world.
	Any external adversary cannot see single-shot measurement outcomes, and can only query the approximate expectation values
	of arbitrarily chosen test functions $\varphi:X\to[-1,1]$.
	
	\item[3.] During the authentication phase, the verifier randomly samples several Hadamard-type test functions 
	$\varphi^{(r)}$ from an exponentially large test-function family of size $|\mathcal F_{\mathrm{Had}}|=2^L$,
	computes the target values $u_r$ using the internally stored probability fingerprint $\{P_\theta(x)\}$,
	and checks whether the responses of the adversary’s simulator are within a small error from $u_r$ in each round.
\end{itemize}

This design brings the following advantages in terms of SQ query complexity, and thereby guarantees the security of the secret:

\begin{itemize}
	\item[1.] An honest prover only needs to invoke the SQ interface a polynomial number of times to pass authentication:
	in the initialization phase the sample complexity is controlled by Hoeffding’s inequality and scales as $\mathrm{poly}(n)$,
	and in the authentication phase one SQ query to $\operatorname{Stat}_\tau(P_\theta)$ per round suffices
	(see Remark~\ref{rem:Hoeffding_sq} and the completeness analysis below).
	
	\item[2.] If an adversary attempts a table-lookup attack --- i.e., during the attack phase it tries to pre-query and store approximate values
	of $P_\theta[\varphi_{\mathrm{BIT}}]$ for as many challenge bit strings $\mathrm{BIT}$ as possible --- then, because the Hadamard test-function family
	has size $2^L$, covering a large fraction of challenges necessarily requires SQ queries to exponentially many test functions.
	In our SQ security model, the total number $q$ of SQ queries is treated as the main resource, hence such table-lookup strategies
	are infeasible if $q=\mathrm{poly}(M)$.
	
	\item[3.] More generally, irrespective of whether the adversary uses table-lookup, interpolation, or any other learning-type or non-learning-type strategy,
	as long as it passes sufficiently many random Hadamard tests in the authentication phase with a constant success probability,
	Lemma~\ref{lem:hadamard_single_round} and Lemma~\ref{lem:hadamard_reconstruct_dist} together imply that its simulator $w_{\theta,R}$ must approximate
	$P_\theta[\varphi_{\mathrm{BIT}}]$ extremely well on a large fraction of challenges $\mathrm{BIT}$.
	Consequently, in the SQ model one can reconstruct an approximate distribution $D_{\theta,R}$ satisfying
	$d_{\mathrm{TV}}(P_\theta,D_{\theta,R})\le\varepsilon_{\mathrm{SQ}}$.
	In other words, any adversary that can pass authentication is, from an information-theoretic point of view,
	equivalent to an SQ learner that has successfully learned $P_\theta$, and therefore necessarily triggers
	the exponential average-case SQ lower bound for random Lindbladians established in Sec.~\ref{result 2}
	(see Theorem~\ref{thm:crqpuf_sq_hadamard_security}).
\end{itemize}

We now formalize the test-function family and the protocol, and then state and prove the security theorem.

\medskip

We first describe in detail the Hadamard-type test-function family and the SQ interface, using the SQ model notation of Sec.~\ref{result 2}
(see also Definition~\ref{def:SQ_oracle}):
for a distribution $P\in\mathcal D_X$ over $X$ and a tolerance $\tau>0$,
the statistical-query oracle $\operatorname{Stat}_\tau(P)$ returns, for any query function $\varphi:X\to[-1,1]$, a real number $v$ such that
\[
|v-P[\varphi]|\le\tau.
\]

In this subsection, $X$ is a finite classical measurement output space (for example, the finite set of outcomes obtained from some coarse-grained block-projective measurement
$\{M_x\}_{x\in X}$; in principle one could also take $X=\{0,1\}^n$ to be the computational basis measurement outcomes of an $n$-qubit system, but we will explain later why this is not ideal).
To construct an exponentially large challenge space, we introduce the following encoding and test functions.

\begin{definition}[Encoding and Hadamard-type test-function family]
	\label{def:hadamard_tests}
	Let $X$ be a finite set with $|X|=N$.
	Choose $L\ge\lceil\log_2 N\rceil$, and fix a public injective encoding
	\[
	h:X\hookrightarrow\{0,1\}^L.
	\]
	The number of distinct functions $\varphi_{\mathrm{BIT}}(x)=(-1)^{\langle \mathrm{BIT},h(x)\rangle}$ realized on $X$ equals $2^{r}$, where $$r:=\dim_{\mathbb F_2}\!\bigl(\mathrm{span}_{\mathbb F_2}(h(X))\bigr) =\mathrm{rank}_{\mathbb F_2}(H)$$with $H$ the $|X|\times L$ matrix of rows $h(x)$, so the full exponential size $2^L$ is attained iff $r=L$, and otherwise the challenge family collapses to size $2^r$~\cite{roth2006introduction}. We choose $h$ such that the row span of the $|X| \times L$ matrix $H$ has full rank $r=L$ (this is always possible once $|X| \geq L$ ), so that $\left|\mathcal{F}_{\text {Had }}\right|=2^L$ distinct test functions are realized on $X$. For each bit string $\mathrm{BIT}\in\{0,1\}^L$, define the test function
	\begin{equation}
		\varphi_{\mathrm{BIT}}(x)
		:=(-1)^{\langle\mathrm{BIT},h(x)\rangle},
		\qquad x\in X,
	\end{equation}
	where $\langle\mathrm{BIT},h(x)\rangle$ denotes the bitwise inner product modulo $2$, i.e.,
	\[
	\langle\mathrm{BIT},h(x)\rangle
	:=\Bigl(\sum_{j=1}^L \mathrm{BIT}_j\,h(x)_j\Bigr)\bmod 2.
	\]
	Clearly $\varphi_{\mathrm{BIT}}(x)\in\{-1,+1\}\subset[-1,1]$ for all $x$, so $\varphi_{\mathrm{BIT}}$ is a valid SQ query function.
	
	Define the exponentially large test-function family
	\[
	\mathcal F_{\mathrm{Had}}
	:=\{\varphi_{\mathrm{BIT}}:\ \mathrm{BIT}\in\{0,1\}^L\},
	\]
	whose size is $|\mathcal F_{\mathrm{Had}}|=2^L$.
\end{definition}

Intuitively, $\mathcal F_{\mathrm{Had}}$ is the restriction to $h(X)$ of the Walsh--Hadamard basis on the Boolean hypercube $\{0,1\}^L$.
For any distribution $P\in\mathcal D_X$, the test values $P[\varphi_{\mathrm{BIT}}]$ can be viewed as Hadamard-type Fourier coefficients of $P$~\cite{o2014analysis}.

\begin{remark}[Size of the challenge space and feasibility of the construction]
	Once $h$ is fixed, the size of the challenge space
	$\mathcal F_{\mathrm{Had}}$ is exactly $2^L$.
	In the theoretical analysis we typically take $L=\Theta(n)$ or treat it as an independent security parameter.
	The verifier only needs to sample $N_{\mathrm{chal}}$ test functions at random from this exponentially large family to carry out $N_{\mathrm{chal}}$ rounds of authentication,
	so any table-lookup attack that aims to cover the entire challenge space must pay an exponential SQ query cost.
\end{remark}

\medskip

We now explain why we introduce a block decomposition of the Hilbert space and coarse-grained projective measurements.
In this scheme we care about the classical output distribution after random Lindbladian evolution.
To avoid the extreme situation that the probability of each individual bit string is exponentially small,
we first coarse-grain the Hilbert space into blocks at the quantum level and then perform projective measurements on these blocks as the POVM elements.

\begin{definition}[Orthogonal block decomposition and block-projective POVM]
	\label{def:block_povm}
	Consider an $n$-qubit system with Hilbert space $\mathcal H\simeq(\mathbb C^2)^{\otimes n}$ and dimension $d=2^n$.
	Take a finite index set $X$ and choose a family of mutually orthogonal subspaces
	\[
	\mathcal H = \bigoplus_{x\in X}\mathcal H_x,
	\]
	where the subspaces $\mathcal H_x\subset\mathcal H$ are pairwise orthogonal.
	Let $\Pi_x$ be the orthogonal projector onto $\mathcal H_x$, and define the POVM elements
	\[
	M_x := \Pi_x,\qquad x\in X.
	\]
\end{definition}

\begin{lemma}[Operator norm and completeness of block-projective POVM]
	\label{lem:block_povm_norm}
	In the setting of Definition~\ref{def:block_povm}, we have:
	\begin{enumerate}
		\item For each $x\in X$, $\Pi_x$ is an orthogonal projector:
		\(
		\Pi_x^\dagger=\Pi_x,\ \Pi_x^2=\Pi_x.
		\)
		\item For each $x\in X$,
		\[
		0\le \Pi_x\le I,
		\qquad
		\|\Pi_x\|_\infty = 1.
		\]
		\item $\{M_x\}_{x\in X}$ forms a projective measurement (POVM):
		\(
		\sum_{x\in X} M_x = I,
		\)
		and $M_x M_y = 0$ for all $x\neq y$.
	\end{enumerate}
\end{lemma}

\begin{proof}
	From the standard properties of the orthogonal decomposition
	$\mathcal H = \bigoplus_{x\in X}\mathcal H_x$,
	each $\Pi_x$ is the orthogonal projector onto $\mathcal H_x$,
	hence $\Pi_x^\dagger=\Pi_x$ and $\Pi_x^2=\Pi_x$.
	
	For any unit vector $\ket{\psi}$ we have
	\(
	\|\Pi_x\ket{\psi}\|\le \|\ket{\psi}\| = 1
	\),
	so $\|\Pi_x\|_\infty\le 1$.
	On the other hand, if we take any normalized vector $\ket{\phi}\in\mathcal H_x$,
	then $\Pi_x\ket{\phi}=\ket{\phi}$,
	which implies $\|\Pi_x\|_\infty\ge 1$.
	Combining these gives $\|\Pi_x\|_\infty=1$ and $0\le\Pi_x\le I$.
	
	By the orthogonal direct sum decomposition of $\mathcal H$,
	any vector $\ket{\psi}\in\mathcal H$ can be uniquely written as
	\(
	\ket{\psi} = \sum_{x\in X}\ket{\psi_x},
	\)
	with $\ket{\psi_x}\in\mathcal H_x$.
	Hence
	\[
	\Bigl(\sum_{x\in X} \Pi_x\Bigr)\ket{\psi}
	= \sum_{x\in X}\Pi_x\ket{\psi}
	= \sum_{x\in X}\ket{\psi_x}
	= \ket{\psi},
	\]
	i.e. $\sum_{x\in X} \Pi_x = I$,
	and for $x\neq y$ we have $\Pi_x \Pi_y=0$.
	Therefore $\{M_x\}_{x\in X}$ forms a POVM.
\end{proof}

\begin{remark}[Coarse-grained measurements and preservation of Lipschitz bounds]
	Lemma~\ref{lem:block_povm_norm} tells us that
	even if each measurement operator $M_x$ may have very high rank,
	its operator norm still satisfies
	\(
	\|M_x\|_\infty = 1.
	\)
	Therefore in our previous estimates of Lipschitz constants,
	if we take $O=M_x$ we still have
	\[
	\bigl|\operatorname{Tr}\bigl[M_x\,(E_\theta-E_{\theta'})(\rho)\bigr]\bigr|
	\le \|M_x\|_\infty\,\|(E_\theta-E_{\theta'})(\rho)\|_1
	\le \|E_\theta-E_{\theta'}\|_{1\to 1},
	\]
	so the Lipschitz structure is not spoiled by making the measurement more coarse-grained or higher rank.
\end{remark}

In concrete designs, we may further choose $\mathcal H$ to be decomposed into only polynomially many blocks,
so as to avoid the situation that each $P_\theta(x)$ becomes exponentially small under an approximately maximally mixed state
(which is the typical behavior of random Lindbladians for typical parameter choices).
The following simple lemma illustrates this idea.

\begin{lemma}[Probability lower bound under coarse-grained projective measurement (under approximate maximally mixed assumption)]
	\label{lem:block_povm_prob_lower}
	Consider an $n$-qubit system with
	$\mathcal H\simeq(\mathbb C^2)^{\otimes n}$ and dimension $d=2^n$,
	and take an orthogonal decomposition as in Definition~\ref{def:block_povm}:
	\[
	\mathcal H = \bigoplus_{x\in X}\mathcal H_x.
	\]
	Suppose there exists a polynomial $\kappa(n)$ such that
	\[
	|X| \le \kappa(n),
	\qquad
	\dim\mathcal H_x \ge \frac{d}{\kappa(n)}
	\quad\text{for all }x\in X.
	\]
	For each $\theta$, let
	\(
	\rho_\theta := E_\theta(\rho_{\mathrm{in}})
	\),
	and define the output distribution
	\(
	P_\theta(x):=\operatorname{Tr}\bigl(M_x\rho_\theta\bigr).
	\)
	If there exists a constant $\varepsilon_{\mathrm{mix}}>0$ such that on some typical parameter set we have
	\[
	\bigl\|\rho_\theta - \tfrac{I}{d}\bigr\|_1
	\le \varepsilon_{\mathrm{mix}},
	\]
	then for all $\theta$ in this set and all $x\in X$,
	\begin{equation}
		\bigl|P_\theta(x) - \tfrac{\dim\mathcal H_x}{d}\bigr|
		\le \varepsilon_{\mathrm{mix}},
	\end{equation}
	and thus, if
	$\varepsilon_{\mathrm{mix}}\le \tfrac12\cdot \tfrac{\dim\mathcal H_x}{d}$,
	\begin{equation}
		P_\theta(x)
		\;\ge\;
		\frac{\dim\mathcal H_x}{2d}
		\;\ge\;
		\frac{1}{2\kappa(n)},
	\end{equation}
	i.e., the probability of each coarse-grained outcome is lower bounded by $1/\mathrm{poly}(n)$
	and does not become exponentially small as in the single-bit-string case.
\end{lemma}

\begin{proof}
	Since $0\le M_x\le I$ and $\|M_x\|_\infty\le 1$,
	we have
	\[
	\bigl|P_\theta(x) - \tfrac{\dim\mathcal H_x}{d}\bigr|
	= \bigl|\operatorname{Tr}\bigl[M_x(\rho_\theta - \tfrac{I}{d})\bigr]\bigr|
	\le \|M_x\|_\infty\,\bigl\|\rho_\theta - \tfrac{I}{d}\bigr\|_1
	\le \varepsilon_{\mathrm{mix}}.
	\]
	On the other hand,
	\(
	\operatorname{Tr}(M_x\,\tfrac{I}{d}) = \tfrac{\dim\mathcal H_x}{d}
	\),
	hence
	\[
	\bigl|P_\theta(x) - \tfrac{\dim\mathcal H_x}{d}\bigr|
	\le \varepsilon_{\mathrm{mix}}.
	\]
	If
	\(
	\varepsilon_{\mathrm{mix}}\le \tfrac12\cdot \tfrac{\dim\mathcal H_x}{d}
	\),
	then
	\[
	P_\theta(x)
	\ge \tfrac{\dim\mathcal H_x}{d} - \varepsilon_{\mathrm{mix}}
	\ge \tfrac12\cdot \tfrac{\dim\mathcal H_x}{d}.
	\]
	Using the assumption $\dim\mathcal H_x \ge d/\kappa(n)$, we obtain
	\[
	P_\theta(x)
	\ge \frac{1}{2}\cdot\frac{d/\kappa(n)}{d}
	= \frac{1}{2\kappa(n)}.
	\]
	This yields the desired polynomial lower bound.
\end{proof}

\begin{remark}[Why use block-projective POVMs rather than single-bit-string measurements?]
	\label{rem:why_block_povm}
	We deliberately choose to decompose the Hilbert space into orthogonal subspaces
	$\mathcal H=\bigoplus_{x\in X}\mathcal H_x$ and use the block projectors $M_x=\Pi_x$ as POVM elements,
	instead of directly measuring in the computational basis $\{\ket{z}\}_{z\in\{0,1\}^n}$ bit string by bit string, for several reasons:
	\begin{enumerate}
		\item From the perspective of statistical learning and sample complexity, bit-string-wise measurements lead to exponentially small probabilities and a blow-up in sample complexity.
		For an $n$-qubit system, a typical state under computational-basis measurement yields a probability of order $2^{-n}$ for each bit string.
		In this case, if one wants to accurately learn the full distribution in total variation distance, one typically needs resolution on each coordinate comparable to $2^{-n}$,
		which corresponds to an exponentially small tolerance $\eta$.
		Hoeffding’s inequality (see Remark~\ref{rem:Hoeffding_sq}) then implies that the sample complexity
		$T=\Theta(\eta^{-2} \log (|X| / \delta_{\mathrm{puf}}))$ becomes exponential.
		In other words, completeness of the protocol would rely on resolving exponentially rare events,
		which is infeasible for any sampling-based practical implementation.
		
		\item Coarse-grained blocking lifts probabilities to a polynomial scale, keeping Hoeffding sample complexity polynomial.
		Via Definition~\ref{def:block_povm} we decompose the Hilbert space into subspaces $\mathcal H_x$,
		and Lemma~\ref{lem:block_povm_prob_lower} shows that
		once the random Lindbladian has mixed $\rho_\theta$ close to the maximally mixed state $I/d$ for typical parameters,
		and each block satisfies
		$\dim\mathcal H_x \ge d/\kappa(n)$ with $|X|\le\kappa(n)$,
		the coarse-grained outcome probabilities
		\(
		P_\theta(x)=\operatorname{Tr}(M_x\rho_\theta)
		\)
		are at least $1/\mathrm{poly}(n)$.
		Under this setup, we can fix the SQ tolerance $\tau$ at a constant scale;
		choosing $\eta=\tau/|X|$, the Hoeffding sample complexity
		$T=\Theta\bigl(\eta^{-2}\log(|X|/\delta_{\mathrm{puf}})\bigr)$
		remains $\mathrm{poly}(|X|)=\mathrm{poly}(n)$.
		Thus, information-theoretically, learning the probability fingerprint in the initialization phase is feasible,
		and we avoid the exponential sample complexity that arises with bit-string-wise measurement.
		
		\item From a complexity-theoretic viewpoint, bit-string-wise measurement corresponds to learning individual output bit-string probabilities,
		whereas coarse-grained block measurements correspond to estimating expectation values of relatively simple observables.
		These lie at different complexity levels.
		In settings such as random circuit sampling, for typical families of universal quantum circuits, additively approximating a single output probability
		$p(z)=\Pr[Z=z]$ is tightly connected to \#P counting problems; under standard assumptions and oracle models,
		approximating such probabilities is \#P-hard, and learning the full output distribution in TV distance is widely believed to be a \#P-type task,
		as reflected in the complexity analyses of random circuit sampling~\cite{aaronson2011computational,bouland2019complexity,arute2019quantum}.
		By contrast, our coarse-grained probabilities
		\[
		P_\theta(x)=\operatorname{Tr}(M_x\rho_\theta)
		\]
		can be viewed as expectation values of relatively simple, efficiently implementable projectors $M_x$.
		Such expectation values can be approximated in BQP using standard quantum algorithms such as the Hadamard test or amplitude estimation~\cite{nielsen2010quantum}.
		Therefore, in an ideal noise model, learning the coarse-grained output distribution naturally falls within BQP,
		instead of being a \#P-type problem of ``approximating individual amplitudes''.
	\end{enumerate}
\end{remark}

In what follows, we assume that $\{M_x\}_{x\in X}$ is induced by a fixed orthogonal block decomposition
$\mathcal H=\bigoplus_x \mathcal H_x$ via block projectors, with $|X|\le\kappa(n)=\mathrm{poly}(n)$.

Let $\{E_\theta\}_{\theta\in\Theta}$ be the family of random Lindbladian channels,
$\rho_{\mathrm{in}}\in\mathsf S(\mathcal H)$ a fixed input state,
and $\{M_x\}_{x\in X}$ the block-projective measurement induced by an orthogonal decomposition
$\mathcal H=\bigoplus_x\mathcal H_x$ as in the previous subsection.
For each $\theta$ define the output distribution
\[
P_\theta(x)
:=\operatorname{Tr}\bigl(M_x E_\theta(\rho_{\mathrm{in}})\bigr),
\qquad x\in X.
\]

\begin{definition}[Lindbladian-PUF protocol with SQ interface and Hadamard test functions]
	\label{def:crqpuf_sq_hadamard}
	Fix an SQ tolerance parameter $\tau>0$ and an authentication round parameter $N_{\mathrm{chal}}\in\mathbb N$.
	The protocol between verifier $V$ and prover $P$ consists of two phases.
	
	\paragraph{Initialization phase (inside the factory, with sample access).}
	\begin{enumerate}
		\item[i.] In the manufacturer’s laboratory, $V$ samples a parameter $\theta$ from $\mu_\Theta$ at random,
		obtaining sample-level access to the physical device.
		At this stage the manufacturer can repeatedly prepare $\rho_{\mathrm{in}}$,
		evolve it via $\mathcal E_\theta := E_\theta$ and measure with $\{M_x\}$,
		obtaining independent samples $X_1,X_2,\dots$.
		
		\item[ii.] Let $\eta>0$ be the target pointwise accuracy of the fingerprint.
		The verifier $V$ runs the device $T$ times to obtain independent samples
		$X_1,\dots,X_T\in X$, and defines the empirical frequencies
		\[
		\hat P_\theta(x)
		:= \frac{1}{T}\sum_{t=1}^T \mathbf 1\{X_t=x\}
		\quad\text{for each }x\in X.
		\]
		By Hoeffding’s inequality (see Remark~\ref{rem:Hoeffding_sq}),
		choosing
		\[
		T = \Theta\bigl(\eta^{-2}\log(|X|/\delta_{\mathrm{puf}})\bigr)
		\]
		guarantees that with probability at least $1-\delta_{\mathrm{puf}}$ we have
		\[
		\bigl|\hat P_\theta(x)-P_\theta(x)\bigr|\le\eta
		\quad\forall x\in X.
		\]
		Conditional on this high-probability event, $V$ stores
		\[
		\mathbf p
		:=\{\hat P_\theta(x)\}_{x\in X}
		\]
		as the probability fingerprint of the device, kept secret.
		In our parameter choices below, we will set
		\begin{equation}
			\label{eq:eta_choice}
			\eta := \frac{\tau}{|X|},
		\end{equation}
		so that the total error in computing expectations of any bounded test function using $\mathbf p$
		is at most $\tau$.
		
		\item[iii.] After initialization, the manufacturer irreversibly disables the internal sample interface,
		leaving only the external SQ interface active.
		In the cryptographic model, we abstract the device as an ideal statistical-query oracle
		$\operatorname{Stat}_\tau(P_\theta)$ (we do not care how it is implemented internally; it may use its own sampling to realize the oracle):
		for any external query function $\varphi:X\to[-1,1]$,
		it returns $v$ such that
		\[
		|v-P_\theta[\varphi]|\le\tau.
		\]
		The device is then sold to the end user.
	\end{enumerate}
	
	\paragraph{Authentication phase (with SQ interface, table-lookup attacks restricted).}
	\begin{enumerate}
		\item[i.] When $P$ claims to hold a genuine device (for example, when the device requires maintenance and the manufacturer is asked to send a technician),
		$V$ initiates authentication.
		
		\item[ii.] $V$ first fixes the public encoding $h:X\to\{0,1\}^L$,
		and hence determines the Hadamard test-function family $\mathcal F_{\mathrm{Had}}$ of Definition~\ref{def:hadamard_tests}.
		
		\item[iii.] For $r=1,2,\dots,N_{\mathrm{chal}}$, the following challenge–response procedure is repeated:
		\begin{enumerate}
			\item[A.] $V$ samples a bit string
			$\mathrm{BIT}^{(r)}$ uniformly from $\{0,1\}^L$,
			and constructs the test function
			$\varphi^{(r)}:=\varphi_{\mathrm{BIT}^{(r)}}$.
			Using the internally stored probability fingerprint $\mathbf p$,
			$V$ computes
			\[
			u_r
			:= \sum_{x\in X}\varphi^{(r)}(x)\,\hat P_\theta(x),
			\]
			which approximates the true expectation
			\[
			P_\theta[\varphi^{(r)}]
			:= \sum_{x\in X}\varphi^{(r)}(x)\,P_\theta(x).
			\]
			Since for all $x\in X$ we have
			$\bigl|\hat P_\theta(x)-P_\theta(x)\bigr|\le\eta$
			and $|\varphi^{(r)}(x)|\le 1$, it follows that
			\begin{equation}
				\label{eq:u_r_error_bound}
				\bigl|u_r - P_\theta[\varphi^{(r)}]\bigr|
				\le \sum_{x\in X}\bigl|\hat P_\theta(x)-P_\theta(x)\bigr|
				\le |X|\,\eta.
			\end{equation}
			Combining this with the parameter choice~\eqref{eq:eta_choice},
			on the high-probability event of the initialization phase we have
			\[
			\bigl|u_r - P_\theta[\varphi^{(r)}]\bigr|
			\le \tau.
			\]
			
			\item[B.] $V$ sends the bit string $\mathrm{BIT}^{(r)}$ to $P$.
			
			\item[C.] The honest prover $P$ queries the device’s SQ interface,
			using $\varphi^{(r)}$ as the query function to $\operatorname{Stat}_\tau(P_\theta)$,
			and receives an output $v_r$ satisfying
			$|v_r-P_\theta[\varphi^{(r)}]|\le\tau$,
			which is then returned to $V$.
		\end{enumerate}
		
		\item[iv.] $V$ accepts the authentication if and only if for all $r=1,\dots,N_{\mathrm{chal}}$,
		\begin{equation}
			\bigl|v_r-u_r\bigr|
			\le 2\tau.
		\end{equation}
		The parameter $2\tau$ arises from the triangle inequality: one part comes from the approximation error of $u_r$ to $P_\theta[\varphi^{(r)}]$
		(see~\eqref{eq:u_r_error_bound} and~\eqref{eq:eta_choice}), and another part comes from the tolerance $\tau$ of the SQ interface itself.
	\end{enumerate}
\end{definition}

\begin{remark}[Hoeffding’s inequality and sample complexity in the initialization phase]
	\label{rem:Hoeffding_sq}
	In the initialization phase, the verifier $V$ repeatedly runs the device $T$ times,
	obtaining independent samples $X_1,\dots,X_T\in X$ and defining
	\[
	\hat P_\theta(x)
	:= \frac{1}{T}\sum_{t=1}^T \mathbf 1\{X_t=x\},
	\]
	which is the empirical average of Bernoulli random variables with parameter $P_\theta(x)$.
	Hoeffding’s inequality states that for any $\eta>0$,
	\[
	\Pr\Bigl(
	\bigl|\hat P_\theta(x)-P_\theta(x)\bigr|>\eta
	\Bigr)
	\le 2\exp(-2T\eta^2).
	\]
	Applying the union bound over all $x\in X$ yields
	\[
	\Pr\Bigl(
	\exists x\in X:\ \bigl|\hat P_\theta(x)-P_\theta(x)\bigr|>\eta
	\Bigr)
	\le 2|X|\exp(-2T\eta^2).
	\]
	Thus, as long as
	\[
	T = \Theta\bigl(\eta^{-2}\log(|X|/\delta_{\mathrm{puf}})\bigr),
	\]
	we have
	\[
	\Pr\Bigl(
	\bigl|\hat P_\theta(x)-P_\theta(x)\bigr|\le\eta\ \forall x\in X
	\Bigr)\ge 1-\delta_{\mathrm{puf}}.
	\]
	In this protocol we subsequently fix
	$\eta=\tau/|X|$ (see~\eqref{eq:eta_choice}),
	so that the total error in computing the expectation of any test function using the fingerprint $\mathbf p$ does not exceed the target SQ tolerance $\tau$.
	Together with Lemma~\ref{lem:block_povm_prob_lower}, we can constrain $|X|$ to be at most
	$\kappa(n)=\mathrm{poly}(n)$,
	and therefore the above $T$ is feasible on a polynomial scale in $n$.
\end{remark}

\begin{remark}[Completeness and efficiency of the honest prover]
	For an honest prover, the device internally has access to $P_\theta$ via the SQ interface,
	and each query requires only $\Theta(\tau^{-2})$ samples to ensure
	$|v_r-P_\theta[\varphi^{(r)}]|\le\tau$.
	On the other hand, in the initialization phase the fingerprint satisfies
	$\bigl|\hat P_\theta(x)-P_\theta(x)\bigr|\le\eta$
	with $\eta=\tau/|X|$, so by~\eqref{eq:u_r_error_bound} we have
	\[
	\bigl|u_r-P_\theta[\varphi^{(r)}]\bigr|
	\le |X|\eta = \tau.
	\]
	Therefore in each challenge round,
	\[
	\bigl|v_r-u_r\bigr|
	\le \bigl|v_r-P_\theta[\varphi^{(r)}]\bigr|
	+\bigl|P_\theta[\varphi^{(r)}]-u_r\bigr|
	\le \tau + \tau
	= 2\tau,
	\]
	so the honest prover passes all $N_{\mathrm{chal}}$ rounds with high probability.
	The verifier’s main cost is:
	in the initialization phase, $T=\Theta\bigl(\eta^{-2}\log(|X|/\delta_{\mathrm{puf}})\bigr)$
	physical runs of the device, which with $\eta=\tau/|X|$ scales as
	$T=\mathrm{poly}(|X|)=\mathrm{poly}(n)$;
	in the authentication phase, each challenge costs only $\mathcal O(|X|)$ classical operations,
	and the total cost is $N_{\mathrm{chal}}\cdot\mathrm{poly}(|X|)$.
	Thus the protocol is feasible on a polynomial scale in the number of qubits $n$.
\end{remark}

\begin{remark}[Adversary capability model from the SQ viewpoint]
	In the above protocol, we only allow the adversary to query the device’s SQ interface
	$\operatorname{Stat}_\tau(P_\theta)$ at most $q$ times during the attack phase.
	After this phase the device is assumed to be retrieved or deployed in a trusted environment,
	and the adversary can only attempt to construct a simulator that answers the verifier’s challenges during authentication.
	
	Hence the interface available to the adversary in the security analysis is exactly the same as that of a learner in the SQ learning problem of Sec.~\ref{result 2}:
	in each query the adversary can only specify $\varphi:X\to[-1,1]$ and obtain an approximate expectation value;
	it cannot see individual samples, nor can it access the physical device during the authentication phase.
\end{remark}

\medskip

To translate the event of ``passing authentication'' into a statement that the output distribution has been learned with high accuracy,
we need a completeness lemma for the Hadamard test-function family.

\begin{lemma}[Parseval identity for Hadamard test functions]
	\label{lem:hadamard_parseval}
	Embed $X$ into $\{0,1\}^L$ via the encoding $h$, and define
	\[
	\chi_{\mathrm{BIT}}(z)
	:=(-1)^{\langle\mathrm{BIT},z\rangle},
	\qquad z\in\{0,1\}^L.
	\]
	For any two distributions $P,Q\in\mathcal D_X$, define their extensions
	to $\{0,1\}^L$ by
	\[
	\widetilde P(z)
	:=
	\begin{cases}
		P(x), & z=h(x)\text{ for some }x\in X,\\[.3em]
		0, & \text{otherwise},
	\end{cases}
	\qquad
	\widetilde Q\ \text{defined analogously}.
	\]
	Then we have the Parseval-type identity
	\begin{equation}
		\frac{1}{2^L}
		\sum_{\mathrm{BIT}\in\{0,1\}^L}
		\Bigl(
		P[\varphi_{\mathrm{BIT}}]
		-Q[\varphi_{\mathrm{BIT}}]
		\Bigr)^2
		=
		\|\widetilde P-\widetilde Q\|_2^2
		=
		\|P-Q\|_2^2.
	\end{equation}
\end{lemma}

\begin{proof}
	On the Boolean hypercube $\{0,1\}^L$, the family $\{\chi_{\mathrm{BIT}}\}$ forms an orthonormal basis with respect to the inner product
	\(
	\langle f,g\rangle
	:= 2^{-L}\sum_{z} f(z)g(z)
	\)
	(the Walsh--Hadamard basis).
	Therefore, for any function $f:\{0,1\}^L\to\mathbb R$ we have the Parseval identity
	(a basic fact about Fourier expansions; see~\cite{o2014analysis}):
	\[
	\frac{1}{2^L}
	\sum_{\mathrm{BIT}} \Bigl(
	\sum_{z} \chi_{\mathrm{BIT}}(z) f(z)
	\Bigr)^2
	= \sum_z f(z)^2
	= \|f\|_2^2.
	\]
	Take $f:=\widetilde P-\widetilde Q$.
	Then
	\[
	\sum_{z} \chi_{\mathrm{BIT}}(z) f(z)
	= \sum_{x\in X} \chi_{\mathrm{BIT}}(h(x))\bigl(P(x)-Q(x)\bigr)
	= \sum_{x\in X} \varphi_{\mathrm{BIT}}(x)\bigl(P(x)-Q(x)\bigr)
	= P[\varphi_{\mathrm{BIT}}]-Q[\varphi_{\mathrm{BIT}}],
	\]
	and hence
	\[
	\frac{1}{2^L}
	\sum_{\mathrm{BIT}\in\{0,1\}^L}
	\Bigl(
	P[\varphi_{\mathrm{BIT}}]
	-Q[\varphi_{\mathrm{BIT}}]
	\Bigr)^2
	= \sum_{z} f(z)^2
	= \|\widetilde P-\widetilde Q\|_2^2.
	\]
	On the other hand, since $f(z)\neq 0$ only when $z=h(x)$ for some $x\in X$, we have
	\[
	\|\widetilde P-\widetilde Q\|_2^2
	= \sum_{z} f(z)^2
	= \sum_{x\in X}\bigl(\widetilde P(h(x))-\widetilde Q(h(x))\bigr)^2
	= \sum_{x\in X}\bigl(P(x)-Q(x)\bigr)^2
	= \|P-Q\|_2^2.
	\]
	Combining the two identities yields the statement.
\end{proof}

\medskip

\begin{lemma}[From multi-round Hadamard authentication success to single-round prediction accuracy]
	\label{lem:hadamard_single_round}
	Fix a parameter $\theta$ and let $P_\theta\in\mathcal D_X$ be the corresponding target distribution.
	Suppose the adversary’s simulator in the authentication phase is
	\[
	w:\{0,1\}^L\to\mathbb R,
	\qquad
	\mathrm{BIT}\mapsto w(\mathrm{BIT}),
	\]
	and the verifier’s target value is
	\[
	u(\mathrm{BIT})
	:= \sum_{x\in X}\varphi_{\mathrm{BIT}}(x)\,\hat P_\theta(x),
	\]
	where the fingerprint $\hat P_\theta$ satisfies
	\[
	\bigl|\hat P_\theta(x)-P_\theta(x)\bigr|
	\le \eta_{\mathrm{fp}}
	\quad\forall x\in X.
	\]
	Hence for all test functions we have
	\[
	\bigl|u(\mathrm{BIT})-P_\theta[\varphi_{\mathrm{BIT}}]\bigr|
	\le |X|\eta_{\mathrm{fp}}.
	\]
	Suppose the protocol uses a constant $C_0>0$ (in our protocol we take $C_0=2$), and the acceptance condition is:
	in each round $r=1,\dots,N_{\mathrm{chal}}$,
	\[
	\bigl|w(\mathrm{BIT}^{(r)})-u(\mathrm{BIT}^{(r)})\bigr|
	\le C_0\tau.
	\]
	
	Let
	\[
	\Delta_\theta(\mathrm{BIT})
	:= w(\mathrm{BIT})-P_\theta[\varphi_{\mathrm{BIT}}],
	\]
	and assume that the fingerprint parameter satisfies
	\[
	|X|\eta_{\mathrm{fp}}\le\tau,
	\]
	so that the event ``the protocol rejects this round'' is equivalent to
	\[
	\bigl|\Delta_\theta(\mathrm{BIT})\bigr|
	> (C_0+1)\tau
	\quad\Longrightarrow\quad
	\bigl|w(\mathrm{BIT})-u(\mathrm{BIT})\bigr|>C_0\tau.
	\]
	
	If under this acceptance rule the probability of passing all $N_{\mathrm{chal}}$ rounds satisfies
	\[
	\Pr_{\mathrm{chal}}\bigl[\text{all $N_{\mathrm{chal}}$ rounds are accepted}\bigr]\;\ge\;\alpha_0,
	\]
	where $\Pr_{\mathrm{chal}}$ averages only over the uniform and independent draws
	of $\mathrm{BIT}^{(1)},\dots,\mathrm{BIT}^{(N_{\mathrm{chal}})}$ (with $w$ treated as deterministic),
	then for a uniformly random
	$\mathrm{BIT}\sim\mathrm{Unif}(\{0,1\}^L)$ we have
	\begin{equation}\label{0721}
		\Pr\Bigl(
		\bigl|\Delta_\theta(\mathrm{BIT})\bigr|>(C_0+1)\tau
		\Bigr)
		\;\le\;
		\delta_{\mathrm{Had}},
		\qquad
		\delta_{\mathrm{Had}}
		:= 1-\alpha_{0}^{1/N_{\mathrm{chal}}}.
	\end{equation}
	In other words, if the adversary can pass $N_{\mathrm{chal}}$ many random challenges with overall success probability at least $\alpha_0$
	(with $\alpha_0$ close to $1$), then on almost all challenges $\mathrm{BIT}$,
	the adversary’s prediction $w(\mathrm{BIT})$ must approximate $P_\theta[\varphi_{\mathrm{BIT}}]$
	within error $(C_0+1)\tau$.
\end{lemma}

\begin{proof}
	Let the set of ``bad'' challenge bit strings be
	\[
	B_\theta
	:=\Bigl\{\mathrm{BIT}:\ 
	\bigl|\Delta_\theta(\mathrm{BIT})\bigr|>(C_0+1)\tau
	\Bigr\},
	\qquad
	p_\theta
	:=\Pr_{\mathrm{BIT}\sim\mathrm{Unif}}[\mathrm{BIT}\in B_\theta].
	\]
	Using the assumption $|X|\eta_{\mathrm{fp}}\le\tau$ and the triangle inequality, we get for any $\mathrm{BIT}$:
	\[
	\bigl|w(\mathrm{BIT})-u(\mathrm{BIT})\bigr|
	\ge
	\bigl|\Delta_\theta(\mathrm{BIT})\bigr|
	-\bigl|u(\mathrm{BIT})-P_\theta[\varphi_{\mathrm{BIT}}]\bigr|
	> (C_0+1)\tau - \tau
	= C_0\tau.
	\]
	Hence if the challenge bit string of a given round belongs to $B_\theta$, that round necessarily fails.
	Therefore the event ``all $N_{\mathrm{chal}}$ rounds are accepted'' is contained in the event
	``none of the $N_{\mathrm{chal}}$ challenge bit strings lies in $B_\theta$''.
	Since the challenge strings in different rounds are independent and identically distributed, we have
	\[
	\Pr_{\mathrm{chal}}[\text{all $N_{\mathrm{chal}}$ rounds are accepted}]
	\;\le\;
	\Pr\bigl[\mathrm{BIT}^{(1)},\dots,\mathrm{BIT}^{(N_{\mathrm{chal}})}\notin B_\theta\bigr]
	= (1-p_\theta)^{N_{\mathrm{chal}}}.
	\]
	By the assumption $\Pr_{\mathrm{chal}}[\text{all rounds accepted}]\ge\alpha_{0}$, we obtain
	\[
	(1-p_\theta)^{N_{\mathrm{chal}}}\ge\alpha_{0},
	\quad
	p_\theta \le 1-\alpha_{0}^{1/N_{\mathrm{chal}}}
	=:\delta_{\mathrm{Had}}.
	\]
	But $p_\theta$ is exactly
	\(
	\Pr_{\mathrm{BIT}}\bigl(
	|\Delta_\theta(\mathrm{BIT})|>(C_0+1)\tau
	\bigr)
	\),
	so we obtain \eqref{0721}.
\end{proof}

\medskip

\begin{lemma}[Reconstructing an approximate distribution from Hadamard-test predictions]
	\label{lem:hadamard_reconstruct_dist}
	Let $X$ be a finite set, $P\in\mathcal D_X$ a target distribution,
	and $\{\varphi_{\mathrm{BIT}}\}_{\mathrm{BIT}\in\{0,1\}^L}$ the Hadamard test-function family of
	Definition~\ref{def:hadamard_tests}.
	Let $w:\{0,1\}^L\to\mathbb R$ be any function (for example, the output of an adversary’s simulator),
	and suppose that for some $\kappa_{\mathrm{*}}>0$ we have
	\begin{equation}
		\label{eq:hadamard_L2_assumption_new}
		\mathbb E_{\mathrm{BIT}\sim\mathrm{Unif}}\Bigl[
		\bigl(w(\mathrm{BIT})-P[\varphi_{\mathrm{BIT}}]\bigr)^2
		\Bigr]
		\;\le\; \kappa_{\mathrm{*}}^2.
	\end{equation}
	Then there exists a distribution $D\in\mathcal D_X$ such that
	\begin{equation}
		d_{\mathrm{TV}}(P,D)
		\;\le\;
		\sqrt{|X|}\,\kappa_{\mathrm{*}}.
	\end{equation}
\end{lemma}

\begin{proof}
	For any $D\in\mathcal D_X$, consider
	\[
	G(D)
	\;:=\;
	\mathbb E_{\mathrm{BIT}\sim\mathrm{Unif}}\Bigl[
	\bigl(w(\mathrm{BIT})-D[\varphi_{\mathrm{BIT}}]\bigr)^2
	\Bigr].
	\]
	The map
	\[
	\mathcal D_X\ni D
	\;\mapsto\;
	\bigl(D[\varphi_{\mathrm{BIT}}]\bigr)_{\mathrm{BIT}\in\{0,1\}^L}
	\in \mathbb R^{2^L}
	\]
	is continuous, and $\mathcal D_X$ is compact in the $\ell_2$ topology.
	Hence $G(D)$ attains its minimum on $\mathcal D_X$:
	there exists $D^\star\in\mathcal D_X$ such that
	\[
	D^\star
	:=
	\arg\min_{D\in\mathcal D_X} G(D).
	\]
	In the proof we only need the existence of $D^\star$.
	If one actually wants to reconstruct a distribution, one need not explicitly solve this minimization problem.
	The most direct algorithm is to apply the inverse Walsh--Hadamard transform to $w$ to obtain a vector, and then project this vector onto the probability simplex (the projection procedure can be found in~\cite{wang2013projection}).
	We briefly mention that the result of such a constructive procedure is equivalent to that of taking a minimizer $D^\star$; see the derivation below for the minimizer case.
	Let $\tilde{p}$ be the vector obtained by the inverse transform; by Lemma~\ref{lem:hadamard_parseval}, we have
	$\mathbb{E}_{\mathrm{BIT}}\left[\left(w_{\theta, R}(\mathrm{BIT})-F_{P_\theta}(\mathrm{BIT})\right)^2\right]=\left\|\tilde{p}-\tilde{P}_\theta\right\|_2^2\leq \kappa_{*}^2$.
	Since the projection onto a convex set does not increase the $\ell_2$ distance
	($\left\|\Pi_C(a)-\Pi_C(b)\right\|_2 \leq\|a-b\|_2$), if we denote by $D_{\mathrm{pro}}$ the projected distribution, then
	$d_{\mathrm{TV}}\left(P_\theta, D_{\mathrm{pro}}\right) \leq \frac{1}{2} \sqrt{|X|}\left\|P_\theta-D_{\mathrm{pro}}\right\|_2 \leq \frac{1}{2} \sqrt{|X|} \kappa_{*}$.
	
	We now return to the case of the minimizer $D^\star$.
	Since $P$ is also a candidate distribution, the minimum is no worse than the value at $P$, i.e.
	\begin{equation}
		\label{eq:best_approx_no_worse_than_P_new}
		\mathbb E_{\mathrm{BIT}}\Bigl[
		\bigl(w(\mathrm{BIT})-D^\star[\varphi_{\mathrm{BIT}}]\bigr)^2
		\Bigr]
		\;\le\;
		\mathbb E_{\mathrm{BIT}}\Bigl[
		\bigl(w(\mathrm{BIT})-P[\varphi_{\mathrm{BIT}}]\bigr)^2
		\Bigr]
		\;\le\; \kappa_{\mathrm{*}}^2,
	\end{equation}
	where the last inequality follows from the assumption~\eqref{eq:hadamard_L2_assumption_new}.
	
	For each $\mathrm{BIT}\in\{0,1\}^L$, define
	\[
	\Delta(\mathrm{BIT})
	\;:=\;
	D^\star[\varphi_{\mathrm{BIT}}]-P[\varphi_{\mathrm{BIT}}].
	\]
	Let
	\[
	a(\mathrm{BIT})
	:= D^\star[\varphi_{\mathrm{BIT}}]-w(\mathrm{BIT}),\qquad
	b(\mathrm{BIT})
	:= w(\mathrm{BIT})-P[\varphi_{\mathrm{BIT}}],
	\]
	then for each $\mathrm{BIT}$ we have
	\[
	\Delta(\mathrm{BIT})
	= a(\mathrm{BIT})+b(\mathrm{BIT}).
	\]
	For any real numbers $a,b$, the inequality
	\(
	(a+b)^2 \le 2a^2+2b^2
	\)
	holds, hence for each $\mathrm{BIT}$,
	\[
	\Delta(\mathrm{BIT})^2
	\le
	2\,a(\mathrm{BIT})^2
	+2\,b(\mathrm{BIT})^2.
	\]
	Taking expectation over uniform random $\mathrm{BIT}$ and using
	\eqref{eq:hadamard_L2_assumption_new} and
	\eqref{eq:best_approx_no_worse_than_P_new}, we obtain
	\begin{equation}
		\label{eq:Delta_L2_bound_new}
		\mathbb E_{\mathrm{BIT}}[\Delta(\mathrm{BIT})^2]
		\;\le\;
		2\,\mathbb E_{\mathrm{BIT}}\bigl[a(\mathrm{BIT})^2\bigr]
		+2\,\mathbb E_{\mathrm{BIT}}\bigl[b(\mathrm{BIT})^2\bigr]
		\;\le\; 4\kappa_{\mathrm{*}}^2.
	\end{equation}
	
	On the other hand, applying Lemma~\ref{lem:hadamard_parseval} to $P$ and $D^\star$ with $Q:=D^\star$, we have
	\[
	\frac{1}{2^L}
	\sum_{\mathrm{BIT}\in\{0,1\}^L}
	\bigl(P[\varphi_{\mathrm{BIT}}]-D^\star[\varphi_{\mathrm{BIT}}]\bigr)^2
	= \|\widetilde P-\widetilde D^\star\|_2^2
	= \|P-D^\star\|_2^2.
	\]
	The left-hand side is exactly
	\(\mathbb E_{\mathrm{BIT}}[\Delta(\mathrm{BIT})^2]\),
	hence
	\[
	\|P-D^\star\|_2^2
	=
	\mathbb E_{\mathrm{BIT}}[\Delta(\mathrm{BIT})^2].
	\]
	Combining this with~\eqref{eq:Delta_L2_bound_new}, we obtain
	\[
	\|P-D^\star\|_2^2
	\le 4\kappa_{\mathrm{*}}^2,
	\qquad
	\text{and hence}\qquad
	\|P-D^\star\|_2
	\le 2\kappa_{\mathrm{*}}.
	\]
	
	Finally, by the Cauchy--Schwarz inequality,
	\[
	d_{\mathrm{TV}}(P,D^\star)
	= \frac12\|P-D^\star\|_1
	\le \frac12\sqrt{|X|}\,\|P-D^\star\|_2
	\le \frac12\sqrt{|X|}\cdot 2\kappa_{\mathrm{*}}
	= \sqrt{|X|}\,\kappa_{\mathrm{*}}.
	\]
	Taking $D:=D^\star$ completes the proof.
\end{proof}

\medskip

\begin{theorem}[SQ-query-complexity security of the Hadamard-type Lindbladian-PUF with SQ interface]
	\label{thm:crqpuf_sq_hadamard_security}
	Consider the protocol of Definition~\ref{def:crqpuf_sq_hadamard},
	using the same random Lindbladian ensemble as in the SQ hardness theorem of Sec.~\ref{result 2}.
	Let $A$ be any (possibly randomized) adversary that, in the attack phase, may query the device’s SQ interface
	$\operatorname{Stat}_\tau(P_\theta)$ at most $q$ times (i.e., makes at most $q$ statistical queries),
	and then attempts to construct a simulator (classical or quantum) that answers all challenges in the authentication phase.
	
	The verifier carries out $N_{\mathrm{chal}}$ rounds of independent random Hadamard tests during authentication.
	We view $N_{\mathrm{chal}}$ as a security parameter.
	By standard repetition and amplification arguments, if the soundness error per round is a constant,
	then to push the overall authentication error probability down to
	$\delta_{\mathrm{auth}}$, it suffices to take
	$N_{\mathrm{chal}} = \Theta(\log(1/\delta_{\mathrm{auth}}))$.
	Below we simply use the notation $\Theta(\log(1/\delta_{\mathrm{auth}}))$ to emphasize this scaling, without fixing the constant.
	
	Suppose there exist constants $\alpha>1/2$ and $\beta>0$ such that
	\[
	\mu_\Theta\Bigl(
	\theta:\ 
	\Pr_{\substack{\text{over adversary}\\\text{and protocol randomness}}}\bigl[
	\text{$A$ is accepted in the authentication phase}
	\bigr]\ge\alpha
	\Bigr)
	\ \ge\ \beta.
	\]
	Then for sufficiently large parameter dimension $M$, there exists a constant $c>0$ such that
	\begin{equation}
		q \ \ge\ \exp(cM),
	\end{equation}
	i.e., in terms of SQ query complexity, the adversary must issue exponentially many SQ queries in order to break this Lindbladian-PUF protocol with non-negligible success probability.
\end{theorem}

\begin{proof}
	To avoid confusion of notation, we denote
	\[
	\alpha_{\mathrm{auth}}:=\alpha,
	\qquad
	\beta_{\mathrm{auth}}:=\beta
	\]
	for the two constants appearing in the theorem assumption that describe the authentication success probability and the measure of ``good'' parameters.
	They are parameters of the protocol’s security and will later be related to
	$(\alpha_{\mathrm{learn}},\beta_{\mathrm{learn}})$
	in the SQ hardness theorem.
	
	\paragraph*{Step 0: Fixing the ``good event'' for the fingerprint.}
	
	In the initialization phase, the manufacturer uses Hoeffding’s inequality
	(Remark~\ref{rem:Hoeffding_sq})
	to choose a sufficiently large sample size so that the fingerprint $\hat P_\theta$ satisfies
	\[
	\bigl|\hat P_\theta(x)-P_\theta(x)\bigr|
	\le \eta_{\mathrm{fp}}
	= \frac{\tau}{|X|}
	\quad\forall x\in X
	\]
	with probability at least $1-\delta_{\mathrm{puf}}$.
	In the following we always work conditional on this ``good event'' and absorb $\delta_{\mathrm{puf}}$ into the overall authentication error probability
	$\delta_{\mathrm{auth}}$; we do not write it explicitly anymore.
	
	On this good event, for all Hadamard test functions we have
	\[
	\bigl|u(\mathrm{BIT})-P_\theta[\varphi_{\mathrm{BIT}}]\bigr|
	\le \sum_{x\in X}|\varphi_{\mathrm{BIT}}(x)|\,
	\bigl|\hat P_\theta(x)-P_\theta(x)\bigr|
	\le |X|\eta_{\mathrm{fp}}
	= \tau.
	\]
	
	\paragraph*{Step 1: From high overall authentication success to an $L_2$ bound for single-round Hadamard predictions.}
	
	Define the ``good parameter set''
	\[
	\Theta_{\mathrm{good}}
	:= \Bigl\{
	\theta:\ 
	\Pr_{\substack{\text{over adversary}\\\text{and protocol randomness}}}\bigl[
	\text{$A$ is accepted in the authentication phase}
	\bigr]\ge\alpha_{\mathrm{auth}}
	\Bigr\}.
	\]
	By the theorem assumption,
	\[
	\mu_\Theta(\Theta_{\mathrm{good}})\ \ge\ \beta_{\mathrm{auth}}.
	\]
	That is, for at least a $\beta_{\mathrm{auth}}$ fraction of parameters
	$\theta$, the overall success probability of the adversary --- measured over its internal randomness, the SQ oracle randomness, and the authentication challenges --- is at least $\alpha_{\mathrm{auth}}$.
	
	Now fix any
	\(\theta\in\Theta_{\mathrm{good}}\).
	To separate the randomness in the attack phase from the randomness in the authentication phase,
	let
	\[
	R
	= \text{the joint random seed for the adversary’s internal coins and the SQ oracle responses during the attack phase}.
	\]
	For each fixed $R$, we have:
	\begin{itemize}
		\item[1.] After the attack phase, the adversary outputs a deterministic simulator
		\[
		w_{\theta,R}:\{0,1\}^L\to\mathbb R,
		\qquad
		\mathrm{BIT}\mapsto w_{\theta,R}(\mathrm{BIT}).
		\]
		\item[2.] During the authentication phase, the only remaining randomness concerns whether this simulator passes the random challenge bit strings
		$\mathrm{BIT}^{(1)},\dots,\mathrm{BIT}^{(N_{\mathrm{chal}})}$.
	\end{itemize}
	For each fixed $(\theta,R)$, define
	\[
	\mathrm{acc}_\theta(R)
	:= \Pr_{\mathrm{chal}}\Bigl[
	\text{all $N_{\mathrm{chal}}$ rounds are accepted}
	\Bigr],
	\]
	where $\Pr_{\mathrm{chal}}$ averages only over the random challenge bit strings and the verifier’s randomness.
	On the other hand, the overall authentication success probability in the theorem assumption can be written as
	\[
	\Pr_{\substack{\text{over adversary}\\\text{and protocol randomness}}}\bigl[
	\text{$A$ is accepted in the authentication phase}
	\bigr]
	= \mathbb E_R\bigl[\mathrm{acc}_\theta(R)\bigr]
	\ge \alpha_{\mathrm{auth}}
	\qquad
	\forall\,\theta\in\Theta_{\mathrm{good}}.
	\]
	
	To turn the statement ``the mean is at least $\alpha_{\mathrm{auth}}$'' into a statement that ``for most seeds $R$, the acceptance probability is not too small'',
	we introduce a threshold that we are free to choose.
	Here we take
	\[
	\alpha_0
	:= \frac{2\alpha_{\mathrm{auth}}-1}{2},
	\]
	so that $1/2<\alpha_{\mathrm{auth}}\le 1$ implies
	\(\alpha_0\in(0,1/2]\subset(0,1)\).
	Define
	\[
	G_\theta
	:= \{R:\ \mathrm{acc}_\theta(R)\ge\alpha_0\}.
	\]
	Since $0\le \mathrm{acc}_\theta(R)\le 1$, we have
	\[
	\alpha_{\mathrm{auth}}
	\le \mathbb E_R[\mathrm{acc}_\theta(R)]
	\le \Pr[R\in G_\theta]\cdot 1
	+ \Pr[R\notin G_\theta]\cdot \alpha_0.
	\]
	Hence
	\begin{align*}
		\Pr[R\in G_\theta]
		&\ge \frac{\alpha_{\mathrm{auth}}-\alpha_0}{1-\alpha_0} \\
		&= \frac{\alpha_{\mathrm{auth}} - \frac{2\alpha_{\mathrm{auth}}-1}{2}}
		{1 - \frac{2\alpha_{\mathrm{auth}}-1}{2}} \\
		&= \frac{1}{3-2\alpha_{\mathrm{auth}}}
		\;=: p_0.
	\end{align*}
	Since $\alpha_{\mathrm{auth}}>1/2$, we have $3-2\alpha_{\mathrm{auth}}<2$, and thus
	\[
	p_0 = \frac{1}{3-2\alpha_{\mathrm{auth}}} > \frac12.
	\]
	That is, for each fixed good parameter
	$\theta\in\Theta_{\mathrm{good}}$, at least a $p_0>1/2$ fraction of random seeds
	$R$ satisfy $\mathrm{acc}_\theta(R)\ge\alpha_0$.
	
	Now fix such a $\theta$ and some
	\(R\in G_\theta\).
	Conditioned on these, the simulator $w_{\theta,R}$ is deterministic,
	and its overall success probability in $N_{\mathrm{chal}}$ authentication rounds is at least
	$\alpha_0$, i.e.
	\[
	\Pr_{\mathrm{chal}}\bigl[
	\text{all $N_{\mathrm{chal}}$ rounds are accepted}
	\bigr]
	= \mathrm{acc}_\theta(R)
	\ge \alpha_0.
	\]
	
	Since we did not constrain the range of the adversary’s outputs beforehand, we can apply a simple clipping operation:
	if $w_{\theta,R}(\mathrm{BIT})$ lies outside the interval $[-1,1]$,
	the verifier can clip it to $[-1,1]$ before checking the acceptance condition.
	This can only reduce
	$|w_{\theta,R}(\mathrm{BIT})-u(\mathrm{BIT})|$,
	and therefore cannot decrease the adversary’s acceptance probability.
	Thus we may assume without loss of generality that
	\[
	w_{\theta,R}(\mathrm{BIT})\in[-1,1]
	\quad\text{for all }\mathrm{BIT}.
	\]
	
	On the good fingerprint event, we have
	\(|u(\mathrm{BIT})-P_\theta[\varphi_{\mathrm{BIT}}]|\le\tau\) for all test functions.
	Let
	\[
	\Delta_\theta(\mathrm{BIT})
	:= w_{\theta,R}(\mathrm{BIT}) - P_\theta[\varphi_{\mathrm{BIT}}],
	\]
	then whenever
	\(|\Delta_\theta(\mathrm{BIT})|>(C_0+1)\tau\), we have
	\[
	\bigl|w_{\theta,R}(\mathrm{BIT})-u(\mathrm{BIT})\bigr|
	\ge
	\bigl|\Delta_\theta(\mathrm{BIT})\bigr|
	-\bigl|u(\mathrm{BIT})-P_\theta[\varphi_{\mathrm{BIT}}]\bigr|
	> (C_0+1)\tau - \tau
	= C_0\tau,
	\]
	so that round necessarily fails.
	Applying Lemma~\ref{lem:hadamard_single_round},
	we obtain, for uniformly random
	$\mathrm{BIT}\sim\mathrm{Unif}(\{0,1\}^L)$,
	\begin{equation}
		\Pr\Bigl(
		\bigl|\Delta_\theta(\mathrm{BIT})\bigr|>(C_0+1)\tau
		\Bigr)
		\le
		\delta_{\mathrm{Had}},
		\qquad
		\delta_{\mathrm{Had}}
		:= 1-\alpha_0^{1/N_{\mathrm{chal}}}.
	\end{equation}
	
	On the other hand, since
	$w_{\theta,R}(\mathrm{BIT})\in[-1,1]$ and $P_\theta[\varphi_{\mathrm{BIT}}]\in[-1,1]$,
	we have $|\Delta_\theta(\mathrm{BIT})|\le 2$ for all
	$\mathrm{BIT}$.
	Splitting the second moment into ``good'' and ``bad'' events, we obtain
	\begin{align}
		\mathbb E_{\mathrm{BIT}}\bigl[\Delta_\theta(\mathrm{BIT})^2\bigr]
		&= \mathbb E\bigl[\Delta_\theta^2
		\mathbf 1_{|\Delta_\theta|\le (C_0+1)\tau}\bigr]
		+ \mathbb E\bigl[\Delta_\theta^2
		\mathbf 1_{|\Delta_\theta|>(C_0+1)\tau}\bigr] \\
		&\le (C_0+1)^2\tau^2
		+ 4\,\Pr\bigl(|\Delta_\theta|>(C_0+1)\tau\bigr) \nonumber\\
		&\le (C_0+1)^2\tau^2 + 4\,\delta_{\mathrm{Had}} \nonumber\\
		&=: \kappa_{\mathrm{Had}}^2. \label{eq:kappa_had_def}
	\end{align}
	
	\paragraph*{Step 2: Reconstruction via inverse transform and projection; obtaining TV-approximation.}
	
	For each fixed $(\theta,R)$, apply Lemma~\ref{lem:hadamard_reconstruct_dist} with
	\[
	P := P_\theta,
	\qquad
	w := w_{\theta,R},
	\]
	and set
	$\kappa = \kappa_{\mathrm{Had}}$ in \eqref{eq:kappa_had_def}.
	The lemma guarantees the existence of a distribution
	$D_{\theta,R}\in\mathcal D_X$ such that
	\begin{equation}
		\label{eq:SQ_epsilon_def}
		d_{\mathrm{TV}}(P_\theta,D_{\theta,R})
		\le
		\frac{1}{2}\sqrt{|X|}\,\kappa_{\mathrm{Had}}
		\;=:\;\varepsilon_{\mathrm{SQ}}.
	\end{equation}
	The distribution $D_{\theta,R}$ is obtained solely from the simulator $w_{\theta,R}$ via the reconstruction procedure (see
	Lemma~\ref{lem:hadamard_reconstruct_dist}),
	and requires no additional SQ queries.
	
	As discussed above, for each $\theta\in\Theta_{\mathrm{good}}$,
	the probability that $R$ lies in $G_\theta$ is at least $p_0$, where
	\[
	p_0=\frac{1}{3-2\alpha_{\mathrm{auth}}}>\frac12.
	\]
	For these ``good seeds'', equation \eqref{eq:SQ_epsilon_def} holds; for $R\notin G_\theta$ we make no guarantee.
	Therefore, for each $\theta\in\Theta_{\mathrm{good}}$,
	\[
	\Pr_{R}\Bigl(
	d_{\mathrm{TV}}(P_\theta,D_{\theta,R})
	\le \varepsilon_{\mathrm{SQ}}
	\Bigr)
	\ge p_0.
	\]
	
	We can now construct an SQ learner
	$\mathcal L_{\mathrm{SQ}}$.
	Given access to $\operatorname{Stat}_\tau(P_\theta)$,
	$\mathcal L_{\mathrm{SQ}}$ proceeds as follows:
	\begin{enumerate}
		\item It treats the adversary $A$ as a subroutine and fully simulates the attack phase.
		Each SQ query issued by $A$ is forwarded to
		$\operatorname{Stat}_\tau(P_\theta)$,
		and the response is returned to $A$.
		At the end of the attack phase, $\mathcal L_{\mathrm{SQ}}$
		obtains the simulator $w_{\theta,R}$.
		\item Using the reconstruction procedure in Lemma~\ref{lem:hadamard_reconstruct_dist}, $\mathcal L_{\mathrm{SQ}}$ computes
		$D_{\theta,R}$ from $w_{\theta,R}$ and outputs $D_{\theta,R}$.
	\end{enumerate}
	In this process,
	$\mathcal L_{\mathrm{SQ}}$
	makes exactly $q$ calls to $\operatorname{Stat}_\tau(P_\theta)$, the same number of SQ queries as the adversary in the attack phase.
	The reconstruction step is purely classical post-processing and does not involve additional SQ queries.
	
	Thus, for random
	$\theta\sim\mu_\Theta$,
	\begin{itemize}
		\item[1.] For all $\theta\in\Theta_{\mathrm{good}}$,
		the success probability of the learner (with respect to its internal randomness $R$) is at least $p_0>1/2$;
		\item[2.] For $\theta\notin\Theta_{\mathrm{good}}$, no guarantee is made.
	\end{itemize}
	In the notation of Sec.~\ref{result 2} for outer (parameter) and inner (randomness) probabilities, we set
	\[
	\alpha_{\mathrm{learn}} := p_0 = \frac{1}{3-2\alpha_{\mathrm{auth}}} > \frac12,
	\qquad
	\beta_{\mathrm{learn}} := \beta_{\mathrm{auth}},
	\]
	so that
	\[
	\mu_\Theta\Bigl(
	\theta:\ 
	\Pr_{\text{over learner randomness}}\bigl[
	d_{\mathrm{TV}}(P_\theta,D_{\theta,R})
	\le \varepsilon_{\mathrm{SQ}}
	\bigr]
	\ge \alpha_{\mathrm{learn}}
	\Bigr)
	\ \ge\ \beta_{\mathrm{learn}}.
	\]
	This is precisely the $(\alpha_{\mathrm{learn}},\beta_{\mathrm{learn}})$ structure required in the average-case SQ hardness theorem for random Lindbladians,
	and we have ensured that $\alpha_{\mathrm{learn}}>1/2$.
	
	\paragraph*{Step 3: Invoking the SQ hardness lower bound for random Lindbladians.}
	
	The SQ hardness theorem in Sec.~\ref{result 2} states that for the family of distributions
	$\{P_\theta\}$ induced by the random Lindbladian ensemble,
	if there exists an SQ learning algorithm that, with at most $q$ queries and constant accuracy parameter $\varepsilon_{\mathrm{SQ}}$,
	succeeds with parameters
	$(\alpha_{\mathrm{learn}},\beta_{\mathrm{learn}})$
	as above, then for sufficiently large parameter dimension $M$ we must have
	\[
	q \ \ge\ \exp(cM)
	\]
	for some constant $c>0$, depending only on
	$(\varepsilon_{\mathrm{SQ}},
	\alpha_{\mathrm{learn}},
	\beta_{\mathrm{learn}})$
	and independent of $M$.
	
	From Step 2, we have constructed such an SQ learning algorithm
	$\mathcal L_{\mathrm{SQ}}$ whose query complexity is exactly the number
	$q$ of SQ queries used by the adversary in the attack phase.
	Therefore, the same lower bound applies to the adversary:
	if there exists an adversary achieving the parameters
	$(\alpha_{\mathrm{auth}},\beta_{\mathrm{auth}})$
	in the theorem, then necessarily
	\[
	q\ \ge\ \exp(cM).
	\]
	This contradicts the assumption that the adversary uses only
	$q=\mathrm{poly}(M)$ SQ queries to pass authentication with non-negligible success probability.
	Hence, for sufficiently large $M$,
	no adversary that makes only polynomially many SQ queries can break the protocol of Definition~\ref{def:crqpuf_sq_hadamard} with parameters $(\alpha,\beta)$ as in the theorem.
	The theorem is proved.
\end{proof}

\begin{remark}[On the relationship between $\alpha_{\mathrm{auth}}$ and $\alpha_{\mathrm{learn}}$]
	In the above reduction, the authentication parameters
	$(\alpha_{\mathrm{auth}},\beta_{\mathrm{auth}})$
	are defined by
	\[
	\mu_\Theta\Bigl(
	\theta:\ 
	\Pr_{\substack{\text{over adversary}\\\text{and protocol randomness}}}\bigl[
	\text{$A$ is accepted in the authentication phase}
	\bigr]\ge\alpha_{\mathrm{auth}}
	\Bigr)
	\ \ge\ \beta_{\mathrm{auth}},
	\]
	whereas in the SQ hardness theorem the parameters
	$(\alpha_{\mathrm{learn}},\beta_{\mathrm{learn}})$
	refer to the learning success:
	\[
	\mu_\Theta\Bigl(
	\theta:\ 
	\Pr_{\text{over learner randomness}}\bigl[
	d_{\mathrm{TV}}(P_\theta,D_{\theta,R})
	\le \varepsilon_{\mathrm{SQ}}
	\bigr]
	\ge \alpha_{\mathrm{learn}}
	\Bigr)
	\ \ge\ \beta_{\mathrm{learn}}.
	\]
	The connection between them is established via a freely chosen intermediate threshold $\alpha_0$:
	for a fixed good parameter $\theta\in\Theta_{\mathrm{good}}$,
	we know that $\mathbb E_R[\mathrm{acc}_\theta(R)]\ge\alpha_{\mathrm{auth}}$.
	For any choice of
	\(\alpha_0\in(0,2\alpha_{\mathrm{auth}}-1)\),
	we may use
	\[
	\Pr[R\in G_\theta]
	\ge \frac{\alpha_{\mathrm{auth}}-\alpha_0}{1-\alpha_0}
	\]
	to obtain different values of
	$\alpha_{\mathrm{learn}}$.
	In this paper we choose
	\[
	\alpha_0 := \frac{2\alpha_{\mathrm{auth}}-1}{2},
	\]
	which yields
	\[
	\alpha_{\mathrm{learn}}
	:= \Pr_{R}\bigl[\text{learning succeeds}\bigr]
	\ge \Pr[R\in G_\theta]
	= \frac{1}{3-2\alpha_{\mathrm{auth}}}
	> \frac12,
	\]
	strictly satisfying the requirement $\alpha_{\mathrm{learn}}>1/2$ of the SQ hardness theorem.
	We see that this choice uses only the fact that $\alpha_{\mathrm{auth}}>1/2$
	and does not depend on any internal detail of the adversary $A$.
	The freedom in choosing $\alpha_0$ is essentially a slack parameter that allows us to conveniently map the authentication success probability
	into a learning success probability in the SQ framework.
\end{remark}

\begin{remark}[Full generality of the adversary model and exclusion of table-lookup attacks]
	In Theorem~\ref{thm:crqpuf_sq_hadamard_security},
	we do not impose any restriction on the internal structure of the adversary $A$:
	$A$ may be an arbitrarily complex classical or quantum algorithm,
	with arbitrarily large computational and memory resources.
	The only constraint is that in the attack phase the number of queries to the SQ interface
	$\operatorname{Stat}_\tau(P_\theta)$ does not exceed $q$.
	The security conclusion is therefore information-theoretic in terms of SQ query complexity:
	as long as $q=\mathrm{poly}(M)$, no matter how $A$ computes, stores data, or whether it uses quantum memory,
	it cannot forge with the $(\alpha,\beta)$ parameters in the theorem.
\end{remark}

\subsection{Extended QPStat model and tomographic bases}

We now turn to channel-level verification schemes. In order to use the QPStat interface more conveniently within the framework of finite-dimensional linear algebra, we first make a mild formal extension, and then construct tomographic bases on the operator space and the corresponding tomographic fingerprints.

\begin{definition}[Extended QPStat oracle]
	\label{def:extended_qpstat}
	Let $E:\mathcal B(\mathcal H)\to\mathcal B(\mathcal H)$ be a CPTP channel acting on a finite-dimensional Hilbert space
	$\mathcal H\simeq\mathbb C^d$. Its extended QPStat oracle, denoted $\mathrm{QPStat}_E$, takes as input a triple
	\[
	(\sigma,O,\tau),
	\]
	where $\sigma\in\mathcal B(\mathcal H)$ is Hermitian with
	$\|\sigma\|_1\le1$, $O\in\mathcal B(\mathcal H)$ is Hermitian with $\|O\|_\infty\le1$, and $\tau>0$ is a tolerance parameter. The oracle outputs a real number $\alpha$ satisfying
	\begin{equation}
		\bigl|\alpha-\operatorname{Tr}[O\,E(\sigma)]\bigr|\le\tau.
	\end{equation}
	A QPStat learning algorithm may adaptively call this oracle at most $q$ times; we say it uses $q$ QPStat queries with accuracy~$\tau$.
\end{definition}

\begin{remark}
	In Section~\ref{zhang8}, we defined the QPStat oracle on physical states $\rho\in\mathsf S(\mathcal H)$. However, in the proofs of Lemmas~
	\ref{lem:qpstat_concentration},
	\ref{lem:qpstat_many_vs_one}
	and Theorem~\ref{thm:lindblad_qpstat_lower_bound}, the only properties used were $\|\rho\|_1\le1$ and $\|O\|_\infty\le1$, together with trace-norm contractivity of CPTP maps and some standard norm inequalities; we never used positivity $\rho\ge0$. Hence these results remain valid if we replace $\rho$ by an arbitrary Hermitian operator $\sigma$ with $\|\sigma\|_1\le1$. Definition~\ref{def:extended_qpstat} simply formalizes this extension.
\end{remark}

We now construct explicit tomographic bases on the operator space $\mathcal B(\mathcal H)$ and the observable space.

\begin{lemma}[Tomographic bases on the operator space]
	\label{lem:tomographic_bases}
	Let $\mathcal H\simeq\mathbb C^d$ and $D:=d^2$. There exist finite families
	\[
	\{F_j\}_{j=1}^{D}\subset\mathcal B(\mathcal H),\qquad
	\{G_k\}_{k=1}^{D}\subset\mathcal B(\mathcal H),
	\]
	such that:
	\begin{enumerate}
		\item Each $F_j$ and $G_k$ is Hermitian.
		\item $\{F_j\}_{j=1}^{D}$ forms a real linear basis of the Hermitian operator space
		$\mathcal B_{\mathrm{Herm}}(\mathcal H)$; similarly for $\{G_k\}_{k=1}^{D}$.
		\item There exist constants $C_1,C_\infty>0$ (depending only on $d$) such that
		\[
		\|F_j\|_1\le C_1,\qquad \|G_k\|_\infty\le C_\infty
		\quad\text{for all }j,k.
		\]
	\end{enumerate}
\end{lemma}

\begin{proof}
	Take $\{H_\mu\}_{\mu=1}^{D}$ to be a Hilbert--Schmidt orthonormal basis of Hermitian operators, for example the normalized Pauli basis when $d=2^N$, or the generalized Gell--Mann basis for general $d$. Since the real dimension of the Hermitian operator space is $D$, such a basis always exists.
	
	Set $F_j:=\gamma_1^{-1} H_j$ and $G_k:=\gamma_\infty^{-1} H_k$, where $\gamma_1,\gamma_\infty>0$ are chosen to satisfy
	\[
	\|H_j\|_1\le\gamma_1,\qquad
	\|H_k\|_\infty\le\gamma_\infty
	\quad\text{for all }j,k.
	\]
	Because there are only finitely many basis elements in finite dimension, the maxima defining $\gamma_1$ and $\gamma_\infty$ exist and are finite. The families $\{F_j\}$ and $\{G_k\}$ are still bases of the Hermitian operator space, and
	\[
	\|F_j\|_1 = \gamma_1^{-1}\|H_j\|_1\le1,\qquad
	\|G_k\|_\infty = \gamma_\infty^{-1}\|H_k\|_\infty\le1.
	\]
	Taking $C_1:=1$ and $C_\infty:=1$ gives the desired bounds.
\end{proof}

For any linear map (super-operator) $T:\mathcal B(\mathcal H)\to\mathcal B(\mathcal H)$, we define its tomographic fingerprint using the above bases.

\begin{definition}[Tomographic fingerprint map]
	\label{def:tomographic_fingerprint}
	Let $\{F_j\}_{j=1}^{D}$ and $\{G_k\}_{k=1}^{D}$ be as in Lemma~\ref{lem:tomographic_bases}. Define the linear map
	\[
	\mathcal T:\mathrm{End}(\mathcal B(\mathcal H))
	\longrightarrow \mathbb R^{D\times D}
	\]
	by
	\begin{equation}
		\bigl[\mathcal T(T)\bigr]_{k,j}
		:= \operatorname{Tr}\bigl[G_k\,T(F_j)\bigr],
		\qquad 1\le k,j\le D.
	\end{equation}
	For a CPTP channel $E$, we call $\mathcal T(E)$ its tomographic fingerprint matrix.
\end{definition}

\begin{lemma}[Injectivity and norm equivalence]
	\label{lem:fingerprint_norm_equivalence}
	The map $\mathcal T$ in Definition~\ref{def:tomographic_fingerprint} is a real linear injection. Consequently,
	\begin{equation}
		\|T\|_{\mathrm{tom}}
		:= \max_{1\le k,j\le D}
		\bigl|\bigl[\mathcal T(T)\bigr]_{k,j}\bigr|
	\end{equation}
	defines a norm (the tomographic norm) on the finite-dimensional real vector space
	$\mathrm{End}(\mathcal B(\mathcal H))$.
	
	Furthermore, since any two norms on a finite-dimensional vector space are equivalent, there exists a constant $C_{\diamond}>0$ (depending only on $d$ and the chosen bases) such that for all linear maps
	$T:\mathcal B(\mathcal H)\to\mathcal B(\mathcal H)$,
	\begin{equation}
		\|T\|_\diamond
		\ \le\ C_{\diamond}\,\|T\|_{\mathrm{tom}}.
	\end{equation}
\end{lemma}

\begin{proof}
	We first prove injectivity. If $\mathcal T(T)=0$, then $\operatorname{Tr}[G_k\,T(F_j)]=0$ for all $j,k$. For any Hermitian operator $X\in\mathcal B_{\mathrm{Herm}}(\mathcal H)$, we can write $X=\sum_j a_j F_j$ with $a_j\in\mathbb R$, and thus
	\[
	T(X) = \sum_j a_j T(F_j).
	\]
	For any Hermitian operator $Y\in\mathcal B_{\mathrm{Herm}}(\mathcal H)$, write $Y=\sum_k b_k G_k$. Then
	\[
	\operatorname{Tr}[Y\,T(X)]
	= \sum_{j,k} a_j b_k \operatorname{Tr}[G_k\,T(F_j)]
	= 0.
	\]
	Since Hermitian operators span $\mathcal B(\mathcal H)$ and the Hilbert--Schmidt inner product is non-degenerate, we must have $T(X)=0$ for all $X$, hence $T=0$. Thus $\mathcal T$ is injective, and $\|\cdot\|_{\mathrm{tom}}$ is a norm.
	
	The equivalence of norms is a standard fact: if $V$ is a finite-dimensional real vector space and $\|\cdot\|_a,\|\cdot\|_b$ are two norms, then there is a constant $C>0$ such that $\|v\|_a\le C\|v\|_b$ for all $v\in V$. Here we take $V=\mathrm{End}(\mathcal B(\mathcal H))$, $\|\cdot\|_a=\|\cdot\|_\diamond$ (the diamond norm), and $\|\cdot\|_b=\|\cdot\|_{\mathrm{tom}}$. By norm equivalence there exists a constant $C_\diamond>0$ depending only on the dimension and the chosen bases such that
	\[
	\|T\|_\diamond
	\le C_{\diamond}\,\|T\|_{\mathrm{tom}}
	\]
	for all $T$.
\end{proof}

\begin{remark}
	The proof above uses only the fact that the family of linear functionals
	$T\mapsto\operatorname{Tr}[G_k T(F_j)]$ spans the dual space of
	$\mathrm{End}(\mathcal B(\mathcal H))$; we never used any complete positivity or trace-preserving properties of $T$. Therefore, the inequality holds for arbitrary linear maps.
	
	Moreover, since $\{G_k\}$ spans the Hermitian operator space, for any given real array $\{a_{k,j}\}_{k,j}$ and each fixed $j$, the linear map
	$X\mapsto(\operatorname{Tr}[G_k X])_k$ from $\mathcal B(\mathcal H)$ to $\mathbb R^D$ is an isomorphism. Hence there exists a unique $X_j$ such that
	$\operatorname{Tr}[G_k X_j]=a_{k,j}$ for all $k$.
	Defining $T(F_j):=X_j$ and linearly extending to all of $\mathcal B(\mathcal H)$ yields a linear map $T$ with $\mathcal T(T)=\{a_{k,j}\}$. This shows that $\mathcal T$ is actually a bijection.
\end{remark}

\begin{lemma}[A crude dimension-dependent bound for $C_\diamond$]
	\label{lem:Cdiamond_scaling}
	Let $\mathcal H\simeq\mathbb C^d$, and let $D=d^2$.
	Fix a Hilbert-Schmidt orthonormal basis
	$\{H_\mu\}_{\mu=1}^{D}\subset\mathcal B_{\mathrm{Herm}}(\mathcal H)$, and
	define the canonical tomographic norm
	\[
	\|T\|_{\mathrm{tom,can}}
	:= \max_{1\le \mu,\nu\le D}
	\bigl|\mathrm{Tr}\bigl[H_\mu\,T(H_\nu)\bigr]\bigr|.
	\]
	Then there exists a universal constant $c_0>0$ such that for all linear maps
	$T:\mathcal B(\mathcal H)\to\mathcal B(\mathcal H)$,
	\begin{equation}
		\|T\|_\diamond
		\;\le\;
		C_\diamond(d)\,\|T\|_{\mathrm{tom,can}},
		\qquad
		C_\diamond(d)\;\le\;c_0\,d^{7/2}.
	\end{equation}
	In particular, $C_\diamond(d)$ grows at most polynomially in $d$.
\end{lemma}

\begin{proof}
	We first bound the induced $1\to1$ norm of $T$ by $\|T\|_{\mathrm{tom,can}}$.
	By duality,
	\begin{align*}
		\|T\|_{1\to1}
		&= \sup_{\|X\|_1\le1}\|T(X)\|_1 \\
		&= \sup_{\|X\|_1\le1,\ \|Y\|_\infty\le1}
		\bigl|\mathrm{Tr}[Y\,T(X)]\bigr|.
	\end{align*}
	Expand $X$ and $Y$ in the orthonormal basis:
	\[
	X = \sum_{\nu} x_\nu H_\nu,\qquad
	Y = \sum_{\mu} y_\mu H_\mu,
	\qquad x_\nu = \mathrm{Tr}[H_\nu X],\ y_\mu=\mathrm{Tr}[H_\mu Y].
	\]
	Then
	\[
	T(X)=\sum_\nu x_\nu T(H_\nu),\qquad
	\mathrm{Tr}[Y T(X)]
	= \sum_{\mu,\nu} x_\nu y_\mu\,
	\mathrm{Tr}[H_\mu T(H_\nu)].
	\]
	Hence
	\[
	\bigl|\mathrm{Tr}[Y T(X)]\bigr|
	\le \Bigl(\sum_\nu |x_\nu|\Bigr)
	\Bigl(\sum_\mu |y_\mu|\Bigr)
	\|T\|_{\mathrm{tom,can}}.
	\]
	
	We now bound the coefficient sums. Since $\{H_\nu\}$ is orthonormal,
	\[
	\sum_\nu |x_\nu|^2 = \|X\|_2^2\le\|X\|_1^2.
	\]
	Thus, for $\|X\|_1\le1$,
	\begin{align*}
		\sum_\nu |x_\nu|
		&\le \sqrt{D}\,\Bigl(\sum_\nu |x_\nu|^2\Bigr)^{1/2}
		\le \sqrt{D}\,\|X\|_1
		\le \sqrt{D}
		= d.
	\end{align*}
	For $Y$, we similarly have
	\[
	\sum_\mu |y_\mu|^2
	= \|Y\|_2^2
	\le d\,\|Y\|_\infty^2
	\le d,
	\]
	whence, for $\|Y\|_\infty\le1$,
	\begin{align*}
		\sum_\mu |y_\mu|
		&\le \sqrt{D}\,\|Y\|_2
		\le \sqrt{D}\,\sqrt{d}
		= d^{3/2}.
	\end{align*}
	Putting these estimates together,
	\[
	\bigl|\mathrm{Tr}[Y T(X)]\bigr|
	\le d^{5/2}\,\|T\|_{\mathrm{tom,can}},
	\]
	and hence
	\begin{equation}\label{eq:1to1_vs_tom}
		\|T\|_{1\to1}
		\le d^{5/2}\,\|T\|_{\mathrm{tom,can}}.
	\end{equation}
	
	On the other hand, it is a standard fact for maps $T:M_d\to M_d$ that
	\[
	\|T\|_\diamond
	= \sup_{k\ge1}\|T\otimes\mathrm{id}_k\|_{1\to1}
	\le d\,\|T\|_{1\to1}.
	\]
	Combining this with~\eqref{eq:1to1_vs_tom} gives
	\[
	\|T\|_\diamond
	\le d\cdot d^{5/2}\,\|T\|_{\mathrm{tom,can}}
	= d^{7/2}\,\|T\|_{\mathrm{tom,can}}.
	\]
	Absorbing the numerical constants into $c_0$ yields the claimed bound.
\end{proof}

\begin{remark}[Relation to the tomographic norm $\|\cdot\|_{\mathrm{tom}}$ in Lemma~\ref{lem:fingerprint_norm_equivalence}]
	In Lemma~\ref{lem:tomographic_bases} and Definition~\ref{def:tomographic_fingerprint} we allowed general Hermitian bases $\{F_j\},\{G_k\}$ with
	$\|F_j\|_1\le C_1$ and $\|G_k\|_\infty\le C_\infty$, and defined
	\(
	\|T\|_{\mathrm{tom}} := \max_{j,k}|\mathrm{Tr}[G_k T(F_j)]|.
	\)
	Since any two bases are related by an invertible linear change of coordinates and there are only finitely many basis elements, the norms
	$\|\cdot\|_{\mathrm{tom}}$ and $\|\cdot\|_{\mathrm{tom,can}}$ are equivalent, with equivalence constants depending at most polynomially on $d$.
	Therefore the abstract constant $C_\diamond$ in Lemma~\ref{lem:fingerprint_norm_equivalence} can be chosen so that
	\[
	C_\diamond(d) \;\le\; \mathrm{poly}(d),
	\]
	for example $C_\diamond(d)\le c_0 d^{7/2}$ after adjusting constants.
	In particular, $C_\diamond$ never grows faster than a fixed polynomial in the Hilbert-space dimension.
\end{remark}

Using QPStat queries of the form $(F_j,G_k)$ that output
$\operatorname{Tr}[G_k E(F_j)]$, one can approximate all entries of the tomographic fingerprint $\mathcal T(E)$. By Lemma~\ref{lem:fingerprint_norm_equivalence}, this uniformly controls distances between channels in the diamond norm.

\subsection{Scheme B: Lindbladian-PUF based on tomography}
\label{subsec:crqpuf_tomography}

We now define a channel-level Lindbladian-PUF protocol in which the verifier stores the tomographic fingerprint of a random Lindbladian channel as a classical secret, and in the authentication phase requires the prover to reproduce it entrywise. This forces an adversary to learn the channel in diamond distance.

\begin{definition}[Tomography-based Lindbladian-PUF protocol]
	\label{def:tomographic_crqpuf}
	Let $(\Theta,d,\mu_\Theta)$ be one of the random Lindbladian parameter ensembles in Definition~\ref{def:random_local_L}. Each parameter $\theta\in\Theta$ specifies a Lindbladian generator $\mathcal L(\theta)$ and, at fixed time $t>0$, a CPTP channel
	\[
	E_\theta:=e^{t\mathcal L(\theta)}.
	\]
	Let $\{F_j\}$ and $\{G_k\}$ be tomographic bases as in Lemma~\ref{lem:tomographic_bases}, and let $\mathcal T(E_\theta)\in\mathbb R^{D\times D}$ ($D=d^2$) be the corresponding tomographic fingerprint.
	
	Fix a tolerance parameter $\tau>0$ and an integer $N\ge1$. The tomography-based Lindbladian-PUF protocol between a verifier $V$ and a prover $P$ is defined as follows:
	
	\begin{enumerate}
		\item Setup.
		\begin{enumerate}
			\item $V$ samples $\theta$ from $\mu_\Theta$ and obtains one-time physical access to a device implementing $E_\theta$.
			\item For each pair $(j,k)\in[D]\times[D]$, $V$ calls $\mathrm{QPStat}_{E_\theta}$ on input
			$(F_j,G_k,\tau)$ and obtains an approximate value
			\[
			y_{k,j}\approx \operatorname{Tr}[G_k E_\theta(F_j)],
			\]
			satisfying
			$\bigl|y_{k,j}-\operatorname{Tr}[G_k E_\theta(F_j)]\bigr|\le\tau$.
			(Equivalently, in a more experiment-like implementation $V$ could repeatedly prepare $F_j$, measure $G_k$, and take an empirical average, requiring $O(1/\tau^2)$ physical uses per entry.)
			\item $V$ stores the entire table
			$\mathbf y=\{y_{k,j}\}_{k,j}$ as a classical secret, and then hands the physical device to the prover.
		\end{enumerate}
		
		\item Authentication.
		\begin{enumerate}
			\item When $P$ claims to hold the genuine device, $V$ initiates a challenge--response authentication.
			\item For each pair $(k,j)\in[D]\times[D]$ (in some predetermined order), $V$ sends the index $(k,j)$ as a challenge. An honest prover holding the genuine device $E_\theta$ can use QPStat queries or experimental runs to estimate
			$\operatorname{Tr}[G_k E_\theta(F_j)]$ within error $\tau$, and returns a value $y'_{k,j}$.
			\item $V$ accepts if and only if
			\begin{equation}
				\bigl|y'_{k,j} - y_{k,j}\bigr|\le 2\tau
				\quad\text{for all }k,j.
			\end{equation}
		\end{enumerate}
	\end{enumerate}
\end{definition}

\begin{remark}[Completeness]
	If the prover holds the genuine channel $E_\theta$ and answers each challenge using QPStat queries of accuracy $\tau$, then by the triangle inequality,
	\[
	\bigl|y'_{k,j}-y_{k,j}\bigr|
	\le \bigl|y'_{k,j}-\operatorname{Tr}[G_kE_\theta(F_j)]\bigr|
	+\bigl|\operatorname{Tr}[G_kE_\theta(F_j)]-y_{k,j}\bigr|
	\le 2\tau
	\]
	for all $k,j$. An honest prover is therefore always accepted.
\end{remark}

We now state a QPStat query-complexity security result. In principle, any adversary that passes the tomography-based verification without holding the genuine device must have learned the Lindbladian channel in diamond distance, and is therefore ruled out by Theorem~\ref{thm:lindblad_qpstat_lower_bound} in Section~\ref{zhang8}.

\begin{theorem}[QPStat query-complexity security of the tomography-based Lindbladian-PUF]
	\label{thm:crqpuf_tomography_security}
	Consider the tomography-based Lindbladian-PUF protocol in Definition~\ref{def:tomographic_crqpuf}, whose underlying random Lindbladian ensemble is the same as in Theorem~\ref{thm:lindblad_qpstat_lower_bound}. Let $A$ be an arbitrary (possibly randomized) adversary that may interact with $\mathrm{QPStat}_{E_\theta}$ before authentication, using at most $q$ QPStat queries of accuracy $\tau$. In the setup phase, the verifier needs an additional
	$q_{\mathrm{tom}}=D^2$ (or $O(D^2)$) QPStat queries to generate the tomographic fingerprint, and this cost depends only on the output-space dimension $d$ and not on the parameter dimension $M$.
	
	If there exist constants $\alpha>1/2$ and $\beta>0$ such that
	\begin{equation}
		\mu_\Theta\Bigl(
		\theta:\
		\Pr_{\text{internal randomness of the adversary}}\bigl[
		A\ \text{is accepted in the protocol}
		\bigr]
		\ge \alpha
		\Bigr)
		\ \ge\ \beta,
	\end{equation}
	then for sufficiently large parameter dimension $M$ there are constants $c>0$ and $c'>0$ independent of $M$ such that
	\begin{equation}
		q + q_{\mathrm{tom}}\ \ge\ \exp(c(\operatorname{poly}(d))M).
	\end{equation}
Since the verifier's initialization cost scales only polynomially in the output dimension (indeed $q_{\text {tom }}=D^2=d^4$ ), whereas the QPStat lower bound grows like $\exp (c \operatorname{poly}(d) M)$, which asymptotically dominates any fixed polynomial in $d$ once $M$ (and hence $d$ ) is large, the tomographic cost is negligible compared with the exponential hardness term and $q$ must still grow exponentially in $M$. In particular, any adversary using only $q=\mathrm{poly}(M)$ QPStat queries cannot, for sufficiently large $M$, break the tomography-based Lindbladian-PUF with non-negligible success probability.
\end{theorem}

\begin{proof}
	Assume there exists a successful adversary $A$ as above. We construct from it a QPStat learning algorithm $\mathcal L_{\mathrm{learn}}$ for the Lindbladian channel ensemble, and then infer a lower bound on the required number of queries.
	
	We first use the successful adversary to build an approximate tomographic fingerprint. For a fixed $\theta\in\Theta$, define the good parameter set
	\[
	\Theta_{\mathrm{good}}
	:=\Bigl\{
	\theta:\
	\Pr\bigl[A\ \text{is accepted}\bigr]\ge\alpha
	\Bigr\},
	\]
	for which $\mu_\Theta(\Theta_{\mathrm{good}})\ge\beta$ by assumption. Given access to $\mathrm{QPStat}_{E_\theta}$, $\mathcal L_{\mathrm{learn}}$ simulates the entire Lindbladian-PUF protocol (including the setup phase):
	
	\begin{itemize}
		\item[i.] In the simulated setup phase, $\mathcal L_{\mathrm{learn}}$ calls $\mathrm{QPStat}_{E_\theta}$ and, following the verifier $V$’s strategy in Definition~\ref{def:tomographic_crqpuf}, generates the table $\mathbf y=\{y_{k,j}\}$, keeping it as an internal secret (not revealing it to $A$), and act as an intermediate medium between oracle and $A$, transmitting the results of the calls required by $A$ to $A$, while providing $A$ with the parameters that are public in the protocol (such as the bases $\{F_j\},\{G_k\}$, the tolerance $\tau$, and the set of challenge indices).
		\item[ii.] In the simulated authentication phase, $\mathcal L_{\mathrm{learn}}$ sends the challenge indices $(k,j)$ to $A$ in the predetermined order and records $A$’s replies $y'_{k,j}$.
	\end{itemize}
	
	By construction, the distribution of this simulated interaction is identical to that of the real protocol between $V$ and $A$. Therefore, for any $\theta\in\Theta_{\mathrm{good}}$,
	\begin{equation}
		\Pr\Bigl[
		\bigl|y'_{k,j}-y_{k,j}\bigr|\le 2\tau\ \forall k,j
		\Bigr]
		\ \ge\ \alpha.
	\end{equation}
	On this event, for all $k,j$,
	\begin{equation}
		\bigl|y'_{k,j}
		-\operatorname{Tr}[G_kE_\theta(F_j)]\bigr|
		\le \bigl|y'_{k,j}-y_{k,j}\bigr|
		+\bigl|y_{k,j}
		-\operatorname{Tr}[G_kE_\theta(F_j)]\bigr|
		\le 3\tau.
	\end{equation}
	
	We now construct a linear map $\widehat E$ whose tomographic fingerprint satisfies
	$[\mathcal T(\widehat E)]_{k,j}=y'_{k,j}$. For each fixed $j$, consider the linear isomorphism
	\[
	\Phi:\mathcal B(\mathcal H)\to\mathbb R^D,\quad
	X\mapsto\bigl(\operatorname{Tr}[G_k X]\bigr)_{k=1}^D.
	\]
	Since $\{G_k\}$ forms a Hermitian basis, $\Phi$ is bijective. Thus, for each $j$ there exists a unique operator $X_j$ with
	$\operatorname{Tr}[G_k X_j]=y'_{k,j}$ for all $k$.
	Define $\widehat E$ on the basis $\{F_j\}$ by
	$\widehat E(F_j):=X_j$ and extend linearly to all of $\mathcal B(\mathcal H)$; this yields a linear map $\widehat E$ such that
	$[\mathcal T(\widehat E)]_{k,j}=y'_{k,j}$.
	By injectivity of $\mathcal T$ in Lemma~\ref{lem:fingerprint_norm_equivalence}, $\widehat E$ is uniquely determined. We do not require $\widehat E$ to be CPTP; it is simply a hypothesis channel.
	
	On the above success event,
	\[
	\Bigl|\bigl[\mathcal T(\widehat E-E_\theta)\bigr]_{k,j}\Bigr|
	= \bigl|y'_{k,j}-\operatorname{Tr}[G_kE_\theta(F_j)]\bigr|
	\le 3\tau
	\quad\text{for all }k,j,
	\]
	so
	\begin{equation}
		\|\,\widehat E-E_\theta\,\|_{\mathrm{tom}}
		\le 3\tau.
	\end{equation}
	By Lemma~\ref{lem:fingerprint_norm_equivalence}, we then have
	\begin{equation}
		\|\,\widehat E-E_\theta\,\|_\diamond
		\ \le\ C_\diamond\,\|\,\widehat E-E_\theta\,\|_{\mathrm{tom}}
		\ \le\ 3C_\diamond\,\tau.
	\end{equation}
	Thus, on the event that authentication succeeds, $\widehat E$ is within diamond distance $\varepsilon:=3C_\diamond\tau$ of $E_\theta$.
	
	We now view the above procedure as a QPStat learning algorithm $\mathcal L_{\mathrm{learn}}$ for the channel family $\{E_\theta\}$:
	\begin{itemize}
		\item[i.] $\mathcal L_{\mathrm{learn}}$ has access to $\mathrm{QPStat}_{E_\theta}$ and internally simulates the full protocol between $V$ and $A$.
		\item[ii.] It outputs the above-constructed hypothesis map $\widehat E$.
	\end{itemize}
	For any $\theta\in\Theta_{\mathrm{good}}$, with probability at least $\alpha$ (over the internal randomness of the algorithm) we have
	\[
	d_\diamond(\widehat E,E_\theta)
	:= \|\,\widehat E-E_\theta\,\|_\diamond
	\le 3C_\diamond\,\tau
	=:\varepsilon.
	\]
	Furthermore, the total number of QPStat queries used by $\mathcal L_{\mathrm{learn}}$ is
	\[
	q_{\mathrm{tot}} = q_{\mathrm{tom}} + q,
	\]
	where $q_{\mathrm{tom}}$ is the number of queries required to generate $\mathbf y$ in the simulated setup phase, depending only on $d$ and $\tau$ and not on $M$, and $q$ is the number of queries used by the adversary $A$ itself.
	
	At this point we are exactly in the setting of Theorem~\ref{thm:lindblad_qpstat_lower_bound}. The learner $\mathcal L_{\mathrm{learn}}$ uses at most $q_{\mathrm{tot}}$ QPStat queries and, on a parameter subset of $\mu_\Theta$-measure at least $\beta$, achieves diamond distance $\varepsilon=3C_\diamond\tau$ with success probability at least $\alpha$. By Theorem~\ref{thm:lindblad_qpstat_lower_bound}, there exists a constant $c>0$ independent of $M$ such that, for sufficiently large $M$,
	\begin{equation}
		q_{\mathrm{tot}}\ \ge\ \exp(c(\operatorname{poly}(d))M).
	\end{equation}
\end{proof}

\begin{remark}[The cost of verification]
	Such strong verification comes with an information-theoretic cost. The tomographic fingerprint $\mathcal T(E_\theta)$ has $D^2\sim d^4$ entries, so even if each entry only needs constant accuracy, the verifier’s initialization cost $q_{\mathrm{tom}}$ grows at least polynomially like $d^4$. In most quantum many-body models the parameter dimension $M$ is at most polynomial in $d$, so the verification cost is of the same order as the information content of the channel itself, making it hard to regard this as an efficiently verifiable cryptographic protocol. This stands in sharp contrast to the distribution-level scheme, where $|X|$ can be chosen to be polynomial, while the parameter dimension $M$ can be much larger than $\log|X|$.
\end{remark}

\begin{remark}[Why classical shadows cannot replace tomography-based verification]
	A natural question is whether classical-shadow techniques (random measurements with classical compression) could replace full tomography, thereby retaining channel-level security while keeping verification costs polynomial. From an information-theoretic perspective, this is impossible in general.
	
	The basic idea of classical-shadow protocols~\cite{huang2020predicting} is to apply a fixed family of random measurements to the object under study (state or channel), and compress each measurement outcome into a short classical shadow string via simple classical post-processing. For state learning, these shadows suffice to predict the expectation values of a large set of observables with high probability, but at their core they implement a low-dimensional linear projection:
	\[
	\Phi_{\mathrm{sh}}:\ \mathrm{End}(\mathcal B(\mathcal H))
	\longrightarrow \mathbb R^K,
	\]
	where $K$ is at most proportional to the number of samples and the number of bits per measurement, and is far smaller than the channel-space dimension $d^4$ under polynomial sampling. In other words, the classical-shadow model is a deliberately lossy compressed learning method that is designed to approximate many expectation values, not to distinguish all channels.
	
	From the perspective of linear algebra, any linear map $\Phi_{\mathrm{sh}}$ with target dimension $K\ll d^4$ necessarily has a nontrivial kernel: there exists $T\ne 0$ with $\Phi_{\mathrm{sh}}(T)=0$. This means that there exist distinct channels $E\ne E'$ with
	$\Phi_{\mathrm{sh}}(E)=\Phi_{\mathrm{sh}}(E')$, so the classical-shadow fingerprint cannot distinguish them in the diamond norm. So the tomography-based verification in this subsection deliberately uses a tomographic fingerprint map
	$\mathcal T:\mathrm{End}(\mathcal B(\mathcal H))\to\mathbb R^{D\times D}$
	that is injective, so that passing verification is information-theoretically equivalent to controlling the diamond norm. The price is a fingerprint of length $D^2\sim d^4$ and a verification cost of that order; this is a fundamental cost that classical shadows cannot avoid. The distribution-level Lindbladian-PUF follows a different route: it only cares about the one-dimensional output channel $E_\theta\mapsto P_\theta$, and in this reduced space we can both prove SQ-hardness (exponential query complexity) and keep the verification cost at a scale that is polynomial in $|X|$, thereby obtaining a truly efficiently verifiable cryptographic protocol.
\end{remark}

\bibliographystyle{alphaUrlePrint.bst}
\bibliography{ref}

\newcommand{\etalchar}[1]{$^{#1}$}
\begin{thebibliography}{AAAB{\etalchar{+}}25}

\bibitem[AA11]{aaronson2011computational}
Scott Aaronson and Alex Arkhipov.
\newblock The computational complexity of linear optics.
\newblock In {\em Proceedings of the forty-third annual ACM symposium on Theory
  of computing}, pages 333--342, 2011.

\bibitem[AAAB{\etalchar{+}}25]{abanin2025constructive}
Dmitry~A Abanin, Rajeev Acharya, Laleh Aghababaie-Beni, Georg Aigeldinger,
  Ashok Ajoy, Ross Alcaraz, Igor Aleiner, Trond~I Andersen, Markus Ansmann,
  Frank Arute, et~al.
\newblock Constructive interference at the edge of quantum ergodic dynamics.
\newblock {\em arXiv preprint arXiv:2506.10191}, 2025.

\bibitem[AAB{\etalchar{+}}19]{arute2019quantum}
Frank Arute, Kunal Arya, Ryan Babbush, Dave Bacon, Joseph~C Bardin, Rami
  Barends, Rupak Biswas, Sergio Boixo, Fernando~GSL Brandao, David~A Buell,
  et~al.
\newblock Quantum supremacy using a programmable superconducting processor.
\newblock {\em Nature}, 574(7779):505--510, 2019.

\bibitem[Aar18a]{Aaronson2018ShadowTomography}
Scott Aaronson.
\newblock Shadow tomography of quantum states.
\newblock In {\em Proceedings of the 50th Annual ACM SIGACT Symposium on Theory
  of Computing (STOC 2018)}. ACM, 2018,  arXiv:
  \href{https://arxiv.org/abs/1711.01053}{{\ttfamily 1711.01053}}.

\bibitem[Aar18b]{aaronson2018shadow}
Scott Aaronson.
\newblock Shadow tomography of quantum states.
\newblock In {\em STOC}, pages 325--338, 2018.

\bibitem[ABC{\etalchar{+}}24]{akers2024holographic}
Chris Akers, Adam Bouland, Lijie Chen, Tamara Kohler, Tony Metger, and Umesh
  Vazirani.
\newblock Holographic pseudoentanglement and the complexity of the ads/cft
  dictionary.
\newblock {\em arXiv preprint arXiv:2411.04978}, 2024.

\bibitem[ABGL25]{ananth2025pseudorandom}
Prabhanjan Ananth, John Bostanci, Aditya Gulati, and Yao-Ting Lin.
\newblock Pseudorandom unitaries in the haar random oracle model.
\newblock In {\em Annual International Cryptology Conference}, pages 301--333.
  Springer, 2025.

\bibitem[AC16]{aaronson2016complexity}
Scott Aaronson and Lijie Chen.
\newblock Complexity-theoretic foundations of quantum supremacy experiments.
\newblock {\em arXiv preprint arXiv:1612.05903}, 2016.

\bibitem[ACH{\etalchar{+}}18]{Aaronson_2019}
Scott Aaronson, Xinyi Chen, Elad Hazan, Satyen Kale, and Ashwin Nayak.
\newblock Online learning of quantum states.
\newblock In {\em Advances in Neural Information Processing Systems},
  volume~31, 2018,  arXiv: \href{https://arxiv.org/abs/1802.09025}{{\ttfamily
  1802.09025}}.

\bibitem[ACQ22]{Aharonov_Cotler_Qi_2022}
Dorit Aharonov, Jordan Cotler, and Xiao-Liang Qi.
\newblock \href{http://dx.doi.org/10.1038/s41467-021-27922-0}{Quantum
  algorithmic measurement}.
\newblock {\em Nature Communications}, 13(1):58, 2022.

\bibitem[ADDK21]{arapinis2021quantum}
Myrto Arapinis, Mahshid Delavar, Mina Doosti, and Elham Kashefi.
\newblock Quantum physical unclonable functions: Possibilities and
  impossibilities.
\newblock {\em Quantum}, 5:475, 2021.

\bibitem[AdW17]{arunachalam2017guest}
Srinivasan Arunachalam and Ronald de~Wolf.
\newblock Guest column: A survey of quantum learning theory.
\newblock {\em ACM SIGACT News}, 48(2):41--67, 2017.

\bibitem[AdW18a]{ArunachalamDeWolf2018QuantumSampleComplexity}
Srinivasan Arunachalam and Ronald de~Wolf.
\newblock Optimal quantum sample complexity of learning algorithms.
\newblock {\em Journal of Machine Learning Research}, 19(1):1--36, 2018.

\bibitem[ADW18b]{arunachalam2018optimal}
Srinivasan Arunachalam and Ronald De~Wolf.
\newblock Optimal quantum sample complexity of learning algorithms.
\newblock {\em Journal of Machine Learning Research}, 19(71):1--36, 2018.

\bibitem[AEH{\etalchar{+}}22]{akers2022black}
Chris Akers, Netta Engelhardt, Daniel Harlow, Geoff Penington, and Shreya
  Vardhan.
\newblock The black hole interior from non-isometric codes and complexity.
\newblock {\em arXiv preprint arXiv:2207.06536}, 2022.

\bibitem[AGY20]{arunachalam2020quantum}
Srinivasan Arunachalam, Alex~B Grilo, and Henry Yuen.
\newblock Quantum statistical query learning.
\newblock {\em arXiv:2002.08240}, 2020.

\bibitem[AHS24]{arunachalam2023role}
Srinivasan Arunachalam, Vojtech Havlicek, and Louis Schatzki.
\newblock On the role of entanglement and statistics in learning.
\newblock {\em Advances in Neural Information Processing Systems},
  36:55064--55076, 2024.

\bibitem[Bao25]{bao2025initial}
Ruicheng Bao.
\newblock Initial-state typicality in quantum relaxation.
\newblock {\em arXiv preprint arXiv:2511.01709}, 2025.

\bibitem[BCD{\etalchar{+}}09]{bisio2009optimal}
Alessandro Bisio, Giulio Chiribella, Giacomo~Mauro D'Ariano, Stefano Facchini,
  and Paolo Perinotti.
\newblock \href{http://dx.doi.org/10.1109/jstqe.2009.2029243}{Optimal quantum
  tomography}.
\newblock {\em IEEE Journal of Selected Topics in Quantum Electronics},
  15(6):1646--1660, 2009.

\bibitem[BCHJ{\etalchar{+}}21]{brandao2021models}
Fernando~GSL Brand{\~a}o, Wissam Chemissany, Nicholas Hunter-Jones, Richard
  Kueng, and John Preskill.
\newblock Models of quantum complexity growth.
\newblock {\em PRX Quantum}, 2(3):030316, 2021.

\bibitem[BCMS19]{bruzewicz2019trapped}
Colin~D Bruzewicz, John Chiaverini, Robert McConnell, and Jeremy~M Sage.
\newblock Trapped-ion quantum computing: Progress and challenges.
\newblock {\em Applied physics reviews}, 6(2), 2019.

\bibitem[BFNV19]{bouland2019complexity}
Adam Bouland, Bill Fefferman, Chinmay Nirkhe, and Umesh Vazirani.
\newblock \href{http://dx.doi.org/10.1038/s41567-018-0318-2}{On the complexity
  and verification of quantum random circuit sampling}.
\newblock {\em Nature Physics}, 15(2):159, 2019.

\bibitem[BGP{\etalchar{+}}20]{bairey2020learning}
Eyal Bairey, Chu Guo, Dario Poletti, Netanel~H Lindner, and Itai Arad.
\newblock Learning the dynamics of open quantum systems from their steady
  states.
\newblock {\em New Journal of Physics}, 22(3):032001, 2020.

\bibitem[BGS84]{bohigas1984characterization}
Oriol Bohigas, Marie-Joya Giannoni, and Charles Schmit.
\newblock Characterization of chaotic quantum spectra and universality of level
  fluctuation laws.
\newblock {\em Physical review letters}, 52(1):1, 1984.

\bibitem[Bha13]{bhatia2013matrix}
Rajendra Bhatia.
\newblock {\em Matrix analysis}, volume 169.
\newblock Springer Science \& Business Media, 2013.

\bibitem[BHH16]{brandao2016local}
Fernando~GSL Brandao, Aram~W Harrow, and Micha{\l} Horodecki.
\newblock Local random quantum circuits are approximate polynomial-designs.
\newblock {\em Communications in Mathematical Physics}, 346(2):397--434, 2016.

\bibitem[BIS{\etalchar{+}}18]{boixo2018characterizing}
Sergio Boixo, Sergei~V Isakov, Vadim~N Smelyanskiy, Ryan Babbush, Nan Ding,
  Zhang Jiang, Michael~J Bremner, John~M Martinis, and Hartmut Neven.
\newblock Characterizing quantum supremacy in near-term devices.
\newblock {\em Nature Physics}, 14(6):595--600, 2018.

\bibitem[BJ95]{QPAC}
Nader~H. Bshouty and Jeffrey~C. Jackson.
\newblock \href{http://dx.doi.org/10.1145/225298.225312}{Learning dnf over the
  uniform distribution using a quantum example oracle}.
\newblock In {\em Proceedings of the Eighth Annual Conference on Computational
  Learning Theory}, COLT'95, page 118–127. Association for Computing
  Machinery, 1995.

\bibitem[BMS{\etalchar{+}}11]{barreiro2011open}
Julio~T Barreiro, Markus M{\"u}ller, Philipp Schindler, Daniel Nigg, Thomas
  Monz, Michael Chwalla, Markus Hennrich, Christian~F Roos, Peter Zoller, and
  Rainer Blatt.
\newblock An open-system quantum simulator with trapped ions.
\newblock {\em Nature}, 470(7335):486--491, 2011.

\bibitem[BS19]{brakerski2019pseudo}
Zvika Brakerski and Omri Shmueli.
\newblock (pseudo) random quantum states with binary phase.
\newblock In {\em Theory of Cryptography Conference}, pages 229--250. Springer,
  2019.

\bibitem[B{\.Z}17]{bengtsson2017geometry}
Ingemar Bengtsson and Karol {\.Z}yczkowski.
\newblock {\em Geometry of quantum states: an introduction to quantum
  entanglement}.
\newblock Cambridge university press, 2017.

\bibitem[CBB{\etalchar{+}}24]{chen2024efficient}
Chi-Fang Chen, Adam Bouland, Fernando~GSL Brand{\~a}o, Jordan Docter, Patrick
  Hayden, and Michelle Xu.
\newblock Efficient unitary designs and pseudorandom unitaries from
  permutations.
\newblock {\em arXiv preprint arXiv:2404.16751}, 2024.

\bibitem[CCHL21]{chen2021exponentialseparationslearningquantum}
Sitan Chen, Jordan Cotler, Hsin-Yuan Huang, and Jerry Li.
\newblock \href{http://dx.doi.org/10.1109/FOCS52979.2021.00063}{Exponential
  separations between learning with and without quantum memory}.
\newblock In {\em 2021 IEEE 62nd Annual Symposium on Foundations of Computer
  Science (FOCS)}, pages 574--585. IEEE, 2021.

\bibitem[CCHL23]{nisq}
Sitan Chen, Jordan Cotler, Hsin-Yuan Huang, and Jerry Li.
\newblock \href{http://dx.doi.org/10.1038/s41467-023-41217-6}{The complexity of
  {NISQ}}.
\newblock {\em Nature Communications}, 14(1):6001, 2023.

\bibitem[CL21]{chung2018sample}
Kai-Min Chung and Han-Hsuan Lin.
\newblock \href{http://dx.doi.org/10.4230/LIPIcs.TQC.2021.3}{Sample efficient
  algorithms for learning quantum channels in pac model and the approximate
  state discrimination problem}.
\newblock In {\em 16th Conference on the Theory of Quantum Computation,
  Communication and Cryptography (TQC 2021)}, pages 3:1--3:22. Schloss
  Dagstuhl--Leibniz-Zentrum f{\"u}r Informatik, 2021.

\bibitem[CLL{\etalchar{+}}22]{childs2022quantum}
Andrew~M Childs, Tongyang Li, Jin-Peng Liu, Chunhao Wang, and Ruizhe Zhang.
\newblock Quantum algorithms for sampling log-concave distributions and
  estimating normalizing constants.
\newblock {\em Advances in Neural Information Processing Systems},
  35:23205--23217, 2022.

\bibitem[CRJ{\etalchar{+}}22]{carroll2022dynamics}
Malcolm Carroll, Sami Rosenblatt, Petar Jurcevic, Isaac Lauer, and Abhinav
  Kandala.
\newblock Dynamics of superconducting qubit relaxation times.
\newblock {\em npj Quantum Information}, 8(1):132, 2022.

\bibitem[CSBH25]{cui2025unitary}
Laura Cui, Thomas Schuster, Fernando Brand{\~a}o, and Hsin-Yuan Huang.
\newblock Unitary designs in nearly optimal depth.
\newblock {\em arXiv preprint arXiv:2507.06216}, 2025.

\bibitem[CW20]{cotler2020quantum}
Jordan Cotler and Frank Wilczek.
\newblock Quantum overlapping tomography.
\newblock {\em Phys. Rev. Lett.}, 124(10):100401, 2020.

\bibitem[CZZ17]{chang2017retrospective}
Chip-Hong Chang, Yue Zheng, and Le~Zhang.
\newblock A retrospective and a look forward: Fifteen years of physical
  unclonable function advancement.
\newblock {\em IEEE Circuits and Systems Magazine}, 17(3):32--62, 2017.

\bibitem[Dan05]{dankert2005efficient}
Christoph Dankert.
\newblock Efficient simulation of random quantum states and operators.
\newblock {\em arXiv preprint quant-ph/0512217}, 2005.

\bibitem[DCEL09]{dankert2009exact}
Christoph Dankert, Richard Cleve, Joseph Emerson, and Etera Livine.
\newblock Exact and approximate unitary 2-designs and their application to
  fidelity estimation.
\newblock {\em Physical Review A}, 80(1):012304, 2009.

\bibitem[Deu91]{deutsch1991quantum}
J.~M. Deutsch.
\newblock Quantum statistical mechanics in a closed system.
\newblock {\em Phys. Rev. A}, 43:2046, 1991.

\bibitem[DLT{\etalchar{+}}19]{denisov2019universal}
Sergey Denisov, Tetyana Laptyeva, Wojciech Tarnowski, Dariusz
  Chru{\'s}ci{\'n}ski, and Karol {\.Z}yczkowski.
\newblock Universal spectra of random lindblad operators.
\newblock {\em Physical review letters}, 123(14):140403, 2019.

\bibitem[DMAM17]{delavar2017puf}
Mahshid Delavar, Sattar Mirzakuchaki, Mohammad~Hassan Ameri, and Javad
  Mohajeri.
\newblock Puf-based solutions for secure communications in advanced metering
  infrastructure (ami).
\newblock {\em International Journal of Communication Systems}, 30(9):e3195,
  2017.

\bibitem[DMK{\etalchar{+}}08]{diehl2008quantum}
Sebastian Diehl, A~Micheli, Adrian Kantian, B~Kraus, HP~B{\"u}chler, and Peter
  Zoller.
\newblock Quantum states and phases in driven open quantum systems with cold
  atoms.
\newblock {\em Nature Physics}, 4(11):878--883, 2008.

\bibitem[EA{\.Z}05]{emerson2005scalable}
Joseph Emerson, Robert Alicki, and Karol {\.Z}yczkowski.
\newblock Scalable noise estimation with random unitary operators.
\newblock {\em Journal of Optics B: Quantum and Semiclassical Optics},
  7(10):S347, 2005.

\bibitem[EFH{\etalchar{+}}23]{elben2023randomized}
Andreas Elben, Steven~T Flammia, Hsin-Yuan Huang, Richard Kueng, John Preskill,
  Beno{\^\i}t Vermersch, and Peter Zoller.
\newblock The randomized measurement toolbox.
\newblock {\em Nature Reviews Physics}, 5(1):9--24, 2023.

\bibitem[EN00]{engel2000one}
Klaus-Jochen Engel and Rainer Nagel.
\newblock {\em One-parameter semigroups for linear evolution equations}.
\newblock Springer, 2000.

\bibitem[EWS{\etalchar{+}}03]{emerson2003pseudo}
Joseph Emerson, Yaakov~S Weinstein, Marcos Saraceno, Seth Lloyd, and David~G
  Cory.
\newblock Pseudo-random unitary operators for quantum information processing.
\newblock {\em science}, 302(5653):2098--2100, 2003.

\bibitem[Fel17]{feldman2017general}
Vitaly Feldman.
\newblock A general characterization of the statistical query complexity.
\newblock In {\em Conference on learning theory}, pages 785--830. PMLR, 2017.

\bibitem[GAE07]{gross2007evenly}
David Gross, Koenraad Audenaert, and Jens Eisert.
\newblock Evenly distributed unitaries: On the structure of unitary designs.
\newblock {\em Journal of mathematical physics}, 48(5), 2007.

\bibitem[GHH{\etalchar{+}}25]{grevink2025will}
Lorenzo Grevink, Jonas Haferkamp, Markus Heinrich, Jonas Helsen, Marcel
  Hinsche, Thomas Schuster, and Zolt{\'a}n Zimbor{\'a}s.
\newblock Will it glue? on short-depth designs beyond the unitary group.
\newblock {\em arXiv preprint arXiv:2506.23925}, 2025.

\bibitem[GKKT20]{guta2020fast}
M~Gu{\c{t}}{\u{a}}, Jonas Kahn, Richard Kueng, and Joel~A Tropp.
\newblock Fast state tomography with optimal error bounds.
\newblock {\em Journal of Physics A: Mathematical and Theoretical},
  53(20):204001, 2020.

\bibitem[Gol04]{goldreich2004foundations}
Oded Goldreich.
\newblock {\em Foundations of Cryptography, Volume 2}.
\newblock Cambridge university press Cambridge, 2004.

\bibitem[GTFS16]{ganji2016strong}
Fatemeh Ganji, Shahin Tajik, Fabian F{\"a}{\ss}ler, and Jean-Pierre Seifert.
\newblock Strong machine learning attack against pufs with no mathematical
  model.
\newblock In {\em International Conference on Cryptographic Hardware and
  Embedded Systems}, pages 391--411. Springer, 2016.

\bibitem[Haf22]{haferkamp2022random}
Jonas Haferkamp.
\newblock Random quantum circuits are approximate unitary t-designs in depth.
\newblock {\em Quantum}, 6:795, 2022.

\bibitem[HBC{\etalchar{+}}22]{Huang_2022}
Hsin-Yuan Huang, Michael Broughton, Jordan Cotler, Sitan Chen, Jerry Li, Masoud
  Mohseni, Hartmut Neven, Ryan Babbush, Richard Kueng, John Preskill, and
  Jarrod~R. McClean.
\newblock \href{http://dx.doi.org/10.1126/science.abn7293}{Quantum advantage in
  learning from experiments}.
\newblock {\em Science}, 376(6598):1182--1186, 2022.

\bibitem[HE23]{hangleiter2023computational}
Dominik Hangleiter and Jens Eisert.
\newblock Computational advantage of quantum random sampling.
\newblock {\em Reviews of Modern Physics}, 95(3):035001, 2023.

\bibitem[HHJ{\etalchar{+}}15]{HaahHarrowJiWuYu2017SampleOptimal}
Jeongwan Haah, Aram~W. Harrow, Zhengfeng Ji, Xiaodi Wu, and Nengkun Yu.
\newblock Sample-optimal tomography of quantum states.
\newblock 2015,  arXiv: \href{https://arxiv.org/abs/1508.01797}{{\ttfamily
  1508.01797}}.

\bibitem[HIN{\etalchar{+}}23]{hinsche2023one}
Marcel Hinsche, Marios Ioannou, Alexander Nietner, Jonas Haferkamp, Yihui Quek,
  Dominik Hangleiter, J-P Seifert, Jens Eisert, and Ryan Sweke.
\newblock One t gate makes distribution learning hard.
\newblock {\em Physical review letters}, 130(24):240602, 2023.

\bibitem[HKKU20]{hamazaki2020universality}
Ryusuke Hamazaki, Kohei Kawabata, Naoto Kura, and Masahito Ueda.
\newblock Universality classes of non-hermitian random matrices.
\newblock {\em Physical Review Research}, 2(2):023286, 2020.

\bibitem[HKOT23]{haah2023query}
Jeongwan Haah, Robin Kothari, Ryan O’Donnell, and Ewin Tang.
\newblock \href{http://dx.doi.org/10.1109/focs57990.2023.00028}{Query-optimal
  estimation of unitary channels in diamond distance}.
\newblock In {\em 2023 IEEE 64th Annual Symposium on Foundations of Computer
  Science (FOCS)}, pages 363--390. IEEE, 2023.

\bibitem[HKP20a]{huang2020predicting}
Hsin-Yuan Huang, Richard Kueng, and John Preskill.
\newblock Predicting many properties of a quantum system from very few
  measurements.
\newblock {\em Nature Physics}, 16(10):1050--1057, 2020.

\bibitem[HKP20b]{classical_shadow_tomography}
Hsin-Yuan Huang, Richard Kueng, and John Preskill.
\newblock \href{http://dx.doi.org/10.1038/s41567-020-0932-7}{Predicting many
  properties of a quantum system from very few measurements}.
\newblock {\em Nature Physics}, 16(10):1050--1057, jun 2020.

\bibitem[HKP21]{Huang_2021}
Hsin-Yuan Huang, Richard Kueng, and John Preskill.
\newblock
  \href{http://dx.doi.org/10.1103/PhysRevLett.126.190505}{Information-theoretic
  bounds on quantum advantage in machine learning}.
\newblock {\em Physical Review Letters}, 126(19):190505, 2021.

\bibitem[HP07]{hayden2007black}
Patrick Hayden and John Preskill.
\newblock Black holes as mirrors: quantum information in random subsystems.
\newblock {\em JHEP}, 2007(09):120, 2007.

\bibitem[HTFS23]{Huang_2023}
Hsin-Yuan Huang, Yu~Tong, Di~Fang, and Yuan Su.
\newblock \href{http://dx.doi.org/10.1103/PhysRevLett.130.200403}{Learning
  many-body hamiltonians with heisenberg-limited scaling}.
\newblock {\em Physical Review Letters}, 130(20):200403, 2023.

\bibitem[JLS18]{ji2018pseudorandom}
Zhengfeng Ji, Yi-Kai Liu, and Fang Song.
\newblock Pseudorandom quantum states.
\newblock In {\em Advances in Cryptology--CRYPTO 2018: 38th Annual
  International Cryptology Conference, Santa Barbara, CA, USA, August 19--23,
  2018, Proceedings, Part III 38}, pages 126--152. Springer, 2018.

\bibitem[Kam23]{kamenev2023field}
Alex Kamenev.
\newblock {\em Field theory of non-equilibrium systems}.
\newblock Cambridge University Press, 2023.

\bibitem[Kat13]{kato2013perturbation}
Tosio Kato.
\newblock {\em Perturbation theory for linear operators}, volume 132.
\newblock Springer Science \& Business Media, 2013.

\bibitem[KBD{\etalchar{+}}08]{kraus2008preparation}
Barbara Kraus, Hans~P B{\"u}chler, Sebastian Diehl, Adrian Kantian, Andrea
  Micheli, and Peter Zoller.
\newblock Preparation of entangled states by quantum markov processes.
\newblock {\em Physical Review A—Atomic, Molecular, and Optical Physics},
  78(4):042307, 2008.

\bibitem[Kea98]{SQLearning}
Michael Kearns.
\newblock \href{http://dx.doi.org/10.1145/293347.293351}{Efficient
  noise-tolerant learning from statistical queries}.
\newblock {\em Journal of the ACM (JACM)}, 45(6):983--1006, 1998.

\bibitem[KG19]{khalafalla2019pufs}
Mahmoud Khalafalla and Catherine Gebotys.
\newblock Pufs deep attacks: Enhanced modeling attacks using deep learning
  techniques to break the security of double arbiter pufs.
\newblock In {\em 2019 Design, automation \& test in Europe conference \&
  exhibition (DATE)}, pages 204--209. IEEE, 2019.

\bibitem[KHM19]{khaymovich2019eigenstate}
Ivan~M Khaymovich, Masudul Haque, and Paul~A McClarty.
\newblock Eigenstate thermalization, random matrix theory, and behemoths.
\newblock {\em Physical review letters}, 122(7):070601, 2019.

\bibitem[KLR{\etalchar{+}}08]{knill2008randomized}
Emanuel Knill, Dietrich Leibfried, Rolf Reichle, Joe Britton, R~Brad Blakestad,
  John~D Jost, Chris Langer, Roee Ozeri, Signe Seidelin, and David~J Wineland.
\newblock Randomized benchmarking of quantum gates.
\newblock {\em Physical Review A}, 77(1):012307, 2008.

\bibitem[KTP20]{kim2020ghost}
Isaac Kim, Eugene Tang, and John Preskill.
\newblock The ghost in the radiation: Robust encodings of the black hole
  interior.
\newblock {\em Journal of High Energy Physics}, 2020(6):1--65, 2020.

\bibitem[Kub57]{kubo1957statistical}
Ryogo Kubo.
\newblock Statistical-mechanical theory of irreversible processes. i. general
  theory and simple applications to magnetic and conduction problems.
\newblock {\em Journal of the physical society of Japan}, 12(6):570--586, 1957.

\bibitem[LDNT25]{lami2025anticoncentration}
Guglielmo Lami, Jacopo De~Nardis, and Xhek Turkeshi.
\newblock Anticoncentration and state design of random tensor networks.
\newblock {\em Physical Review Letters}, 134(1):010401, 2025.

\bibitem[Led01]{ledoux2001concentration}
Michel Ledoux.
\newblock {\em The concentration of measure phenomenon}.
\newblock Number~89. American Mathematical Soc., 2001.

\bibitem[LL24]{laracuente2024approximate}
Nicholas LaRacuente and Felix Leditzky.
\newblock Approximate unitary $ k $-designs from shallow, low-communication
  circuits.
\newblock {\em arXiv preprint arXiv:2407.07876}, 2024.

\bibitem[LP61]{lumer1961dissipative}
G{\"u}nter Lumer and Ralph~S Phillips.
\newblock Dissipative operators in a banach space.
\newblock 1961.

\bibitem[Meh04]{mehta2004random}
Madan~Lal Mehta.
\newblock {\em Random matrices}, volume 142.
\newblock Elsevier, 2004.

\bibitem[Mel24]{mele2024introduction}
Antonio~Anna Mele.
\newblock Introduction to haar measure tools in quantum information: A
  beginner's tutorial.
\newblock {\em Quantum}, 8:1340, 2024.

\bibitem[MH25]{ma2025construct}
Fermi Ma and Hsin-Yuan Huang.
\newblock How to construct random unitaries.
\newblock In {\em Proceedings of the 57th Annual ACM Symposium on Theory of
  Computing}, pages 806--809, 2025.

\bibitem[MHB{\etalchar{+}}04]{muller2004semiclassical}
Sebastian M{\"u}ller, Stefan Heusler, Petr Braun, Fritz Haake, and Alexander
  Altland.
\newblock Semiclassical foundation of universality in quantum chaos.
\newblock {\em Physical review letters}, 93(1):014103, 2004.

\bibitem[Mon17a]{montanaro2017learning}
Ashley Montanaro.
\newblock Learning stabilizer states by bell sampling.
\newblock {\em arXiv:1707.04012}, 2017.

\bibitem[Mon17b]{monthus2017dissipative}
C{\'e}cile Monthus.
\newblock Dissipative random quantum spin chain with boundary-driving and
  bulk-dephasing: magnetization and current statistics in the
  non-equilibrium-steady-state.
\newblock {\em Journal of Statistical Mechanics: Theory and Experiment},
  2017(4):043302, 2017.

\bibitem[MPSY24]{metger2024simple}
Tony Metger, Alexander Poremba, Makrand Sinha, and Henry Yuen.
\newblock Simple constructions of linear-depth t-designs and pseudorandom
  unitaries.
\newblock {\em arXiv preprint arXiv:2404.12647}, 2024.

\bibitem[MRL08]{mohseni2008quantum}
Masoud Mohseni, Ali~T Rezakhani, and Daniel~A Lidar.
\newblock Quantum-process tomography: Resource analysis of different
  strategies.
\newblock {\em Physical Review A}, 77(3):032322, 2008.

\bibitem[MVM{\etalchar{+}}23]{morvan2023phase}
Alexis Morvan, B~Villalonga, X~Mi, S~Mandra, A~Bengtsson, PV~Klimov, Z~Chen,
  S~Hong, C~Erickson, IK~Drozdov, et~al.
\newblock Phase transition in random circuit sampling.
\newblock {\em arXiv preprint arXiv:2304.11119}, 2023.

\bibitem[NC10]{nielsen2010quantum}
Michael~A Nielsen and Isaac~L Chuang.
\newblock {\em Quantum computation and quantum information}.
\newblock Cambridge university press, 2010.

\bibitem[NH95]{najfeld1995derivatives}
Igor Najfeld and Timothy~F Havel.
\newblock Derivatives of the matrix exponential and their computation.
\newblock {\em Advances in applied mathematics}, 16(3):321--375, 1995.

\bibitem[NHKW16]{nakata2016efficient}
Yoshifumi Nakata, Christoph Hirche, Masato Koashi, and Andreas Winter.
\newblock Efficient unitary designs with nearly time-independent hamiltonian
  dynamics.
\newblock {\em arXiv preprint arXiv:1609.07021}, 2016.

\bibitem[Nie23]{nietner2023unifying}
Alexander Nietner.
\newblock \href{https://arxiv.org/abs/2310.17716}{Unifying (quantum)
  statistical and parametrized (quantum) algorithms}.
\newblock {\em arXiv preprint arXiv:2310.17716}, 2023.

\bibitem[NIS{\etalchar{+}}23]{nietner2023average}
Alexander Nietner, Marios Ioannou, Ryan Sweke, Richard Kueng, Jens Eisert,
  Marcel Hinsche, and Jonas Haferkamp.
\newblock On the average-case complexity of learning output distributions of
  quantum circuits.
\newblock {\em arXiv preprint arXiv:2305.05765}, 2023.

\bibitem[NIS{\etalchar{+}}25]{nietner2025average}
Alexander Nietner, Marios Ioannou, Ryan Sweke, Richard Kueng, Jens Eisert,
  Marcel Hinsche, and Jonas Haferkamp.
\newblock On the average-case complexity of learning output distributions of
  quantum circuits.
\newblock {\em Quantum}, 9:1883, 2025.

\bibitem[NO14]{nakata2014thermal}
Yoshifumi Nakata and Tobias~J Osborne.
\newblock Thermal states of random quantum many-body systems.
\newblock {\em Physical Review A}, 90(5):050304, 2014.

\bibitem[O'D14]{o2014analysis}
Ryan O'Donnell.
\newblock {\em Analysis of boolean functions}.
\newblock Cambridge University Press, 2014.

\bibitem[PASA{\etalchar{+}}21]{phalak2021quantum}
Koustubh Phalak, Abdullah Ash-Saki, Mahabubul Alam, Rasit~Onur Topaloglu, and
  Swaroop Ghosh.
\newblock Quantum puf for security and trust in quantum computing.
\newblock {\em IEEE Journal on Emerging and Selected Topics in Circuits and
  Systems}, 11(2):333--342, 2021.

\bibitem[PPS22]{pirnay2022learning}
Niklas Pirnay, Anna Pappa, and Jean-Pierre Seifert.
\newblock Learning classical readout quantum pufs based on single-qubit gates.
\newblock {\em Quantum Machine Intelligence}, 4(2):14, 2022.

\bibitem[PRTG02]{pappu2002physical}
Ravikanth Pappu, Ben Recht, Jason Taylor, and Neil Gershenfeld.
\newblock Physical one-way functions.
\newblock {\em Science}, 297(5589):2026--2030, 2002.

\bibitem[RDO08]{rigol2008thermalization}
Marcos Rigol, Vanja Dunjko, and Maxim Olshanii.
\newblock Thermalization and its mechanism for generic isolated quantum
  systems.
\newblock {\em Nature}, 452(7189):854--858, 2008.

\bibitem[RH12]{rivas2012open}
Angel Rivas and Susana~F Huelga.
\newblock {\em Open quantum systems}, volume~10.
\newblock Springer, 2012.

\bibitem[Rot06]{roth2006introduction}
Ron Roth.
\newblock {\em Introduction to coding theory}.
\newblock Cambridge University Press, 2006.

\bibitem[RSS{\etalchar{+}}10]{ruhrmair2010modeling}
Ulrich R{\"u}hrmair, Frank Sehnke, Jan S{\"o}lter, Gideon Dror, Srinivas
  Devadas, and J{\"u}rgen Schmidhuber.
\newblock Modeling attacks on physical unclonable functions.
\newblock In {\em Proceedings of the 17th ACM conference on Computer and
  communications security}, pages 237--249, 2010.

\bibitem[SHH25]{schuster2024random}
Thomas Schuster, Jonas Haferkamp, and Hsin-Yuan Huang.
\newblock Random unitaries in extremely low depth.
\newblock {\em Science}, 389(6755):92--96, 2025.

\bibitem[SHU21]{sugimoto2021test}
Shoki Sugimoto, Ryusuke Hamazaki, and Masahito Ueda.
\newblock Test of the eigenstate thermalization hypothesis based on local
  random matrix theory.
\newblock {\em Physical Review Letters}, 126(12):120602, 2021.

\bibitem[SKYH25]{schuster2025hardness}
Thomas Schuster, Dominik Kufel, Norman~Y Yao, and Hsin-Yuan Huang.
\newblock Hardness of recognizing phases of matter.
\newblock {\em Forthcoming}, 2025.

\bibitem[SML{\etalchar{+}}25]{schuster2025strong}
Thomas Schuster, Fermi Ma, Alex Lombardi, Fernando Brand\~{a}o, and Hsin-Yuan
  Huang.
\newblock Strong random unitaries and fast scrambling.
\newblock {\em arXiv preprint arXiv:2509.26310}, 2025.

\bibitem[TCD{\.Z}23]{tarnowski2023random}
Wojciech Tarnowski, Dariusz Chru{\'s}ci{\'n}ski, Sergey Denisov, and Karol
  {\.Z}yczkowski.
\newblock Random lindblad operators obeying detailed balance.
\newblock {\em Open Systems \& Information Dynamics}, 30(02):2350007, 2023.

\bibitem[TPI19]{tebelmann2019side}
Lars Tebelmann, Michael Pehl, and Vincent Immler.
\newblock Side-channel analysis of the tero puf.
\newblock In {\em International Workshop on Constructive Side-Channel Analysis
  and Secure Design}, pages 43--60. Springer, 2019.

\bibitem[Val84]{valiant1984theory}
Leslie~G Valiant.
\newblock A theory of the learnable.
\newblock {\em Commun. ACM}, 27(11):1134--1142, 1984.

\bibitem[VWIC09]{verstraete2009quantum}
Frank Verstraete, Michael~M Wolf, and J~Ignacio~Cirac.
\newblock Quantum computation and quantum-state engineering driven by
  dissipation.
\newblock {\em Nature physics}, 5(9):633--636, 2009.

\bibitem[WAG{\etalchar{+}}25]{wallace2025learning}
Stewart Wallace, Yoann Altmann, Brian~D Gerardot, Erik~M Gauger, and Cristian
  Bonato.
\newblock Learning the dynamics of markovian open quantum systems from
  experimental data.
\newblock {\em Physical Review Research}, 7(3):033217, 2025.

\bibitem[Wat04]{watrous2004notes}
John Watrous.
\newblock Notes on super-operator norms induced by schatten norms.
\newblock {\em arXiv preprint quant-ph/0411077}, 2004.

\bibitem[Wat18]{watrous2018theory}
John Watrous.
\newblock {\em The theory of quantum information}.
\newblock Cambridge university press, 2018.

\bibitem[WCP13]{wang2013projection}
Weiran Wang and Miguel~A Carreira-Perpinn.
\newblock Projection onto the probability simplex: An efficient algorithm with
  a simple proof, and an application.
\newblock {\em arXiv preprint arXiv:1309.1541}, 2013.

\bibitem[WD24]{wadhwa2024noise}
Chirag Wadhwa and Mina Doosti.
\newblock Noise-tolerant learnability of shallow quantum circuits from
  statistics and the cost of quantum pseudorandomness.
\newblock {\em arXiv preprint arXiv:2405.12085}, 2024.

\bibitem[Wig67]{wigner1967random}
Eugene~P Wigner.
\newblock Random matrices in physics.
\newblock {\em SIAM review}, 9(1):1--23, 1967.

\bibitem[WKPT22]{wu2022erasure}
Yue Wu, Shimon Kolkowitz, Shruti Puri, and Jeff~D Thompson.
\newblock Erasure conversion for fault-tolerant quantum computing in alkaline
  earth rydberg atom arrays.
\newblock {\em Nature communications}, 13(1):4657, 2022.

\bibitem[WL24]{wang2024simulation}
Ke~Wang and Xiantao Li.
\newblock Simulation-assisted learning of open quantum systems.
\newblock {\em Quantum}, 8:1407, 2024.

\bibitem[WPS{\etalchar{+}}17]{Wang_2017}
Jianwei Wang, Stefano Paesani, Raffaele Santagati, Sebastian Knauer, Antonio~A.
  Gentile, Nathan Wiebe, Maurangelo Petruzzella, Jeremy~L. {O}'Brien, John~G.
  Rarity, Anthony Laing, and Mark~G. Thompson.
\newblock \href{http://dx.doi.org/10.1038/nphys4074}{Experimental quantum
  hamiltonian learning}.
\newblock {\em Nature Physics}, 13(6):551--555, 2017.

\bibitem[YE25]{yang2025complexity}
Lisa Yang and Netta Engelhardt.
\newblock The complexity of learning (pseudo) random dynamics of black holes
  and other chaotic systems.
\newblock {\em Journal of High Energy Physics}, 2025(3):1--65, 2025.

\bibitem[Yos12]{yosida2012functional}
K{\^o}saku Yosida.
\newblock {\em Functional analysis}, volume 123.
\newblock Springer Science \& Business Media, 2012.

\bibitem[ZRM21]{zhao2021fermionic}
Andrew Zhao, Nicholas~C Rubin, and Akimasa Miyake.
\newblock Fermionic partial tomography via classical shadows.
\newblock {\em Physical Review Letters}, 127(11):110504, 2021.

\end{thebibliography}

\end{document}